%% file: thesis.tex
\documentclass[
a4paper, 
11pt, 
onecolumn, 
openright, 
oldfontcommands,
]{memoir}

\input{preamble.tex}

\author{J. R. Ipsen}
\title{Products of random matrices}

\begin{document}

\frontmatter*

\firmlists

\raggedbottom

\input{cover-etc.tex}

\addtocounter{page}{2}

\input{preface.tex}

\tableofcontents

\mainmatter*

\input{motivation.tex}

\input{prologue.tex}

\input{singular.tex}

\input{complex.tex}

\input{lyapunov.tex}



\appendix

\input{decomposition}

\input{special}



\raggedright

\bibliographystyle{amsplain}
\bibliography{ref}


\end{document}

%% file: preamble.tex
\usepackage[utf8]{inputenc} 
\usepackage[T1]{fontenc}    %
\usepackage[english]{babel} 
\usepackage[final]{microtype} 

\usepackage{amsmath,amssymb,mathtools} 

\usepackage{paralist} 

\usepackage{lipsum}

\setlrmarginsandblock{0.15\paperwidth}{*}{1} 
\setulmarginsandblock{0.2\paperwidth}{*}{1}  
\checkandfixthelayout

\usepackage{bbm} 

\maxsecnumdepth{subsubsection} 
\setsecnumdepth{subsubsection}

\makeatletter %
\makechapterstyle{ipsen}{
  \setlength{\beforechapskip}{5\baselineskip}
  \setlength{\midchapskip}{0\baselineskip}
  \setlength{\afterchapskip}{10\baselineskip}
  
  \renewcommand{\chapnamefont}{\large\normalfont}
  \renewcommand{\printchaptername}{\large\chapnamefont \@chapapp}

}
\makeatother
\chapterstyle{ipsen}

\setsecheadstyle{\large\bfseries}
\setsubsecheadstyle{\bfseries}
\setsubsubsecheadstyle{\bfseries}
\setparaheadstyle{\itshape}
\setparaindent{0pt}\setafterparaskip{0pt}

\usepackage{tikz} 
\usepackage{graphicx} 

\captiondelim{\space\space } 
\captionnamefont{\small\bfseries} 
\captiontitlefont{\small\normalfont} 
\changecaptionwidth          
\captionwidth{\textwidth} %

\maxtocdepth{subsection} 
\settocdepth{subsection}

\newleadpage{appendices}{\appendixpagename}

\renewcommand{\partnumberline}[1]{}
\cftpagenumbersoff{part} 

\makepagestyle{standard} 

\makeatletter                 
\makeevenfoot{standard}{}{\thepage}{} %
\makeoddfoot{standard}{}{\thepage}{}  %
\makeevenhead{standard}{}{}{}
\makeoddhead{standard}{}{}{}
\makeatother                  %

\makeatletter
\makepsmarks{standard}{
\createmark{chapter}{both}{shownumber}{}{ \  }
\createmark{section}{right}{shownumber}{}{ \  }
\createplainmark{toc}{both}{\contentsname}
\createplainmark{lof}{both}{\listfigurename}
\createplainmark{lot}{both}{\listtablename}
\createplainmark{bib}{both}{\bibname}
\createplainmark{index}{both}{\indexname}
\createplainmark{glossary}{both}{\glossaryname}
}
\makeatother                               %

\makepagestyle{chap} 

\makeatletter
\makeevenfoot{chap}{}{\thepage}{} 
\makeoddfoot{chap}{}{\thepage}{}  %
\makeevenhead{chap}{}{}{}   %
\makeoddhead{chap}{}{}{}    %
\makeatother

\nouppercaseheads
\aliaspagestyle{chapter}{chap} %
\aliaspagestyle{part}{empty}

\usepackage{amsthm}

\theoremstyle{plain}
  \newtheorem{theorem}{Theorem}[chapter]
  \newtheorem{proposition}[theorem]{Proposition}
  \newtheorem{lemma}[theorem]{Lemma}
  \newtheorem{corollary}[theorem]{Corollary}
\theoremstyle{definition}
  \newtheorem{definition}[theorem]{Definition}
  \newtheorem{remark}[theorem]{Remark}
\theoremstyle{remark}

\epigraphtextposition{flushleftright}
\newcommand{\N}{\mathbb{N}} \newcommand{\Z}{\mathbb{Z}} 
\newcommand{\R}{\mathbb{R}} \newcommand{\C}{\mathbb{C}} \renewcommand{\H}{\mathbb{H}}
\newcommand{\F}{\mathbb{F}} \newcommand{\E}{\mathbb{E}} \renewcommand{\P}{\mathbb{P}} 

 \newcommand{\cM}{\mathcal M}

 \newcommand{\cP}{\mathcal P} 
\newcommand{\cH}{\mathcal H}  \newcommand{\cZ}{\mathcal Z}
\newcommand{\cI}{\mathcal I}  

\newcommand{\gO}{\textup{O}}     \newcommand{\gU}{\textup{U}}     \newcommand{\gUSp}{\textup{USp}}
      
\newcommand{\gSO}{\textup{SO}}   
 
\newcommand{\gGL}{\textup{GL}}      \newcommand{\gSp}{\textup{Sp}}

\newcommand{\p}{\partial} 
\newcommand{\one}{\mathbbm{1}} 
\renewcommand{\imath}{\mathrm{i}} 
\newcommand{\transpose}{\mathsf{T}} 
\renewcommand{\phi}{\varphi} 
\renewcommand{\epsilon}{\varepsilon} 
\renewcommand{\Re}{\operatorname{Re}} 
\renewcommand{\Im}{\operatorname{Im}} 
\newcommand{\Ai}{\operatorname{Ai}} 

\newcommand{\qi}{\mathbbm{i}} 
\newcommand{\qj}{\mathbbm{j}}
\newcommand{\qk}{\mathbbm{k}}

\DeclareMathOperator{\tr}{Tr} 
\DeclareMathOperator*{\pf}{Pf} 
\DeclareMathOperator{\diag}{diag} 
\DeclareMathOperator{\sign}{sign} 
\DeclareMathOperator*{\Span}{span} 
\DeclareMathOperator{\erfc}{erfc} 
\DeclareMathOperator{\erf}{erf} 
\DeclareMathOperator*{\per}{per} 
\DeclareMathOperator{\vol}{vol} 
\DeclareMathOperator{\var}{var} 
\DeclareMathOperator{\cov}{cov} 

\let\vec\relax
\DeclareMathOperator{\vec}{vec}

\newcommand{\deq}{\stackrel{d}{=}}

\DeclarePairedDelimiter{\abs}{\lvert}{\rvert} 
\DeclarePairedDelimiter{\norm}{\lVert}{\rVert} 
\DeclarePairedDelimiter{\inner}{\langle}{\rangle} 

\DeclarePairedDelimiter{\curly}{[}{]}
\newcommand{\average}{\E\curly}
\DeclarePairedDelimiter{\naverage}{\langle}{\rangle} 

\newcommand{\nn}{\nonumber}

\newcommand{\jpdf}{\textup{jpdf}}

\newcommand{\gauss}{\textup{G}}

\newcommand{\fuss}{\textup{FC}}
\newcommand{\macro}{\textup{macro}}
\newcommand{\bulk}{\textup{bulk}}
\newcommand{\edge}{\textup{edge}}
\newcommand{\origin}{\textup{origin}}
\newcommand{\meijer}{\textup{Meijer}}
\newcommand{\bessel}{\textup{Bessel}}
\newcommand{\sine}{\textup{sine}}
\newcommand{\airy}{\textup{Airy}}

\newcommand{\MP}{\textup{MP}}
\newcommand{\hard}{\textup{hard}}
\newcommand{\soft}{\textup{soft}}
\newcommand{\conn}{\textup{conn}}
\newcommand{\ev}{\textup{e.v.}}

\newcommand{\MeijerG}[8][\Big]{G^{{ #2 },{ #3 }}_{{ #4 },{ #5 }} #1( \begin{matrix} #6 \\ #7 \end{matrix}\, #1\vert\, #8 #1)}
\newcommand{\hypergeometric}[6][\Big]{\,{}_{#2} F_{#3} #1( \begin{matrix} #4 \\ #5 \end{matrix}\, #1\vert\, #6 #1)}

\newcommand{\FoxH}[8][\Big]{H^{{ #2 },{ #3 }}_{{ #4 },{ #5 }} #1( \begin{matrix} #6 \\ #7 \end{matrix}\, #1\vert\, #8 #1)}

\usepackage{hyperref}   
\hypersetup{
colorlinks=true,        
linktoc=page,           
linkcolor=blue,         
citecolor=blue,         
pdfauthor={J. R. Ipsen} 
}
\usepackage{memhfixc}   %

%% file: cover-etc.tex
\begin{titlingpage}
\thispagestyle{empty}

\raggedright

\vspace*{.2\textheight}

{\Large J.~R.~Ipsen}

\vspace*{\baselineskip}

{\Huge\bfseries Products of Independent}\\

~\hfill {\Huge\bfseries Gaussian Random Matrices}

\vfill

Doctoral dissertation\\[.5\baselineskip]
Department of Physics\\
Bielefeld University

\newpage

%


~

\vfill

%

\raggedright

Typeset in 11\,pt. Computer Modern using the \textsf{\itshape memoir} class.\\
All figures are made with Ti\textit{k}Z and \textsc{pgfplots}.

\medskip 

Printed in Germany on acid-free paper.

%
%

\end{titlingpage}

%
%
%
%
%
%
%
%
%
%
%

%% file: preface.tex
\chapter{Preface}

The study of products of random matrices dates back to the early days of random matrix theory. Pioneering work by Bellman, Furstenberg, Kesten, Oseledec and others were contemporary to early contributions by Wigner, Dyson, Mehta and others regarding the spectrum of a single large matrix. It is not unreasonable to divide these early results into two different schools separated both by the questions asked and the techniques used. One school focused on the Lyapunov spectrum of products of finite size matrices as the number of factors tended to infinity, while the other focused on eigenvalues of a single matrix as the matrix dimension tended to infinity. 

From a physical point of view a restriction to Hermitian matrices is often natural when considering a single random matrix, since the random matrix typically is imagined to approximate the properties of a Hamiltonian or another self-adjoint operator. On the other hand, a restriction to Hermitian matrices is no longer natural when considering products. This is illustrated by the fact that a product of two Hermitian matrices is, in general, non-Hermitian. 

When considering products it is more natural to study random matrices chosen according to a probability measure on some matrix semi-group. Historically, one of the first examples was provided by considering a product of random matrices with positive entries~\cite{Bellman:1954}; the relevance of such models in physics may be realised by considering the transfer matrix representation of one-dimensional lattice models with random coupling between spins (see section~\ref{sec:moti:lattice}). As another example we could consider products of random unitary matrices describing a unitary time evolution~\cite{JW:2004} or a random Wilson loop~\cite{NR:2007,BN:2008}. We emphasise that choosing unitary matrices uniformly with respect to the Haar measure constitutes a trivial example since this corresponds to studying the evolution of a system starting in the equilibrium state. Thus, the circular unitary ensemble is rather boring when considering products. Moreover, the circular orthogonal and symplectic 
ensembles do not even qualify as semi-groups if the ordinary matrix product is used.

The semi-groups which will be important in this thesis are the space of all $N\times N$ matrices over the (skew-)field of real numbers, complex numbers and quaternions endowed with usual matrix multiplication; the threefold classification in accordance with the associative division algebras corresponds to Dyson's classification of the classical Wigner--Dyson ensembles~\cite{Dyson:1962}. An important difference between these matrix spaces and the unitary group from the previous example is that they are non-compact, thus, a priori, there is no natural equilibrium measure.

Historically, the research on products of random matrices was centred around the Lyapunov spectrum and, in particular, the largest Lyapunov exponent, which in physical models may be related to e.g. the stability of dynamical systems or the free energy of disordered lattice systems, see~\cite{CPV:1993} for a review of applications. A ``law of large numbers'' for the largest Lyapunov exponent as the number of factors tends to infinity was established early on by Furstenberg and Kesten~\cite{FK:1960} leading up to Oseledec's celebrated multiplicative ergodic theorem~\cite{Oseledec:1968,Raghunathan:1979}. However, universal laws for the fluctuations of the Lyapunov exponents are more challenging. Nonetheless, for certain classes of matrices a central limit theorem has been established for the largest Lyapunov exponent, see e.g.~\cite{CN:1984,lePage:1982}. The fact that the largest Lyapunov exponent follows a Gaussian law is rather remarkable when we compare this with our knowledge about a single random matrix. Under 
quite 
general conditions the largest singular value of a large random matrix will follow the so-called Tracy--Widom law~\cite{TW:1994}; this is expected to extend to products of independent random matrices as long as the number of factors is finite (this has been shown explicitly for products of Gaussian random matrices~\cite{LWZ:2014}). Thus, when considering products of random matrices, we are led to believe that it has fundamental importance for the microscopic spectral properties whether we first take the matrix dimensions to infinity and then the number of factors \emph{or} we first take number factors to infinity and then the matrix dimensions. Double scaling limits are undoubtedly a subtle matter. 

The more recent interest in products of random matrices (and, more generally, the algebra of random matrices) is partly due to progress in free probability, see~\cite{Burda:2013} for a short review. However, a limitation of the techniques from free probability and related methods in random matrix theory is that they only consider macroscopic spectra. It is highly desirable to extend these known results to include higher point correlations as well as microscopic spectral properties. The reasons for this is not only because such quantities are expected to be universal and are relevant for applications, but also because we are interested in the connection to older results about Lyapunov exponents in the limit where the number of factors tends to infinity.

Considerable progress on the microscopic spectral properties of finite products of random matrices has appeared very recently with the introduction of matrix models which are exactly solvable for an arbitrary number of factors as well as arbitrary matrix dimensions. The first of such models considered the eigenvalues of a product of independent square complex Gaussian random matrices~\cite{AB:2012}; this was later extended to include rectangular and quaternionic matrices~\cite{Ipsen:2013,IK:2014,ARRS:2013} and to some extent real matrices~\cite{Forrester:2014a,IK:2014}; explicit expressions for the singular values of the complex Gaussian matrix model were obtained in~\cite{AKW:2013,AIK:2013}. Subsequently, treatments of models involving products of inverse Gaussian matrices and truncated unitary matrices have followed~\cite{ARRS:2013,IK:2014,ABKN:2014,Forrester:2014b,KKS:2015}, see~\cite{AI:2015} for a review. These new integrable models reveal determinantal and Pfaffian structures much like the classical 
matrix ensembles. With the long history of research on products of random matrices and with strong traditions for exactly solvable models (including multi-matrix models) in random matrix theory, it is rather surprising that none of these models have been found earlier.

Obviously, the detailed knowledge of all eigen- and singular value correlation functions for arbitrary matrix dimensions and an arbitrary number of factors has opened up the possibility to study microscopic spectral properties, and the search for known and new universality classes.
For a finite number of matrices, new universality classes have been observed near the origin~\cite{AB:2012,Ipsen:2013,ABKN:2014,KZ:2014,Forrester:2014b,KS:2014,KKS:2015} while familiar random matrix kernels have been reobtained in the bulk and near ``free'' edges~\cite{AB:2012,Ipsen:2013,ABKN:2014,LW:2014,LWZ:2014}. The claim of universality of the new classes near the origin is justified, since several exactly solvable models (also beyond products of random matrices) have the same correlation kernels after proper rescaling. More general universality criteria are highly expected but still unproven. However, it would be a mistake to think of the new exactly solvable matrix models merely as precursors for universality theorems in random matrix theory. There are good reasons (physical as well as mathematical) for giving a prominent r\^ole to the integrable and, in particular, the Gaussian matrix models. Let us emphasise one of these: Gaussian integrations appear as an integral part of the Hubbard--Stratonovich 
transformation which is one way to establish a link between random matrix models and dual non-linear sigma models which appear as effective field theories in condensed matter theory and high energy physics, see e.g.~\cite{VW:2000,Beenakker:1997}.

The new exactly solvable models have also provided new insight to the limit where the number of factors tends to infinity~\cite{Forrester:2014a,ABK:2014,Ipsen:2015,Forrester:2015}. If the matrix dimensions are kept fixed, then it was shown that the eigen- and singular values separate exponentially compared to the interaction range. As a consequence, the determinantal and Pfaffian point processes valid for a finite number of matrices turn into permanental (or symmetrised) processes. Moreover, the Stability and Lyapunov exponents were shown to be Gaussian distributed. A surprising property presents itself when considering products of real Gaussian matrices: the eigenvalue spectrum becomes real (albeit the matrix is asymmetric) when the number of factors tends to infinity. This was first observed numerically~\cite{Lakshminarayan:2013} and was shown analytically for square Gaussian matrices~\cite{Forrester:2014a} while an alternative proof including the possibility of rectangular 
matrices was presented in~\cite{Ipsen:2015}. Numerical evidence suggests that this phenomenon extends to a much wider class of matrices~\cite{HJL:2015}. The fact that the spectrum becomes real is remarkable since when we consider finite product matrices then (under certain assumptions) the macroscopic eigenvalue spectrum becomes rotational symmetric in the complex plane in the limit of large matrix dimension. Again this shows that the two scaling limits do not commute and suggests that interesting behaviour may appear in a double scaling limit.

\section*{Outline of thesis}
\label{sec:moti:outline}

This thesis reviews recent progress on products of random matrices from the perspective of exactly solved Gaussian random matrix models. Our reason for taking the viewpoint of the Gaussian matrices is twofold. Firstly, the Gaussian models have a special status since they are both unitary invariant and have independent entries which are properties related to two typical generalisations within random matrix theory. Secondly, we believe that the Gaussian models are a good representative for the other models which are now known to be exactly solvable, since many techniques observed in the Gaussian case reappear in the description of products involving inverse and truncated unitary matrices. 

For obvious reasons, our main attention must be directed towards results published in papers where the author of this thesis is either the author or a co-author~\cite{AI:2015,AIK:2013,AIS:2014,Ipsen:2013,Ipsen:2015,IK:2014}.
However, not all results presented in this thesis can be found in these papers neither will all results from the aforementioned papers be repeated in this thesis. Proper citation will always be given, both for results originally obtained by other authors and for results from the aforementioned publications. There are several reasons for our deviation from a one-to-one correspondence between thesis and publications. Firstly (and most important), the study of product of random matrices has experienced considerable progress over the last couple of years due to the work of many authors; this thesis would be embarrassingly incomplete if we failed to mention these results. Secondly, we have tried to fill some minor gaps between known results. In particular, we have attempted to generalise to the rectangular matrices whenever these results were not given in the literature. Lastly, certain results deviating from our main theme have been left out in an attempt to keep a consistent tread throughout the thesis and to 
make the presentation as concise (and short) as possible.

The rest of thesis is devived into four main parts. (i) The two first chapters contain introductory material regarding applications of products of random matrices and the step from products of random scalars to products of random matrices; a few general concepts which are essential for the following chapters are also introduced. (ii) The next two chapters derive explicit results for products of \emph{Gaussian} random matrices and consider the asymptotic behaviour for large matrix dimensions. (iii) Chapter~\ref{chap:lyapunov} revisits the matrix models from chapter~\ref{chap:singular} and~\ref{chap:complex}, but focuses on the limit where the number of factors tends to infinity. (iv) Finally, results regarding matrix decompositions and special functions, which are used consistently throughout the thesis, are collected in two appedices.

\begin{flexlabelled}{itshape}{1.75em}{.5em}{0pt}{2.25em}{0pt}

 \item[Chapter~\ref{chap:motivation}.] We ask ``\emph{Why products of random matrices?}'' and discuss a number of application of products of random matrices in physics and beyond. Readers only interested in mathematical results may skip this chapter.
 
 \item[Chapter~\ref{chap:prologue}.] We first recall some well-known results for products of scalar-valued random variables, which will be helpful to keep in mind when considering products of matrices. Thereafter, we turn our attention towards products of random matrices and provide proofs of a weak commutation relation for so-called isotropic random matrices as well as a reduction formula for rectangular random matrices. Even though these results are not used explicitly in the proceeding chapters, they are used implicitly to provide an interpretation for general \emph{rectangular} products. Finally, we introduce definitions for the Gaussian ensembles which will be the main focus for the rest of the thesis. 
 
 \item[Chapter~\ref{chap:singular}.] The squared singular values for a product of complex Gaussian random matrices are considered and explicit formulae are obtained for the joint probability density function and the correlation functions. These exact formulae are valid for arbitrary matrix dimension as well as an arbitrary number of factors. Furthermore, the formulae are used to study asymptotic behaviour as the matrix dimension tends to infinity. In particular, we find the macroscopic density and the microscopic correlations at the hard edge, while scaling limits for the bulk and the soft edge is stated without proof. The chapter ends with a discussion of open problems.
 
 \item[Chapter~\ref{chap:complex}.] We consider the (generally complex) eigenvalues for products of real, complex, and quaternionic matrices. Explicit formulae for the joint probability density function and the correlations functions are obtained for complex and quaternionic matrices, while partial results are presented for the real case. Asymptotic behaviour for large matrix dimension are derived in the known limits; this includes a new microscopic kernel at the origin. Finally, open problems are discussed. 
 
 \item[Chapter~\ref{chap:lyapunov}.] The formulae for the joint densities for the eigen- and singular values of products of Gaussian matrices are used to obtain the asymptotic behaviour for a large number of factors. Explicit formulae are given for the stability and Lyapunov exponents as well as their fluctuations. Certain aspects of double scaling limits are discussed together with open problems.
 
 \item[Appendix~\ref{app:decompositions}.] Several matrix decompositions are discussed. In particular, we provide proofs for some recent generalised decompositions which play an important r\^ole for products of random matrices (some of these are not explicitly given in the literature).
 
 \item[Appendix~\ref{app:special}.] For easy reference, we summarise known properties for certain higher transcedental functions: the gamma, digamma, hypergeometric, and Meijer $G$-functions. The formulae stated in this appendix are frequently used throughout the thesis.
 
\end{flexlabelled}

\plainbreak{1}

Note that we have provided a summary of results and open problems in the end of each of the three main chapters (chapter~\ref{chap:singular}, \ref{chap:complex}, and~\ref{chap:lyapunov}) rather than collecting it all in a final chapter. The intention is that it should be possible to read each of these three chapters separately.

\section*{Acknowledgements}

The final and certainly the most pleasant duty is, of course, to thank friends, colleagues and collaborators who have helped me throughout my doctoral studies. First and foremost my gratitude goes to my advisor G.~Akemann for sharing his knowledge and experience with me; his guidance has been invaluable. I am also indebted to Z.~Burda, P.~J.~Forrester, A.~B.~J.~Kuijlaars, T.~Neuschel, H. Schomerus, D.~Stivigny, E.~Strahov, K.~\. Zyczkowski and, in particular, M.~Kieburg for fruitful and stimulating discussions on the topic of this thesis. 

Naturally, special thanks are owed to my co-authors: G.~Akemann, M.~Kieburg, and E.~Strahov; without them the results presented in this thesis would have been much different. I would also like to take this opputunity to thank P.~J.~Forrester, A.~B.~J.~Kuijlaars, E.~Strahov and L.~Zhang for sharing copies of unpublished drafts.

I am pleased to be able to thank P.~J.~Forrester, A.~B.~J.~Kuijlaars and G.~Schehr for inviting me to visit their institutions and for giving me the possibility to present my work. I am grateful for their generous hospitality.

D.~Conache and A.~Di~Stefano are thanked for reading through parts of this thesis. Any remaining errors are due to the author.

Doing my doctoral studies I have had the opportunity to talk to many brilliant scientists and mathematicians as well as talented students. Special thanks are owed to M.~Atkin, Ch.~Chalier, B.~Fahs, B.~Garrod, R.~Marino, T.~Nagao, A.~Nock, X.~Peng, G.~Silva, R.~Speicher, K.~Splittorff, A.~Swiech, J.~J.~M.~Verbaarschot, P.~Vivo, P.~Warcho\l, L.~Wei, T.~Wirtz, and Z.~Zheng for discussions on various topics.

Lastly, I would like to thank my friends and colleagues at Bielefeld University. It is my pleasure to thank G.~Bedrosian, P.~Bei\ss ner, V.~Bezborodov, S.~Cheng, M.~Cikovic, D.~Conache, M.~Dieckmann, T.~Fadina, D.~K\"ampfe, M.~Lebid, T.~L\"obbe, K.~von~der~L\"uhe, A.~Reshetenko, J.~Rodriguez, M.~Serti\'c, A.~Di~Stefano, Y.~Sun, M.~Venker, P.~Vidal, P.~Voigt, L.~Wresch, and D.~Zhang for making the university a pleasant place to work and for our shared chinese adventure. H.~Litschewsky, R.~Reischuk, and K.~Zelmer are thanked for administrative support and showing me the way through the jungle of bureaucracy. My thanks extend to scientists and staff at the Chinese Academy of Sciences, Department of Mathematics and Statistics at University of Melbourne, and LPTMS Université Paris-Sud.

The author acknowledge financial support by the German science foundation (DFG) through the International Graduate College \emph{Stochastics and Real World Models} (IRTG~1132) at Bielefeld University.

~

\noindent
\emph{Bielefeld, Germany} \hfill J.~R.~Ipsen\\
\emph{August 2015}

\cleardoublepage

%% file: motivation.tex
\chapter{Why products of random matrices?}
\label{chap:motivation}

\begin{quote}
\itshape
The properties of random matrices and their products form a basic tool,
whose importance cannot be underestimated. They play a role as important
as Fourier transforms for differential equations.
~\hfill\normalfont
---Giorgio Parisi~\cite{CPV:1993}
\end{quote}

~

\noindent
It is broadly accepted that random matrix theory is an essential tool in a great variety of topics in both mathematics and physics (and beyond). Moreover, random matrix theory is an extremely rich research area in its own right. We consider the importance of random matrix theory in the theoretical sciences to be so well-established that it is unnecessary to provide further motivation for the study of random matrices themselves. For this reason, we jump directly to the sub-field considered in this thesis with the question: \emph{Why products of random matrices?} If unsatisfied with this leap, the reader is referred to the vast literature on the subject of random matrix theory; we emphasize the contemporary and extensive handbook~\cite{ABF:2011} in which many applications are discussed.

This chapter is intended to give a more physical motivation for the study of products of random matrices and the intriguing questions arising in this sub-field of random matrix theory. To do so, we will introduce a few possible applications of products of random matrices in the sciences. We emphasise that it is not our intention to present exhaustive technical derivations. Neither do we attempt to give an exhaustive nor extensive list of applications for products of random matrices. Rather, we sketch a few illustrative examples from which we hope it is possible to untangle the threads of the much larger pattern.

The applications considered in this chapter include: 
wireless telecommunication (section~\ref{sec:moti:mimo}), 
disordered spin chains (section~\ref{sec:moti:lattice}),
stability of large complex systems (section~\ref{sec:moti:stab}),
symplectic maps and Hamiltonian mechanics (section~\ref{sec:moti:hamilton}),
quantum transport in disordered wires (section~\ref{sec:moti:dmpk}), and
QCD at non-zero chemical potential (section~\ref{sec:moti:qcd}).

\section{Wireless telecommunication}
\label{sec:moti:mimo}

In this section, we look at an application of products of random matrices stemming from wireless telecommunication, see~\cite{Telatar:1999,Muller:2002,TV:2004}. We will consider a so-called multiple-input multiple-output (MIMO) communication channel. This is a single user system with $M$ transmitting and $N$ receiving antennae. As usual, it is convenient to write the amplitude and phase of our signal as the modulus and phase of a complex number. The most general MIMO communication channel may be written as
\begin{equation}
y=Xx+\eta,
\end{equation}
where $\eta$ is an $N$-dimensional vector representing the background noise, $X$ is an $N\times M$ complex matrix representing the channel, while $x$ and $y$ are $M$- and $N$-dimensional complex vectors which represent the signal at the transmitting and receiving antennae, respectively. The canonical choice for the channel matrix, $X$, is to take its entries as independent complex Gaussian variables, i.e. the phases are uniformly distributed and the moduli are Rayleigh distributed. This is known as a Rayleigh fading environment and it is a reasonable approximation for channels with many scatterers and no line-of-sight.  

The typical question asked by engineers concerns the channels information capacity. One of the most frequently used performance measures is the so-called mutual information which gives an upper bound for the spectral efficiency measured as bit-rate per bandwidth. Assuming that the channel matrix is known to the receiver and that the input-signal consists of independent and identically distributed random variables, then the mutual information is given by (see e.g.~\cite{TV:2004})
\begin{equation}
\cI^N(\gamma)=\frac1N\tr\log_2(1+\gamma X^\dagger X)
\end{equation}
with
\begin{equation}
\gamma=\frac{N\E\norm x^2}{M\E \norm\eta^2}
\end{equation}
denoting the signal-to-noise ratio. This means that the mutual information depends on the squared singular values of the channel matrix. For a Rayleigh fading environment we know that the distribution of squared singular values converges in probability to the so-called Mar\v cenko--Pastur law~\cite{MP:1967} in the limit where $N,M\to\infty$ and $N/M\to\alpha\in(0,\infty)$ (see chapter~\ref{chap:singular}). Consequently, the mutual information converges to
\begin{equation}
\cI(\gamma)=\int_0^\infty dx\,\rho_\MP(x)\log_2(1+\gamma x)
\end{equation}
in the limit with a large number of antennae. Here $\rho_\MP(x)$ denotes the density for the Mar\v cenko--Pastur law.

Let us look at a model introduced in~\cite{Muller:2002}, which is more complicated than the Rayleigh fading environment. We will consider a communication channel consisting of $n$ scattering environments separated by some major obstacles. We could imagine that the transmitter and the receiver were located in the same building but on different floors, such that the floors act as the obstacles. Our signal will not pass through a floor equivalently well everywhere, there will be certain spots of preferred penetration referred to as ``key holes''. Assuming that the $i$-th floor has $N_i$ ``key holes'', our communication channel becomes
\begin{equation}
y=X_n\cdots X_1x+\eta,
\end{equation}
where $X_i$ is an $N_i\times N_{i-1}$ matrix representing the $i$-th scattering environment (e.g. between floor number $i$ and $i-1$) with $N_0=M$ and $N_n=N$. The mutual information is given as before; except for a replacement of $X$ with the product matrix $X_n\cdots X_1$. Thus, we need knowledge about the distribution of the singular values of a product of random matrices in order to determine the mutual information in this model. We will return to this question in chapter~\ref{chap:singular}.

\section{Disordered spin chains}
\label{sec:moti:lattice}

The next application we will look at arises in the study of disordered spin chains. 
Consider a periodic chain with nearest neighbour interaction consisting of $n$ spins, $\{s_i\}\in\{1,\ldots,N\}^n$, described by a Hamiltonian,
\begin{equation}
\cH=-\sum_{i=1}^{n} J_{i}(s_{i},s_{i-1})
\end{equation}
where $J_{i}(s_i,s_{i-1})$ denote the coupling constants at the $i$-th link, i.e. the coupling between the spin at the $i$-th and the $(i-1)$-th site. Using standard techniques (see e.g.~\cite{Yeomans:1992}), the partition function (at temperature $1/\beta$) may be written in terms of a product of transfer matrices,
\begin{equation}
\cZ_n=\tr X_nX_{n-1}\cdots X_1,
\label{moti:lattice:partition-function}
\end{equation}
where the trace stems from the periodic boundary condition and each $X_i$ denotes an $N\times N$ transfer matrix given by
\begin{equation}
X_i=
\begin{bmatrix}
 e^{\beta J_i(1,1)} & \cdots &  e^{\beta J_i(1,N)} \\
 \vdots             &        &  \vdots \\
 e^{\beta  J_i(N,1)} & \cdots &   e^{\beta J_i(N,N)}
\end{bmatrix}.
\end{equation}
Note that the eigenvalues of such matrices may be complex even though the trace is real. However, it is known from the Perron--Frebenius theorem (see e.g.~\cite{HJ:2012}) that there is at least one real eigenvalue. Furthermore, this eigenvalue is strictly larger than the rest of the eigenvalues in absolute value. 

First, let us consider the case where all transfer matrices are identical. We will denote the eigenvalues of the transfer matrix by $\lambda_1,\ldots,\lambda_N$ and, due to the Perron--Frebenius theorem, we may order them as $\lambda_1>\abs{\lambda_2}\geq\cdots\geq\abs{\lambda_N}$. For $N$ fixed, the free energy per site becomes
\begin{equation}
\beta f=-\lim_{n\to\infty}\frac1n\log \tr (X_1)^n=-\lim_{n\to\infty}\frac1n\log(\lambda_1^n+\cdots+\lambda_N^n)=-\log \lambda_1
\label{moti:lattice:free-energy}
\end{equation}
in the thermodynamic limit. We can, of course, also consider other physical quantities, e.g. if $\abs{\lambda_2}>\abs{\lambda_3}$ then the correlation length (in units of the lattice spacing) is given by $\xi=(\log \lambda_1/\abs{\lambda_2})^{-1}$. 

Now, imagine that we want to consider a disordered system. In the physical literature, we typically model disorder by introducing randomness to the system, e.g. replacing the coupling constants $J_{i}(s_i,s_{i-1})$ by random variables. Thus, the transfer matrices, $X_i$ ($i=1,\ldots,N$), become random matrices distributed with respect to some probability measure on the multiplicative semi-group of positive matrices and the partition function~\eqref{moti:lattice:partition-function} is determined by the spectral properties of a product of random transfer matrices. Consequently, physical quantities are random variables in the disordered models, hence it is natural to ask for their distributions and whether these are universal. Actually, a few results are known for relatively general distributions. Typically, the free energy~\eqref{moti:lattice:free-energy} will be a Gaussian (see~\cite{Bellman:1954,Ishitani:1977} for precise statements), while the correlation length will tend to zero in the thermodynamic limit 
(this is the so-called Anderson localisation). We refer to~\cite{CPV:1993} and references for further discussion (in particular related to the disordered Ising chain).

\section{Stability of large complex systems}
\label{sec:moti:stab}

Let us follow the idea in~\cite{May:1972} and construct a simple model for the stability of large complex systems. 
We imagine a dynamical system in the variables $u(n)=\{u_1(n),\ldots,u_N(n)\}$ evolving in discrete time as $u_i(n+1)=f_i[u_1(n),\ldots,u_N(n)]$, where each $f_i$ is some smooth function. We assume that there is a fixed point, $u^*$, about which we make an expansion,
\begin{equation}
\delta u(n+1)=X\,\delta u(n)+O(\epsilon) \qquad\text{with}\qquad X_{ij}:= \frac{\p f_i(u)}{\p u_j}\bigg\vert_{u^*}.
\label{moti:stab:fix-expand}
\end{equation}
Here $X$ is the so-called stability matrix and $\epsilon\delta u(0)$ denotes a small initial perturbation to the fixed point. To leading order, we have the solution
\begin{equation}
\delta u(n)=X^n\delta u(0).
\end{equation}
The system is said to be (asymptotically) stable if the spectral norm tends to zero,
\begin{equation}
\norm{X^n}=\sup_{\delta u(0)\neq0}\frac{\norm{\delta u(n)}}{\norm{\delta u(0)}}\xrightarrow{n\to\infty}0.
\end{equation}
The interpretation of this definition is that given some small initial perturbation to the fixed point then the system will stay close to the fixed point as time evolves. An equivalent definition for stability would be to require that the spectral radius is less than unity, which in terms of the eigen- or singular values means that (i) if $z_1$ denotes the largest eigenvalue in terms of absolute values of the stability matrix, $X$, then the system is stable if $\abs{z_1}<1$ and (ii) if $\sigma_{1,n}$ denotes the largest singular value of $X^n$ then the system is stable if $(\sigma_{1,n})^{1/n}\to\sigma_1<1$. We recall that $\abs{z_1}\leq (\sigma_{1,n})^{1/n}$ for all $n$ and that $\abs{z_1}=\sigma_1$.

As an example, let us consider a large ecosystem containing $N$ interacting species. A full description of the system would, of course, be extremely complicated and highly non-linear. However, we are only interested in some small neighbourhood of a fixed point described by a stability matrix $X$, where the entry $X_{ij}$ tells us how a small fluctuation in population of the $j$-th species will affect the population of the $i$-th species. Thus, if the links $X_{ij}$ and $X_{ji}$ are both positive then the species will benefit from an increase in the population of the other (symbiosis); likewise if both links are negative then the species will have a competitive relation and if the links have opposite signs then the species will have a predator--prey relation. Rather than study each interaction between two species individually and build up the stability matrix entry by entry, we will replace it by a random matrix. The hope is that if the system is both large and complex, then it will be self-averaging. In our 
example, minimal requirements demand that the stability matrix is an asymmetric real
matrix (note that this is the discrete time version of the model considered in~\cite{May:1972}). For this reason, we will choose the entries of our random stability matrix as independent and identically distributed real-valued (Gaussian) random variables with variance $\sigma^2/N$. The circular law theorem (see~\cite{BC:2012} for a review) states that such matrices tend to a uniform law on a disk with radius $\sigma$ centred at the origin. It follows that the large-$N$ limit of our model has a phase transition between a stable and unstable phase at $\sigma=1$. Additional knowledge about this phase transition requires knowledge about the largest singular value of $X^n$ (see~\cite{MS:2014} and references within for a discussion of the phase transition in a different but closely related model).

Now, we are equipped to consider a generalisation, which requires knowledge about products of random matrices. Rather than considering an expansion around a fixed point, we might imagine expanding around a low energy path through a complicated and highly irregular landscape described by the dynamical system (see also the next section). In this case, we will have to evaluate the stability matrix along the path,
\begin{equation}
\delta u(n+1)=X_n\,\delta u(n) \qquad\text{with}\qquad (X_n)_{ij}:= \frac{\p f_i(u)}{\p u_j}\bigg\vert_{u(n)},
\label{moti:stab:path-expand}
\end{equation}
which gives rise to a solution
\begin{equation}
\delta u(n)=X_n\cdots X_1\delta u(0).
\end{equation}
Here, the question is whether two initially close trajectories will diverge or remain close as time evolves. As above, a first (and perhaps crude) approximation of such a system would be to replace the matrices $X_i$ ($i=1,\ldots,N$) with random matrices subject to symmetry constraints determined by physical considerations. This replacement will turn the question of stability into a question about the spectral properties of a product of random matrices as the number of factors tends to infinity. We will return to such models in chapter~\ref{chap:lyapunov}.

\section{Symplectic maps and Hamiltonian mechanics}
\label{sec:moti:hamilton}

Let us consider a slightly more concrete model inspired by~\cite{PV:1986} (see also~\cite{CPV:1993,Benettin:1984}), which may be thought of as included in the discussion from the previous section. We imagine a Hamiltonian system evolving in discrete time according to a $2N$ dimensional symplectic map,
\begin{align}
q(n+1)&=q(n)+p(n), \\
p(n+1)&=p(n)-\nabla V(q(n+1)),
\end{align}
where $q(n)$ and $p(n)$ are $N$ dimensional real vectors and $V$ is a twice differentiable function introducing a (non-integrable) deviation from the trivial map. In order to study the chaoticity of this evolution process we introduce a small perturbation to the trajectory, $\epsilon \delta q(n)$ and $\epsilon \delta p(n)$, which gives rise to a linearised problem,
\begin{equation}
u(n+1)=X_nu(n)
\label{moti:sym:map}
\end{equation}
with
\begin{equation}
u(n):=\begin{bmatrix} \delta q({n}) \\ \delta p({n}) \end{bmatrix},\quad
X_n:=
\begin{bmatrix}
 1 & 1 \\ 
H_n & 1+H_n 
\end{bmatrix},
\quad\text{and}\quad
(H_n)_{ij}:=-\frac{\p^2V(q(n+1))}{\p q_i\p q_j}.
\end{equation}
Here $H_n$ is a real symmetric matrix, which is in agreement with the fact that $X_n$ has to belong to $\gSp(2N,\R)$.

Given an initial perturbation, $u(0)$, then the solution of the linearised problem is trivially seen to be
\begin{equation}
\delta u(n)=X_n\cdots X_1\delta u(0).
\end{equation}
Now, the idea is the same as in the previous section. We replace either the symmetric matrices $H_i$ ($i=1,\ldots,N$) or the symplectic matrices $X_i$ ($i=1,\ldots,N$) with random matrices, which turns our problem into a study of the spectral properties of a product of random matrices. In~\cite{CPV:1993,PV:1986,Benettin:1984}, the randomness was introduced as a small perturbation of the integrable system and used to study critical exponents close to the transition between integrable and chaotic motion. However, many questions remain unanswered due to the lack of good analytic methods.

\section{Quantum transport in disordered wires}
\label{sec:moti:dmpk}

In this section, we look at how products of random matrices enter the study of quantum transport through quasi one-dimensional wires (see~\cite{Beenakker:1997} for a review). Here, ``quasi one-dimensional'' refers to a situation where we have a large number of conducting channels even though we consider a wire geometry.

Before we can understand the wire, we need to look at the transport properties for a quantum dot, i.e. the point geometry. The dot will be a chaotic cavity with two ideal leads: lead I and lead II. For simplicity, it is assumed that the leads are identical. The longitude modes of the wave function in lead I consists of $N$ incoming modes with amplitudes collectively denoted by the vector $c_\text{in}^I$ and $N$ outgoing modes with amplitudes collectively denoted by the vector $c_\text{out}^I$ and likewise for lead II. 

Two popular ways to describe the transport properties of a quantum dot are to use either a scattering matrix or a transfer matrix. The scattering matrix relates the incoming flux to the outgoing flux, while the transfer matrix relates the flux through one lead to the flux through the other. In matrix notation, we have
\begin{equation}
\begin{bmatrix}c_\text{out}^I \\ c_\text{out}^{II} \end{bmatrix}
=
\underbrace{\begin{bmatrix} r & t' \\ t & r' \end{bmatrix}}_{\displaystyle S}
\begin{bmatrix}c_\text{in}^I \\ c_\text{in}^{II} \end{bmatrix}
\qquad\text{and}\qquad
\begin{bmatrix}c_\text{in}^{II} \\ c_\text{out}^{II} \end{bmatrix}
=
\underbrace{\begin{bmatrix} a & b \\ c & d \end{bmatrix}}_{\displaystyle X}
\begin{bmatrix}c_\text{in}^I \\ c_\text{out}^{I} \end{bmatrix},
\end{equation}
respectively. Under the assumption of flux conservation it follows that the scattering matrix must be unitary, $S\in\gU(2N)$, while the transfer matrix must be split-unitary, $X\in \gU(N,N)$. For a dot given as a chaotic cavity, these matrices will be represented by random matrices. The conductance of the quantum dot is given in terms of the Landauer formula,
\begin{equation}
G/G_0=\tr t^\dagger t=\tr (a^\dagger a)^{-1},
\end{equation}
where $G_0=2e^2/h$ is the conductance quantum.

\begin{figure}[htbp]
\centering
\includegraphics{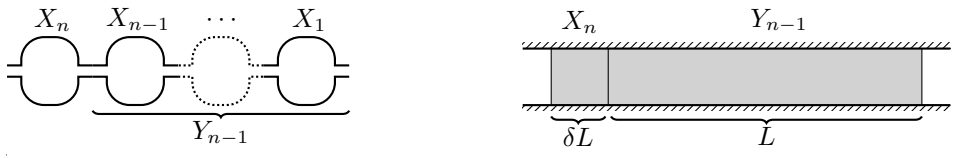}
\caption{Schematic illustrations of a disordered wire. The left panel emphasizes an interpretation of
the wire as $n$ quantum dots coupled in series, where each dot is described by a transfer matrix, $X_i$.
The right panel emphasizes an interpretation of the wire as divided into two segments; a long segment of
length $L$ and a short segment (thought of as infinitesimal) of length $\delta L$ described by transfer matrices $Y_{n-1}$ and $X_n$, respectively.
In both interpretations, we may think of the disordered wire as constructed successively one dot (or one segment of length $\delta L$) at the time.}
\label{fig:moti:wire}
\end{figure}

Now, we are ready to look at the wire geometry. In principle, the transport properties of the wire are described exactly like the dot except that the probability distribution for the scattering and the transfer matrices have to be chosen differently. However, it turns out to be a highly non-trivial task to find the correct distribution for these matrices. The usual trick is to divide the wire up into smaller pieces which are easier to understand and then rebuild the wire piece by piece, see e.g.~\cite{MPK:1988,IWZ:1990}. Figure~\ref{fig:moti:wire} illustrates two possible ways to construct a wire. The transfer matrix description seems particularly suited for such descriptions, since it links flux at one lead to the flux at the other lead. Thus, if we have a wire with an unknown transfer matrix, $Y_n$, divided into $n$ pieces each described by a transfer matrix $X_i$ (numbered successively), then the transfer matrix for the wire may be written as
\begin{equation}
Y_n=X_nY_{n-1}=X_nX_{n-1}\cdots X_1,
\end{equation}
i.e. a product of random matrices. A construction using scattering matrices is slightly more complicated, since they provide a relation between incoming and outgoing flux rather than a relation between leads.

\section{QCD at non-zero chemical potential}
\label{sec:moti:qcd}

Quantum chromodynamics (QCD) is broadly accepted as \emph{the} theory for the strong interaction. However, even with a known fundamental theory, many intriguing questions remain unanswered partly due to the non-perturbative nature of QCD at low energies. One of the most successful approaches in this non-perturbative regime is the use of lattice simulations. However, lattice simulations may be prohibited in certain regions of the phase diagram due to technical difficulties.

A major open problem in the description of strongly interacting matter is to understand the behaviour at a non-zero (baryon) chemical potential (and therefore non-zero baryon density). In this case, lattice simulations are plagued by a notorious sign problem. The core of the problem is that the fermion determinant is not ensured to be real and non-negative (it becomes complex) when the chemical potential differs from zero. For this reason, the fermion determinant cannot be included in the weight used for Monte Carlo sampling, which prohibits a standard approach, see~\cite{Schmidt:2006} for a review.

Some insight into this problem may be achieved using a random matrix model related to the product of two random matrices~\cite{Osborn:2004}. The model is defined through the partition function
\begin{equation}
\cZ(\mu)=\int\limits_{\C^{N\times(N+\nu)}}d^2X_1 \int\limits_{\C^{(N+\nu)\times N}}d^2X_2\, w_\mu(X_1,X_2)\prod_{f=1}^{N_f}\det[D+m_f],
\label{moti:qcd:partition}
\end{equation}
where $N_f$ is the number of quark flavours, $m_f$ denotes the mass of the $f$-th flavour, and $\mu\in(0,1]$ is the chemical potential. The matrix $D$ is given by
\begin{equation}
D=\begin{bmatrix} 0 & X_1 \\ X_2 & 0 \end{bmatrix}
\label{moti:qcd:dirac}
\end{equation}
and corresponds to the Dirac operator, while the weight function, $w_\mu(X_1,X_2)$, is given by
\begin{equation}
w_\mu(X_1,X_2)=\exp\bigg[-\frac{N(1+\mu^2)}{4\mu^2}\tr(X_1^\dagger X_1+X_2^\dagger X_2)-\frac{N(1-\mu^2)}{4\mu^2}\tr(X_1X_2+X_1^\dagger X_2^\dagger)\bigg].
\label{moti:qcd:weight}
\end{equation}
The scaling regime relevant for QCD is when $\mu^2=O(N^{-1})$ as $N$ tends to infinity; this is the limit of weak non-Hermiticity. The limit of strong non-Hermiticity, $\mu\to1$, is interesting as well (albeit not relevant for applications to QCD). In this limit the weight function~\eqref{moti:qcd:weight} splits into two separate Gaussian weights, hence $X_1$ and $X_2$ become independent Gaussian random matrices. 

For physical applications, we are interested in the (generalised) spectral density of the Dirac operator~\eqref{moti:qcd:dirac} or equivalently of the product $X_1X_2$. If $z_1,\ldots,z_N$ denote the non-zero eigenvalues of the Dirac operator~\eqref{moti:qcd:dirac} (this means that $z_1^2,\ldots,z_N^2$ are the eigenvalues of the product $X_1X_2$), then we define the spectral density as
\begin{equation}
\rho_\mu(z):=\naverage[\bigg]{\frac{1}{N}\sum_{k=1}^N\delta^2(z_k-z)}_{\cZ(\mu)},
\end{equation}
where the average is taken according to the partition function~\eqref{moti:qcd:partition}. We stress that this is a generalised density in the sense that it integrates to unity, but it is only ensured to be real and non-negative if $N_f=0$ (this is the so-called quenched approximation). 

Using the above given matrix model it was shown in~\cite{OSV:2005,AOSV:2005} that the complex phase arising due to the QCD sign problem contains essential physical information which must be included in order to obtain correct physical predictions, see~\cite{Splittorff:2006} for a review. 

%
%
%
%
%

%% file: prologue.tex
\chapter{From random scalars to random matrices}
\label{chap:prologue}


The purpose of the this chapter is two-fold: firstly, we want to review a few general properties of products of random variables to illustrate some similarities as well as differences between random scalars and random matrices; it will be helpful to keep the well-known structures for random scalars in mind, when considering products of independent Gaussian random matrices in the following chapters. Secondly, we want to introduce a few concepts which will be extensively used in the following chapters; isotropy and induced Ginibre matrices will be of particular interest. 

The chapter is divided into three sections: in section~\ref{sec:prologue:scalar} we will recollect some important structures for products of random scalars; section~\ref{sec:prologue:square} concerns (isotropic) random matrices and is partially based on the paper~\cite{IK:2014}, while section~\ref{sec:prologue:gauss} introduces the well-known Gaussian random matrix ensembles, which will be the central object for the rest of the thesis.

\section{Products of independent random scalars}
\label{sec:prologue:scalar}

This section is devoted to products of independent random scalars. However, we do not attempt to give an exhaustive nor extensive description of such products. For a more thorough account of classical probability the reader is referred to~\cite{Feller:1968,Kallenberg:1997}, while a thorough description of the algebra of (real-valued) random scalars can be found in~\cite{Springer:1979}.

\subsection{Finite products of random scalars}

Let $x_i$ ($i=1,\ldots,n$) be a family of continuous independent real ($\beta=1$) or complex ($\beta=2$) random scalars distributed with respect to probability density functions $p_i^\beta(x_i)$. We can construct a new random scalar as a product of the old, $y_n:= x_n\cdots x_1$. The density of the new random scalar can formally be written as
\begin{equation}
p^\beta(y_n)=\bigg[\prod_{i=1}^n\int_{\F_\beta} d^\beta x_i\,p_i^\beta(x_i)\bigg]\,\delta^\beta(x_n\cdots x_1-y_n),
\label{prologue:scalar:general}
\end{equation}
where $\delta^\beta(x)$ is the Dirac delta function and $d^\beta x$ denotes the flat (Lebesgue) measure on the real line ($\F_{\beta=1}:=\R$) or complex plane ($\F_{\beta=2}:=\C$), respectively.  By definition, the individual random scalars are non-zero almost surely, and therefore, so is any product with a finite number of factors.  

An alternative expression for the density~\eqref{prologue:scalar:general} is obtained by a simple change of variables, $y_{i+1}=x_{i+1}y_{i}$ with $y_1=x_1$. This yields
\begin{equation}
p^\beta(y_n)=
\bigg[\prod_{i=1}^{n-1}\int_{\F_\beta} \frac{d^\beta y_i}{\abs{y_i}^\beta} p_{i+1}^\beta\bigg(\frac{y_{i+1}}{y_i}\bigg)\bigg]
p_1^\beta(y_1),
\label{prologue:scalar:y-product}
\end{equation}
where we explicitly use that the random scalars are non-zero almost surely. For notational simplicity, it is sometimes convenient to introduce the convolution defined by
\begin{equation}
f*g(y):= \int_{\gGL(1,\F_\beta)} d\mu(x) f(y/x)g(x)
\end{equation}
where $d\mu(x):= d^\beta x/\abs{x}^\beta$ is the Haar (invariant) measure on the group of non-zero real or complex numbers with multiplication. With this notation, the density~\eqref{prologue:scalar:y-product} reduces to
\begin{equation}
p^\beta(y_n)=p_n^\beta *\cdots *p_1^\beta(y_n).
\label{prologue:scalar:y-product-simple}
\end{equation}
It is worth noting that the equivalent expression for sums of random scalars is obtained by replacing the convolution on the multiplicative group $\gGL(1,\F_\beta)$ with the convolution on the additive group $(\F_\beta,+)$. Both convolutions inherit commutativity from the scalar operations. 

Isotropic probability distributions will be of particular interest in this thesis, that is distributions which are invariant under bi-unitary transformations, see definition~\ref{def:prologue:isotropy}. For a random scalar with density $p_i^\beta(x)$ that is 
\begin{equation}
p_i^\beta(ux)=p_i^\beta(x)
\label{prologue:scalar:isotropic}
\end{equation}
with $u=\pm 1$ for $\beta=1$ and $u=e^{i\theta}\in \gU(1)$ for $\beta=2$. Both the flat and the Haar measure are invariant under such transformation as well. Isotropy is an important symmetry, since it allows us to describe random scalars solely in terms of their absolute value, i.e. a problem restricted to the positive half-line rather than the full real line or the complex plane.  

Let us return to the product density~\eqref{prologue:scalar:general} and consider a product with isotropic densities. If the $s$-th moment is well-defined for the individual distributions, then we see that
\begin{equation}
\average{\abs{y_n}^s}:=\int_{\F_\beta} d^\beta y_n\,p^\beta(y_n)\abs{y_n}^s=\prod_{i=1}^n\int_{\F_\beta} d^\beta x_i\,p_i^\beta(x_i) \abs{x_i}^s=:\prod_{i=1}^n\E_i[{\abs{x_i}^s}]\,,
\label{prologue:scalar:moments}
\end{equation}
which is an immediate consequence of the independence. In words, this means that the moments of the product are given by the product of the moments. Note that $s$ is not necessarily an integer. Thus, a description in terms of moments is straightforward. However, it might be a non-trivial task to obtain an explicit expression for the density. A main observation is that~\eqref{prologue:scalar:moments} may be interpreted as a Mellin transform; as a consequence it is often possible to find the corresponding density by means of an inverse Mellin transform. The application of the Mellin transform in this context dates back to the seminal paper~\cite{Epstein:1948}; the reader is referred to~\cite{Springer:1979}, and references within, for a thorough description of products of (real) random scalars. 

Let us illustrate the above mentioned procedure with a simple, but important, example. Namely, a product of $n$ independent Gaussian random scalars with zero mean and unit variance. From~\eqref{prologue:scalar:general}, we have
\begin{equation}
p^\beta(y_n)=\bigg[\prod_{i=1}^n\Big(\frac{\beta}{2\pi}\Big)^{\beta/2}\int_{\F_\beta} d^\beta x_i\,e^{-\beta\abs{x_i}^2/2}\bigg]\,\delta^\beta(x_n\cdots x_1-y_n).
\end{equation}
Isotropy suggests a change to polar coordinates, which after integration over the phases yields
\begin{equation}
p^\beta(y_n)=\frac1Z\bigg[\prod_{i=1}^n\frac2\beta\int_0^\infty dr_i\,e^{-\beta r_i/2}\bigg]\,\delta(r_n\cdots r_1-\abs{y_n}^{1/2}).
\end{equation}
with $Z:=\pi^{\beta-1}((2/\beta)^{(\beta-2)/2}\Gamma[\beta/2])^n$.
This expression has a natural interpretation as the probability density for a product of $n$ gamma distributed random scalars. The Mellin transform, or equivalently the $(s-1)$-th moment, is given
\begin{equation}
\cM[p^\beta](s):=\average{\abs{y_n}^{s-1}}=\frac1Z\Big[\Big(\frac2\beta\Big)^{s+1}\Gamma[s]\Big]^n.
\end{equation}
The inverse Mellin transform is immediately recognised as a Meijer $G$-function (see definition~\ref{def:special:meijer}),
\begin{equation}
p^\beta(y_n)=\frac{\pi^{1-\beta}}{\Gamma(\beta/2)^n}\Big(\frac{\beta}{2}\Big)^{\beta n/2}
\MeijerG{n}{0}{0}{n}{-}{0,\ldots,0}{\Big(\frac{\beta}{2}\Big)^n\abs{y_n}^2}.
\label{prologue:scalar:Meijer}
\end{equation}
We stress that the appearance of the Meijer $G$-function is by no means restricted to the problem involving Gaussian random scalars. On the contrary, the Meijer $G$-function possesses a prominent position in the study of products of random scalars due to its intimate relation with the Mellin transform. In fact, the Meijer $G$-function turns out to be important in the study of products of random matrices as well. A discussion of the Meijer $G$-function as well as references to the relevant literature can be found in appendix~\ref{app:special}.

\subsection{Asymptotic behaviour}

In certain cases our problem simplifies when the number of factors tends to infinity, since this allows us to employ the law of large numbers and the central limit theorem. 

Consider a set of independent and identically distributed random scalars $x_i$ ($i=1,2,\ldots$) and assume that the expectation $\average{\log\abs{x_1}}$ is finite, then it follows from the (strong) law of large numbers that the geometric mean converges almost surely,
\begin{equation}
\lim_{n\to\infty}\abs{x_n\cdots x_1}^{1/n}=\lim_{n\to\infty}\exp\bigg[\frac1n\sum_{i=1}^n\log\abs{x_1}\,\bigg] =\exp\big[\average{\log\abs{x_i}}\big].
\end{equation}
Note that the equality uses the commutative property of the scalar product; and that the absolute value is generally required in order to ensure a unique limit. 
If we additionally assume that $\average{(\log\abs{x_1}-\E\log\abs{x_1})^2}=\sigma^2<\infty$ then the central limit theorem states
\begin{equation}
\lim_{n\to\infty}\P\Big[\frac{\log \abs{x_1\cdots x_n}^{1/n}-\E\log\abs{x_1}}{\sqrt{\sigma^2/n}}\leq y\Big]
=\int_{-\infty}^yd\lambda\, \frac{e^{-\lambda^2/2}}{\sqrt{2\pi}}.
\label{prologue:scalar:CLT}
\end{equation}
In words, the law of large numbers tells us that a product of random scalars with a large number of factors grows (decays) exponentially with a growth rate $\average{\log\abs{x_1}}$, while the central limit theorem tells us that the fluctuations of the growth rate converge in distribution to a Gaussian on the scale $1/\sqrt n$. Both results are universal in the sense that they do not depend on the explicit form of the underlying distribution. 

Let us verify that the large-$n$ limit of the random scalar $y_n$ distributed with respect to the density~\eqref{prologue:scalar:Meijer} indeed behaves according to the law of large numbers and the central limit theorem. In order to do so, we introduce a new random scalar defined as $\lambda:= (\log\abs{y_n})/n$. The cumulant generating function and the corresponding $k$-th cumulant for $\lambda$ are given by
\begin{equation}
g^\beta(t)=\log\bigg(2n\pi^{\beta-1}\int_\R d\lambda\,e^{(\beta n+t)\lambda}p^\beta(e^{n\lambda})\bigg)
\qquad\text{and}\qquad
\kappa_k^\beta=\frac{\p^k g^\beta(t)}{\p t^k}\bigg\vert_{t=0},
\label{prologue:scalar:generate-cumulant}
\end{equation}
respectively. In our case, the density $p^\beta(y)$ is given by~\eqref{prologue:scalar:Meijer} and the integral within the logarithm in~\eqref{prologue:scalar:generate-cumulant} can be performed using an integration formula for the Meijer $G$-function~\eqref{special:meijer:meijer-moment}. A short computation yields
\begin{align}
\mu:=\kappa_1^\beta&=\frac12\log\frac{2}\beta+\frac12\psi\Big(\frac\beta2\Big), \qquad
\sigma^2:=\kappa_2^\beta=\frac1{4n}\psi'\Big(\frac\beta2\Big), \\
\kappa_k^\beta&=\frac12\Big(\frac1{2n}\Big)^{k-1}\psi^{(k-1)}\Big(\frac\beta2\Big)
\qquad\text{for}\qquad k\geq3.
\end{align}
Here $\psi(x)$ denotes the digamma function, while $\psi^{(k)}(x)$ is its $k$-th derivative also known as the polygamma function, see appendix~\ref{app:special}. To find the large-$n$ behaviour, we switch from the random scalar $\lambda$ to the standardised random scalar $\widetilde \lambda:= (\lambda-\mu)/\sigma$, which has zero mean and unit variance. The standardised cumulants are
\begin{equation}
\widetilde \kappa_1^\beta=0, \qquad \widetilde \kappa_2^\beta=1,
\qquad\text{and}\qquad
\widetilde \kappa_k^\beta=O(n^{1-k/2})
\quad\text{for}\quad k\geq3.
\end{equation}
We see that the higher order cumulants tend to zero as $n$ tends to infinity. It follows from standard arguments that the limiting distribution is a Gaussian. Returning from the standardised variable $\widetilde \lambda$ to the original variable $\lambda$, we find the asymptotic behaviour
\begin{equation}
\P[\lambda\leq t]\sim\sqrt{\frac{n}{2\pi\sigma^2}}\int_{-\infty}^{t}\exp\bigg[-n\frac{(\lambda-\mu)^2}{2\sigma^2}\bigg],
\end{equation}
which is in agreement with the law of large numbers and the central limit theorem.

\section{Products of independent random matrices}
\label{sec:prologue:square}

We are now ready to discuss products of random matrices which is the main topic in this thesis. Here, we focus on a few general properties related to matrix-multiplication and to isotropy, while a discussion of more classical results from random matrix theory (that is statements about spectral correlations) is postponed to the following chapters.
For an introduction to random matrix theory, we refer to~\cite{Mehta:2004,Forrester:2010,Tao:2012,AGZ:2010} and the review~\cite{ER:2005}; while a large variety of applications is discussed in a contemporary and extensive handbook~\cite{ABF:2011}. Some of the well-known properties for \emph{products} of random matrices are summarised in~\cite{CPV:1993,CKN:1986} and~\cite{NS:2006}, where the latter takes the viewpoint of free probability.

\subsection{Finite products of finite size square random matrices}

We will generally be interested in statistical properties of a product of $n$ independent square random matrices. We write this product matrix as
\begin{equation}
Y_n:= X_n\cdots X_1,
\label{prologue:square:product}
\end{equation}
where each $X_i$ ($i=1,\ldots,n$) is a real ($\beta=1$), complex ($\beta=2$) or quaternionic ($\beta=4$) $N\times N$ random matrix distributed with respect to a probability density $P_i^\beta(X_i)$, which by assumption is integrable with respect to the flat measure on the corresponding matrix space. For quaternions we use the canonical representation as $2\times 2$ matrices, see appendix~\ref{app:decompositions}. Thus, an $N\times N$ quaternionic matrix should be understood as a $2N\times 2N$ complex matrix which satisfies the quaternionic symmetry requirements. 

The probability density for the matrix $Y_n$ is formally defined as
\begin{equation}
P_{\{n,\ldots,1\}}^\beta(Y_n):=
\bigg[\prod_{i=1}^n\int_{\F_\beta^{N\times N}} d^\beta X_i\,P_i^\beta(X_i)\bigg]\,\delta^\beta(X_n\cdots X_1-Y_n),
\label{prologue:square:general}
\end{equation}
where $\delta^\beta(x)$ is the Dirac delta function of matrix argument and $d^\beta X$ denotes the flat measure on space of real, complex or quaternionic $N\times N$ matrices, i.e. on $\F_\beta^{N\times N}$ with $\F_\beta:=\R,\C,\H$. The multi-index on the density~\eqref{prologue:square:general} incorporates the ordering of the factors; this is necessary since matrix-multiplication is non-commutative.

By assumption, the matrices $X_i$ ($i=1,\ldots,n$) are non-singular almost surely and an alternative expression for the density~\eqref{prologue:scalar:general} can be found by a change of variables, $Y_{i+1}:= X_{i+1}Y_{i}$ with $Y_1:= X_1$. We find
\begin{align}
P_{\{n,\ldots,1\}}^\beta(Y_n)&=P_n^\beta *\cdots *P_1^\beta(Y_n) \nn \\
&=\bigg[\prod_{i=1}^{n-1}\int_{\gGL(N,\F_\beta)} d\mu(Y_i) P_{i+1}^\beta(Y_{i+1}Y_i^{-1})\bigg]P_1^\beta(Y_1),
\label{prologue:square:density-product}
\end{align}
where `$*$' and $d\mu(Y):= d^\beta Y/(\det Y^\dagger Y)^{\beta N/2\gamma}$ ($\gamma=1,1,2$ for $\beta=1,2,4$) denote the convolution and the Haar measure on the the group of real, complex or quaternionic invertible matrices, respectively. 

Both~\eqref{prologue:square:general} and~\eqref{prologue:square:density-product} appear as direct generalisations of the formulae for products of random scalars. However, this similarity is to some extent deceiving, since we typically are interested in spectral properties rather than the matrices themselves. 

\subsection{Weak commutation relation for isotropic random matrices}
\label{sec:prologue:weak}

One of the key differences between products of random scalars and random matrices is that matrix-multiplication generally is non-commutative. Nonetheless, we may consider matrix products~\eqref{prologue:square:product} which commute in a weak sense, such that
\begin{equation}
X_n\cdots X_1\deq X_{\sigma(n)}\cdots X_{\sigma(1)}
\label{prologue:square:weak-X}
\end{equation}
for any permutation $\sigma\in S_n$. The trivial example is when $X_i$ ($i=1,\ldots,n$) are independent and identically distributed (square) random matrices. In this section we will show that the restriction to identical distributions may be replaced by a symmetry requirement. We follow the idea in~\cite{IK:2014}.

\begin{definition}\label{def:prologue:isotropy}
Let $X$ be an $N\times M$ continuous random matrix distributed according to a probability density $P^\beta(X)$ on the matrix space $\F_\beta^{N\times M}$ with $\F_{\beta=1,2,4}=\R,\C,\H$. If
\begin{equation}
P^\beta(UXV)=P^\beta(X)\quad\text{for all}\quad (U,V)\in \gU(N,\F_\beta)\times\gU(M,\F_\beta) \\
\label{prologue:square:isotropy}
\end{equation}
then we say that the density $P^\beta(X)$ is isotropic, while the matrix $X$ is said to be statistically isotropic. Above, we have used the notation
\begin{equation}
\gU(N,\F_{\beta=1,2,4})=\gO(N),\gU(N),\gUSp(2N)
\end{equation}
for the maximal compact subgroups. 
\end{definition}

\begin{remark}
It is evident that isotropy implies that the density only depends on the singular values of its matrix argument.
\end{remark}

\begin{proposition}
If $\{X_i\}_{i=1,\ldots,n}$ is a set of independent statistically isotropic square random matrices, then the weak commutation relation~\eqref{prologue:square:weak-X} holds.
\end{proposition}

\begin{proof}
Our starting point is the density~\eqref{prologue:square:general} which is valid for independent matrices. It is sufficient to show that
\begin{equation}
P_{\{n,\ldots,j+1,j,\ldots,1\}}^\beta(Y_n)=P_{\{n,\ldots,j,j+1,\ldots,1\}}^\beta(Y_n)
\label{prologue:square:reorder}
\end{equation}
for any $j$, since such permutations are the generators of the permutation group, $S_n$.
We use that for any two matrices $X_j$ and $X_{j+1}$ there exists a singular value decomposition such that
\begin{equation}
X_{j+1}X_j=V\Sigma U=V\Sigma^\dagger U=VU(X_{j+1}X_{j})^\dagger VU=VUX_j^\dagger X_{j+1}^\dagger VU,
\label{prologue:square:XXtoXX}
\end{equation}
where $\Sigma$ is a positive semi-definite diagonal matrix, while $U$ and $V$ are orthogonal ($\beta=1$), unitary ($\beta=2$), or unitary symplectic ($\beta=4$) matrices. We insert identity~\eqref{prologue:square:XXtoXX} into the delta function in~\eqref{prologue:square:general} and use the isotropy~\eqref{prologue:square:isotropy} to absorb the unitary transformations $U$ and $V$ into the measure. This yields
\begin{equation}
P_{\{n,\ldots,j+1,j,\ldots,1\}}^\beta(Y_n)=\bigg[\prod_{i=1}^n\int d^\beta X_i\,P_i^\beta(X_i)\bigg]\,
\delta^\beta(X_n\cdots X_{j}^\dagger X_{j+1}^\dagger\cdots X_1-Y_n).
\label{prologue:square:density-step}
\end{equation}
We can now repeat the same idea for the individual matrices $X_j$ and $X_{j+1}$. Similar to~\eqref{prologue:square:XXtoXX}, we have
\begin{equation}
X_j^\dagger=V_j\Sigma_j U_j=V_jU_jX_jV_jU_j,
\label{prologue:square:XtoX}
\end{equation}
where $\Sigma_j$ is a positive semi-definite diagonal matrix, while $U_j$ and $V_j$ are orthogonal, unitary, or unitary symplectic matrices. We insert~\eqref{prologue:square:XtoX} and an equivalent identity for $X_{j+1}^\dagger$ into the delta function in~\eqref{prologue:square:density-step}. As before, the unitary transformations can be absorbed into the measure due to isotropy, hence
\begin{equation}
P_{\{n,\ldots,j+1,j,\ldots,1\}}^\beta(Y_n)=\bigg[\prod_{i=1}^n\int d^\beta X_i\,P_i^\beta(X_i)\bigg]\,
\delta^\beta(X_n\cdots X_{j} X_{j+1}\cdots X_1-Y_n).
\end{equation}
This is the identity~\eqref{prologue:square:reorder}, which proves the weak commutation relation for isotropic densities. 
\end{proof}

\subsection{From rectangular to square matrices}
\label{sec:prologue:rectangular}

So far we have looked solely on products of square matrices. However, it is desirable to extend the description to the general case including \emph{rectangular} matrices. Let us consider a product of independent random matrices,
\begin{equation}
\widetilde Y_n:= \widetilde X_n\cdots \widetilde X_1,
\label{prologue:rectangular:product}
\end{equation}
where each $\widetilde X_i$ ($i=1,\ldots,n$) is a real, complex or quaternionic $N_i\times N_{i-1}$ random matrix distributed with respect to a probability density $\widetilde P_i^\beta(\widetilde X_i)$.

The probability density for the product matrix is defined like in the square case,
\begin{equation}
\widetilde P_{\{n,\ldots,1\}}^\beta(\widetilde Y_n):=\bigg[\prod_{i=1}^n\int_{\F_\beta^{N_i\times N_{i-1}}} d^\beta\widetilde X_i\,\widetilde P_i^\beta(\widetilde X_i)\bigg]\,
\delta^\beta(\widetilde X_n\cdots \widetilde X_1-\widetilde Y_n).
\label{prologue:rectangular:general}
\end{equation}
However, we have no direct analogue of~\eqref{prologue:square:density-product} nor does isotropy imply weak commutativity in the sense of~\eqref{prologue:square:weak-X}.
In order to reclaim these useful properties of square matrices, we will reformulate the product of rectangular matrices defined through~\eqref{prologue:rectangular:product} and~\eqref{prologue:rectangular:general} in terms of square matrices. We follow the idea presented in~\cite{IK:2014}.

The generalised block QR decomposition (proposition~\ref{thm:decomp:gen-QR-block} and corollary~\ref{thm:decomp:gen-QR-block-2}) tells us that given a product matrix~\eqref{prologue:rectangular:product} with the smallest matrix dimension denoted by $N$, we can find a pair of orthogonal ($\beta=1$), unitary ($\beta=2$), or unitary symplectic ($\beta=4$) matrices, $\widetilde U_1$ and $\widetilde U_n$, so that
\begin{equation}
\widetilde Y_n=\widetilde U_n\begin{bmatrix}Y_n & 0 \\ 0 & 0 \end{bmatrix}(\widetilde U_1)^{-1},
\end{equation}
where $Y_n$ is an $N\times N$ matrix. This immediately reveals that $\widetilde Y_n$ has at most rank $N$, or equivalently that at least $\max\{N_n,N_0\}-N$ singular values are equal to zero (a similar statement may be formulated for the eigenvalues if $N_n=N_0$). If we additionally require that the individual matrices $\widetilde X_i$ ($i=1,\ldots,n$) are statistically isotropic, then we can establish a stronger statement:

\begin{proposition}\label{thm:pro:rect-square}
Consider a product of independent random matrices~\eqref{prologue:rectangular:product} with matrix density~\eqref{prologue:rectangular:general} where each of the individual densities, $\widetilde P_i(\widetilde X_i)$, is isotropic. Let $N=N_j$ denote the smallest matrix dimension (not necessarily unique) and let $\nu_i$ ($i=1,\ldots,n$) be a collection of non-negative integers such that $N_i=N+\nu_{i+1}$ for $i<j$ and $N_i=N+\nu_{i}$ for $i>j$, then
\begin{equation}
\int_{\F_\beta^{N_0\times N_n}} d^\beta\widetilde Y_n\widetilde P_{\{n,\ldots,1\}}^\beta(\widetilde Y_n)\delta^\beta\bigg(\widetilde Y_n-\bigg[\,\begin{matrix}Y_n & 0 \\ 0 & 0 \end{matrix}\,\bigg]\bigg)
=P_{\{n,\ldots,1\}}^\beta(Y_n)
\label{prologue:rectangular:rect-to-square}
\end{equation}
where $Y_n$ is an $N\times N$ matrix and
\begin{equation}
P_{\{n,\ldots,1\}}^\beta(Y_n):=\bigg[\prod_{i=1}^n\int_{\F_\beta^{N\times N}} d^\beta X_i\,P_i^\beta(X_i)\bigg]\,\delta^\beta(X_n\cdots X_1-Y_n)
\label{prologue:rectangular:square-density}
\end{equation}
where $P_i(X_i)$ are probability densities for a family of $N\times N$ matrices, $X_i$. Moreover, the densities are explicitly given by
\begin{subequations}
\label{prologue:rectangular:induced}
\begin{align}
P_{i}^\beta(X_i)&=\vol_i\,\det (X_i^\dagger X_i)^{\beta\nu_{i}/2\gamma}\int_{\F_\beta^{\nu_{i+1}\times(N+\nu_i)}} d^\beta T_i\,
\widetilde P_i^\beta\bigg( \bigg[\begin{array}{@{}c @{}} \begin{matrix} X_{i} & 0 \end{matrix} \\ \hline T_i \end{array}\bigg] \bigg),
\qquad i<j \label{prologue:rectangular:induced-a}\\
P_i^\beta(X_i)&=\vol_i\,\det (X_i^\dagger X_i)^{\beta\nu_i/2\gamma} \int_{\F_\beta^{(N+\nu_{i})\times\nu_{i-1}}} d^\beta T_i\,
\widetilde P_i^\beta\bigg( \bigg[\begin{matrix} X_{i} \\ 0 \end{matrix}\, \bigg\vert\, T_i\,\bigg] \bigg),
\qquad i>j\label{prologue:rectangular:induced-b}
\end{align}
\end{subequations}
with
\begin{equation}
\vol_i:=\vol [\gU(N+\nu_{i},\F_\beta)/\gU(N,\F_\beta)\times\gU(\nu_{i},\F_\beta)]
\label{prologue:rectangular:volume}
\end{equation}
denoting the volumes of the Grassmannians.
\end{proposition}

\begin{proof} 
We factorise the product~\eqref{prologue:rectangular:product} into two partial products $\widetilde X_n\cdots\widetilde X_{j+1}$ and $\widetilde X_j\cdots\widetilde X_{1}$. From proposition~\ref{thm:decomp:gen-QR-block} and corollary~\ref{thm:decomp:gen-QR-block-2}, we have the parametrisation
\begin{equation}
\widetilde X_i=\widetilde U_i\bigg[\begin{array}{@{}c @{}} \begin{matrix} X_{i} & 0 \end{matrix} \\ \hline T_i \end{array}\bigg](\widetilde U_{i-1})^{-1}
\quad\text{and}\quad
\widetilde X_i=\widetilde U_{i+1}\bigg[\begin{matrix} X_{i} \\ 0 \end{matrix}\, \bigg\vert\, T_i\,\bigg](\widetilde U_{i})^{-1}
\end{equation}
for $i<j$ and $i>j$, respectively. Here $X_{i}$ are $N\times N$ matrices and $\widetilde U_{i}\in\gU(N+\nu_{i},\F_\beta)/\gU(N,\F_\beta)\times\gU(\nu_{i},\F_\beta)$, while each $T_i$ is either a $\nu_{i+1}\times (N+\nu_{i})$ matrix ($i<j$) or an $(N+\nu_{i})\times\nu_{i-1}$ matrix ($i>j$). The corresponding change of measure is
\begin{equation}
\prod_{i=1}^nd^\beta\widetilde X_i
=\prod_{i=1}^n\det (X_i^\dagger X_i)^{\beta\nu_{i}/2\gamma}d^\beta X_id^\beta T_id\mu(\widetilde U_i)
\end{equation}
with $d\mu(\widetilde U_i):=(\widetilde U_i)^{-1}d\widetilde U_i$ denoting the Haar measure on $\gU(N+\nu_{i},\F_\beta)/\gU(N,\F_\beta)\times\gU(\nu_{i},\F_\beta)$.
We insert this parametrisations into the density~\eqref{prologue:rectangular:general} and use isotropy to absorb the unitary transformation into the measures, which yields
\begin{align}
\widetilde P_{\{n,\ldots,1\}}^\beta(\widetilde Y_n)
&=\bigg[\prod_{i=1}^j\int d^\beta X_i\det (X_i^\dagger X_i)^{\beta\nu_{i}/2\gamma} \int d^\beta T_i\int d\mu(\widetilde U_i)
\widetilde P_i^\beta\bigg( \bigg[\begin{array}{@{}c @{}} \begin{matrix} X_{i} & 0 \end{matrix} \\ \hline T_i \end{array}\bigg] \bigg)\bigg]\nn\\
&\times\bigg[\prod_{i=j+1}^n\int d^\beta X_i\det (X_i^\dagger X_i)^{\beta\nu_{i}/2\gamma} \int d^\beta T_i\int d\mu(\widetilde U_i)
\widetilde P_i^\beta\bigg( \bigg[\begin{matrix} X_{i} \\ 0 \end{matrix}\, \bigg\vert\, T_i\,\bigg] \bigg)\bigg]\nn\\
&\times\delta^\beta\bigg(\bigg[\begin{matrix}X_n\cdots X_1 & 0 \\ 0 & 0 \end{matrix}\,\bigg]-\widetilde Y_n\bigg).
\label{prologue:rectangular:square-rectangular}
\end{align}
We can now insert this expression into~\eqref{prologue:rectangular:rect-to-square}. The formulae~\eqref{prologue:rectangular:square-density} and~\eqref{prologue:rectangular:induced} are obtained after integration over $\widetilde Y_n$ and $\widetilde U_i$ ($i=1,\ldots,n$).

It remains to verify that the densities~\eqref{prologue:rectangular:induced} are normalised to unity. In order to show this, we introduce matrices $\widetilde V_i\in\gU(N+\nu_{i},\F_\beta)/\gU(N,\F_\beta)\times\gU(\nu_{i},\F_\beta)$. By definition, we have $\vol_i=\int d\mu(\widetilde V_i)$ where $d\mu(\widetilde V_i):=[\widetilde V_i^{-1}d\widetilde V_i]$ is the Haar measure. It follows that
\begin{equation}
\int d^\beta X_i P^\beta_i(X_i)=\frac{1}{\vol_i}\int d\mu(\widetilde V_i)\int d^\beta X_i P^\beta_i(X_i)
\end{equation}
and by isotropy that
\begin{equation}
\int d^\beta X_i P^\beta_i(X_i)=\int d\mu(\widetilde V_i)\int d^\beta X_i\det (X_i^\dagger X_i)^{\beta\nu_{i}/2\gamma}\int d^\beta T_i\,
\widetilde P_i^\beta\bigg(\bigg[\begin{array}{@{}c @{}}\begin{matrix}X_{i}&0\end{matrix}\\ \hline T_i\end{array}\bigg]\widetilde V_i\bigg)
\end{equation}
for $i<j$ (with an equivalent expression for $i>j$). Here, we recognise the right hand side as a block-QR decomposition, thus
\begin{equation}
\int d^\beta X_i P^\beta_i(X_i)=\int d^\beta\widetilde  X_i\widetilde  P^\beta_i(\widetilde X_i)
\end{equation}
and the normalisation follows from the definition of $\widetilde  P^\beta_i(\widetilde X_i)$.
\end{proof}

\begin{corollary}\label{thm:pro:rect-weak}
The matrix densities~\eqref{prologue:rectangular:induced} are isotropic.
\end{corollary}
\begin{proof}
The isotropy of $P_i(X_i)$ follows from the isotropy of $\widetilde P_i(\widetilde X_i)$ together with invariance of the determinantal prefactor under bi-unitary transformations.
\end{proof}

\begin{remark}
Proposition~\ref{thm:pro:rect-square} tells us, that given a product of independent statistically isotropic \emph{rectangular} random matrices, we can find a product of independent \emph{square} matrices which has the same spectral properties (up to a number of trivial zeros). Furthermore, corollary~\ref{thm:pro:rect-weak} tells us that the square matrices inherit isotropy from their rectangular counter parts and, thus, the square matrices commute in a weak sense.
\end{remark}

\section{Gaussian random matrix ensembles}
\label{sec:prologue:gauss}

In the rest of this thesis, we will focus on Gaussian random matrix ensembles. For future reference, we summarise the precise definitions for these ensembles in this section.

\begin{definition}\label{def:pro:ginibre}
The real, complex, and quaternionic Ginibre ensembles (or non-Hermitian Gaussian ensembles) are defined as the space of $N\times M$ matrices $\widetilde X$ whose entries are independent and identically distributed real-, complex-, or quaternion-valued Gaussian random variables with zero mean and unit variance, i.e. matrices distributed according to the density
\begin{equation}
\widetilde P_\gauss^\beta(\widetilde X)
=\bigg(\frac\beta{2\pi}\bigg)^{\beta NM/2}\exp\Big[-\frac{\beta}{2\gamma}\tr \widetilde X^\dagger \widetilde X\Big]
\label{prologue:outline:gauss}
\end{equation}
on the matrix space $\F_\beta^{N\times M}$ with $\F_\beta=\R,\C,\H$ and $\gamma=1,1,2$ for $\beta=1,2,4$.
\end{definition}

\begin{remark}
Note that~\eqref{prologue:outline:gauss} is an isotropic density. In the light of section~\ref{sec:prologue:square}, this will obviously be an important observation when considering products.
\end{remark}

Typically, the eigenvalues of a (square) Ginibre matrix are scattered in the complex plane due to the non-Hermiticity. More precisely, the number of real eigenvalues is zero almost surely for complex and quaternionic Ginibre matrices~\cite{Ginibre:1965,Mehta:2004}, but given a real Ginibre matrix then there is non-zero probability that all eigenvalues are real, however, this probability tends to zero as the matrix dimension increases~\cite{EKS:1994}. Rather than considering the generally complex eigenvalue spectra of Ginibre matrices, we can use $\widetilde X$ to construct other (Hermitian) matrix ensembles. We note that the ensemble of Gaussian non-Hermitian matrices, $\widetilde X$, may be considered as a building block for other Gaussian ensembles through the following constructions:

\begin{flexlabelled}{itshape}{1.75em}{.5em}{0pt}{2.25em}{0pt}

\item[Wishart ensemble.] The ensembles of positive semi-definite Hermitian matrix constructed as $\widetilde X^\dagger\widetilde X$ ($\beta=1,2,4$) are known as the Wishart ensembles~\cite{Wishart:1928}; their eigenvalues are identical to the squared singular values of $\widetilde X$ except for $\gamma(M-N)$ trivial zeros if $M>N$. Note that the Wishart ensembles are essentially equivalent to the so-called chiral Gaussian ensembles~\cite{Verbaarschot:1994}. We will return to Wishart matrices and their product generalisations in chapter~\ref{chap:singular}.

\item[Hermitian Gaussian ensembles.] If $X=\widetilde X$ is a square matrix, then we can construct ensembles of Hermitian matrices by $H:=(X+X^\dagger)/2$. These constitute the Gaussian orthogonal ($\beta=1$), unitary ($\beta=2$), and symplectic ($\beta=4$) ensembles (GOE, GUE, and GSE). We refer to~\cite{Mehta:2004,Forrester:2010} for an elaborate description.

\item[Elliptic Gaussian ensembles.] Let $\tau\in[-1,+1]$, if $X=\widetilde X$ is a square matrix, then we can construct an ensemble of matrices with the form:
\begin{equation}
E:=\sqrt{\frac{1+\tau}{2}}\frac{X+X^\dagger}{2}+\sqrt{\frac{1-\tau}{2}}\frac{X-X^\dagger}{2}.
\end{equation}
This is the so-called Gaussian elliptic ensemble~\cite{SCSS:1988}, which reduces to the (square) Ginibre ensemble for $\tau=0$ and to the Hermitian Gaussian ensembles for $\tau=1$. 
\end{flexlabelled}
Here, Wishart ensembles preserve isotropy, while the Hermitian and elliptic Gaussian ensembles explicitly break the symmetry from a bi-unitary to a (single) unitary invariance, i.e.
\begin{equation}
H\mapsto UHU^{-1}
\qquad\text{and}\qquad
E\mapsto UEU^{-1}
\qquad\text{for}\qquad
U\in\gU(N,\F_\beta)
\end{equation}
are still invariant transformations.

\begin{definition}\label{def:pro:induced}
Let $\nu$ be a non-negative constant, then the real, complex, and quaternionic induced Ginibre ensembles with charge $\nu$ are defined as matrices distributed according to the density
\begin{equation}
P_\nu^\beta(X)=\frac1{Z^\beta}\det (X^\dagger X)^{\beta\nu/2\gamma}\exp\Big[-\frac{\beta}{2\gamma}\tr X^\dagger X\Big]
\label{prologue:outline:induced}
\end{equation}
on the (square) matrix space $\F_\beta^{N\times N}$ with $\F_\beta=\R,\C,\H$ and $\gamma=1,1,2$ for $\beta=1,2,4$. Here, $Z^\beta$ is a normalisation constant.
\end{definition}

\begin{corollary}\label{thm:pro:ginibre-to-induced}
The product of independent $N\times N$ induced Ginibre matrices, $Y_n=X_n\cdots X_1$, with non-negative integer charges $\nu_1,\ldots,\nu_n$ and density
\begin{equation}
P^\beta(Y_n)=\bigg[\prod_{i=1}^n\int_{\F_\beta^{N\times N}} d^\beta X_i\,P_{\nu_i}^\beta(X_i)\bigg]\,\delta^\beta(X_n\cdots X_1-Y_n)
\label{prologue:outline:induced-product}
\end{equation}
has, up to a number of trivial zeros, the same spectral properties as a product of independent rectangular Ginibre matrices with dimensions as in proposition~\ref{thm:pro:rect-square}.
\end{corollary}

\begin{proof}
Follows from proposition~\ref{thm:pro:rect-square}.
\end{proof}

\begin{remark}
In all following chapters, we restrict our attention to products of induced Ginibre matrices, but due to corollary~\ref{thm:pro:ginibre-to-induced} this incorporates the general structure of products of \emph{rectangular} matrices.
\end{remark}

\begin{remark}
We have dropped the multi-index on the right hand side of~\eqref{prologue:outline:induced-product}, since the induced matrices are statistically isotropic and therefore commute in the weak sense. Furthermore, we can choose to order the charges $\nu_1\leq\cdots\leq\nu_n$ without loss of generality.
\end{remark}

The relation between the Wishart ensemble and densities of the form~\eqref{prologue:outline:induced} has been known for a longer time, but applications in relations to complex spectra are more recent. The induced Ginibre ensemble as a truncation of a rectangular Ginibre matrix was first presented in~\cite{FBKSZ:2012}, where also the name was coined. Their aim was to describe statistical properties of evolution operators in quantum mechanical systems. However, similar structures had appeared prior in the literature. In~\cite{Akemann:2001,Akemann:2002} an induced version of the elliptic ensemble was studied as a toy-model for quantum chromodynamics at finite chemical potential. A succeeding model describing the same system~\cite{Osborn:2004} had the clear physical benefit that it could be mapped exactly to the corresponding effective field theory. In this case the model included the product of two random 
matrices and the induced structure appeared (as in our case) because the product of two matrices can be square even though individual matrices are rectangular. As a consequence of its origin the charge $\nu$ was restricted to the integers, and it represented the topological charge (or winding number) on the field theoretical side.
 
We note that the induced density~\eqref{prologue:outline:induced} equivalently can be written as
\begin{equation}
P_\nu^\beta(X)=\frac1{Z^\beta}\exp\Big[-\frac{\beta}{2\gamma}\tr (X^\dagger X+\nu\log X^\dagger X)\Big],
\label{prologue:outline:induced-log}
 \end{equation}
 which illustrates the fact that $\nu$ represents the charge of a logarithmic singularity at the origin. Furthermore, the induced density is a special case of the more general class of ensembles,
 \begin{equation}
 P^\beta(X)=\frac1{Z^\beta}\exp\Big[-\frac{\beta}{2\gamma}\tr V(X^\dagger X)\Big],
 \label{pro:outline:feinberg-zee}
 \end{equation}
 where $V$ is a confining potential (subject to certain regularity conditions). If we are interested in the eigenvalues of the Hermitian matrix, $X^\dagger X$, then~\eqref{pro:outline:feinberg-zee} belongs to the canonical generalisation of the (Hermitian) Gaussian ensembles, which have been studied in great detail. On the other hand, if we are interested in the complex eigenvalues of $X$, then the density~\eqref{pro:outline:feinberg-zee} is of so-called Feinberg--Zee-type~\cite{FZ:1997}. The logarithmic singularity moves the microscopic neighbourhood of the origin out of the regime of known universality results.

%% file: singular.tex
\chapter{Wishart product matrices}
\label{chap:singular}

In this chapter, we will consider the statistical properties of the eigenvalues of a product generalisation of the Wishart ensemble. However, it seems appropriate to briefly recall the well-known structure of the (standard) Wishart ensemble before we embark on this description. We emphasise that our intention with this introductory remark is to recollect some well-known results rather than providing a comprehensive description. A more thorough account on the Wishart ensemble as well as references to relevant literature can be found in~\cite{Forrester:2010}.

For reasons which will become clear when we consider products, we restrict our discussion to the complex Wishart ensemble. As explained in section~\ref{sec:prologue:gauss}, we say that $XX^\dagger$ is a complex Wishart matrix if $X$ is an $N\times (N+\nu)$ complex random matrix distributed according to the density
\begin{equation}
P_\nu(X)=\Big(\frac1\pi\Big)^{N(N+\nu)}e^{-\tr X^\dagger X}.
\end{equation}
We are interested in properties of the eigenvalues $\lambda_i$ ($i=1,\ldots,N$) of the matrix $XX^\dagger$, i.e. the squared singular values of $X$. The joint probability density function for the eigenvalues is readily obtained by means of a singular value decomposition (proposition~\ref{thm:decomp:SVD}); after integration over the unitary groups we have the point process
\begin{equation}
\cP_\jpdf(\lambda_1,\ldots,\lambda_N)=\frac1\cZ\prod_{k=1}^Ne^{-\lambda_k}\lambda_k^\nu \prod_{1\leq i<j\leq N}(\lambda_j-\lambda_i)^2
\label{singular:intro:jpdf-laguerre}
,\qquad \cZ=\prod_{k=0}^{N-1}k!(k+\nu)!,
\end{equation}
where the eigenvalues $\lambda_i$ ($i=1,\ldots,N$) are restricted to the positive half-line (the Wishart matrix is positive definite). This point process may be thought of as log-gas i.e. $N$ logarithmically repulsive particles trapped by a confining potential $V(x)=x-\nu\log x$. Moreover, due to the special form of the eigenvalue repulsion we can study such ensembles using the method of orthogonal polynomials. In fact, the polynomials related to~\eqref{singular:intro:jpdf-laguerre} are the well-known Laguerre polynomials (for this reason the Wishart ensemble is often also referred to as the Wishart--Laguerre or simply Laguerre ensemble). The $k$-point correlation function is given by
\begin{equation}
R_k(\lambda_1,\ldots,\lambda_k):=\frac{N!}{(N-k)!}\bigg[\prod_{i=k+1}^N\int_0^\infty d\lambda_i\bigg] \cP_\jpdf(\lambda_1,\ldots,\lambda_N)
=\det_{1\leq i,j\leq N}\big[K_N(\lambda_i,\lambda_j)\big]
\end{equation}
with correlation kernel
\begin{align}
K_N(x,y)&=e^{-\frac{x+y}2}(xy)^{\nu/2}\sum_{k=0}^{N-1}\frac{\widetilde L_k^\nu(x)\widetilde L_k^\nu(y)}{k!(k+\nu)!} \label{singular:intro:CD}\\
&=\begin{cases}
\displaystyle{\frac{e^{-\frac{x+y}2}(xy)^{\nu/2}}{\Gamma[N]\Gamma[N+\nu]}\frac{\widetilde L_{k+1}^\nu(x)\widetilde L_k^\nu(y)-\widetilde L_k^\nu(x)\widetilde L_{k+1}^\nu(y)}{x-y}} & (x\neq y) \\
\displaystyle{\frac{e^{-x}x^{\nu}}{\Gamma[N]\Gamma[N+\nu]}\Big[\frac{d\widetilde L_{k+1}^\nu(x)}{dx}\widetilde L_k^\nu(x)-\frac{d\widetilde L_k^\nu(x)}{dx}\widetilde L_{k+1}^\nu(x)\Big]} & (x= y) \nn
\end{cases}.
\end{align}
Thus, the eigenvalues form a determinantal point process.
Here $\widetilde L_{k}^{\nu}(x)=(-1)^kk!L_k^{\nu_1}(x)$ denotes the (associated) Laguerre polynomial in monic normalisation. The latter equality in~\eqref{singular:intro:CD} is the celebrated Christoffel--Darboux formula which is a consequence of the three-step recurrence relation for orthogonal polynomials, see e.g.~\cite{Szego:1939}.

We are interested in the asymptotic properties as the matrix dimension, $N$, tends to infinity. Typically, we distinguish between two types of scaling regimes: (i) a macroscopic (or global) regime in which the eigenvalue interspacing decays with $N$, and (ii) a microscopic (or local) regime in which the eigenvalue interspacing is kept at order unity as $N$ tends to infinity. Both regimes are important for applications.

Without rescaling, the largest eigenvalue of a typical Wishart matrix is of order $N$; this suggests that the appropriate scaling for the macroscopic density (one-point correlation function) is $R_1(N\lambda)/N$. One way to get an explicit expression for this density is to calculate the asymptotic value of the integer-moments and then consider the corresponding moment problem. If $\nu$ is kept fixed as $N$ tends to infinity, then the moments converge to the Catalan numbers and consequently we have (see proposition~\ref{thm:singular:fuss-catalan} with $n=1$)
\begin{equation}
\lim_{N\to\infty}\frac{1}{N}R_1(Nx)=\rho_\MP(x)=\frac1{2\pi}\sqrt{\frac{4-x}{x}}\one_{0<x<4},
\label{singular:intro:mp-density}
\end{equation}
which is a special case of the Mar\v cenko--Pastur density. If $\nu/N\to\alpha>0$ for $N\to\infty$ then the moments are given as Narayana polynomials in $1+\alpha$ and there will be a macroscopic gap between the the origin and the infimum of the spectrum.

The interest in macroscopic properties is, of course, not restricted to the spectral density. Macroscopic (smoothed) correlation functions (often referred to as wide correlators) are of great interest as well (partly due to their relation with the variance of linear statistics, see e.g.~\cite{Beenakker:1997,Forrester:2010}). A description of such structures requires different techniques than those used in the description of the density (e.g. loop equations) and very little is known for products of matrices, even though wide correlators are well studied within the classical matrix ensembles~\cite{AJM:1990,BZ:1993,Beenakker:1994,KKP:1995,DZ:1996,Itoi:1997}. For this reason, we will not include a discussion of wide correlators in this introductory remark; a few comments regarding this open problem are given in section~\ref{sec:singular:discuss}.

\begin{figure}[tp]
\centering
\includegraphics{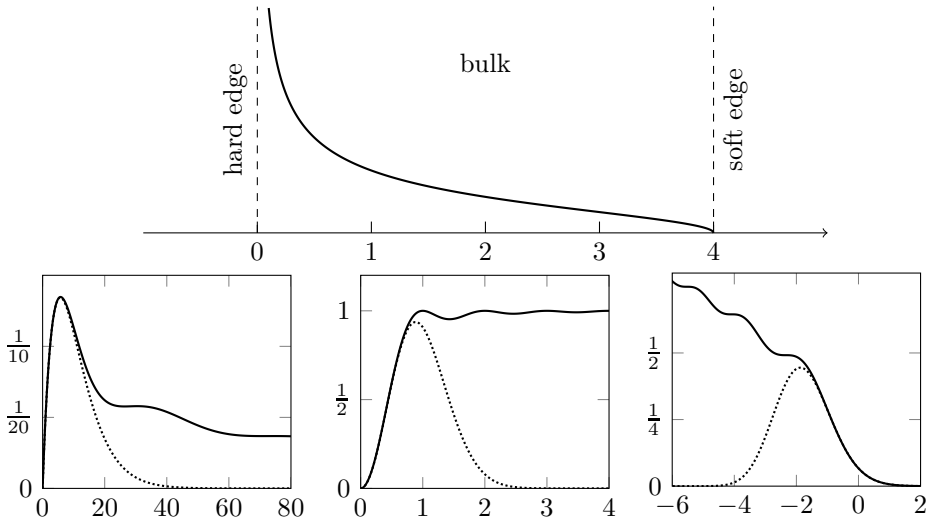}
\caption{The upper panel shows the macroscopic density~\eqref{singular:intro:mp-density}. The solid curves on lower panels show from left to right: (i) the microscopic density near the hard edge, $\rho_\hard^{\nu=1}(x)=K_\bessel^{\nu=1}(x,x)$, (ii) the microscopic two-point correlations in the bulk, $\rho_\bulk^2(0,x)=1-K_\sine(0,x)^2$, and (iii) the microscopic density near the hard edge, $\rho_\soft(x)=K_\airy(x,x)$. For comparison, the dotted curves show the density for the smallest eigenvalue~\cite{Forrester:1993}, the Wigner surmise, and the density for the largest eigenvalue~\cite{TW:1994}, respectively. Note that, unlike the bulk, the scaling at the hard and soft edge are traditionally not ``unfolded'' and, as a consequence, the relevant length scales differ considerably between the three lower panels.}
\label{fig:singular:intro}
\end{figure}

Instead, let us turn to the microscopic scaling regimes, where we are interested in correlations on the same scale as the eigenvalue interspacing.  A na\"ive use of~\eqref{singular:intro:mp-density} gives (it is assumed that $\nu$ is kept fixed for large $N$)
\begin{flalign}
&\text{(i)} & \int_0^{a/N^2} dx\,R_1(Nx)&\approx \frac2\pi\sqrt a, && \\
&\text{(ii)} & \int_{x_0}^{x_0+a/N} dx\,R_1(Nx)&\approx \rho_\MP(x_0)\,a, && \\
&\text{(iii)} & \int_{4-a/N^{2/3}}^4 dx\,R_1(Nx)&\approx \frac{a^{3/2}}{6\pi}, &&
\end{flalign}
for $x_0\in(0,4)$, $N\gg1$ and $a$ denotes a fixed constant. This approximation is too crude to be used in any actual computations but, nonetheless, it suggests that we have three different scaling regions: (i) a hard edge (close to the origin) characterised by an inverse square root divergence of the density, (ii) the bulk of the spectrum, and (iii) a soft edge (close to the supremum of the spectrum) characterised by a square root decay of the density. Rigorous results for these three scaling regimes can be obtained using known formulae for the asymptotics of Laguerre polynomials~\cite{Szego:1939}. The scaling limit for the $k$-point correlation function can for all three regions be written as
\begin{equation}
\lim_{N\to\infty}\frac1{(c_N^*)^k}R_k\Big(N\Big(x_\ast+\frac{x_1}{c_N^*}\Big),\ldots,N\Big(x_\ast+\frac{x_k}{c_N^*}\Big)\Big)=\det_{1\leq i,j\leq k}\big[K_*(x_i,x_j)\big],
\label{singular:intro:corr-scaling}
\end{equation}
where $c_N^*$ is an $N$-dependent constant and $K_*(x,y)$ is the limiting correlation kernel. 
\begin{enumerate}[(i)]
\item At the hard edge we set $x_*=0$ and $c^*_N=(2N)^2$ and the $k$-point correlation function~\eqref{singular:intro:corr-scaling} converges uniformly for $x_i$ ($i=1,\ldots,k$) in compact subsets of the positive half-line and the limiting kernel reads
\begin{equation}
 K_\bessel^\nu(x,y)=\frac{J_\nu(x^{1/2})y^{1/2}J_\nu'(y^{1/2})-x^{1/2}J_\nu'(x^{1/2})J_\nu(y^{1/2})}{2(x-y)},
\label{singular:intro:bessel}
\end{equation}
where $J_\nu(x)$ is the Bessel function.
\item
In the bulk we fix some $x_*\in(0,4)$ and set $c^*_N=N\rho_\MP(x_*)$, where $\rho_\MP(x_*)$ denotes the macroscopic density~\eqref{singular:intro:mp-density} at $x_*$. With this scaling~\eqref{singular:intro:corr-scaling} converges uniformly for $x_i$ ($i=1,\ldots,k$) in compact subsets of the real line and the celebrated sine kernel is obtained,
\begin{equation}
 K_\sine(x,y)=\frac{\sin\pi(x-y)}{\pi(x-y)}.
\label{singular:intro:sine}
\end{equation} 
Note that this scaling is independent of the exact value of $x_*$ and is in this way universal.
\item
At the soft edge we set $x_*=4$ and $c^*_N=(4N)^{2/3}$. In this case the limiting kernel becomes
\begin{equation}
K_\airy(x,y)=\frac{\Ai(x)\Ai'(y)-\Ai'(x)\Ai(y)}{x-y},
\label{singular:intro:airy}
\end{equation}
where $\Ai(x)$ is the Airy function.
\end{enumerate}
For obvious reasons, the three limiting kernels are known as the Bessel, the sine, and the Airy kernel, respectively. The limiting structures are visualised on figure~\ref{fig:singular:intro}.

Two ways of generalising the Wishart ensemble dominate the literature: either the assumption of i.i.d. entries is taken as fundamental and the Gaussian variables are replaced with another set of i.i.d. random variables, or the log-gas picture is taken as fundamental and we consider matrix densities of the form
\begin{equation}
P(H)=\frac1Ze^{-\tr V(H)},
\end{equation}
where $Z$ is a normalisation constant, $V(x)$ is a confining potential and $H$ is a Hermitian matrix. The former case has the benefit that such matrices are very easy to simulate numerically, but the assumption of independent entries is not always appropriate for applications. The latter case has the benefit that it may studied using the method of orthogonal polynomials; these ensembles are known as trace-class or invariant ensembles since they are invariant under unitary similarity transformations $H\mapsto UHU^{-1}$ (in the special case $H=X^\dagger X$, the matrix is isotropic in the sense of definition~\ref{def:prologue:isotropy}); this type of ensembles has a close connection with symmetry classification schemes such as~\cite{AZ:1997}. One of the interesting features of these generalisations is that the limiting behaviour of the Wishart matrix discussed above can be re-obtained for these more general classes of ensembles if certain constraints on either the distribution of the entries or the properties 
of the confining potential are imposed; we say that these limiting formulae are universal, see~\cite{Kuijlaars:2011,Erdos:2011} and references within for precise statements.

Here, we will consider a generalisation of the Wishart matrix constructed by considering products of independent Gaussian matrices. We note that the entries of such a product matrix will be strongly correlated even though the entries of the individual factors are independent. The product generalisation will not belong to the trace-class ensembles either. Thus, none of the classical universality results apply. It is intriguing to ask: \emph{Are the familiar universality classes valid in the study of product ensembles as well, or do new universality classes appear?} And: \emph{Do we have something similar to the method of orthogonal polynomials valid for the trace-class ensembles?}

The chapter is organised as follows: in section~\ref{sec:singular:exact} we describe exact results valid for finite matrix dimension and a finite number of factors, section~\ref{sec:singular:asymp} considers the asymptotic behaviour as the matrix dimension tends to infinity. In the final section we summarise known results and discuss some problems which remain open.

\section{Exact results for Wishart product matrices}
\label{sec:singular:exact}

Let $Y_n=X_n\cdots X_1$ be a product of independent random matrices, where each $X_i$ ($i=1,\ldots,n$) is an $N\times N$ induced Gaussian matrix, i.e. distributed with respect to the density~\eqref{prologue:outline:induced} with charge $\nu_i>-1$. Then we say that
\begin{equation}
Y_n^\dagger Y_n:=(X_n\cdots X_1)^\dagger X_n\cdots X_1,
\label{singular:intro:product}
\end{equation}
is a Wishart product matrix and we call the set of all such matrices the Wishart product ensemble. We are interested in the statistical properties of eigenvalues of such matrices.

\begin{definition}
Let $\lambda_i$ ($i=1,\ldots,N$) denote the eigenvalues of the Wishart product matrix $Y_n^\dagger Y_n$. The joint probability density function for the eigenvalues is defined as
\begin{equation}
\cP_\jpdf^\beta(x_1,\ldots,x_N):=\int d^\beta Y_nP^\beta(Y_n)\prod_{k=1}^N\delta(\lambda_{k}-x_k),
\label{singular:intro:jpdf-def}
\end{equation}
where $P^\beta(Y_n)$ is the matrix density of the product matrix $Y_n=X_n\cdots X_1$ given in terms of~\eqref{prologue:square:general} and~\eqref{prologue:outline:induced} with parameters $0\leq\nu_1\leq\cdots\leq\nu_n$.
\end{definition}

\begin{remark} We recall that if the constants $\nu_i$ ($i=1,\ldots,n$) which appear in~\eqref{prologue:outline:induced} are non-negative integers, then they may be interpreted as incorporating the structure of a product of rectangular Gaussian random matrices~\eqref{prologue:outline:gauss}. Moreover, Wishart product matrices are statistically isotropic, thus the ordering of the factors is irrelevant; for this reason we can choose the ordering $\nu_1\leq\cdots\leq\nu_n$ without loss of generality.
\end{remark}

It was realised in~\cite{AKW:2013,AIK:2013} that the complex ($\beta=2$) Wishart product ensemble is exactly solvable, meaning that we can find explicit expressions for the joint density~\eqref{singular:intro:jpdf-def} as well as any $k$-point correlations function for arbitrary matrix dimension $N$ and any number of factors $n$. Here, we follow~\cite{AIK:2013} and re-obtain the results for the joint probability density function (section~\ref{sec:singular:jpdf}) and the correlation functions in terms of a family of bi-orthogonal functions (section~\ref{sec:singular:correlations}). We will see in section~\ref{sec:singular:microscopic} that such exact results allow a study of universal behaviour at the microscopic scale in contrast to previous results which have been restricted to macroscopic densities.

\subsection{Joint probability density function}
\label{sec:singular:jpdf}

Our first step is to go from a joint density expressed in terms of matrix-valued variables~\eqref{singular:intro:jpdf-def} to an expression solely in terms of the eigenvalues, i.e. integrating out all irrelevant degrees of freedom. We have~\cite{AIK:2013}:

\begin{proposition}\label{thm:singular:jpdf-determinantal}
The eigenvalues of the complex Wishart product matrix form a determinantal point process on the positive half-line with joint density
\begin{equation}
\cP_\jpdf^{\beta=2}(x_1,\ldots,x_N)=\frac1{\cZ}
\prod_{1\leq i<j\leq N}(x_j-x_i)
\det_{1\leq i,j\leq N}\big[w_{j-1}(x_i)\big],
\label{singular:intro:jpdf}
\end{equation}
where $\cZ$ is a normalisation constant and $w_{j-1}(x)\ (j=1,\ldots,N)$ is a family of positive weight functions given by
\begin{equation}
\cZ=N!\prod_{k=0}^{N-1}k!\prod_{i=1}^n\Gamma[\nu_i+k+1]
\quad\text{and}\quad
w_{j}(x)=\MeijerG{n}{0}{0}{n}{-}{\nu_1+j,\nu_2,\ldots,\nu_n}{x},
\label{singular:intro:weight+normal}
\end{equation}
respectively. Here $N$ and $n$ are positive integers denoting the matrix dimension and the number of factors, while the charges $0\leq\nu_1\leq\cdots\leq\nu_n$ are constants.
\end{proposition}

\begin{proof}
The matrix density for complex ($\beta=2$) Wishart product matrices follows directly from~\eqref{prologue:square:density-product} and~\eqref{prologue:outline:induced}. Up to a factor of normalisation, we have
\begin{multline}
P^{\beta=2}(Y_n)\propto(\det Y_n^\dagger Y_n)^{\nu_{n}}\\
\times\bigg[\prod_{i=1}^{n-1}\int_{\gGL(N,\C)} d\mu(Y_i) 
(\det Y_{i}^\dagger Y_{i}^{})^{\nu_{i}-\nu_{i+1}}
e^{-\tr (Y_{i+1}^\dagger Y_{i+1}^{}) (Y_{i}^\dagger Y_i^{})^{-1}}\bigg]
e^{-\tr Y_1^\dagger Y_1^{}},
\label{singular:intro:density-Y}
\end{multline}
where $d\mu(Y)=d^2Y/(\det Y^\dagger Y)^N$ is the invariant measure on the group of complex $N\times N$ invertible matrices, $\gGL(N,\C)$. Note that we have dropped the multi-index from~\eqref{prologue:square:density-product}, since the ordering of the individual matrices is irrelevant for the properties of the product matrix $Y_n$ due the weak commutation relation.

We are interested in the eigenvalues of the Wishart product matrix~\eqref{singular:intro:product}, and we therefore change variables. We parametrise each matrix $Y_i$ ($i=1,\ldots,n$) in~\eqref{singular:intro:density-Y} using a singular value decomposition,
\begin{equation}
Y_i=U_i\Lambda_i^{1/2}V_i,
\end{equation}
where $\Lambda_i=\diag(\lambda_{i,1},\ldots,\lambda_{i,N})$ with $\lambda_{i,j}$ ($j=1,\ldots,N$) denoting the eigenvalues of $Y_i^\dagger Y_i^{}$, while $U_i$ and $V_i$ are unitary matrices. It is known that this change of variables gives (see proposition~\ref{thm:decomp:SVD})
\begin{equation}
d\mu(Y_i)=d\mu(U_i)d\mu(V_i)\prod_{j=1}^N\frac{d\lambda_{i,j}}{\lambda_{i,j}^N}\prod_{1\leq k<\ell\leq N} (\lambda_{i,\ell}-\lambda_{i,k})^2
\label{singular:intro:SVD-jacobian}
\end{equation}
with $d\mu(U_i)$ and $d\mu(V_i)$ denoting the Haar measures on $\gU(N)$ and $\gU(N)/U(1)^N$, respectively. For notational simplicity, we also introduce the symbol
\begin{equation}
\Delta_N(\Lambda_i):=\det_{1\leq k,\ell\leq N}\big[\lambda_{i,\ell}^{k-1}\big]=\prod_{1\leq k<\ell\leq N} (\lambda_{i,\ell}-\lambda_{i,k})
\end{equation}
for the Vandermonde determinant.

With the above given parametrisation and the explicit expression for the matrix density~\eqref{singular:intro:density-Y} we may write the joint density~\eqref{singular:intro:jpdf-def} as
\begin{align}
&\cP_\jpdf^{\beta=2}(x_1,\ldots,x_N)\propto 
\prod_{k=1}^N\int_0^\infty d\lambda_{n,k}\,\lambda_{n,k}^{\nu_{n}}\delta(\lambda_{n,k}-x_k)\Delta_N(\Lambda_n)^2 \label{singular:intro:jpdf-step1} \\
&\times\bigg[\prod_{i=1}^{n-1}\prod_{j=1}^N\int_0^\infty \frac{d\lambda_{i,j}}{\lambda_{i,j}^N} \lambda_{i,j}^{\nu_{i}-\nu_{i+1}}\Delta_N(\Lambda_i)^2
\int_{\gU(N)} d\mu(U_{i+1})e^{-\tr U\Lambda_{i+1}U^{-1}\Lambda_i^{-1}}\bigg]\prod_{\ell=1}^Ne^{-\lambda_{1,\ell}}, \nn
\end{align}
The remaining unitary integrals are of Harish-Chandra--Itzykson--Zuber type~\cite{HC:1957,IZ:1980} which allow explicit evaluation
\begin{equation}
\int_{\gU(N)} d\mu(U_{i+1})e^{-\tr U\Lambda_{i+1}U^{-1}\Lambda_i^{-1}}\propto
\frac{\det\limits_{1\leq k,\ell\leq N}\left[e^{-\lambda_{i+1,k}/\lambda_{i,\ell}}\right]}
{\Delta_{N}(\Lambda_{i+1})\Delta_{N}(\Lambda_{i}^{-1})}.
\label{singular:intro:IZ-integral}
\end{equation}
Using this integration formula~\eqref{singular:intro:IZ-integral} in the joint density~\eqref{singular:intro:jpdf-step1} yields
\begin{align}
&\cP_\jpdf^{\beta=2}(x_1,\ldots,x_N)\propto 
\prod_{k=1}^N\int_0^\infty d\lambda_{n,k}(\lambda_{n,k})^{\nu_{n}}\delta(\lambda_{n,k}-x_k)\Delta_{N}(\Lambda_{n}) \label{singular:intro:jpdf-step2}\\
&\times\bigg[\prod_{i=1}^{n-1}\prod_{j=1}^N\int_0^\infty \frac{d\lambda_{i,j}}{\lambda_{i,j}} (\lambda_{i,j})^{\nu_{i}-\nu_{i+1}}
\det_{1\leq k,\ell\leq N}\left[e^{-\lambda_{i+1,k}/\lambda_{i,\ell}}\right]\bigg]
\det_{1\leq k,\ell\leq N}\left[ \lambda_{1,\ell}^{k-1}e^{-\lambda_{1,\ell}} \right]. \nn
\end{align}
Here the integrals on the second line can be absorbed into a single determinant by successive use of Andr\'eief's integration formula~\cite{Andreief:1883,deBruijn:1955}, while the integration over $\lambda_{n,k}$ ($k=1,\ldots,n$) may be performed due to the delta function. This yields~\eqref{singular:intro:jpdf} with the weight functions given as $n$-fold integrals
\begin{equation}
w_{k}(x)=\int_0^\infty d\lambda_n \lambda_n^{\nu_n}
\bigg[\prod_{i=1}^{n-1}\int_0^\infty \frac{d\lambda_{i}}{\lambda_i} \lambda_{i}^{\nu_{i}-\nu_{i+1}}e^{-\lambda_{i+1}/\lambda_{i}}\bigg]
\lambda_{1}^{k}e^{-\lambda_{1}}\delta(\lambda_n-x),
\label{singular:intro:weight-nfold}
\end{equation}
hence it only remains to evaluate the integral~\eqref{singular:intro:weight-nfold} and determine the normalisation.

We note that after proper normalisation~\eqref{singular:intro:weight-nfold} can be interpreted as a probability density for a product of $n$ gamma distributed random scalars (expressed as an ($n-1$)-fold multiplicative convolution on the positive half-line). This problem may be solved using the Mellin transform and its inverse, see e.g.~\cite{Springer:1979} (similar to the discussion in section~\ref{sec:prologue:scalar}). Alternatively, we can rewrite the exponentials in~\eqref{singular:intro:weight-nfold} as Meijer $G$-functions and then perform the integrals using~\eqref{special:meijer:int-meijer-meijer}. Either way, we find the Meijer $G$-function formulation in~\eqref{singular:intro:weight+normal}, which essentially means that we have replaced the $n$-fold integral with a single contour integral. The normalisation constant $\cZ$ is found by integrating out $x_i$ ($i=1,\ldots,N$) in~\eqref{singular:intro:jpdf}, which can be done by applying Anr\'eief's integration formula~\cite{Andreief:1883,deBruijn:1955}
and then performing the integral inside 
the determinant using~\eqref{special:meijer:meijer-moment}.
\end{proof}

\begin{remark}
The reason we have restricted ourselves to the complex ($\beta=2$) case is the explicit use of the Itzykson--Zuber integral~\eqref{singular:intro:IZ-integral}. We have no analogues formulae for integrals over the orthogonal ($\beta=1$) and unitary symplectic ($\beta=4$) group, although there still exist expressions in terms of Zonal polynomials~\cite{GR:1987}. One way to circumvent such difficulties and obtain large-$N$ behaviour could be the use of supersymmetric techniques~\cite{Kieburg:2015}.  
\end{remark}

\subsection{Correlations and bi-orthogonal functions}
\label{sec:singular:correlations}

The point process formed by the eigenvalues of the Wishart product matrix belongs to a special class of bi-orthogonal ensembles~\cite{Muttalib:1995,Borodin:1998} which recently has been coined polynomial ensembles~\cite{KS:2014,Kuijlaars:2015}. Before we continue our study of Wishart product matrices, let us briefly consider a generic polynomial ensemble of $N$ points on an interval $(a,b)\subseteq\R$ with joint density
\begin{equation}
\cP_\jpdf(x_1,\ldots,x_N)=\frac1{\cZ}
\prod_{1\leq i<j\leq N}(x_j-x_i)
\det_{1\leq i,j\leq N}\big[w_{j-1}(x_i)\big],
\label{singular:correlation:jpdf-general}
\end{equation}
where $\{w_{j-1}(x)\}$ is a family of weight functions. It is assumed that~\eqref{singular:correlation:jpdf-general} is a well-defined probability density, which implies certain constraints on the weight functions, see~\cite{Kuijlaars:2015} for a brief discussion.

Typically, we are interested in different large-$N$ limits, where universal behaviour is expected. For this reason it is convenient to introduce the $k$-point correlation function defined by (see e.g.~\cite{Mehta:2004})
\begin{equation}
R_k(x_1,\ldots,x_k):=\frac{N!}{(N-k)!}\bigg[\prod_{i=k+1}^N\int_a^b dx_i\bigg] \cP_\jpdf(x_1,\ldots,x_N).
\label{singular:correlation:correlation-def}
\end{equation}
To find an explicit expression for this correlation functions, we might exploit the fact that the determinants in~\eqref{singular:correlation:jpdf-general} are invariant under permutations of rows or columns.
Suppose we have managed to bi-orthogonalise~\eqref{singular:correlation:correlation-def}, i.e. we have found a family of (monic) functions
\begin{subequations}
\label{singular:correlation:monic}
\begin{align}
p_j(x)-x^j&\in\Span(1,\ldots,x^{j-1}),\\
\phi_j(x)-w_j(x)&\in\Span(w_0(x),\ldots,w_{j-1}(x)),
\end{align}
\end{subequations}
which satisfy the bi-orthogonality relation
\begin{equation}
\inner{p_i,\phi_j}:=\int_a^b dx\, p_i(x)\phi_j(x)=h_i\delta_{ij}
\label{singular:correlation:orthogonal}
\end{equation}
where $h_i\ (i=1,2,\ldots)$ are positive constants. If we construct a kernel,
\begin{equation}
K_N(x,y):=\sum_{k=0}^{N-1}\frac{p_k(x)\phi_k(y)}{h_k},
\label{singular:correlations:kernel-general}
\end{equation}
then it follows from the bi-orthogonality~\eqref{singular:correlation:orthogonal} that the kernel satisfies
\begin{equation}
\int_a^bdx\,K_N(x,x)=N
\quad\text{and}\quad
\int_a^bdu\,K_N(x,u)K_N(u,y)=K_N(x,y).
\end{equation}
Following the same arguments as for orthogonal ensembles shows that the $k$-point correlation function~\eqref{singular:correlation:correlation-def} is given by~\cite{Muttalib:1995,Borodin:1998}
\begin{equation}
R_k(x_1,\ldots,x_k)=\det_{1\leq i,j\leq k}\big[K_N(x_i,x_j)\big]
\label{singular:correlation:correlation-kernel}
\end{equation}
with correlation kernel given by~\eqref{singular:correlations:kernel-general}. Note that the correlation kernel is not unique, e.g. if $g(x)$ is a non-zero function on the relevant domain then the correlation function~\eqref{singular:correlation:correlation-kernel} is unaffected by a gauge transfermation $K_N(x,y)\mapsto (g(x)/g(y))  K_N(x,y)$.

A major difference between orthogonal and bi-orthogonal ensembles is that orthogonal polynomials satisfy a three-step recurrence formula (see e.g.~\cite{Szego:1939,Ismail:2005}), while there is no guarantee for any recurrence relation for the bi-orthogonal functions. The three-step recurrence relation is an integral part of several proofs related to orthogonal ensembles, thus different techniques are sometimes required for bi-orthogonal ensembles. This also opens the possibility for new universality classes and we will see in section~\ref{sec:singular:microscopic} that new classes indeed are present at the hard edge.

A useful property of polynomial ensembles~\eqref{singular:correlation:jpdf-general} are that the bi-orthogonal functions are related to the expectation (with respect to the joint density) of the characteristic polynomial and its inverse,
\begin{subequations}
\label{singular:correlations:average-poly}
\begin{align}
\average[\bigg]{\prod_{i=1}^N(z-x_i)}&=p_N(z), && z\in\C, \\
\average[\bigg]{\prod_{i=1}^N\frac1{(z-x_i)}}&=\int_a^bdx\frac{\phi_{N-1}(x)}{z-x}, && z\in\C\setminus[a,b].
\end{align}
\end{subequations}
Here, the former formula is often referred to as the Heine formula~\cite{Szego:1939} and gives the bi-orthogonal polynomial of order $N$, while the latter gives the Cauchy transform of the other bi-orthogonal function of order $N-1$. These formulae follow exactly like in the case of orthogonal ensembles since the joint density~\eqref{singular:correlation:jpdf-general} contains a Vandermonde determinant, but they will not generally hold for the larger class of bi-orthogonal ensembles (clearly the expectation of the characteristic polynomial must yield another polynomial). The canonical generalisations of~\eqref{singular:correlations:average-poly} to products of characteristic polynomials are valid as well; again the proofs follow the usual procedure~\cite{BH:2000,FS:2003,BDS:2003,DF:2006}.

The main purpose of this section is to bi-orthogonalise the polynomial ensemble~\eqref{singular:intro:jpdf} and thereby find an explicit expression for the $k$-point correlation function of the eigenvalues of a Wishart product matrix. The bi-orthogonal functions and the corresponding correlation kernel were originally obtained in~\cite{AKW:2013,AIK:2013} using a link to a two-matrix model and the Eynard--Mehta theorem~\cite{EM:1998}. An alternative, but closely related, derivation of the same result was given later in~\cite{KZ:2014}. Here, we present a slightly modified version of the proof from~\cite{AIK:2013}.

\begin{proposition}\label{thm:singular:biorthogonalise}
The complex Wishart product ensemble~\eqref{singular:intro:jpdf} is bi-orthogonalised by
\begin{subequations}
\label{singular:correlations:bi-system}
\begin{align}
p_k^n(x)&=(-1)^kh_k\,\MeijerG{1}{0}{1}{n+1}{k+1}{0,-\nu_1,\ldots,-\nu_n}{x}, \label{singular:correlations:p} \\
\phi_k^n(x)&=(-1)^k\MeijerG{n}{1}{1}{n+1}{-k}{\nu_1,\ldots,\nu_n,0}{x}, \label{singular:correlations:phi} \\
h_k^n&=k!\prod_{\ell=1}^n\Gamma[k+\nu_\ell+1]. \label{singular:correlations:norm}
\end{align}
\end{subequations}
With parameters as in proposition~\ref{thm:singular:jpdf-determinantal}.
\end{proposition}

\begin{proof}
In order to find the bi-orthogonal functions, we start by computing the bi-moments using~\eqref{special:meijer:meijer-moment},
\begin{equation}
M_{i,j}^n:= \int_0^\infty dx\, x^i w_{j}^n(x)= \Gamma[i+j+\nu_1+1]\prod_{\ell=2}^{n}\Gamma[i+\nu_\ell+1].
\label{singular:correlations:bimoment}
\end{equation}
Due to the total positivity of the bi-moment matrix, there exists a family of functions~\eqref{singular:correlation:monic} which satisfies the bi-orthogonality relation~\eqref{singular:correlation:orthogonal}. Using Cramer's rule, the bi-orthogonal functions can be expressed in terms of the bi-moments,
\begin{equation}
p_k^n(x)=\frac{1}{D_{k-1}^n}
\det_{\substack{i=0,\ldots,k\\j=0,\ldots,k-1}}
\begin{bmatrix}
M_{i,j}^n \bigg\vert\, x^i
\end{bmatrix},\qquad
\phi_k^n(x)=\frac{1}{D_{k-1}^n}
\det_{\substack{i=0,\ldots,k-1\\j=0,\ldots,k}}
\bigg[\begin{array}{@{}cc @{}}
M_{i,j}^n \\ \hline
w_j^n(x)
\end{array}\bigg],
\label{singular:correlations:function-step1}
\end{equation}
with
\begin{equation}
D_k^n:=\det_{0\leq i,j\leq k} \big[M_{ij}^n\big]=\prod_{i=0}^ki!\prod_{\ell=1}^n\Gamma[i+\nu_\ell+1].
\end{equation}
The squared norms are given by $h_k^n=D_k^n/D_{k-1}^n$ which immediately yields~\eqref{singular:correlations:norm}. It remains to show that the bi-orthogonal functions~\eqref{singular:correlations:function-step1} are equivalent to those in~\eqref{singular:correlations:bi-system}.

We observe that the last gamma functions in~\eqref{singular:correlations:bimoment} only depend on a single index. For this reason, we may think of their contribution as multiplication by a diagonal matrix within the determinants in~\eqref{singular:correlations:function-step1}. We pull this contribution out of the determinant, which yields
\begin{subequations}
\label{singular:correlations:function-step2}
\begin{align}
p_k^n(x)&=\frac{\prod_{i=0}^k\prod_{\ell=2}^n\Gamma[i+\nu_\ell+1]}{D_{k-1}^n}
\det_{\substack{i=0,\ldots,k\\j=0,\ldots,k-1}}
\begin{bmatrix}
 M_{i,j}^{n=1} \bigg\vert\, \dfrac{x^i}{\prod_{\ell=2}^n\Gamma[i+\nu_\ell+1]}
\end{bmatrix},\\
\phi_k^n(x)&=\frac{\prod_{i=0}^{k-1}\prod_{\ell=2}^n\Gamma[i+\nu_\ell+1]}{D_{k-1}^n}
\det_{\substack{i=0,\ldots,k-1\\j=0,\ldots,k}}
\bigg[\begin{array}{@{}cc @{}}
M_{i,j}^{n=1} \\ \hline
w_j^n(x)
\end{array}\bigg].
\end{align}
\end{subequations}
The determinants are now related to the $n=1$ case, i.e. the ordinary Wishart ensemble. It is well-known (see e.g.~\cite{Forrester:2010}) that the Wishart ensemble may be orthogonalised by means of Laguerre polynomials, which in monic normalisation reads
\begin{equation}
\widetilde L_k^{\nu_1}(x)=(-1)^kk!L_k^{\nu_1}(x)
=\sum_{j=0}^k\frac{(-1)^{k+j}k!}{(k-j)!}\frac{\Gamma[k+\nu_1+1]}{\Gamma[j+\nu_1+1]}\frac{x^j}{j!}.
\label{singular:correlations:laguerre}
\end{equation}
We can use this fact together with the replacements $x^i\mapsto x^i/\prod_{\ell=2}^n\Gamma[i+\nu_\ell+1]$ and $x^i\mapsto w_i^n(x)$ to evaluate the determinants in~\eqref{singular:correlations:function-step2}. We find
\begin{subequations}
\begin{align}
p_k^n(x)&=\sum_{j=0}^k\frac{(-1)^{k+j}k!}{(k-j)!}\prod_{\ell=1}^n\frac{\Gamma[k+\nu_\ell+1]}{\Gamma[j+\nu_\ell+1]}\frac{x^j}{j!}, \\
\phi_k^n(x)&=\sum_{j=0}^k\frac{(-1)^{k+j}k!}{(k-j)!}\frac{\Gamma[k+\nu_1+1]}{\Gamma[j+\nu_1+1]}\frac{w_j^n(x)}{j!}. \label{singular:correlations:function-step3-phi}
\end{align}
\end{subequations}
The polynomial $p_k^n(x)$ is immediately recognized as a hypergeometric function
\begin{equation}
p_k^n(x)=(-1)^k\prod_{\ell=1}^n\frac{\Gamma[k+\nu_\ell+1]}{\Gamma[\nu_\ell+1]}\hypergeometric{1}{n}{-k}{\nu_1+1,\ldots,\nu_n+1}{x}
\end{equation}
or equivalent the Meijer $G$-function~\eqref{singular:correlations:p}, see appendix~\ref{app:special}. In order to get a compact expression for the function $\phi_k^n(x)$, we first rewrite the weight function as an integral
\begin{equation}
w_k^n(x)=\int_0^\infty \frac{dy}{y}\, y^{\nu_1+k}e^{-y}\MeijerG{n-1}{0}{0}{n-1}{-}{\nu_2,\ldots,\nu_n}{\frac{x}{y}}.
\end{equation}
We insert this expression into~\eqref{singular:correlations:function-step3-phi} and use~\eqref{singular:correlations:laguerre} to write
\begin{equation}
\phi_k^n(x)=(-1)^kk!\int_0^\infty \frac{dy}{y}\, y^{\nu_1}e^{-y}L_k^{\nu_1}(y)\MeijerG{n-1}{0}{0}{n-1}{-}{\nu_2,\ldots,\nu_n}{\frac{x}{y}}.
\end{equation}
After rewriting the Laguerre polynomial as a Meijer $G$-function~\eqref{special:meijer:meijer-other}, we perform this integral using~\eqref{special:meijer:int-meijer-meijer} which yields~\eqref{singular:correlations:phi} and concludes the proof.
\end{proof}

\begin{remark}
It was pointed out in~\cite{KZ:2014} that the bi-orthogonal functions~\eqref{singular:correlations:bi-system} satisfy a stronger requirement than bi-orthogonality: they are multiple orthogonal in the sense of~\cite{Ismail:2005,Kuijlaars:2010}. Multiple orthogonality is useful since it can be used to find recurrence relations for the functions. However, the same multiple orthogonality has not been observed for other product ensembles~\cite{Forrester:2014b,KS:2014,KKS:2015} and we will not pursue that direction in this thesis. 
\end{remark}

\begin{corollary}
The $k$-point correlation function describes a determinantal point process~\eqref{singular:correlation:correlation-kernel} on the positive half-line with kernel
\begin{equation}
K_N^n(x,y)=\sum_{k=0}^{N-1}\MeijerG{1}{0}{1}{n+1}{k+1}{0,-\nu_1,\ldots,-\nu_n}{x}\MeijerG{n}{1}{1}{n+1}{-k}{\nu_1,\ldots,\nu_n,0}{y}
\label{singular:correlations:kernel}
\end{equation}
with parameters as in proposition~\ref{thm:singular:biorthogonalise}.
\end{corollary}

\begin{proof}
Follows from proposition~\ref{thm:singular:biorthogonalise} and~\eqref{singular:correlations:kernel-general}.
\end{proof}

For later use, we also reproduce an alternative expression for the kernel originally given by Kuijlaars and Zhang \cite{KZ:2014}.

\begin{corollary}
The kernel~\eqref{singular:correlations:kernel} can be written as
\begin{equation}
K_N^n(x,y)=\int_0^1 du\,\MeijerG{1}{0}{1}{n+1}{N}{0,-\nu_1,\ldots,-\nu_n}{ux}\MeijerG{n}{1}{1}{n+1}{-N}{\nu_1,\ldots,\nu_n,0}{uy}.
\label{singular:correlations:kernel-int}
\end{equation}
\end{corollary}

\begin{proof}
The kernel~\eqref{singular:correlations:kernel} has a natural representation as a double contour integral through the definition of the Meijer $G$-function~\eqref{special:meijer:meijer-def}. We have
\begin{equation}
K_N^n(x,y)=\frac{1}{(2\pi i)^2}\int_{-1-i\infty}^{-1+i\infty} dt\oint_{\Sigma_k} ds\, x^sy^t\frac{\Gamma[-s]}{\Gamma[1+t]}
\prod_{\ell=1}^n\frac{\Gamma[\nu_\ell-t]}{\Gamma[\nu_\ell+1+s]}
\sum_{k=0}^{N-1}\frac{\Gamma[k+1+t]}{\Gamma[k+1-s]},
\label{singular:correlations:kernel-step1}
\end{equation}
where $\Sigma_k$ is a closed contour which encircles $0,\ldots,k$ in the negative direction such that $\Re s>-1$ for all $s\in\Sigma_k$. Using the recursive property $\Gamma[z+1]=z\Gamma[z]$, it is straightforward to show that
\begin{equation}
(t+s+1)\frac{\Gamma[k+t]}{\Gamma[k-s]}=\frac{\Gamma[k+1+t]}{\Gamma[k-s]}-\frac{\Gamma[k+t]}{\Gamma[k-1-s]}.
\end{equation}
This identity turns~\eqref{singular:correlations:kernel-step1} into a telescopic sum, hence
\begin{multline}
K_N^n(x,y)=\frac{1}{(2\pi i)^2}\int_{-1-i\infty}^{-1+i\infty} dt\oint_{\Sigma_k} ds\, \frac{x^sy^t}{s+t+1}\frac{\Gamma[-s]}{\Gamma[1+t]}
\frac{\Gamma[N+1+t]}{\Gamma[N-s]}\prod_{\ell=1}^n\frac{\Gamma[\nu_\ell-t]}{\Gamma[\nu_\ell+1+s]}\\
-\frac{1}{(2\pi i)^2}\int_{-1-i\infty}^{-1+i\infty} dt\oint_{\Sigma_k} ds\, \frac{x^sy^t}{s+t+1}
\prod_{\ell=1}^n\frac{\Gamma[\nu_\ell-t]}{\Gamma[\nu_\ell+1+s]}.
\label{singular:correlations:kernel-step2}
\end{multline}
Note that $s+t+1$ is non-zero with the above given contours. The integral on the second line in~\eqref{singular:correlations:kernel-step2} is zero, since $\Sigma_k$ encircles no singularities. In order to re-obtain an expression in terms of Meijer $G$-functions, we use that
\begin{equation}
\frac{x^sy^t}{s+t+1}=\int_0^1 du (ux)^s(uy)^t.
\end{equation}
Inserting this into the first line in~\eqref{singular:correlations:kernel-step2} and using the definition of the Meijer $G$-function~\eqref{special:meijer:meijer-def} completes the proof.
\end{proof}


\section{Asymptotic results for Wishart product matrices}
\label{sec:singular:asymp}

\subsection{Macroscopic density}

The macroscopic density for the Wishart product ensemble has been obtained previously in the literature (without knowledge about the structure presented above) using planar diagrams, free probability, or probabilistic techniques~\cite{Muller:2002,BJLNS:2010,BG:2010,BBTCC:2011,PZ:2011}. In this section we show that the finite-$N$ expression is in agreement with these results. In order to do so, we first calculate the moments and then look at the corresponding moment problem.

We first look at the moment at finite-$N$ originally presented in~\cite{AIK:2013}.

\begin{lemma}
Let $s$ be a positive integer, then the $s$-th moment for the spectral density (one-point correlation function) reads
\begin{equation}
\int_0^\infty dx\,R_1^n(x)x^s
=\frac1{s!}\sum_{j=0}^{s-1}(-1)^j\binom{s-1}{j}\frac{\Gamma[N-j+s]}{\Gamma[N-j]}\prod_{\ell=1}^n\frac{\Gamma[N+\nu_\ell-j+s]}{\Gamma[N+\nu_\ell-j]}
\label{singular:asymp:moments-finite}
\end{equation}
with parameters as in proposition~\ref{thm:singular:jpdf-determinantal}. The $0$-th moment is normalised to $N$ in the present notation.
\end{lemma}

\begin{proof}
We use the kernel~\eqref{singular:correlations:kernel-int} in the formula for the one-point correlation function. Writing the first Meijer $G$-function as a polynomial, we find
\begin{multline}
\int_0^\infty dx\,R_1^n(x)x^s=\sum_{j=0}^{N-1}\frac{(-1)^j}{j!(N-j-1)!}\prod_{\ell=1}^n\frac1{\Gamma[\nu_\ell+j+1]}\\
\times\int_0^\infty dx\,x^{s+j}\int_0^1du\,u^j\MeijerG{n}{1}{1}{n+1}{-N}{\nu_1,\ldots,\nu_n,0}{ux},
\end{multline}
which is continuous in $s$ for $s>-\nu_1-1$. The integrals over the Meijer $G$-function can be performed using~\eqref{special:meijer:int-meijer-meijer} and~\eqref{special:meijer:meijer-moment}. After a reordering of the sum, we obtain
\begin{equation}
\int_0^\infty dx\,R_1^n(x)x^s
=\frac1{s}\sum_{j=0}^{N-1}\frac{(-1)^j}{j!\Gamma[s-j]}\frac{\Gamma[N-j+s]}{\Gamma[N-j]}\prod_{\ell=1}^n\frac{\Gamma[N+\nu_\ell-j+s]}{\Gamma[N+\nu_\ell-j]}
\end{equation}
with integer values of $s$ uniquely determined by the corresponding limits. If $s$ is a positive integer then the sum terminates for $j\geq s$ and~\eqref{singular:asymp:moments-finite} is obtained, while $s\to0$ yields the normalisation.
\end{proof}


\begin{definition}The Fuss--Catalan distribution is defined on the positive half-line through the density
\begin{equation}
\rho_\fuss^n(x):=
\frac1{\sqrt{2\pi}}\frac{n^{n-3/2}}{(n+1)^{n+1/2}}\MeijerG[\bigg]{n}{0}{n}{n}{\frac1n,0,-\frac1n,\ldots,-\frac{n-2}n}{-\frac1{n+1},-\frac2{n+1},\ldots,-\frac{n}{n+1}}{\frac{n^n}{(n+1)^{n+1}}x}.
\label{singular:asymp:density-fc}
\end{equation}
\end{definition}

\begin{remark}
The Fuss--Catalan distribution is named after its moments, 
\begin{equation}
\int_0^\infty dx\,\rho_\fuss^n(x)x^s=\frac1{ns+1}\binom{(n+1)s}{s},
\label{singular:asymp:moments-fc}
\end{equation}
which for integer values are the so-called Fuss--Catalan numbers. From general properties of the Meijer $G$-function, we know that~\eqref{singular:asymp:density-fc} is an analytical function on the positive half-line except at $x=(n+1)^{n+1}/n^n$. Furthermore, we note that the complex contour, which appears in the definition of the Meijer $G$-function~\eqref{singular:asymp:density-fc}, depends on whether $x$ is larger or smaller than $K_n:=(n+1)^{n+1}/n^n$. For $x>K_n$ the contour encloses no singularities, thus the Fuss--Catalan distribution has compact support on $[0,K_n]$.
\end{remark}


\begin{proposition}\label{thm:singular:fuss-catalan}
For large matrix dimension, the macroscopic spectral density converges to the Fuss--Catalan density~\eqref{singular:asymp:density-fc}, i.e. 
\begin{equation}
\lim_{N\to\infty}N^{n-1} R_1^n(N^nx)=\rho_\fuss^n(x)
\end{equation}
with $n$ and $\nu_\ell$ ($\ell=1,\ldots,n$) fixed. 
\end{proposition}

\begin{proof}
First, we will show that the moments of the spectral density converges to the Fuss--Catalan numbers. It follows from~\eqref{singular:asymp:moments-finite} that the rescaled moments are 
\begin{equation}
\int_0^\infty dxN^{n-1}{R_1^n(N^nx)}x^s
=\frac1{N^{ns+1}}\frac1{s!}\sum_{j=0}^{s-1}(-1)^j\binom{s-1}{j}\frac{\Gamma[N-j+s]}{\Gamma[N-j]}\prod_{\ell=1}^n\frac{\Gamma[N+\nu_\ell-j+s]}{\Gamma[N+\nu_\ell-j]}
\label{singular:asymp:moments-rescaled}
\end{equation}
for $s=1,2,\ldots$ (the $0$-th moment is unity). In order to evaluate the sum, we notice that
\begin{equation}
\sum_{j=0}^{s-1}(-1)^jj^k\binom{s-1}{j}=(-1)^{s-1}(s-1)!\,\delta_{k,s-1},
\label{singular:asymp:moments-identity}
\end{equation}
which may be seen by taking derivatives of $(x-1)^{s-1}$ with respect to $x$ and then setting $x=1$. With this in mind, we rewrite the gamma functions in~\eqref{singular:asymp:moments-rescaled} as a polynomial in $j$ and $N$,
\begin{equation}
\frac{\Gamma[N-j+s]}{\Gamma[N-j]}\prod_{\ell=1}^n\frac{\Gamma[N+\nu_\ell-j+s]}{\Gamma[N+\nu_\ell-j]}
=\sum_{k=0}^{(n+1)s}\sum_{m=0}^{(n+1)s-k}a_{k,m}j^kN^m,
\end{equation}
where $a_{k,m}$ denotes the coefficients.
Due to the identity~\eqref{singular:asymp:moments-identity} we may restrict ourselves to $k=s-1$, hence the dominant term is of order $N^{ns+1}$ in agreement with the prefactor in~\eqref{singular:asymp:moments-rescaled}. Moreover, the coefficient of $j^{s-1}N^{ns+1}$ is equal to the coefficient of the same quantity in the expansion of $(N-j)^{(n+1)s}$. Inserting this back into~\eqref{singular:asymp:moments-rescaled} yields
\begin{equation}
\int_0^\infty dxN^{n-1}{R_1^n(N^nx)}x^s=\frac1{ns+1}\binom{(n+1)s}{s}+O(N^{-1}).
\end{equation}
Thus, the moments converge to the Fuss--Catalan numbers\footnote{A. Smith is thanked for help with this proof. \url{mathoverflow.net/questions/201496}}. 

It follows from Stieltjes moment problem and Carleman's criterion that the Fuss--Catalan numbers determine a unique probability measure on the positive half-line (see e.g.~\cite{Simon:1998} for a discussion of the moment problem). The probability density~\eqref{singular:asymp:density-fc} is obtained from the moments by means of an inverse Mellin transform after an expansion of the gamma functions using the Gauss' multiplication formula~\eqref{special:gamma:multiplication}, see~\cite{PZ:2011} for details. 
\end{proof}

\begin{remark}
A more compact form for the Fuss--Catalan density~\eqref{singular:asymp:density-fc} can be obtained using the Fox $H$-function,
\begin{equation}
\rho_\fuss^n(x)=\FoxH{1}{0}{2}{1}{(0,1),(2-n,n)}{(-n,n+1)}{x};
\label{singular:asymp:fox}
\end{equation}
see appendix~\ref{app:special} for discussion of this function. On the interval of support, an expression for Fuss--Catalan density in terms elementary function can be obtained through a Plancherel--Rotach-like parametrisation~\cite{HM:2013,Neuschel:2014},
\begin{equation}
\rho_\fuss^n\Big(\frac{\sin^{n+1}((n+1)\alpha)}{\sin\alpha\sin^n(n\alpha)}\Big)=\frac{\sin^2\alpha\sin^{n-1}(n\alpha)}{\pi\sin^n((n+1)\alpha)},
\qquad
0<\alpha<\frac{\pi}{n+1},
\label{singular:asymp:density-sines}
\end{equation}
where the argument of the density is an increasing function of $\alpha$ in the relevant interval.

It can be verified directly from the Meijer $G$-function~\eqref{singular:asymp:moments-fc}, or using~\eqref{singular:asymp:density-sines}, that near the edges of support the Fuss--Catalan density behave like~\cite{FL:2014}
\begin{align}
\rho_\fuss^n(x)&\sim \sin\Big(\frac\pi{n+1}\Big)\frac{x^{-n/(n+1)}}\pi
&\text{for}\quad x&\searrow 0,\\
\rho_\fuss^n(x)&\sim \bigg(\frac{2^{1/3}n^{n-1}}{(n+1)^{n+2/3}}\bigg)^{3/2}\frac{\sqrt{\frac{(n+1)^{n+1}}{n^n}-x}}\pi
&\text{for}\quad x&\nearrow \frac{(n+1)^{n+1}}{n^{n}}.
\label{singular:asymp:limits}
\end{align}
We note that the square root decay near the soft edge is the same as for the Wishart ensemble (or, more generally, orthogonal ensembles~\cite{BIPZ:1978,DKMVZ:1999}), while the asymptotic behaviour near the hard edge differs for $n\geq2$. It was conjectured in~\cite{AKW:2013} that the microscopic correlations at the hard edge would be described by a new universal kernel indexed by the number of factors $n$ (i.e. different from the Bessel kernel for $n\geq2$), while the correlations in the bulk and at the soft edge would be described by the sine and Airy kernel, respectively. For the bulk and the soft edge, this conjecture was proven in~\cite{LWZ:2014}, while an explicit expression for the microscopic kernel at the hard edge was obtained in~\cite{KZ:2014}. A brief discussion will be given in section~\ref{sec:singular:microscopic}.
\end{remark}

\begin{remark}\label{remark:singular:algebraic-eq}
Let us compare proposition~\ref{thm:singular:fuss-catalan} with an earlier result about the one-point Green function (or resolvent). Recall that the Green function is defined as the Stieltjes transform of the density,
\begin{equation}
G_n(z):=\int_0^{K_n}dx\frac{\rho_\fuss^n(x)}{z-x},\qquad z\notin[0,K_n].
\end{equation}
Since we know the density~\eqref{singular:asymp:density-fc}, we can perform this integral using~\eqref{special:meijer:int-meijer-meijer} and thereby write the Green function as a Meijer $G$-function,
\begin{equation}
zG_n(z)=\sqrt{\frac{n+1}{2\pi n^3}}\MeijerG[\bigg]{1}{n+1}{n+1}{n+1}{0,\frac1{n+1},\frac2{n+2},\ldots,\frac n{n+1}}{0,-\frac1n,0,\frac1n,\ldots,\frac{n-2}n}{-\frac{(n+1)^{n+1}}{z\,n^n}}.
\label{singular:asymp:green}
\end{equation}
On the other hand, it was shown in~\cite{BJLNS:2010} that the Green function satisfies an algebraic (trinomial) equation%
\begin{equation}
(zG_n(z))^{n+1}=z(zG_n(z)-1),\qquad \lim_{\abs z\to\infty}zG_n(z)=1.
\label{singular:asymp:green-eq-square}
\end{equation}
For $n\leq3$ we have an equation of degree four or less which can be solved by standard means and the solution can be shown to agree with~\eqref{singular:asymp:green} after explicit evaluation of the Meijer $G$-function. For $n\geq4$ there is no general \emph{algebraic} solution to~\eqref{singular:asymp:green-eq-square} and the relation becomes less trivial. In this case, we follow the approach presented in~\cite{FL:2014} (see also~\cite{Glasser:1994}). Let $f(u)$ and $\phi(u)$ be functions which are analytic in a neighbourhood $\Omega$ of $v\in\C$ and let $t$ be a small parameter so that $\abs{t\phi(u)}<\abs{u-v}$ for $u\in\Omega$. Then the Lagrange reversion formula~\cite{WW:1996} tells that $u=v+t\phi(u)$ has one solution in $\Omega$ and that
\begin{equation}
f(u)=f(v)+\sum_{s=1}^\infty\frac{t^s}{s!}\frac{d^{s-1}}{dv^{s-1}}f'(v)\phi(v)^s.
\end{equation}
We are interested in the equation $g(z)^{n+1}=z(g(z)-1)$ which according to~\eqref{singular:asymp:green-eq-square} should have a solution $g(z)=zG_n(z)$. Rewriting this equation as $g(z)=1+\frac1zg(z)^{n+1}$, we see that we can use the Lagrange reversion formula by taking $v=1$, $t=1/z$, $f(g)=g$ and $\phi(g)=g^{n+1}$. We have
\begin{equation}
g(z)=\sum_{s=0}^\infty \frac{1}{z^s}\frac1{ns+1}\binom{(n+1)s}{s}
\end{equation}
for sufficiently large $z$. We recognise the Fuss--Catalan numbers, hence 
\begin{equation}
g(z)=\sum_{s=0}^\infty \frac{1}{z^s}\int_0^\infty dx\,\rho_\fuss^n(x)x^s=\int_0^{K_n}dx\frac{\rho_\fuss^n(x)}{1-x/z}=zG_n(z),
\end{equation}
which was the statement of~\eqref{singular:asymp:green-eq-square}. Here, the sum converges for $\abs{z}>K_n$ and the result is extended to the rest of the domain by analytic continuation.

In the above given description of the macroscopic density, we have restricted ourselves to the case where the charges $\nu_i$ ($i=1,\ldots,n$) are kept fixed as $N$ tends to infinity. If we instead let the charges scale with $N$ so that $\nu_i/N\to\alpha_i\in[0,\infty)$, then we no longer have an explicit expression for the density for arbitrary $n$. However, we still have an algebraic equation for the Green function~\cite{BJLNS:2010},
\begin{equation}
zG_n^\alpha(z)\prod_{\ell=1}^n\frac{zG_n^\alpha(z)+\alpha_\ell}{\alpha_\ell+1}=z(zG_n^\alpha(z)-1),\qquad \lim_{\abs z\to\infty}zG_n^\alpha(z)=1.
\label{singular:asymp:green-eq-gen}
\end{equation}
For $n\leq3$, it is again a straightforward task to solve this equation and as usual the density can obtained by
\begin{equation}
\rho_\nu^n(x)=\frac1\pi\lim_{\epsilon\to 0^+}\Im G_n^\nu(x-i\epsilon).
\label{singular:asymp:green-to-density}
\end{equation}
For the case of general $n$, an important observation is that there is a macroscopic gap between the origin and the smallest eigenvalues if and only if all scaled charges are non-zero, i.e. $0<\alpha_1\leq\cdots\leq\alpha_n$. This was shown in~\cite{AIK:2013} using a saddle point approximation, but we will not repeat this derivation here. If $0=\alpha_1=\cdots=\alpha_m$ and $0<\alpha_{m+1}\leq\cdots\leq\alpha_n$ for $0<m\leq n$ (i.e. there is no macroscopic gap) then it is seen from~\eqref{singular:asymp:green-eq-gen} and~\eqref{singular:asymp:green-to-density} that the macroscopic density diverges like $x^{-m/(m+1)}$ for $x\to0$~\cite{BJLNS:2010}. This observation is important if we believe that the microscopic universality class at the hard edge is related to the rate of divergence of the macroscopic density (see proposition~\ref{thm:singular:meijer-kernel}). 
\end{remark}

\subsection{Microscopic correlations}
\label{sec:singular:microscopic}

As mentioned in the previous section, the results for microscopic correlations were given by Kuijlaars and Zhang~\cite{KZ:2014} for the hard edge, and by Liu, Wang and Zhang~\cite{LWZ:2014} for the bulk and the soft edge. For completeness we will recall their results. For the hard edge we prove a slight generalisation of the original result, while the reader is referred to the original paper for proofs of the microscopic limits in the bulk and at the soft edge.

\begin{proposition}[Hard edge]\label{thm:singular:meijer-kernel}
Let $n\geq m\geq 1$ and $\nu_i$ for $i=1,\ldots,m$ be fixed and let $\nu_i=\alpha_iN+O(1)$ with $\alpha_i>0$ fixed for $i=m+1,\ldots,n$. With $(c_N)^{-1}=\alpha_{m+1}\cdots\alpha_nN^{n-m-1}$, we have
\begin{equation}
\lim_{N\to\infty}\frac1{c_N}K_N^n\Big(\frac x{c_N},\frac y{c_N}\Big)=K_\meijer^{m,\nu}(x,y)
\end{equation}
uniformly for $x$ and $y$ in compact subsets of the positive half-line, where
\begin{equation}
K_\meijer^{m,\nu}(x,y)=
\int_0^1du\,\MeijerG{1}{0}{0}{m+1}{-}{0,-\nu_1,\ldots,-\nu_m}{ux}\MeijerG{m}{0}{0}{m+1}{-}{\nu_1,\ldots,\nu_m,0}{uy}.
\label{singular:asymp:meijer-kernel}
\end{equation}
\end{proposition}

\begin{proof}
We follow the same steps as in~\cite{KZ:2014} and take expression~\eqref{singular:correlations:kernel-int} for the correlation kernel as our starting point. Writing the kernel as its double contour integral representation, we have
\begin{multline}
\frac1{c_N}K_N^n\Big(\frac x{c_N},\frac y{c_N}\Big)=
\frac1{c_N}\frac{1}{(2\pi i)^2}\int_0^1du\int_{-1-i\infty}^{-1+i\infty} dt\oint_{\Sigma_k} ds\, \Big(\frac{ux}{c_N}\Big)^s\Big(\frac{uy}{c_N}\Big)^t \\
\times\frac{\Gamma[-s]}{\Gamma[1+t]}\frac{\Gamma[N+1+t]}{\Gamma[N-s]}\prod_{\ell=1}^n\frac{\Gamma[\nu_\ell-t]}{\Gamma[\nu_\ell+1+s]}
\label{singular:asymp:hard-proof-long}
\end{multline}
with $\Sigma_k$ as in~\eqref{singular:correlations:function-step1}. In order to evaluate the large-$N$ limit, we need to look at the asymptotic behaviour of the fractions of gamma functions in the second line in~\eqref{singular:asymp:hard-proof-long}. We know that (see appendix~\ref{app:special})
\begin{equation}
\frac{\Gamma[\alpha N+a]}{\Gamma[\alpha N+b]}=(\alpha N)^{a-b}(1+O(N^{-1}))
\end{equation}
for complex constants $a$ and $b$, and $\alpha>0$. We recall that $\nu_i=O(1)$ for $i\leq m$ and $\nu_i=\alpha_iN+O(1)$ for $i>m$, thus
\begin{equation}
\frac{\Gamma[N+t+1]}{\Gamma[N-s]}\prod_{i=m+1}^n\frac{\Gamma[\nu_i-t]}{\Gamma[\nu_i+s+1]}
=(c_N)^{1+s+t}(1+O(N^{-1})).
\end{equation}
Using this asymptotic formula in~\eqref{singular:asymp:hard-proof-long}, we get
\begin{multline}
\frac1{c_N}K_N^n\Big(\frac x{c_N},\frac y{c_N}\Big)=
\frac{1}{(2\pi i)^2}\int_0^1du\int_{-1-i\infty}^{-1+i\infty} dt\int_{\Sigma} ds\, (ux)^s(uy)^t\\
\times\frac{\Gamma[-s]}{\Gamma[1+t]}\prod_{\ell=1}^m\frac{\Gamma[\nu_\ell-t]}{\Gamma[\nu_\ell+1+s]}(1+O(N^{-1})),
\end{multline}
where we have modified the contour of the $s$-integral such that it begins and ends at $+\infty$ and encircles the positive half-line in the negative direction without crossing the other contour. An interchange of the large-$N$ limit and the integrals is justified by the dominated convergence theorem. Finally, the expression~\eqref{singular:asymp:meijer-kernel} is obtained by rewriting the contour integrals as Meijer $G$-functions using~\eqref{special:meijer:meijer-def}. 
\end{proof}

\begin{remark}
The hard edge scaling limit requires that at least one charge is independent of $N$. In the light of remark~\ref{remark:singular:algebraic-eq}, the reason for this requirement is quite obvious: If all charges scale with $N$, i.e. $\nu_i/N\to\alpha_i>0$ for all $i$, then there will be a macroscopic gap between the origin and the smallest eigenvalues, thus the spectrum has no hard edge. Furthermore, we note that the microscopic correlations are classified according to the divergence of the macroscopic density near the origin through the integer $m$. 
In the special case where $m=1$ and the density diverges like an inverse square root, the Meijer $G$-kernel agrees with the Bessel kernel~\eqref{singular:intro:bessel} as expected. Explicitly, we have 
\begin{equation}
\frac14K_\meijer^{m=1,\nu}\Big(\frac x4,\frac y4\Big)=\Big(\frac yx\Big)^{\nu/2}\int_0^1du\,J_\nu(\sqrt{ux})J_\nu(\sqrt{uy})=\Big(\frac yx\Big)^{\nu/2} K_\bessel(x,y).
\end{equation}
Recall that the prefactor $(y/x)^{\nu/2}$ does not affect the correlations.
\end{remark}

\begin{remark}
We note that the hard edge scaling relation (propostion~\ref{thm:singular:meijer-kernel}) incorporates the reduction formula for $m\geq2$,
\begin{equation}
\lim_{\nu_m\to\infty}\nu_m K_\meijer^{m,\nu}(\nu_m x,\nu_m y)=K_\meijer^{m-1,\nu}(x,y),
\label{singular:asymp:reduction}
\end{equation}
given in~\cite{AI:2015}. This may also be seen starting from the Meijer $G$-kernel, we have 
\begin{multline}
\nu_mK_N^n(\nu_mx,\nu_my)=
\frac{\nu_m}{(2\pi i)^2}\int_0^1du\int_{-1-i\infty}^{-1+i\infty} dt\int_{\Sigma} ds\, (\nu_mux)^s(\nu_muy)^t\\
\times\frac{\Gamma[-s]}{\Gamma[1+t]}\frac{\Gamma[\nu_m-t]}{\Gamma[\nu_m+1+s]}
\prod_{\ell=1}^{m-1}\frac{\Gamma[\nu_\ell-t]}{\Gamma[\nu_\ell+1+s]}.
\label{singular:asymp:reduction-proof-1}
\end{multline}
Using~\eqref{special:gamma:ratio-asymp} it follows that
\begin{equation}
\frac{\Gamma[\nu_m-t]}{\Gamma[\nu_m+1+s]}=\nu_m^{-s-t-1}(1+O(N^{-1}))\qquad\text{for}\qquad \nu_m\to\infty.
\end{equation}
Inserting this into~\eqref{singular:asymp:reduction-proof-1} and interchanging the limit and the integration (justified by the dominated converge theorem), we find the reduction formula~\eqref{singular:asymp:reduction}. 
This reduction of the Meijer $G$-kernel may be thought of as a consequence of the change of the rate of divergence of the macroscopic density at the hard edge for $\nu_m\to\infty$, see remark~\ref{remark:singular:algebraic-eq}.
\end{remark}

For completeness, let us recall (without proof) the scaling limits for the bulk and the soft edge obtained by Liu, Wang and Zhang~\cite{LWZ:2014}:

\begin{proposition}[Bulk]\label{thm:singular:sine}
Let $n$ and $\nu_i>-1$ ($i=1,\ldots,n$) be fixed. Given a point $x_0\in(0,(n+1)^{n+1}/n^n)$ in the bulk with spectral density $\rho_0:=\rho_\fuss^n(x_0)$ (cf. proposition~\ref{thm:singular:fuss-catalan}), then 
\begin{equation}
\lim_{N\to\infty}\bigg(\frac{N^{n}}{\rho_0N}\bigg)^kR_k^n\Big(N^{n}\Big(x_0+\frac{x_1}{\rho_0N}\Big),\ldots,N^{n}\Big(x_0+\frac{x_k}{\rho_0N}\Big)\Big)=\det_{1\leq i,j\leq k}[ K_\sine(x_i,x_j)]
\end{equation}
uniformly for $x$ and $y$ belonging to any compact subset of the real line, where $K_\sine(x,y)$ denotes the sine kernel~\eqref{singular:intro:sine}.
\end{proposition}

\begin{proposition}[Soft edge]\label{thm:singular:airy}
Let
\begin{equation}
x_*=\frac{(n+1)^{n+1}}{n^n}
\qquad\text{and}\qquad
c_*=\bigg(\frac{(n+1)^{n+2/3}}{2^{1/3}n^{n-1}}\bigg)^{3/2},
\end{equation}
i.e. $x_*$ is the location of the soft edge, while $c_*$ should be compared with the prefactor in~\eqref{singular:asymp:limits}. With $n$ and $\nu_i>-1$ ($i=1,\ldots,n$) fixed, we have 
\begin{equation}
\lim_{N\to\infty}\!
\bigg(\frac{N^{n}}{(c_*N)^{\frac23}}\bigg)^k
R_k^n\Big(N^{n}\Big(x_\ast+\frac{x_1}{(c_*N)^{\frac23}}\Big),\ldots,N^{n}\Big(x_\ast+\frac{x_k}{(c_*N)^{\frac23}}\Big)\Big)
=\det_{1\leq i,j\leq k}[K_\airy(x_i,x_j)]\!\!
\label{singular:asymp:airy-kernel}
\end{equation}
uniformly for $x$ and $y$ belonging to any compact subset of the real line, where $K_\airy(x,y)$ is the Airy kernel~\eqref{singular:intro:airy}.
\end{proposition}

\section{Summary, discussion and open problems}
\label{sec:singular:discuss}

We have seen that it is possible to obtain exact expressions for the statistical properties of the singular values of a product constructed from an arbitrary number of independent Gaussian matrices with arbitrary size (proposition~\ref{thm:singular:jpdf-determinantal} and~\ref{thm:singular:biorthogonalise}). However, the ``real'' difference between the joint density from proposition~\ref{thm:singular:jpdf-determinantal} and the familiar Wishart--Laguerre ensemble~\eqref{singular:intro:jpdf-laguerre} is to some extent hidden within the Meijer $G$-function expression for the weight functions. In order to get a better understanding of this difference, it is illustrative to introduce an approximation of the Meijer $G$-function. It follows from~\eqref{special:meijer:meijer-asymp} that the joint density~\eqref{singular:intro:jpdf} may be approximated by
\begin{equation}
\cP_\jpdf^{\beta=2}(N^nx_1,\ldots,N^nx_N)N^{nN}\approx \frac1\cZ\prod_{k=1}^Ne^{-nNx_k^{1/n}}x_k^{(\overline\nu+1-n)/n}\prod_{1\leq i<j\leq N}(x_j-x_i)(x_j^{1/n}-x_i^{1/n})
\label{singular:discuss:jpdf-MB-inv-laguerre}
\end{equation}
or with a change of variables
\begin{equation}
\cP_\jpdf^{\beta=2}((Nx_1)^n,\ldots,(Nx_N)^n)\prod_{k=1}^NnN^nx_k^{n-1}\approx \frac{1}{\cZ}\prod_{k=1}^Ne^{-nNx_k}x_k^{\overline\nu}\prod_{1\leq i<j\leq N}(x_j-x_i)(x_j^{n}-x_i^{n})
\label{singular:discuss:jpdf-MB-laguerre}
\end{equation}
where $\cZ$ is a normalisation constant and $\overline\nu=\nu_1+\cdots+\nu_n$ is the cumulative charge. These approximations give a more useful intuition about how matrix multiplication changes the repulsion between eigenvalues. Note that joint densities with this structure are not new, they belong to a special type of polynomial ensembles, which were originally suggested as an approximate description for transport properties in disordered wires in~\cite{Muttalib:1995}. The $n=2$ case of~\eqref{singular:discuss:jpdf-MB-laguerre} has also been linked to the dilute phase model, $\gO(-2)$, on a random lattice~\cite{EZ:1992,EK:1995} and to systems of disordered bosons~\cite{LSZ:2006}.

A na\"ive comparison of the joint density~\eqref{singular:discuss:jpdf-MB-inv-laguerre} with the scaling limits from section~\ref{sec:singular:asymp} suggests that the macroscopic spectral density as well as the microscopic correlations in the bulk and at the soft edge should agree with those of the Wishart product ensemble. In fact, this na\"ive guess for the limiting behavior is confirmed by rigorous computations~\cite{FL:2014,FW:2015,Zhang:2015}. The approximation~\eqref{singular:discuss:jpdf-MB-inv-laguerre} breaks down at the hard edge, thus we do not expect to find the same scaling in this limit (this is immediately confirmed by noticing that the charges $\nu_1,\ldots,\nu_n$ do not appear individually but only in the cumulative charge $\overline\nu$). However, the macroscopic density still has the same rate of divergence at the hard edge and we might guess that (albeit different) the microscopic correlations will belong to the same class of correlation functions. Again this intuition is confirmed by 
actual computations; it was shown in~\cite{Borodin:1998,KS:2014} that the hard edge scaling limit for~\eqref{singular:discuss:jpdf-MB-inv-laguerre} gives rise to a correlation kernel
\begin{equation}
K_\meijer^{n,\hat\nu}(x,y)=\int_0^1du\MeijerG{1}{0}{0}{n+1}{-}{0,-\hat\nu_1,\ldots,-\hat\nu_n}{ux}\MeijerG{n}{0}{0}{n+1}{-}{\hat\nu_1,\ldots,\hat\nu_n,0}{uy}
\end{equation}
with
\begin{equation}
\hat\nu_i=\frac{\overline\nu+j-n}{n},\qquad i=1,\ldots,n. 
\end{equation}
This reappearance of the Meijer $G$-kernel is important, since it suggests universality. In addition to the two examples mentioned above, the Meijer $G$-kernel has recently been observed in several other product ensembles~\cite{Forrester:2014b,KS:2014,KKS:2015,AI:2015} and multi-matrix models~\cite{BGS:2009,BGS:2014,BB:2015,FK:2014} confirming that a much stronger underlying universality principle is at play; similar to that of the Bessel kernel~\cite{ADMN:1997,KF:1998,KV:2002,KV:2003,Lubinsky:2008,BP:2005,TV:2010}. Perhaps, the simplest starting point for a study of such universality would be to consider ensembles of the form
\begin{equation}
\cP_\jpdf^n(x_1,\ldots,x_N)= \frac{1}{\cZ}\prod_{k=1}^Ne^{-NV(x_k)}\prod_{1\leq i<j\leq N}(x_j-x_i)(x_j^{\theta}-x_i^{\theta}),
\label{singular:discuss:jpdf-MB-general}
\end{equation}
where $V(x)$ is a real analytic confining potential and $\theta$ is a rational constant. However, it should be mentioned that in all the known product ensembles the charges $\nu_1,\ldots,\nu_n$ which appear in the Meijer $G$-kernel~\eqref{singular:asymp:meijer-kernel} have a natural interpretation as incooperating a rectangular structure of the product due to the general reduction procedure (proposition~\ref{thm:pro:rect-square}), such structure cannot, by construction, be included in~\eqref{singular:discuss:jpdf-MB-general}.

\begin{figure}[tp]
\centering
\includegraphics{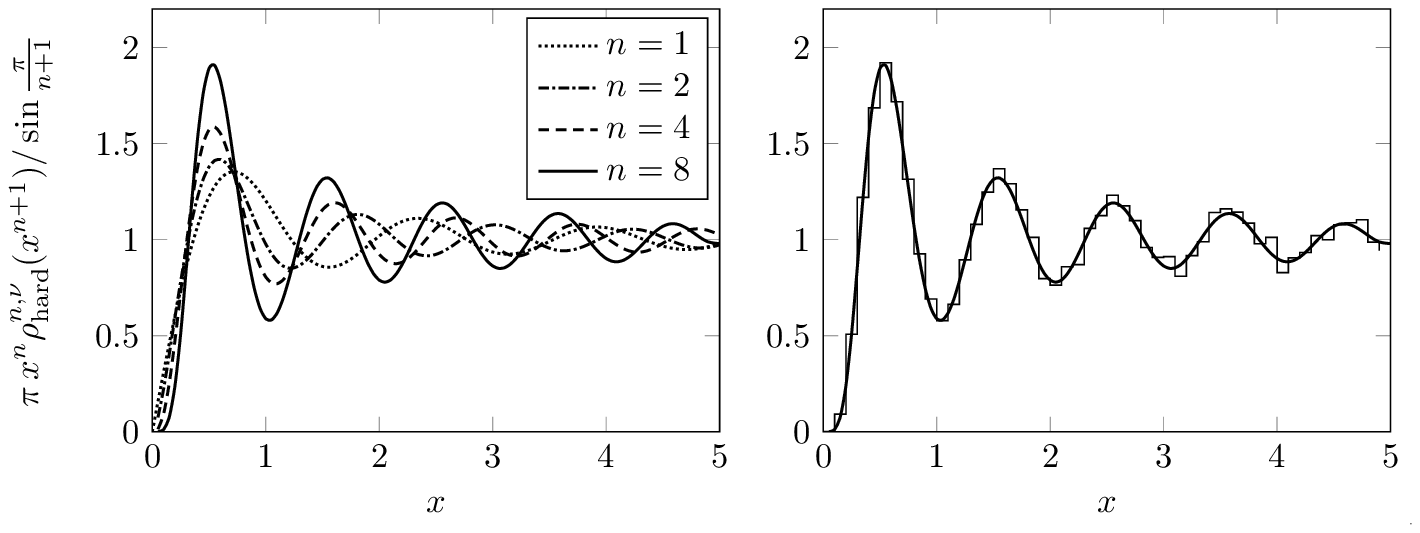}
\caption{The figure shows the microscopic density as described by the Meijer $G$-kernel, $\rho_\hard^{n,\nu}(x)=K_\meijer^{n,\nu}(x,x)$, with $\nu_1=\cdots=\nu_n=0$. The left panel shows the density for $n=1,2,4,8$; for comparison between densities, we have unfolded the densities such that they all tend to unity as their argument tends to infinity. The right panels the $n=8$ case with numerical data from an ensemble of $10\,000$ realisations of a product of $8$ independent $100\times100$ complex Ginibre matrices.}
\label{fig:singular:meijer}
\end{figure}

It is, of course, also highly desirable to find more physical models which belong to this new hard edge universality class. One could hope that links to physical models can be established through the corresponding non-linear sigma models but, so far, little progress has been made in this direction. The structure of the Meijer $G$-kernel is visualised in figure~\ref{fig:singular:meijer}. We note that the amplitude of the oscillations in the microscopic density near the hard edge appears to increase with the index $n$ (i.e. the number of matrices). We will return to large-$n$ limits in chapter~\ref{chap:lyapunov}.

\begin{figure}[tp]
\centering
\includegraphics{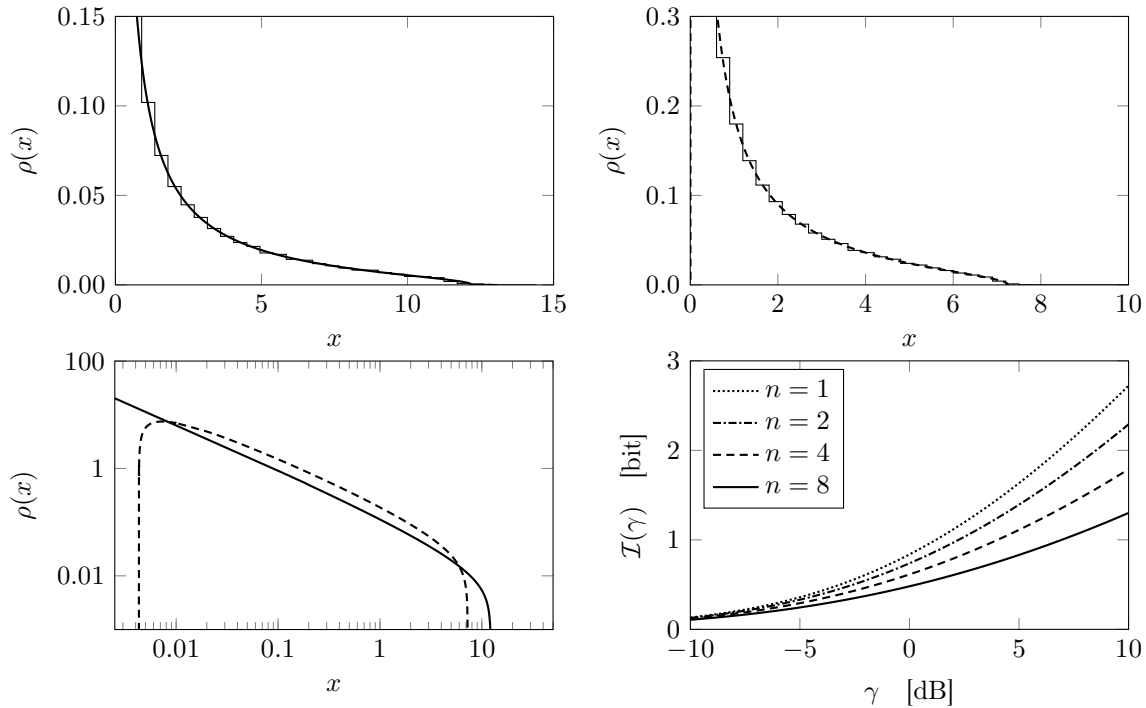}
\caption{The top left panel and the top right panel compare the macroscopic densities for $n=4$ with $\alpha_1=\cdots=\alpha_4=0$ (left) and $\alpha_1=\cdots=\alpha_4=1$ (right) with numerical data generated from an ensemble of $1\,000$ realisations of a product of four independent $100\times100$ complex Ginibre matrices and a product of one $100\times200$ and three $200\times200$ complex Ginibre matrices, respectively. The bottom left panel shows the two same analytic curves on a double logarithmic scale which illustrates that the $\alpha_1=\cdots=\alpha_4=1$ case has a macroscopic gap between the origin and the eigenvalue spectrum even though this gap is too small to be visible on the top right panel. The bottom right panel shows mutual information~\eqref{singular:discuss:info} for $n=1,2,4,8$ as a function of the signal-to-noise ratio measured in decibel.}
\end{figure}

Let us leave the microscopic regimes for now, and turn to macroscopic properties instead. The explicit expression for the macroscopic spectral density of Wishart product ensemble as a Meijer $G$-function~\eqref{singular:asymp:density-fc} together with the powerful integration formula~\eqref{special:meijer:int-meijer-meijer}, immediately allow us calculate explicit formulae for many important types of linear statistics. For example, an interesting quantity in wireless telecommunication is the input-output mutual information for a MIMO communication channel with progressive scattering (see section~\eqref{sec:moti:mimo}) given by
\begin{equation}
\cI_n(\gamma):=\int_0^{K_n} dx\,\rho_\fuss^n(x)\log_2(1+\gamma x),
\end{equation}
which gives an upper bound for the spectral efficiency. Here, $\gamma$ denotes the signal-to-noise ratio and integration gives
\begin{equation}
\cI_n(\gamma)=\frac{1}{\log2}\sqrt{\frac{n+1}{2\pi n^3}}\MeijerG[\bigg]{1}{n+2}{n+2}{n+2}{\frac1{n+1},\frac2{n+1},\ldots,\frac{n}{n+1},1,1}{1,0,-\frac1n,0,\frac1n,\ldots,\frac{n-2}n}{\frac{(n+1)^{n+1}}{n^n}\gamma}.
\label{singular:discuss:info}
\end{equation}
The simplicity of this macroscopic formula is illustrated by comparing with the formulae valid for finite matrix dimensions given in~\cite{AIK:2013,AKW:2013}.
We note that universality results such as~\cite{GKT:2014} do not apply to this model since the assumptions of i.i.d. entries and their Gaussian distribution go hand-in-hand within the Rayleigh fading regime, while leaving this regime results in violation of both (note that in the noise free channel the Gaussian matrices should be replaced by unitary matrices). Nonetheless, the $n=1$ case is known to be in good agreement with experimental data for strongly scattering environments with no line-of-sight between the transmitter and receiver; this agreement is expected to extend to communication channels with progressive scattering described by the $n\geq2$ cases.

We would, of course, also like to know the structure of higher point correlations, i.e. the wide correlators. The wide correlators for product ensembles are expected to differ from the known universality results of the trace-class ensembles, since the eigenvalue repulsion is not purely logarithmic. This conjecture is supported by a heuristic argument originally given in~\cite{Beenakker:1993,Beenakker:1994}.

Consider a point process of $N$ particles with partition function,
\begin{equation}
\cP_\jpdf(x_1,\ldots,x_N)=\frac1{\cZ_V}\exp\bigg[-N\sum_{k=1}^NV(x_k)+\sum_{1\leq i<j\leq N}U(x_j,x_i)\bigg],
\label{singular:discuss:particle}
\end{equation}
where $\cZ_V$ denotes the potential dependent normalisation constant and the exponent is interpreted as a Hamiltonian for the $N$ particles located at $x_1,\ldots,x_N$ subject to a confining potential $V(x)$ and pair interaction $U(x,y)$. If the pair interaction is given by $U(x,y)=\log(x-y)+\log(x^\theta-y^\theta)$ then~\eqref{singular:discuss:particle} corresponds to~\eqref{singular:discuss:jpdf-MB-general} and if $\theta=1$ this reduces to a trace-class ensemble in which case universality is known. In the following we consider $U(x,y)$ as a generic (repulsive) interaction.

Leaving the large-$N$ limit implicit, the macroscopic density and the connected two-point correlator are given by
\begin{align}
\rho_V(x)&=\average[\Big]{\sum_{i}\delta(x_i-x)}
\label{singular:discuss:density-particle}\\
\rho_\conn(x,y)&=\average[\Big]{\sum_{i\neq j}\delta(x_i-x)\delta(x_j-y)}-\rho_V(x)\rho_V(y),
\label{singular:discuss:connect}
\end{align}
where the averages are taken according to the partition function~\eqref{singular:discuss:particle}. The subscript $V$ indicates that the density has a non-universal dependence on the potential. We recall that the macroscopic density, $\rho_V(x)$, with compact support on a single interval $[a,b]$ is such that it minimises the energy functional
\begin{equation}
E[\rho_V]=\int_a^bdx\,\rho_V(x)V(x)-\frac12\int_a^bdx\,\rho_V(x)\int_a^bdy\,\rho_V(y)U(x,y).
\end{equation}
Or stated differently, the density satisfies the macroscopic balance equation,
\begin{equation}
V(y)=\int_a^bdx\,\rho_V(x)U(x,y)+\text{const.},\qquad
\int_a^bdx\,\rho_V(x)=1,
\label{singular:discuss:balance}
\end{equation}
and therefore responds linearly to changes in the potential. 

The idea presented in~\cite{Beenakker:1993,Beenakker:1994} is to consider the functional derivative with respect to the potential of density~\eqref{singular:discuss:density-particle} and balance equation~\eqref{singular:discuss:balance}. Under the assumption that the macroscopic density vanishes at non-fixed (i.e. soft) edges, this yields
\begin{align}
\frac{\delta\rho_V(x)}{\delta V(y)}&=\rho_\conn(x,y)+\rho_V(x)\delta(x-y)\\
\delta(x-y)&=\int_a^bdx'\frac{\delta\rho_V(x')}{\delta V(y)}U(x',x)
\end{align}
Combining these two equations gives
\begin{equation}
\rho_\conn(x,y)=U^{-1}(y,x)-\rho_V(x)\delta(x-y),
\end{equation}
where the inverse pair interaction is defined through the integral equation
\begin{equation}
\int_a^bdx'\,U^{-1}(y,x')U(x',x)=\delta(x-y).
\end{equation}
This heuristic argument shows that the potential $V(x)$ is expected to \emph{only} influence the connected two-point correlation function~\eqref{singular:discuss:connect} through the boundary of support, $a$ and $b$ (i.e. the correlation function is universal). On the other hand, the correlations will dependent on the eigenvalue repulsion and, thus, differ from the classical universality results when the repulsion is not purely logarithmic. Due to the approximation~\eqref{singular:discuss:jpdf-MB-inv-laguerre} we additionally expect that the Wishart product ensemble gives rise to the same universality classes as~\eqref{singular:discuss:jpdf-MB-general}. 

It remains an intriguing (but challenging) problem to find explicit expressions for the wide correlator in the case of product ensembles. In~\cite{Forrester:2014b}, it was argued that the wide correlator on the semi-infinite interval $(0,\infty)$ for a product of two matrices is
\begin{equation}
\rho_\conn^{n=2}(x,y)=-\frac1{6\pi^2}\frac{1+(x/y)^{1/3}+(y/x)^{1/3}}{(x-y)^2}.
\label{singular:discuss:wide2}
\end{equation}
We note that this expression (as expected) differs from the same result the for trace-class ensembles,
\begin{equation}
\rho_\conn^{n=1}(x,y)=-\frac1{4\pi^2}\frac{(x/y)^{1/2}+(y/x)^{1/2}}{(x-y)^2}.
\end{equation}
Formula~\eqref{singular:discuss:wide2} was obtained using asymptotics of the Meijer $G$-kernel (similar to an approach previously applied to the Bessel kernel~\cite{BT:1993}), but it is non-trivial to extend this formula to products of more matrices. Moreover, the general expression for arbitrary boundary conditions $a$ and $b$ is still completely unknown even for two matrices. 

One reason for the interest in the connected two-point wide correlator is its relation to the variance of linear statistics. If $F=\sum_if(x_i)$ is linear statistic (e.g. $f(x)=\log_2(1+\gamma x)$ for the mutual information) then it follows
\begin{equation}
\var F=\int_a^bdx\int_a^bdy\,f(x)f(y)[\rho_\conn(x,y)+\rho_V(x)\delta(x-y)];
\label{singular:discuss:linear}
\end{equation}
assuming that the right hand side is well-defined. The linear statistics for a product of two matrices was studied in~\cite{GNT:2014}, but the direct link to~\eqref{singular:discuss:wide2} is not completely obvious.

Finally, let us mention that it is also interesting to consider matrix products in the limit where the number of factors tends to infinity, but we will return to this question in chapter~\ref{chap:lyapunov}.

%% file: complex.tex
\chapter{Eigenvalues of Gaussian product matrices}
\label{chap:complex}

This chapter is devoted to the study of complex eigenvalues of products of independent random matrices. Unlike the description of the Wishart product ensemble given chapter~\ref{chap:singular}, we will not be limited by the use of the Itzykson--Zuber--Harish-Chandra formula and we can therefore consider all three standard classes: real ($\beta=1$), complex ($\beta=2$), and quaternionic ($\beta=4$) matrices. 
As in the previous chapter, our focus will be on products of independent induced Ginibre matrices. 

Let us briefly recall the known structure of the eigenvalues of induced Ginibre matrices. Details for the Ginibre ensembles can be found in~\cite{Ginibre:1965,MS:1966,Kanzieper:2002,LS:1991,Edelman:1997,AK:2005,FN:2007,Sommers:2007,SW:2008,BS:2009}; while the relevant structures for the product ensembles will be given in the three proceeding sections. 

For $N\times N$ complex ($\beta=2$) and quaternionic ($\beta=4$) matrices distributed according to the density given in definition~\ref{def:pro:induced}, we readily obtain the joint density for the eigenvalues~\cite{Ginibre:1965},
\begin{align}
\cP_{\jpdf}^{\beta=2}(z_1,\ldots,z_N)&=\frac{1}{\cZ^{\beta=2}}\prod_{j=1}^N \abs{z_j}^{2\nu}e^{-\abs{z_j}^2}\prod_{1\leq k<\ell\leq N}\abs{z_k-z_\ell}^2, \label{complex:intro:jpdf-C} \\
\cP_{\jpdf}^{\beta=4}(z_1,\ldots,z_N)&=\frac{1}{\cZ^{\beta=4}}\prod_{j=1}^N \abs{z_j}^{4\nu}e^{-2\abs{z_j}^2}\abs{z_j-z_j^*}^2\prod_{1\leq k<\ell\leq N}\abs{z_k-z_\ell}^2\abs{z_k-z_\ell^*}^2,
\end{align}
using the Schur decomposition (proposition~\ref{thm:decomp:schur-C} and~\ref{thm:decomp:schur-H}, respectively) and integrating out irrelevant degrees of freedom. Here, $\cZ^{\beta}$ are normalisation constants, $\nu$ is a non-negative integer, and $z_i$ ($i=1,\ldots,N$) are the eigenvalues (in the quaternionic case the eigenvalues come in complex conjugate pairs and $z_1,\ldots,z_N$ are chosen to belong to the complex upper half-plane, cf. appendix~\ref{app:decompositions}). The real ($\beta=1$) induced Ginibre ensemble is more complicated since the eigenvalues come in a combination of real and complex conjugate pairs. Assuming that the matrix has $K$ real eigenvalues denoted by $x_1,\ldots,x_K$ and $2L$ complex eigenvalues denoted by $z_1,z_1^*,\ldots,z_L,z_L^*$, then the joint density for the eigenvalues can be written as~\cite{LS:1991,Edelman:1997}
\begin{multline}
\cP_{\jpdf}^{\beta=1,K}(x_1,\ldots,x_K  ,z_1,\ldots,z_L)=
  \frac{1}{\cZ_{K}^{\beta=1}} \prod_{1\leq i<j\leq K} \abs{x_j-x_i} \prod_{j=1}^K\prod_{i=1}^L \abs{x_j-z_i}^2\prod_{j=1}^K \abs{x_j}^\nu e^{-x_j^2/2}  \\
\times \prod_{1\leq i<j\leq L}\abs{z_j-z_i}^2\abs{z_j-z_i^*}^2
\prod_{j=1}^L(z_j-{z}_j^*) \abs{z_j}^{2\nu}\erfc\left[\frac{z_j-z_j^*}{i\sqrt{2}}\right] e^{-(z_j^2+{z}_j^{*2})/2}. 
\end{multline}
Again the joint density is obtained using the Schur decomposition (proposition~\ref{thm:decomp:schur-R}) and integrating out irrelevant degrees of freedom; the error function appears since the real Schur decomposition is incomplete for $L\geq1$ (see appendix~\ref{app:decompositions}). Here, $\cZ_K^{\beta=1}$ is chosen such that integration over the eigenvalues gives the probability that there are $K$ real eigenvalues, hence the sum over $K=0,\ldots,N$ is normalised to unity.

After rescaling by $N$, the macroscopic spectrum for all three Ginibre ensembles is uniformly distributed within a disk if the charge $\nu$ is kept fixed as $N$ tends to infinity and uniformly distributed within an annulus if $\nu/N\to\alpha>0$ for $N\to\infty$~\cite{FBKSZ:2012}; this is equivalent to the behaviour of the macroscopic gap between the origin and the infimum of the real spectrum of the Wishart ensemble. 
In the Ginibre case ($\nu=0$) the uniform distribution is a special case of the circular law, see~\cite{BC:2012} for a review. It is worth noting that the authors of~\cite{FBKSZ:2012} also provided a method for generating induced Ginibre matrices with integer charge $\nu$ from an $N\times(N+\nu)$ rectangular matrix with i.i.d. Gaussian entries (also numerically!).

\begin{figure}[tp]
\centering
\includegraphics{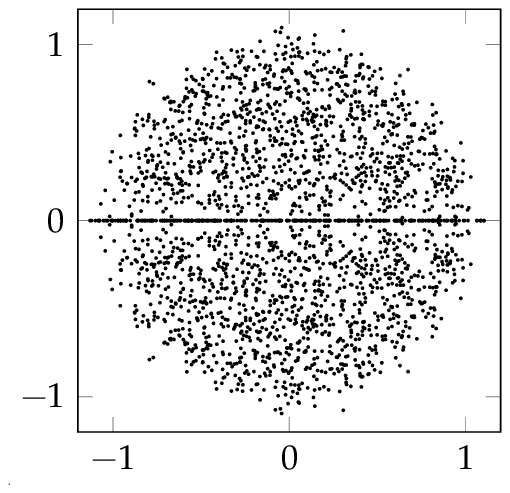}
\includegraphics{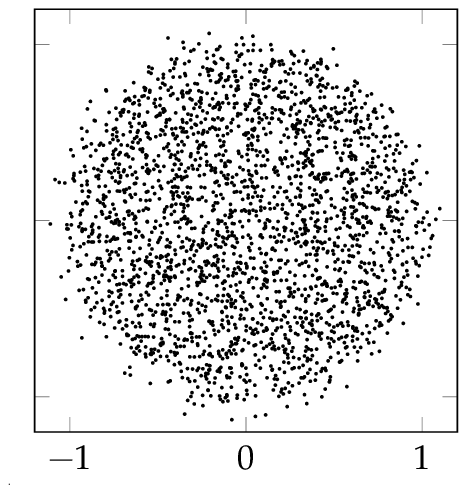}
\includegraphics{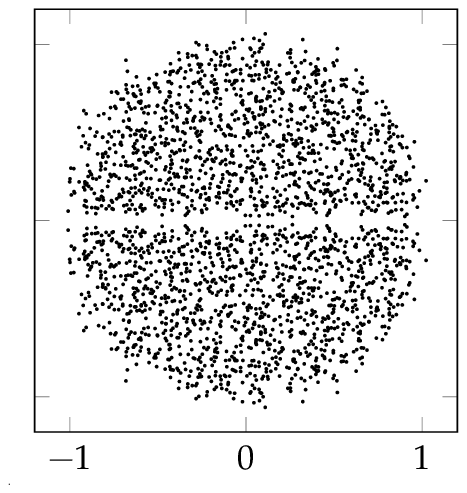}
\caption{Scatter plots of the complex eigenvalues of $50\times 50$ real (left panel), complex (center panel), and quaternionic Ginibre matrices generated from an ensemble of $50$ (real and complex case) and $25$ (quaternionic case) matrices. Note that the apparent poisson-like clustering is due to the ensemble average.}
\label{fig:complex:scat-intro}
\end{figure}

The uniform density of the macroscopic spectrum within the entire region of support does not extend to microscopic structures. As indicated by figure~\ref{fig:complex:scat-intro} the real and quaternionic Ginibre ensembles assign special behaviour to the neighbourhood of the real axis due to the complex pairing of eigenvalues. Special structure is also assigned to the origin, since we are considering induced ensembles. In the complex case, the joint density~\eqref{complex:intro:jpdf-C} may be interpreted as a two-dimensional log-gas trapped by a potential $V(\abs{z})=\abs{z}^2-\nu\log\abs{z}^2$ and the special structure near the origin is due to the logarithmic singularity (the singularity takes us out of the universality regime considered in~\cite{Berman:2008}). Similar interpretations can be given to the real and quaternionic ensembles.

Let us recall the three microscopic scaling limits for the complex ($\beta=2$) induced Ginibre ensemble: the bulk, the origin, and the edge. The real and quaternionic Ginibre ensembles are Pfaffian rather than determinantal point processes and have additional scaling regimes near the real axis but we will not repeat these structures here, see~\cite{Kanzieper:2002,BS:2009} and references within. The joint density~\eqref{complex:intro:jpdf-C} describes a determinantal point process with $k$-point correlation function given by
\begin{align}
R_k(z_1,\ldots,z_k)&=\det_{1\leq i,j\leq k}\big[K_N(z_i,z_j)\big], \\
K_N(u,v)&=\abs{uv}^{\nu}e^{-\abs{u}^2/2-\abs v^2/2}\sum_{\ell=0}^{N-1}\frac{(uv^*)^\ell}{\pi\Gamma[\nu+k+1]}. \nn
\end{align}
The simple structure of the finite-$N$ correlation function enables us to perform the scaling explicitly in all three regions. The macroscopic density is uniform in the region of the support and the microscopic scaling is 
\begin{equation}
\lim_{N\to\infty}\frac1{N^k}R_k\Big(N^{\frac12}\Big(z_\ast+\frac{z_1}{N^{\frac12}}\Big),\ldots,N^{\frac12}\Big(z_\ast+\frac{z_k}{N^{\frac12}}\Big)\Big)=\det_{1\leq i,j\leq k}\big[K_*(z_i,z_j)\big],
\label{complex:intro:corr-scaling}
\end{equation}
where $K_*(u,v)$ is the limiting kernel. 
Assuming $\nu$ is kept fixed, the limiting kernel reads
\begin{equation}
K_\bulk(u,v)=\frac1\pi e^{-\abs u^2/2-\abs v^2/2+uv^*}
\label{complex:intro:bulk}
\end{equation}
for $0<\abs{z_*}<1$. This is the well-known universal microscopic correlation kernel in the bulk. If we look at microscopic correlations close to either the origin ($z_*=0$) or the edge ($\abs{z_*}=1$) then there will be a correction term to~\eqref{complex:intro:bulk} given by
\begin{equation}
\frac{K_\origin(u,v)}{K_\bulk(u,v)}=1-\nu\frac{\Gamma[\nu,uv^*]}{\Gamma[\nu+1]}
\qquad\text{and}\qquad
\frac{K_\edge(u,v)}{K_\bulk(u,v)}=\frac{\erfc\big[\sqrt{2}(zv^*+uz^*)\big]}2,
\label{complex:intro:origin}
\end{equation}
respectively. Here, $\Gamma[\nu,z]$ is the incomplete gamma function and $\erfc[z]$ is the complementary error function. We note that for $\nu=0$, we have $K_\origin(u,v)=K_\bulk(u,v)$ which is the well known result from the Ginibre ensemble. It should be mentioned that the microscopic correlations in the bulk and at the edge are known to be universal in the sense of~\cite{TV:2012}. The microscopic correlations at the origin for $\nu\neq0$ have previously been studied in~\cite{Akemann:2001}. It is also of physical relevance since it may be linked to the microscopic density of the QCD baryon number Dirac operator (see~\cite{IS:2012}) in a sector with non-zero topological charge $\nu$. 

\plainbreak{1}

With the structure of a single induced Ginibre matrix in mind, we turn to products of such matrices. Contrary to our description of the Wishart product ensembles (chapter~\ref{chap:singular}) where we took the density presentation~\eqref{prologue:square:density-product} as a starting point, we will start from the representation~\eqref{prologue:square:general}. Using the density for induced Ginibre matrices~\eqref{prologue:outline:induced} in~\eqref{prologue:square:general} gives
\begin{equation}
P^\beta(Y_n)=\bigg[\prod_{i=1}^n\frac1{Z_i^\beta}\int_{\F_\beta^{N\times N}} d^\beta X_i\,\det (X_i^\dagger X_i)^{\beta\nu_{i}/2\gamma}
e^{-\frac{\beta}{2\gamma}\tr X_i^\dagger X_i}\bigg]\,\delta^\beta(X_n\cdots X_1-Y_n),
\label{complex:intro:density-matrix}
\end{equation}
with $\nu_1,\ldots,\nu_n$ denoting non-negative integers and $\gamma=1,1,2$ for $\beta=1,2,4$.
The joint density for the eigenvalues is obtained by integrating out all irrelevant degrees of freedom, i.e. if $\lambda_i$ ($i=1,\ldots,N$) denote the eigenvalues of the matrix $Y_n$ distributed with respect to~\eqref{complex:intro:density-matrix}, then the joint probability density function for the eigenvalues is given by
\begin{equation}
\cP_\jpdf^\beta(z_1,\ldots,z_N)=\int_{\F_\beta^{N\times N}} d^\beta Y_nP^\beta(Y_n)\prod_{k=1}^N\delta(\lambda_{k}-z_k),
\label{complex:intro:jpdf-def}
\end{equation}
with $\F_{\beta=1,2,4}=\R,\C,\H$ and $n$, $N$ and $\nu_i$ ($i=1,\ldots,n$) as above.
Recall that the constants $\nu_i$ ($i=1,\ldots,n$) which appear in~\eqref{prologue:outline:induced} may be interpreted as incorporating the structure of a product of rectangular Gaussian random matrices~\eqref{prologue:outline:gauss}. Moreover, the induced Ginibre matrices are statistically isotropic, thus the ordering of the factors is irrelevant and we may again choose the ordering $\nu_1\leq\cdots\leq\nu_n$ without loss of generality.

The rest of this chapter is organised such that section~\ref{sec:complex:C},~\ref{sec:complex:H}, and~\ref{sec:complex:R} deals with complex, quaternionic, and real matrices, respectively. A summary and a discussion of open problems will be provided in section~\ref{sec:complex:discuss}. It should be mentioned that some results for the special case of a product of \emph{two} random matrices pre-date the more general results presented in the rest of this chapter. A product of two Ginibre matrices appears as a certain limit in a model of QCD at finite baryon chemical potential~\cite{Osborn:2004,Akemann:2005,APS:2009,APS:2010} (see~\cite{Akemann:2011} for a review) and was also considered in~\cite{KS:2010}. Similarly, the so-called spherical ensembles~\cite{Krishnapur:2009,FM:2012,Mays:2013} are given as a product of a Ginibre and an inverse Ginibre matrix; such ensembles arise when considering a random matrix version of the generalised eigenvalue problem, or matrix pencil~\cite{GL:1996}.

\section{Products of complex Ginibre matrices}
\label{sec:complex:C}

First, let us consider the simplest case: products of complex induced Ginibre matrices. The case of square complex Ginibre matrices (i.e. with all charges $\nu_i$ equal to zero) was originally solved by Akemann and Burda~\cite{AB:2012}. The generalisation to rectangular or induced matrices can be found in~\cite{ARRS:2013,IK:2014}. The first step is to find the joint density for the eigenvalues.

\begin{proposition}\label{thm:complex:jpdf-C}
The eigenvalues of a product of $n$ independent $N\times N$ induced Ginibre matrices with charges $\nu_1,\ldots,\nu_n$ form a point process on the complex plane with joint density  
\begin{equation}
\cP_\jpdf^{\beta=2}(z_1,\ldots,z_N)=\frac{1}{\cZ^{\beta=2}}\prod_{k=1}^Nw_n^{\beta=2}(z_k)\prod_{1\leq i<j\leq N}\abs{z_j-z_i}^2,
\label{complex:C:jpdf}
\end{equation}
where $\cZ^{\beta=2}$ is a normalisation constant and $w_n^{\beta=2}(z)$ is a weight function given by
\begin{equation}
\cZ^{\beta=2}=N!\,\pi^N\prod_{k=1}^N\prod_{\ell=1}^n\Gamma[\nu_\ell+k]
\quad\text{and}\quad
w_n^{\beta=2}(z)=\MeijerG{n}{0}{0}{n}{-}{\nu_1,\ldots,\nu_n}{\abs z^2},
\label{complex:C:norm+weight}
\end{equation}
respectively. 
\end{proposition}

\begin{proof}
In order to find the joint density for the eigenvalues, we need to change variables. Here, we use a generalised complex Schur decomposition (proposition~\ref{thm:decomp:gen-schur-C}) to write
\begin{equation}
X_i=U_i(\Lambda_i+T_i)U_{i-1},\qquad i=1,\ldots,n,
\end{equation}
where each $U_i$ denotes a unitary matrix ($U_0:=U_n$), each $T\in\C^{N(N-1)/2}$ is a strictly upper-triangular complex matrix, and $\Lambda_i=\diag(\lambda_{i,1},\ldots,\lambda_{i,N})$ with $\lambda_{i,j}\in\C$ for all $j=1,\ldots,N$. With this parametrisation the eigenvalues of the product matrix $Y_n$ are given by $\lambda_j=\lambda_{1,j}\lambda_{2,j}\cdots\lambda_{n,j}$ ($j=1,\ldots,N$) but $\lambda_{i,j}$ is (generally) not an eigenvalue of $X_i$. From proposition~\ref{thm:decomp:gen-schur-C}, we know that
\begin{equation}
\prod_{\ell=1}^nd^2 X_\ell=\prod_{1\leq i<j\leq N}\abs{\lambda_j-\lambda_i}^2\prod_{\ell=1}^nd\mu(U_\ell)d^2T_\ell\prod_{k=1}^Nd^2\lambda_{\ell,k} ,
\end{equation}
where $d\mu(U_\ell):=U_\ell^{-1}dU_\ell$ is the Haar measure on $\gU(N)/U(1)^N$, $d^2T_\ell$ is the flat measure on the space of all strictly upper-triangular complex matrices, and $d^2\lambda_{\ell,k}$ is the flat measure on the complex plane. 

We use this parametrisation in~\eqref{complex:intro:jpdf-def}. The integration over $U_i$ and $T_i$ ($i=1,\ldots,n$) decouples and contributes only to the normalisation. It follows that the joint density is~\eqref{complex:C:jpdf} with a weight given as an $n$-fold integral,
\begin{equation}
w_n^{\beta=2}(z)=\pi\bigg[\prod_{\ell=1}^n\frac1\pi\int_{\C}d^2\lambda_\ell \abs{\lambda_\ell}^{2\nu_\ell}e^{-\abs{\lambda_\ell}^2}\bigg]\delta^2(\lambda_n\cdots\lambda_1-z).
\label{complex:C:weight-n-fold}
\end{equation}
After a normalisation, this weight function may be interpreted as the density of a product of complex-valued random scalars. Using the Mellin transform procedure presented in section~\ref{sec:prologue:scalar} we find the Meijer $G$-function formulation of the weight function given in~\eqref{complex:C:norm+weight}. 

It remains to determine the normalisation, which we do by integrating out the variables in~\eqref{complex:C:jpdf}. Using Andr\'eief's integration formula~\cite{Andreief:1883,deBruijn:1955} we find
\begin{equation}
\cZ^{\beta=2}=N!\det_{1\leq k,\ell\leq N}\bigg[\int_\C d^2z\, w_n^{\beta=2}(z) z^{k-1}z^{*\,\ell-1}\bigg].
\end{equation}
The integral within the determinant is zero if $k\neq \ell$, while the $k=\ell$ case can be performed using~\eqref{special:meijer:meijer-moment}. This determines the normalisation and completes the proof.
\end{proof}

\subsection{Correlations for finite size matrices}

Before we look at the correlations of the product ensemble~\eqref{complex:C:jpdf}, we recall a few general properties of orthogonal ensembles in the complex plane. Let us consider an ensemble of $N$ points in the complex plane with joint probability density function
\begin{equation}
\cP_\jpdf(z_1,\ldots,z_N)=\frac{1}{\cZ}\prod_{k=1}^Nw(z_k)\prod_{1\leq i<j\leq N}\abs{z_j-z_i}^2.
\label{complex:C:jpdf-general}
\end{equation}
Here $\cZ$ is a normalisation constant and $w(z)$ is a positive weight function such that
\begin{equation}
\int_\C d^2z\,w(z)\,\abs z^{2k}<\infty,\qquad k=0,1,\ldots,N-1.
\end{equation}
Note that we do not require that all moments are finite, since this would prohibit the study of product ensembles which contain inverse matrices~\cite{Krishnapur:2009,ARRS:2013,AI:2015}. 

We are interested in the $k$-point correlation function defined by 
\begin{equation}
R_k(z_1,\ldots,z_k):=\frac{N!}{(N-k)!}\bigg[\prod_{i=k+1}^N\int_\C d^2z_i\bigg]\cP_\jpdf(z_1,\ldots,z_N).
\end{equation}
Suppose that we have managed to orthogonalise~\eqref{complex:C:jpdf-general}, i.e. we have found a family of monic polynomials, $p_k(z)=z^k+O(z^{k-1})$, with real coefficients, which satisfy the orthogonality relation
\begin{equation}
\inner{p_i,p_j}:=\int_\C d^2z\,w(z)\,p_i(z)p_j(z^*)=h_i\delta_{ij},\qquad i,j=1,2,\ldots,N-1,
\label{complex:C:orthogonal}
\end{equation}
where $h_i$ are the squared norms. We construct the kernel
\begin{equation}
K_N(x,y)=w(x)\sum_{k=0}^{N-1}\frac{p_k(x)p_k(y^*)}{h_k},
\label{complex:C:kernel-general}
\end{equation}
which due the orthogonality relation~\eqref{complex:C:orthogonal} satisfies
\begin{equation}
\int_\C d^2z\,K_N(z,z)=N
\quad\text{and}\quad
\int_\C d^2z\, K_N(x,z)K_N(z,y)=K_N(x,y).
\end{equation}
It follows from general considerations (see e.g.~\cite{Mehta:2004}) that the joint density~\eqref{complex:C:jpdf-general} describes a determinantal point process
\begin{equation}
R_k(z_1,\ldots,z_k)=\det_{1\leq i,j\leq k}\big[K_N(z_i,z_j)\big]
\label{complex:C:R-general}
\end{equation}
with correlation kernel given by~\eqref{complex:C:kernel-general}. In particular, the one-point correlation function (spectral density) is given by
\begin{equation}
R_1^n(z)=K_N(z,z).
\label{complex:C:density-finite}
\end{equation}
As for determinantal point processes on the real line, the correlation kernel is not unique, e.g. if $g(x,x^*)$ is a non-zero function on the relevant domain, then the correlation function~\eqref{complex:C:R-general} is unaffected by a gauge transformation $K_N(x,y)\mapsto (g(x,x^*)/g(y,y^*)) K_N(x,y)$. In fact, our definition of the correlation kernel~\eqref{complex:C:kernel-general} differs from the more conventional choice
\begin{equation}
K_N(x,y)=\sqrt{w(x)w(y)}\sum_{k=0}^{N-1}\frac{p_k(x)p_k(y^*)}{h_k},
\label{complex:C:kernel-sqrt}
\end{equation}
by a gauge transformation $\sqrt{w(x)/w(y)}$. The reason we define the correlation kernel by~\eqref{complex:C:kernel-general} rather than by~\eqref{complex:C:kernel-sqrt} is that we will deal with weight functions defined by a contour integral in the complex plane rather than an exponential of a potential; the square root in~\eqref{complex:C:kernel-sqrt} would force us to do certain computations on the level of correlation functions rather than directly for the kernel.

Similar to the description of real eigenvalues in section~\ref{sec:singular:correlations}, a useful property of the ensemble~\eqref{complex:C:jpdf-general} is that the orthogonal polynomial is related to the expectation (with respect to the joint density) of the characteristic polynomial,
\begin{equation}
\average[\bigg]{\prod_{i=1}^N(z-z_i)}=p_N(z),\qquad z\in\C.
\label{complex:C:average-poly}
\end{equation}
Extensions to products and ratios of characteristic polynomials exist as well~\cite{AV:2003,AP:2004}; all of which follows using standard identities of Vandermonde determinants in complete analogy to the real eigenvalue case. The complex Heine formula~\eqref{complex:C:average-poly} is in general quite useful when searching for orthogonal polynomials, but the description given below turns out to be even simpler since the symmetry of our weight functions implies that the orthogonal polynomials must be monomials. We have the following well-known lemma:

\begin{lemma}\label{thm:complex:monomials}
Given a point process on the complex plane, $\C$, with joint density~\eqref{complex:C:jpdf-general} and weight function such that $w(z)=w(\abs z)$, then the corresponding monic orthogonal polynomials are the monomials, $z$ and $z^*$. 
\end{lemma}

\begin{proof}
The fact that the monic orthogonal polynomials are the monomials, $z$ and $z^*$, follows immediately from the orthogonality relation
\begin{equation}
 \int_{-\pi}^{+\pi}d\theta e^{i(k-\ell)\theta}=2\pi\,\delta_{k\ell},\qquad k,\ell\in\Z,
\label{complex:C:unitary-ortho}
\end{equation}
after a change to polar coordinates, since the weight function is independent of the complex phase. 
\end{proof}

Due to the explicit form of the joint density of the eigenvalues for the product ensemble under consideration (proposition~\ref{thm:complex:jpdf-C}) the correlation kernel follows as a simple corollary~\cite{AB:2012,ARRS:2013,IK:2014}.

\begin{corollary}
The joint density~\eqref{complex:C:jpdf} describes a determinantal point process on the complex plane with kernel
\begin{equation}
K_N^n(x,y)=\MeijerG{n}{0}{0}{n}{-}{\nu_1,\ldots,\nu_n}{\abs{x}^2}\sum_{k=0}^{N-1}\frac{(xy^*)^k}{\pi\prod_{\ell=1}^n\Gamma[\nu_\ell+k+1]}
\label{complex:C:kernel-finite}
\end{equation}
with parameters as in proposition~\ref{thm:complex:jpdf-C}.
\end{corollary}

\begin{proof}
The correlation kernel is given by~\eqref{complex:C:kernel-general} with weight~\eqref{complex:C:norm+weight}. It follows from lemma~\ref{thm:complex:monomials} that the orthogonal polynomials are the monomials, thus the squared norms are
\begin{equation}
h^n_k=\int_\C d^2w_n^{\beta=2}(z)\abs{z}^{2k}=\pi\prod_{\ell=1}^n\Gamma[\nu_\ell+k+1].
\label{complex:C:norm}
\end{equation}
Here the last equality follows from~\eqref{special:meijer:meijer-moment}.
\end{proof}

Yet another interesting property (originally observed for Ginibre matrices in~\cite{Kostlan:1992,Rider:2004}) follows directly from the symmetry of the weight function. 

\begin{proposition}\label{thm:complex:permanent-C}
Given a determinantal point process on the complex plane~\eqref{complex:C:jpdf-general} with weight function such that $w(z)=w(\abs z)$, then
\begin{equation}
\prod_{k=1}^N\int_{-\pi}^\pi d\theta_k\, r_k\,\cP_\jpdf(r_1e^{i\theta_1},\ldots,r_Ne^{i\theta_N})
=\frac1{N!}\per_{1\leq i,j\leq N}\bigg[\frac{2\pi\, r_i^{2j-1}w(r_i)}{h_{j-1}}\bigg]
\label{complex:C:permanent}
\end{equation}
with $r_i\in\R_+$ ($i=1,\ldots,k$), \textup{`$\per$'} denotes a permanent, and the squared norms are given by
\begin{equation}
h_j^n=\int_\C d^2z\,w(z)\abs{z}^{2j}. 
\end{equation}
\end{proposition}

\begin{proof}
Our starting point is the joint density~\eqref{complex:C:jpdf-general}. Expanding the Vandermonde determinants gives
\begin{multline}
\prod_{j=1}^N\int_{-\pi}^\pi d\theta_j\, r_j\,\cP_\jpdf(r_1e^{i\theta_1},\ldots,r_Ne^{i\theta_N})=
\frac{1}{\cZ}\prod_{j=1}^N\int_{-\pi}^\pi d\theta_j\, r_j\,w(r_j)\\
\times\sum_{\sigma,\omega\in S_N}\sign\sigma\sign\omega
\prod_{k,\ell=1}^Nr_k^{\sigma(k)-1}e^{i(\sigma(k)-1)\theta_k}r_\ell^{\sigma(\ell)-1}e^{-i(\sigma(\ell)-1)\theta_\ell}.
\end{multline}
It is trivial to perform the integrals using~\eqref{complex:C:unitary-ortho}, since the weight function is independent of the complex phases. Exploiting that $\cZ=N!\prod_{j=0}^{N-1}h_j$, we write
\begin{equation}
\prod_{k=1}^N\int_{-\pi}^\pi d\theta_k\, r_k\,\cP_\jpdf(r_1e^{i\theta_1},\ldots,r_Ne^{i\theta_N})
=\frac1{N!}\sum_{\sigma\in S_N}\prod_{i=1}^N\frac{2\pi\, r_i^{2\sigma(i)-1}w_n(r_i)}{h_{\sigma(i)-1}}.
\end{equation}
Here the right hand side is recognised as a permanent.
\end{proof}

\begin{remark}
We note that the permanental structure~\eqref{complex:C:permanent} is simpler than those which appear in the statistical description of interacting bosonic systems. In fact, the joint density~\eqref{complex:C:permanent} corresponds to the symmetrised density for $N$ independent random variables $r_k$ ($k=1,\ldots,N$) with densities
\begin{equation}
f_k(r_k)=\frac{2\pi\, r_k^{2k-1}w(r_k)}{h_{k-1}},\qquad k=1,\ldots,N.
\end{equation}
The normalisation of the densities, $f_k$, is ensured by~\eqref{complex:C:orthogonal}.
\end{remark}

\begin{remark}
The permanental structure~\eqref{complex:C:permanent} is useful when considering quantities which only depend on the radial structure, such as hole or overcrowding probabilities at the origin (see e.g.~\cite{HKPV:2009}) or the ``largest'' (meaning furthest away from the origin) eigenvalue. Hole probabilities for products of two matrices were considered in~\cite{APS:2009gap}, while hole and overcrowding probabilities for an arbitrary number of factors were studied in~\cite{AS:2013,AIS:2014}, but we will not repeat that description here. A study of the ``largest'' eigenvalue for a product of complex Ginibre matrices was carried out in~\cite{JQ:2014}.
\end{remark}

\subsection{Macroscopic density}

With exact formulae for correlation functions at arbitrary matrix dimension and an arbitrary number of factors, we can turn our attention towards asymptotic behaviour. First, we consider the macroscopic density which has been studied in several papers~\cite{BJW:2010,BJLNS:2010,GT:2010,OS:2011,AB:2012}. Here, we will follow an approach presented in~\cite{IK:2014} for mixed products of induced Ginibre and truncated unitary matrices. The method is easily extended to include inverse Ginibre and inverse truncated unitary matrices as well, but we will concentrate on the case consisting purely of induced Ginibre matrices.

\begin{lemma}
Let $0\leq\nu_1\leq\cdots\leq\nu_n$ be a collection of real constants and let $s>-\nu_1-1$ be a real parameter, then the $2s$-th absolute moment for the spectral density (one-point correlation function) reads
\begin{equation}
\int_\C d^2z\,R_1^n(z)\,\abs{z}^{2s}
=\sum_{k=0}^{N-1}\prod_{\ell=1}^n\frac{\Gamma[\nu_\ell+k+s+1]}{\Gamma[\nu_\ell+k+1]}.
\label{complex:C:moments-finite}
\end{equation}
The $0$-th moment is normalised to $N$ in the present notation.
\end{lemma}

\begin{proof}
The moments follow directly from~\eqref{complex:C:density-finite} and~\eqref{complex:C:kernel-finite} using the integration formula~\eqref{special:meijer:meijer-moment}.
\end{proof}

\begin{proposition}\label{thm:complex:macro-C}
Let $n$ be a positive integer and $\nu_i=\alpha_iN$ with $\alpha_i\in[0,\infty)$, then there exists a macroscopic limit for the spectral density,
\begin{equation}
\lim_{N\to\infty}N^{n-1} R_1^n(N^{n/2}z)=\rho_\macro^{n,\alpha}(z)
\end{equation}
where the macroscopic density is defined through its Mellin transform
\begin{equation}
\int_\C d^2z\,\rho_\macro^{n,\alpha}(z)\abs{z}^{2s}=\int_0^1dt\prod_{\ell=1}^n(t+\alpha_\ell)^s,\qquad s\in[0,\infty).
\label{complex:C:density-macro}
\end{equation}
In particular, for $\alpha_1=\cdots=\alpha_n=0$  (this includes a product of square matrices), we have
\begin{equation}
\rho_\macro^{n,\alpha=0}(z)=\frac{\abs{z}^{2(1-n)/n}}{\pi\, n}\one_{0\leq \abs{z}\leq 1},
\label{complex:C:density-macro0}
\end{equation}
which has compact support on unit disk.
\end{proposition}

\begin{proof}
In order to find the macroscopic density, we look at its moments. It will be sufficient to investigate the asymptotic form of the absolute moments~\eqref{complex:C:moments-finite}, since the density is invariant under rotations in the complex plane. After rescaling we have 
\begin{equation}
\int_\C d^2z\,N^{n-1}R_1^n(N^{n/2}z)\,\abs{z}^{2s}
=\frac1N\sum_{k=0}^{N-1}\prod_{\ell=1}^n\frac{\Gamma[\nu_\ell+k+s+1]}{N^s\Gamma[\nu_\ell+k+1]}.
\label{complex:C:macro-proof-1}
\end{equation}
We would like to approximate the sum with an integral, which we do using the following bounds:
\begin{equation}
\int_0^N\frac{dt}{N}\prod_{\ell=1}^n\frac{\Gamma[t+\nu_\ell+s]}{N^s\Gamma[t+\nu_\ell]}
\leq\frac{1}{N}\sum_{k=0}^{N-1}\prod_{\ell=1}^n\frac{\Gamma[k+\nu_\ell+s+1]}{N^s\Gamma[k+\nu_\ell+1]}
\leq\int_0^N\frac{dt}{N}\prod_{\ell=1}^n\frac{\Gamma[t+\nu_\ell+s+1]}{N^s\Gamma[t+\nu_\ell+1]}
\end{equation}
This holds since
\begin{equation}
f(t)=\frac{\Gamma[t+\nu_\ell+s]}{N^s\Gamma[t+\nu_\ell]}
\end{equation}
is a non-negative, monotonously increasing function for $t\in[0,\infty)$ given $\nu_\ell,s\geq0$.

In order to evaluate the upper bound, we first recall that $\nu_\ell=N\alpha_\ell$, thus after a change of variables $t/N\mapsto t$,
\begin{equation}
\int_0^N\frac{dt}{N}\prod_{\ell=1}^n\frac{\Gamma[t+\nu_\ell+s+1]}{N^s\Gamma[t+\nu_\ell+1]} 
=\int_0^1dt\prod_{\ell=1}^n\frac{\Gamma[N(t+\alpha_\ell)+s+1]}{N^s\Gamma[N(t+\alpha_\ell)+1]}.
\end{equation}
Here, the asymptotic behaviour of the right hand side may be evaluated using~\eqref{special:gamma:ratio-asymp},
\begin{equation}
\int_0^N\frac{dt}{N}\prod_{\ell=1}^n\frac{\Gamma[t+\nu_\ell+s+1]}{N^s\Gamma[t+\nu_\ell+1]} =\int_0^1dt\prod_{\ell=1}^n(t+\alpha_\ell)^s(1+O(N^{-1})).
\end{equation}
Equivalent manipulations for the lower bound also give
\begin{equation}
\int_0^N\frac{dt}{N}\prod_{\ell=1}^n\frac{\Gamma[t+\nu_\ell+s]}{N^s\Gamma[t+\nu_\ell]}=
\int_0^1dt\prod_{\ell=1}^n(t+\alpha_\ell)^s(1+O(N^{-1})),
\end{equation}
and, thus, agrees with the upper bound to leading order in $N$. 

The integrations and the large-$N$ limits may be interchanged using the dominated convergence theorem. Thus, the two bounds reduce to the same integral and the moments read
\begin{equation}
\lim_{N\to\infty}\int_\C d^2z\,N^{n-1}R_1^n(N^{n/2}z)\,\abs{z}^{2s}=\int_0^1dt\prod_{\ell=1}^n(t+\alpha_\ell)^s.
\end{equation}
The macroscopic density is obtained by an inverse Mellin transform. Note that we can find the density even though the integer-moments do not determine a unique probability measure, since we have all positive moments, which is sufficient for the inverse Mellin transform to be unique. For $\alpha_1=\cdots=\alpha_n=0$, we have
\begin{equation}
\int_0^1dt\,t^{ns}=\frac1{ns+1}
\end{equation}
and it is straightforward to show that the inverse Mellin transform is~\eqref{complex:C:density-macro0}.
\end{proof}

\begin{remark}\label{remark:complex:macro-C}
It is more challenging to find an explicit expression for the macroscopic density for general $\{\alpha_i\}$, but formally we can write
\begin{equation}
\rho_\macro^{n,\alpha}(z)=\int_0^1dt\,\delta(\abs{z}^2-p_\alpha(t)),\qquad p_\alpha(t):=\prod_{\ell=1}^n(t+\alpha_\ell).
\label{complex:C:macro-gen}
\end{equation}
We note that the polynomial $p_\alpha(t)$ is a strictly increasing function for $t>0$, thus the macroscopic density has compact support on an annulus centred in the complex plane with inner and outer radius given by
\begin{equation}
r_\text{in}=p_\alpha(0)=\prod_{\ell=1}^n\alpha_\ell
\qquad\text{and}\qquad
r_\text{out}=p_\alpha(1)=\prod_{\ell=1}^n(1+\alpha_\ell),
\end{equation}
i.e. the macroscopic spectrum has a hole at the origin if and only if $\alpha_\ell>0$ for all $\ell$. 

Also the algebraic equation for the Green function originally given in~\cite{BJLNS:2010} is readily obtained from~\eqref{complex:C:macro-gen}. Inversion of the relation $\pi\rho_\macro^{n,\alpha}(z)=\p G^\alpha_n(z)/\p z^*$ yields
\begin{equation}
G_n^\alpha(z)=\frac1z\int_0^1dt\,\one_{\abs{z}^2\leq p_\alpha(t)}=\frac{p_\alpha^{-1}(\abs{z}^2)}{z}
\end{equation}
or equivalently
\begin{equation}
p_\alpha(zG_n^\alpha(z))=\prod_{\ell=1}^n(zG_n^\alpha(z)+\alpha_\ell)=\abs{z}^2,
\end{equation}
which is the algebraic equation from~\cite{BJLNS:2010} (see also~\cite{GKT:2014}). Finally, we note that if $\alpha_1=\cdots=\alpha_m=0$ and $\alpha_{m+1}=\cdots=\alpha_n>0$ for some $m\geq1$ (i.e. there is no macroscopic hole) then the macroscopic density diverges like $\abs{z}^{2(1-m)/m}$ at the origin~\cite{BJLNS:2010}. 
\end{remark}

\subsection{Microscopic correlations}

Let us consider the microscopic correlations near the origin. An explicit expression for such correlations was first obtained by Akemann and Burda in~\cite{AB:2012} for the $\nu_1=\cdots=\nu_n=0$ case. Here, we present the general result for rectangular matrices, which is a straightforward generalisation due to the reduction formula from chapter~\ref{chap:prologue} (proposition~\ref{thm:pro:rect-square}).

\begin{proposition}[Origin]\label{thm:complex:origin}
Let $n$ be a positive integer and $\nu_1\leq\ldots\leq\nu_n$ be fixed, then we have
\begin{equation}
\lim_{N\to\infty}K_N^n(x,y)=K_\origin^{n,\nu}(x,y)
\end{equation}
uniformly for $z_1,\ldots,z_k$ in compact subsets of the complex plane, where
\begin{equation}
K_\origin^{n,\nu}(x,y)=\frac1\pi\MeijerG{n}{0}{0}{n}{-}{\nu_1,\ldots,\nu_n}{\abs{x}^2}
\MeijerG{1}{1}{1}{n+1}{0}{0,-\nu_1,\ldots,-\nu_n}{-xy^*}.
\label{complex:C:kernel-origin}
\end{equation}
\end{proposition}

\begin{proof}
We need to take the large-$N$ in~\eqref{complex:C:kernel-finite} without any further rescaling. The weight function~\eqref{complex:C:norm+weight} is independent of $N$, hence we only need to evaluate the sum. From~\eqref{special:meijer:hypergeometric} we have
\begin{equation}
\hypergeometric{1}{n}{1}{\nu_1+1,\ldots,\nu_n+1}{xy^*}=
\sum_{k=0}^{\infty}\prod_{\ell=1}^n\frac{\Gamma[\nu_\ell+1]}{\Gamma[\nu_\ell+k+1]}(xy^*)^k,
\label{complex:C:hard-proof-hyper}
\end{equation}
which immediately gives the limiting expression for the sum in~\eqref{complex:C:kernel-finite} as a hypergeometric function. Rewriting the hypergeometric function as a Meijer $G$-function using~\eqref{special:meijer:hyper-meijer}, we find the microscopic kernel~\eqref{complex:C:kernel-origin}.
\end{proof}

We also have a reduction formula similar to~\eqref{singular:asymp:reduction}: 

\begin{proposition}
Let $n\geq2$ and let $\nu_i$ ($i=1,\ldots,n-1$) be fixed, then
\begin{equation}
\lim_{\nu_n\to\infty}\nu_nK_\origin^{n,\nu}(\nu_nx,\nu_ny)=K_\origin^{n-1,\nu}(x,y),
\end{equation}
uniformly for $x$ and $y$ in compact subsets of $\C$.
\end{proposition}

\begin{proof}
The first step is to write the kernel~\eqref{complex:C:kernel-origin} as its double contour representation using~\eqref{special:meijer:meijer-def}, 
\begin{multline}
\nu_nK_\origin^{n,\nu}(\nu_nx,\nu_ny)=
\frac{\nu_n}{\pi}\frac{1}{(2\pi i)^2}\int_{-1-i\infty}^{-1+i\infty}dt\int_\Sigma ds(\nu_n\abs{x}^2)^t(-\nu_nxy^*)^s\\
\times\Gamma[s+1]\Gamma[-s]\prod_{\ell=1}^n\frac{\Gamma[\nu_\ell-t]}{\Gamma[\nu_\ell+s+1]},
\end{multline}
where the contour, $\Sigma$, is a loop which starts and ends at $+\infty$ and encircles all the non-negative integers in the negative direction without crossing the $t$-contour. From~\eqref{special:gamma:ratio-asymp} we know that
\begin{equation}
\frac{\Gamma[\nu_n-t]}{\Gamma[\nu_n+s+1]}=\nu_n^{-t-s-1}(1+O(\nu_n^{-1})),
\end{equation}
for large $\nu_n$, thus
\begin{multline}
\nu_nK_\origin^{n,\nu}(\nu_nx,\nu_ny)=
\frac{1}{\pi}\frac{1}{(2\pi i)^2}\int_{-1-i\infty}^{-1+i\infty}dt\int_\Sigma ds(\abs{x}^2)^t(-xy^*)^s\\
\times\Gamma[s+1]\Gamma[-s]\prod_{\ell=1}^{n-1}\frac{\Gamma[\nu_\ell-t]}{\Gamma[\nu_\ell+s+1]}(1+O(\nu_n^{-1})).
\end{multline}
An interchange of the large-$\nu_n$ limit and the integrations may be justified by means of the dominated convergence theorem. The proposition follows after re-evaluating the contour integrals as Meijer $G$-functions.  
\end{proof}

\begin{remark}
For $n=1$ the kernel reduces to~\eqref{complex:intro:origin} (up to a gauge factor which does not affect the correlations), while for $n=2$ and $(\nu_1,\nu_2)=(0,\nu)$ the kernel reads
\begin{equation}
\frac12K_\origin^{n=2,(0,\nu)}\Big(\frac x2,\frac y2\Big)=\bigg(\frac{x^*}{y^*}\bigg)^{\frac\nu2}K_\nu(\abs{x})\,I_\nu(\sqrt{xy^*}),
\end{equation}
where $K_\nu(x)$ and $I_\nu(x)$ are modified Bessel functions, and the prefactor $(x^*/y^*)^{\nu/2}$ is a gauge factor which does not contribute to the correlations. This type of kernel has previously been observed in a matrix model related to quantum chromodynamics at finite chemical potential, see~\cite{Akemann:2011} for a review.
\end{remark}

A direct derivation of the microscopic correlation in the bulk and at the soft edge from the exact expression~\eqref{complex:C:kernel-finite} was obtained by Akemann and Burda~\cite{AB:2012} and more rigorously by Liu and Wang~\cite{LW:2014}. They showed:

\begin{proposition}\label{thm:complex:bulk+edge}
Let $z_*$ be a fixed point in the complex plane so that $0<\abs{z_*}\leq 1$ and let $\rho:=\abs{z_*}^{2(1-n)/n}/\pi n$. For fixed parameters $n$ and $\nu_\ell$ ($\ell=1,\ldots,n$), we have
\begin{equation}
\lim_{N\to\infty} \bigg(\frac{N^{n}}{N\rho}\bigg)^k R_k^n\Big(N^{n/2}\Big(z_*+\frac{z_1}{\sqrt{N\rho}}\Big),\ldots,N^{n/2}\Big(z_*+\frac{z_k}{\sqrt{N\rho}}\Big)\Big)
=\det_{1\leq i,j\leq k}\big[K_*(z_i,z_j)\big]
\end{equation}
uniformly for $z_1,\ldots,z_k$ in compact subsets of the complex plane, where $K_*(x,y)=K_\bulk(x,y)$ for $0<\abs{z_*}<1$ and $K_*(x,y)=K_\edge(x,y)$ for $\abs{z_*}=1$ with $K_\bulk(x,y)$ and $K_\edge(x,y)$ given by~\eqref{complex:intro:bulk} and~\eqref{complex:intro:origin}, respectively.
\end{proposition}

\begin{remark}
It should be noted that the joint density~\eqref{complex:C:jpdf}, from which the proposition is obtained, takes the form
\begin{equation}
\cP_\jpdf^{\beta=2}(z_1,\ldots,z_N)=\frac{1}{\cZ^{\beta=2}}\prod_{k=1}^Ne^{-V(\abs{z_k}^2)}\prod_{1\leq i<j\leq N}\abs{z_j-z_i}^2,
\end{equation}
where the potential $V(x)=\log w_n^{\beta=2}(\sqrt x)$ is real analytic for $x\in(0,\infty)$. Thus the complex eigenvalues of a product of induced Ginibre matrices may still be thought of as a log-gas in the usual way (this is unlike the Wishart product ensemble from the previous chapter where the eigenvalue repulsion changed). Ensembles of this type have been studied prior to the product ensembles and certain universality results are known, see~\cite{Berman:2008,Zabrodin:2006,AHM:2011} and references therein.
\end{remark}

\section{Products of quaternionic Ginibre matrices}
\label{sec:complex:H}

Now, let us turn to products of quaternionic ($\beta=4$) matrices (see appendix~\ref{sec:decomp:H} for an introduction to quaternions). Unlike complex Ginibre matrices, the finite-$N$ spectrum will not be invariant under rotation in the complex plane, cf. figure~\ref{fig:complex:scat-intro}. The repulsion from the real axis is due to the complex conjugate pairing of eigenvalues induced by the quaternionic (symplectic) symmetry. The description of quaternionic ensembles presented in this section follows~\cite{Ipsen:2013}.

\begin{proposition}\label{thm:complex:jpdf-H}
The eigenvalues of a product of $n$ quaternionic $N\times N$ induced Ginibre matrices form a point process on the upper half-plane, $\C_+:=\{z\in\C:\Im z\geq 0\}$, with joint probability density function
\begin{equation}
\cP_\jpdf^{\beta=4}(z_1,\ldots,z_N)=\frac{1}{\cZ^{\beta=4}}\prod_{k=1}^Nw_n^{\beta=4}(z_k)\,\abs{z_k-z_k^*}^2
\prod_{1\leq i<j\leq N}\abs{z_j-z_i}^2\abs{z_j^*-z_i}^2,
\label{complex:H:jpdf}
\end{equation}
where $\cZ$ is a normalisation constant and $w_n^{\beta=4}(z)$ is a weight function given by
\begin{equation}
\cZ^{\beta=4}=N!\,\pi^N\prod_{k=1}^N\prod_{\ell=1}^n\frac{\Gamma[2(\nu_\ell+k)]}{4^{k}}
\quad\text{and}\quad
w_n^{\beta=4}(z)=\MeijerG{n}{0}{0}{n}{-}{2\nu_1,\ldots,2\nu_n}{2^n\abs z^2},
\label{complex:H:norm+weight}
\end{equation}
respectively.
\end{proposition}

\begin{proof}
We follow the similar steps as in the proof of proposition~\ref{thm:complex:jpdf-C}: We use a generalised quaternionic Schur decomposition (proposition~\ref{thm:decomp:gen-schur-H}) to parametrise the matrices $X_i$ ($i=1,\ldots,n$) as
\begin{equation}
X_i=U_i(\Lambda_i+T_i)U_{i-1},\qquad i=1,\ldots,n,
\end{equation}
where each $U_i$ denotes a unitary symplectic matrix ($U_0:=U_n$), $T\in\H^{N(N-1)/2}$ denotes a strictly upper-triangular quaternionic matrix, and $\Lambda_i=\diag(\lambda_{i,1},\lambda_{i,1}^*,\ldots,\lambda_{i,N},\lambda_{i,N}^*)$ with $\lambda_{i,j}\in\C$ for all $j=1,\ldots,N$. It follows that the eigenvalues of the product matrix $Y_n$ are given by $\lambda_j=\lambda_{1,j}\lambda_{2,j}\cdots\lambda_{n,j}$ and $\lambda_j^*=\lambda_{1,j}^*\lambda_{2,j}^*\cdots\lambda_{n,j}^*$ ($j=1,\ldots,N$). In this case, the corresponding change of measure reads (proposition~\ref{thm:decomp:gen-schur-H})
\begin{equation}
\prod_{\ell=1}^nd^4 X_\ell=
\prod_{1\leq i<j\leq N}\abs{\lambda_j-\lambda_i}^2\abs{\lambda_j-\lambda_i^*}^2\prod_{k=1}^N\abs{\lambda_k-\lambda_k^*}^2
\prod_{\ell=1}^nd\mu(U_\ell)d^4T_\ell\prod_{k=1}^Nd^2\lambda_{\ell,k} ,
\end{equation}
where $d\mu(U_i)=[U_i^{-1}dU_i]$ is the Haar measure on $\gUSp(2N)/U(1)^N$, $d^4T_i$ and $d^2\lambda_{i,j}$ are the flat measure on the space of all strictly upper-triangular quaternionic matrices and the upper-half of the complex plane, respectively. Using this parametrisation in~\eqref{complex:intro:jpdf-def} implies that integration over $U_i$ and $T_i$ ($i=1,\ldots,n$) only contributes with a constant. It follows that the joint density is given by~\eqref{complex:H:jpdf} with a weight function expressed as an $n$-fold integral,
\begin{equation}
w_n^{\beta=4}(z)=\pi\bigg[\prod_{\ell=1}^n\frac1\pi\int_{\C}d^2\lambda_\ell 2^{2\nu_\ell}\abs{\lambda_\ell}^{4\nu_\ell}e^{-2\abs{\lambda_\ell}^2}\bigg]\delta^2(\lambda_n\cdots\lambda_1-z).
\end{equation}
This formula is identical to~\eqref{complex:C:weight-n-fold}, except for factors of two in the exponentials and the powers. As before, the Meijer $G$-function formulation of the weights can be obtained from the moments by means of an inverse Mellin transform. Furthermore, using that the eigenvalues come in complex conjugate pairs, we can choose $z_i$ ($i=1,\ldots,N$) in~\eqref{complex:H:jpdf} to lie in the upper-half of the complex plane without loss of generality. 

In order to determine the normalisation, we apply de~Bruijn's integration formula~\cite{deBruijn:1955}, which yields
\begin{equation}
\cZ^{\beta=4}=N!\pf_{1\leq k,\ell\leq 2N}\bigg[ \frac12\int_\C d^2z\,w_n^{\beta=4}(z)\,(z-z^*)(z^{k-1}z^{*\,\ell-1}-z^{\ell-1}z^{*\,k-1})\bigg].
\end{equation}
If $\ell=k\pm1$, then the integral within the Pfaffian can be performed using~\eqref{special:meijer:meijer-moment}; all other cases are zero. Finally, the evaluating the Pfaffian gives the normalisation in~\eqref{complex:H:norm+weight}.
\end{proof}

\subsection{Correlations for finite size matrices}

In this subsection, we will find an explicit expression for $k$-point correlations corresponding to the joint density~\eqref{complex:H:jpdf}. As for a single matrix, the structure differs considerably from the complex case, since the correlations constitute a Pfaffian (rather than determinantal) point process. For this reason, we will first recall a few well-known properties for such processes~\cite{Mehta:2004,MS:1966,Kanzieper:2002}.

We consider a point process of $N$ points on the upper half-plane, $\C_+:=\{z\in\C:\Im z\geq 0\}$, with joint probability density function
\begin{align}
\cP_\jpdf(z_1,\ldots,z_N)
&=\frac{1}{\cZ}\prod_{k=1}^Nw(z_k)\abs{z_k-z_k^*}^2\prod_{1\leq i<j\leq N}\abs{z_j-z_i}^2\abs{z_j-z_i^*}^2 \nn \\
&=\frac{1}{\cZ}\prod_{k=1}^Nw(z_k)(z_k^*-z_k)
\det_{\substack{i=1,\ldots,N\\j=1,\ldots,2N}}\begin{bmatrix} z_i^{j-1}\\z_i^{*j-1}\end{bmatrix},
\label{complex:H:jpdf-general}
\end{align}
where $\cZ$ is a normalisation constant and $w(z)$ is a positive weight function such that
\begin{equation}
\int_{\C}d^2z\,w(z)\abs{z}^{2k}<\infty,\qquad k=0,1,\ldots,2N-1.
\end{equation}
The $k$-point correlation function is defined analogously to the complex complex case,
\begin{equation}
R_k(z_1,\ldots,z_k):=\frac{N!}{(N-k)!}\bigg[\prod_{i=k+1}^N\frac12\int_{\C} d^2z_i\bigg]\cP_\jpdf(z_1,\ldots,z_N).
\end{equation}
Note that the factor $1/2$ is included in the integration, since we are integrating over the entire complex plane which leads to a double counting due to complex conjugate pairing of the eigenvalues.
In order to find explicit expressions for such correlations, we introduce the skew-symmetric product,
\begin{equation}
\inner{f,g}_s:=\frac12\int_{\C} d^2z\,w(z)(z^*-z) (f(z)g(z^*)-f(z^*)g(z)),
\end{equation}
where the weight function stems from~\eqref{complex:H:jpdf-general}. The main idea is to use the symmetries of the determinant in~\eqref{complex:H:jpdf-general} to write its entries in terms of a family of skew-orthogonal polynomials. More precisely, if $p_k(z)=z^k+\cdots$ ($k=1,2,\ldots$) is a family of monic polynomials, then they are said to be skew-orthogonal if
\begin{subequations}
\begin{align}
\inner{p_{2i},p_{2j}}_s=\inner{p_{2i+1},p_{2j+1}}_s &=0 \label{complex:H:skew-0}\\
\inner{p_{2i+1},p_{2j}}_s=-\inner{p_{2i},p_{2j+1}}_s &=2h_{2i+1}\delta_{ij} \label{complex:H:skew-h}
\end{align}
\end{subequations}
for $i,j=0,1,2,\ldots$. Here, $\{h_{2i+1}\}$ denotes a family of positive constants. Given such a family of skew-orthogonal polynomials, we construct a 
$2\times 2$ matrix kernel
\begin{equation}
K_N(x,y)=\sqrt{(x^*-x)w(x)(y^*-y)w(y)}
\begin{bmatrix}
\kappa_N(x,y^*) & -\kappa_N(x^*,y^*) \\
\kappa_N(x,y) & -\kappa_N(x^*,y) 
\end{bmatrix},
\label{complex:H:kernel}
\end{equation}
with a pre-kernel given terms of the skew-orthogonal polynomials,
\begin{equation}
\kappa_N(x,y):=\sum_{k=0}^{N-1}\frac{p_{2k+1}(x)p_{2k}(y)-p_{2k+1}(y)p_{2k}(x)}{2h_{2k+1}}.
\label{complex:H:pre-kernel}
\end{equation}
Due to the skew-orthogonality, we have
\begin{equation}
\frac12\int_{\C}d^2z\,K_N(z,z)=N\one
\quad\text{and}\quad
\frac12\int_{\C}d^2z\,K_N(x,z)K_N(z,y)=K_N(x,y),
\end{equation}
where $\one$ is the $2\times 2$ identity matrix. This implies that~\eqref{complex:H:jpdf-general} is a Pfaffian point process (see~\cite{Mehta:2004}),
\begin{equation}
R_k(z_1,\ldots,z_k)=\pf_{1\leq i,j\leq N}\big[K_N(z_i,z_j)\big].
\label{complex:H:pfaff}
\end{equation}
In particular, the spectral density (one-point correlation function) is given by
\begin{equation}
R_1(z)=(z^*-z)w(z)\kappa_N(z,z^*).
\label{complex:H:density-general}
\end{equation}
With the above given generality, the correlations were first obtained in~\cite{Kanzieper:2002}. The \mbox{one-,} two- and three-point correlation functions for the quaternionic Ginibre ensemble were originally found in~\cite{MS:1966}, where also the form of the higher correlation functions was conjectured.%

It remains to determine the skew-orthogonal polynomials. For a general weight function, they may be written as~\cite{Kanzieper:2002}
\begin{align}
p_{2k+1}(z)&=\average[\bigg]{\bigg(z-\sum_{j=1}^k(z-z_j)\bigg)\prod_{i=1}^k(z-z_i)(z-z_i^*)}, \\
p_{2k}(z)&=\average[\bigg]{\prod_{i=1}^k(z-z_i)(z-z_i^*)},
\end{align}
where the expectation is with respect to the joint density~\eqref{complex:H:jpdf-general}.
However, in this thesis we may restrict ourselves to a simpler case due to symmetry of the weight function. From~\cite{Ipsen:2013}, we have the following lemma:

\begin{lemma}\label{thm:complex:skew}
Given a Pfaffian point process on the upper half-plane, $\C_+$, with joint density~\eqref{complex:H:jpdf-general} and weight function such that $w(z)=w(\abs z)$, then the corresponding monic skew-orthogonal polynomials are given by
\begin{equation}
p_{2k+1}(z)=z^{2k+1}
\qquad\text{and}\qquad
p_{2k}(z)=\sum_{i=0}^k\bigg[\prod_{j=i+1}^k\frac{h_{2j}}{h_{2j-1}}\bigg]z^{2i}
\label{complex:H:skew-general}
\end{equation}
with
\begin{equation}
h_k:=\frac12\int_{\C}d^2z\,w(z)\abs z^{2k},\qquad k=0,1,\ldots,2N-1.
\label{complex:H:h-thm}
\end{equation}
\end{lemma}

\begin{proof}
It follows from $w(z)=w(\abs z)$ and the orthogonality relation~\eqref{complex:C:unitary-ortho} that
\begin{equation}
\frac12\int_{\C}d^2z\,w(z) z^{k}z^{*\ell}=h_k\delta_{k\ell},\qquad k,\ell=0,1,\ldots,2N-1,
\label{complex:H:semi-ortho}
\end{equation}
which is seen by changing to polar coordinates. Combining this formula with condition~\eqref{complex:H:skew-0} gives
\begin{equation}
p_{2k+1}(z)=\sum_{j=1}^kc_{2j+1}^kz^{2j+1}
\qquad\text{and}\qquad
p_{2k}(z)=\sum_{j=1}^kc_{2j}^kz^{2j},
\end{equation}
respectively. Here, $c_j^k$ denote real coefficients with $c_{2k+1}^k=c_{2k}^k=1$ for all $k$ due to monic normalisation. The fact that all odd polynomials are monomials follows by induction, since $p_1(z)=z$ and if $p_{2j+1}(z)=z^{2j+1}$ for $j=0,1,\ldots,k-1$, then
\begin{equation}
\inner{p_{2k+1},p_{2\ell}}_s=2\sum_{j=0}^{k}c^k_{2j+1}h_{2j+1}\delta_{j\ell},
\end{equation}
which together with~\eqref{complex:H:skew-h} implies that $c^k_{2j+1}=0$ for $j<k$. Finally, we insert the monomial form for the odd polynomial into~\eqref{complex:H:skew-h} and perform the integrals using~\eqref{complex:H:semi-ortho}, which gives the recurrence relation $c_{2j-2}^k=c_{2j}^kh_{2j}/h_{2j-1}$. Combining this recursion relation with the monic normalisation completes the proof.
\end{proof}

\begin{corollary}\label{thm:complex:prekernel-finite}
The eigenvalues of a product of $n$ independent induced $N\times N$ quaternionic Ginibre matrices with indices $\nu_1\leq\cdots\leq\nu_n$ form a Pfaffian point process~\eqref{complex:H:pfaff} on the upper half-plane with kernel~\eqref{complex:H:kernel}, where the weight function and pre-kernel given by~\eqref{complex:H:norm+weight} and
\begin{equation}
\kappa_N^n(x,y)=2\pi^{n/2-1}\sum_{k=0}^{N-1}\sum_{j=0}^k \frac{x^{2k+1}y^{2j}-y^{2k+1}x^{2j}}{\prod_{\ell=1}^n4^{\nu_\ell}\Gamma[\nu_\ell+j+1]\Gamma[\nu_\ell+k+3/2]},
\label{complex:H:pre-kernel-gauss}
\end{equation}
respectively. 
\end{corollary}

\begin{proof}
The weight function~\eqref{complex:H:norm+weight} satisfies the condition $w_n(z)=w(\abs z)$ and according to lemma~\ref{thm:complex:skew} it is sufficient to determine the constants~\eqref{complex:H:h-thm}, which is achieved by means of the integration formula~\eqref{special:meijer:meijer-moment},
\begin{equation}
h_k=\frac12\int_{\C}d^2z\,w_n(z)\,\abs z^{2k}=\frac\pi2\prod_{\ell=1}^n\frac{\Gamma[2\nu_\ell+k+1]}{2^{k+1}}.
\label{complex:H:h}
\end{equation}
Using this together with~\eqref{complex:H:skew-general} and the multiplication formula for gamma functions~\eqref{special:gamma:multiplication} in the expression for the pre-kernel~\eqref{complex:H:pre-kernel} yields~\eqref{complex:H:pre-kernel-gauss}.
\end{proof}

\begin{remark} Note that corollary~\ref{thm:complex:prekernel-finite} is in agreement with previous results for the quaternionic Ginibre ensemble~\cite{Mehta:2004,MS:1966,Kanzieper:2002} and for products of two matrices~\cite{Akemann:2005}, e.g. for $n=1$ and $\nu_1=0$, we have the well-known expression  
\begin{equation}
\frac12\kappa_N^{n=1}\bigg(\frac x{\sqrt2},\frac y{\sqrt2}\bigg)=\frac1{2\pi}\sum_{k=0}^{N-1}\sum_{j=0}^k \frac{x^{2k+1}y^{2j}-y^{2k+1}x^{2j}}{(2j)!!(2k+1)!!}.
\end{equation}
Here, the rescaling appears because we are using the weight function $w_{n=1}(z)=e^{-2\abs z^2}$ rather than $e^{-\abs z^2}$.
\end{remark}

Equivalently to proposition~\ref{thm:complex:permanent-C}, the joint density simplifies considerably if we integrate out the complex phases~\cite{Rider:2004,AIS:2014}.

\begin{proposition}\label{thm:complex:permanent-H}
Given a Pfaffian point process~\eqref{complex:H:jpdf-general} on the upper-half plane, $\C_+$, with weight function such that $w(z)=w(\abs z)$, then
\begin{equation}
\prod_{k=1}^N\frac12\int_{-\pi}^\pi d\theta_k\, r_k\,\cP_\jpdf(r_1e^{i\theta_1},\ldots,r_Ne^{i\theta_N})
=\frac1{N!}\per_{1\leq i,j\leq N}\bigg[\frac{\pi\, r_i^{4j-1}w(r_i)}{h_{2j-1}}\bigg]
\label{complex:H:permanent}
\end{equation}
with $r_i\in\R_+$ ($i=1,\ldots,k$), \textup{`$\per$'} denotes a permanent, and $h_k$ is given by~\eqref{complex:H:h-thm}.
\end{proposition}

\begin{proof}
We start with the joint density~\eqref{complex:H:jpdf-general}
\begin{multline}
\prod_{j=1}^N\frac12\int_{-\pi}^\pi d\theta_j\, r_j\,\cP_\jpdf(r_1e^{i\theta_1},\ldots,r_Ne^{i\theta_N})\\
=\frac{1}{\cZ}\prod_{j=1}^N\frac12\int_{-\pi}^\pi d\theta_j\, r_j\,w(r_j)(r_je^{+i\theta_j}-r_je^{-i\theta_j})
\det_{\substack{k=1,\ldots,N\\ \ell=1,\ldots,2N}}
\begin{bmatrix} r_k^{\ell-1}e^{+i(\ell-1)\theta_k} \\[.1em] r_k^{\ell-1}e^{-i(\ell-1)\theta_k}\end{bmatrix}.
\end{multline}
Expanding the Vandermonde determinant yields
\begin{multline}
\prod_{j=1}^N\frac12\int_{-\pi}^\pi d\theta_j\, r_j\,\cP_\jpdf(r_1e^{i\theta_1},\ldots,r_Ne^{i\theta_N})
=\frac{1}{\cZ}\sum_{\sigma\in S_{2N}} \sign\sigma \prod_{j=1}^Nw(r_j)\,r_j^{\sigma(j)+\sigma(j+N)} \\
\times\frac12\int_{-\pi}^\pi d\theta_j(e^{i(\sigma(j)-\sigma(j+N)+1)\theta_j}-e^{i(\sigma(j)-\sigma(j+N)-1)\theta_j}).
\label{complex:H:permanent-proof-1}
\end{multline}
Here, the integral on the second line can be performed using the orthogonality relation~\eqref{complex:C:orthogonal},
\begin{equation}
\frac12\int_{-\pi}^\pi d\theta_j(e^{i(\sigma(j)-\sigma(j+N)+1)\theta_j}-e^{i(\sigma(j)-\sigma(j+N)-1)\theta_j})
=\pi(\delta_{\sigma(j)+1,\sigma(j+N)}-\delta_{\sigma(j),\sigma(j+N)+1}).
\label{complex:H:delta}
\end{equation}
This implies that that only terms with $\abs{\sigma(j)-\sigma(j+N)}=1$ contribute to the sum in~\eqref{complex:H:permanent-proof-1}, hence the sum over the permutation group $S_{2N}$ reduces to a sum over pairs, i.e. $S_{N}$. It is a straightforward combinatorial task to use~\eqref{complex:H:delta} to show that (see~\cite{Rider:2004,AIS:2014})
\begin{equation}
\prod_{j=1}^N\frac12\int_{-\pi}^\pi d\theta_j\, r_j\,\cP_\jpdf(r_1e^{i\theta_1},\ldots,r_Ne^{i\theta_N})
=\frac{1}{\cZ}\sum_{\sigma\in S_{N}} \prod_{j=1}^N 2\pi\,w(r_{\sigma(j)})\,r_{\sigma(j)}^{4j-1}.
\end{equation}
Using that $\cZ=N!\prod_{j=0}^{N-1}2h_{2j+1}$, we recognise the right hand side as the permanent~\eqref{complex:H:permanent}.
\end{proof}

\begin{remark}
As for complex matrices, the joint density~\eqref{complex:H:permanent} corresponds to the symmetrised density for $N$ independent random variables $r_k$ ($k=1,\ldots,N$) with densities
\begin{equation}
f_k(r_k)=\frac{\pi\, r_k^{4k-1}w(r_k)}{h_{2k-1}},\qquad k=1,\ldots,N.
\end{equation}
The normalisation of the densities, $f_k$, is ensured by~\eqref{complex:H:h-thm}.
\end{remark}

\subsection{Asymptotic formulae for large matrix dimension}

Evaluation of the asymptotic limits for large matrix dimension is more difficult in the quaternionic case than in the complex case (section~\ref{sec:complex:C}). For this reason, some questions remain unanswered even for Gaussian matrices. In this section we repeat the known results originally presented in~\cite{Ipsen:2013,IK:2014}.

First, let us consider the macroscopic density. The macroscopic density is expected to be invariant under rotations in the complex plane (as is well-known for $n=1$) even though the finite-$N$ spectral density is not, see figure~\ref{fig:complex:scat-intro}. For simplicity we will look at the macroscopic limit of the phase-averaged spectral density. 

\begin{lemma}
Let $\nu_1\leq\cdots\leq\nu_n$ be a collection of non-negative integers and let $s>-\nu_1-1$ be a real parameter, then the $2s$-th absolute moment for the spectral density (one-point correlation function) reads
\begin{equation}
\frac12\int_\C d^2z\,R_1^n(z)\,\abs{z}^{2s}
=\sum_{k=0}^{N-1}\prod_{\ell=1}^n\frac{\Gamma[2\nu_\ell+2k+2+s]}{2^s\,\Gamma[2\nu_\ell+2k+2]}.
\label{complex:H:moments-finite}
\end{equation}
The $0$-th moment is normalised to $N$ in the present notation.
\end{lemma}

\begin{proof}
It follows from proposition~\ref{thm:complex:permanent-H} that the absolute moments are given by
\begin{equation}
\frac12\int_\C d^2z\,R_1^n(z)\,\abs{z}^{2s}=\sum_{k=0}^{N-1}\int_0^\infty dr\,\frac{\pi\, r^{4k+2s+1}w_n^{\beta=4}(r)}{h_{2k+1}^n}.
\end{equation}
Using weight~\eqref{complex:H:norm+weight} and the constant~\eqref{complex:H:h} in this expression, we can perform the integral using~\eqref{special:meijer:meijer-moment}.
\end{proof}

\begin{proposition}\label{thm:complex:macro-H}
Let $n$ be a positive integer and $\nu_i=\alpha_iN$ with $\alpha_i\in[0,\infty)$, then there exists a macroscopic limit for the phase-averaged spectral density,
\begin{equation}
\lim_{N\to\infty}\pi\,r \int_{-\pi}^\pi d\theta\,N^{n-1} R_1^n(N^{n/2}r\,e^{i\theta})=\rho_\macro^{n,\alpha}(r)
\end{equation}
where the macroscopic density is defined as in proposition~\ref{thm:complex:macro-C}.
\end{proposition}

\begin{proof}
Similar to the proof of proposition~\ref{thm:complex:macro-C}, we can find the phase-averaged density by looking at the absolute moments~\eqref{complex:H:moments-finite}. With proper rescaling, we have
\begin{equation}
\frac1{2N}\int_\C d^2z\,N^{n-1}R_1^n(N^{n/2}z)\,\abs{z}^{2s}
=\frac1N\sum_{k=0}^{N-1}\prod_{\ell=1}^n\frac{\Gamma[2(\nu_\ell+k+1)+s]}{(2N)^s\Gamma[2(\nu_\ell+k+1)]}.
\end{equation}
Up to a factor of two this expression is identical to~\eqref{complex:C:macro-proof-1}, and we can follow exactly the steps as in the proof of proposition~\ref{thm:complex:macro-C} (the additional factor of two cancels since it appears identically in the denominator and numerator in the fraction).
\end{proof}

\begin{proposition}
Let $n$ be a positive integer and let $\nu_1\leq\cdots\leq\nu_n$ be fixed constants, then
\begin{equation}
\lim_{N\to\infty}\kappa_N^{n}(x,y)=\kappa_\origin^{n,\nu}(x,y)
\end{equation}
for $x$ and $y$ belonging to compact subsets of $\C$, where $\kappa_N^n(x,y)$ is the pre-kernel~\eqref{complex:H:pre-kernel-gauss} and
\begin{equation}
\kappa_\origin^{n,\nu}(x,y)=2\pi^{n/2-1}\sum_{k=0}^{\infty}\sum_{j=0}^k \frac{x^{2k+1}y^{2j}-y^{2k+1}x^{2j}}{\prod_{\ell=1}^n4^{\nu_\ell}\Gamma[\nu_\ell+j+1]\Gamma[\nu_\ell+k+3/2]}.
\label{complex:H:pre-kernel-origin}
\end{equation}
\end{proposition}

\begin{proof}
It is evident from the finite-$N$ expression~\eqref{complex:H:pre-kernel-gauss} that the large-$N$ limit tends to~\eqref{complex:H:pre-kernel-origin} if the limit exists. The series~\eqref{complex:H:pre-kernel-origin} is clearly convergent for $x$ and $y$ in compact subsets of $\C$, thus the proposition follows.
\end{proof}

\begin{remark}
For the Ginibre ensemble ($n=1,\nu=0$), it is known that the series representation~\eqref{complex:H:pre-kernel-origin} may be evaluated as~\cite{Kanzieper:2002}
\begin{equation}
\frac12\kappa_\origin^{n=1,\nu=0}\Big(\frac x{\sqrt2},\frac y{\sqrt2}\Big)
=\frac1{\sqrt\pi}\exp\Big[\frac{x^2+y^2}2\Big]\erf\Big[\frac{x-y}{\sqrt2}\Big].
\label{complex:H:origin-error}
\end{equation}
The idea used in~\cite{Kanzieper:2002} was to first realise that the series was a solution to a pair of first order inhomogeneous differential equations. Solving these equations yielded~\eqref{complex:H:origin-error}. A similar idea was later used for a product of two matrices in~\cite{Akemann:2005}.

It also possible to show that the general expression~\eqref{complex:H:pre-kernel-origin} is a solution to a pair of differential equations,
\begin{multline}
\bigg[\frac{1}{z_\pm}\prod_{\ell=1}^n\Big(z_\pm\frac{\p}{\p z_\pm}+2\nu_\ell\Big)-2z_\pm\bigg]\kappa_\origin^{n,\nu}(z_+,z_-)\\
=\pm4\pi^{\frac{n-3}2}\MeijerG{1}{1}{1}{n+1}{0}{0,-2\nu_1,\ldots,-2\nu_n}{-2^nz_+z_-}.
\label{complex:H:diff-eq-evil}
\end{multline}
For $n=1$ and $\nu=0$, these equations reduces to 
\begin{align}
\bigg[\frac{\p}{\p z_\pm}-z_\pm\bigg]\frac12\kappa_\origin^{n=1,\nu=0}\Big(\frac{z_+}{\sqrt2},\frac{z_-}{\sqrt2}\Big)
=\pm\frac{\sqrt2}\pi e^{z_+z_-}, 
\end{align}
which are equations from~\cite{Kanzieper:2002}. However, it is highly non-trivial to use~\eqref{complex:H:diff-eq-evil} to evaluate the series~\eqref{complex:H:pre-kernel-origin} in the general case. We will not pursue this direction further.
\end{remark}

\begin{remark}
For comparison with numerical data, it is convenient to introduce the phase-averaged density (albeit the microscopic density is not rotational invariant). It follows from~\eqref{complex:H:permanent} with~\eqref{complex:H:norm+weight} and~\eqref{complex:H:h} that
\begin{equation}
\int_{-\pi}^\pi\, d\theta R_1(z)=
2^{n+1}\MeijerG{n}{0}{0}{n}{-}{2\nu_1,\ldots,2\nu_n}{2^nr^2}\sum_{k=0}^{\infty}\frac{(2^{2n}r^4)^{k}}{\prod_{\ell=1}^n\Gamma[2(\nu_\ell+k+1)]}.
\end{equation}
The sum is recognised as a Meijer $G$-function by first rewriting the gamma function using the multiplication formula~\eqref{special:gamma:multiplication} and then using~\eqref{special:meijer:hypergeometric} and~\eqref{special:meijer:hyper-meijer},
\begin{multline}
\int_{-\pi}^\pi\, d\theta R_1(r\,e^{i\theta})=
\frac{2\pi^{n/2}}{\prod_{\ell=1}^n2^{2\nu_\ell}}\MeijerG{n}{0}{0}{n}{-}{2\nu_1,\ldots,2\nu_n}{2^nr^2}\\
\times\MeijerG{1}{1}{1}{2n+1}{0}{0,-\nu_1,\ldots,-\nu_n,-\nu_1-\frac12,\ldots,-\nu_n-\frac12}{-r^4}.
\end{multline}
\end{remark}

\section{Products of real Ginibre matrices}
\label{sec:complex:R}

Finally, let us look at a product of real induced Ginibre matrices. Here, the main difficulty is that the Schur decomposition is incomplete; this means that there exists no (real) orthogonal similarity transformation which brings a real matrix to an upper-triangular form (unless all eigenvalues real). Here, we will follow the approach in~\cite{IK:2014} and write down the joint density for the eigenvalues as a function of the real eigenvalues and $2\times2$ real matrices with eigenvalues that can be either a real pair or a complex conjugate pair. The expression for the joint density simplifies considerably in the case where all eigenvalues are real, which was studied independently by Forrester in~\cite{Forrester:2014a}.  

\begin{proposition}\label{thm:complex:jpdf-R}
Let $N,K,L $ be non-negative integers such that $N=K+2L $ and let $Y_n$ be a product matrix constructed from a product of $n$ independent $N\times N$ real induced Ginibre matrices with parameters as in~\eqref{complex:intro:density-matrix}. Assume that $Y_n$ has at least $K$ real eigenvalues denoted by $x_i$ ($i=1,\ldots,K$), then the joint probability density function of the $K$ real eigenvalues may written as
\begin{multline}
\cP_\jpdf^{\beta=1}(x_1,\ldots,x_K,(Z)_{11},\ldots,(Z)_{L ,L })
=\frac{1}{\cZ^{\beta=1}_{K,L }}\prod_{a=1}^Kw_n^{\beta=1}(x_a)\prod_{b=1}^L  W_n^{\beta=1}((Z)_{bb}) \\
\times\prod_{1\leq i<j\leq K}\abs{x_j-x_i}
\prod_{1\leq i<j\leq L }\abs{\det[(Z)_{jj}^\transpose\otimes\one-\one\otimes (Z)_{ii}]}
\prod_{i=1}^K\prod_{j=1}^L \abs{\det[(Z)_{jj}-\one\otimes x_i]}.
\label{complex:R:jpdf}
\end{multline}
where the remaining $2L $ eigenvalues (real or complex) are the eigenvalues of the $2\times2$ real matrices $(Z)_{ii}$ ($i=1,\ldots,L $).
Here, the normalisation constants are given by
\begin{equation}
\cZ^{\beta=1}_{K,L }=K!L !\pi^{3nL /2}2^{5nL /2}2^{nN(N+1)/4}\prod_{i=1}^n\prod_{j=1}^N\Gamma[(j+\nu_i)/2],
\label{complex:R:normal}
\end{equation}
while the weight functions, $w_n^{\beta=1}(x)$ and $W_n^{\beta=1}(Z)$, are given by
\begin{align}
w_n^{\beta=1}(x)&=\MeijerG[\bigg]{n}{0}{0}{n}{-}{\frac{\nu_1}2,\ldots,\frac{\nu_n}2}{\frac1{2^n}x^2},
\qquad\text{and} \label{complex:R:weight1}\\
W_n^{\beta=1}(Z)&=\prod_{i=1}^n\int_{\R^{2\times 2}}dZ_i\,\det \Big(\frac{Z_i^\transpose Z_i}2\Big)^{\nu_i/2}e^{-\frac12\tr Z_i^\transpose Z_i}\delta^{2\times 2}(Z_n\cdots Z_1-Z).
\label{complex:R:weight2}
\end{align}
\end{proposition}

\begin{proof}
The main idea of the proof is similar to that of proposition~\ref{complex:C:jpdf} and~\ref{complex:H:jpdf}. We use the generalised real Schur decomposition (proposition~\ref{thm:decomp:gen-schur-R}) to parametrise the matrices $X_i$ ($i=1,\ldots,n$) as
\begin{equation}
X_i=U_i\begin{bmatrix}\Lambda_i + T^{11}_i & T^{12}_i \\ 0 & Z_i + T^{22}_i \end{bmatrix}U_{i-1}^{-1},\qquad i=1,\ldots,n,
\end{equation}
with notation as in proposition~\ref{thm:decomp:gen-schur-R}.
The corresponding change of variables gives rise to a Jacobian
\begin{multline}
\prod_{i=1}^ndX_i=\prod_{1\leq i<j\leq K}\abs{\lambda_j-\lambda_i}
\prod_{1\leq i<j\leq L }\abs{\det[(Z)_{jj}^\transpose\otimes\one-\one\otimes (Z)_{ii}]}
\prod_{i=1}^K\prod_{j=1}^L \abs{\det[(Z)_{jj}-\one\otimes \lambda_i]} \\
\times \prod_{i=1}^ndT^{11}_idT^{12}_idT^{22}_id\mu(U_i)\prod_{j=1}^Kd\lambda_{i,j}\prod_{h=1}^L  dZ_{i,h},
\end{multline} 
where $d\mu(U_i)$ is the Haar measure on $\gO(N)/\gO(2)^L $, while $dT^{11}_i$, $dT^{12}$ and $dT^{22}$ are the flat measures on $\R^{K(K-1)/2}$, $(\R^{1\times 2})^{KL }$, and $(\R^{2\times 2})^{L (L -1)/2}$, respectively. Inserting this back into~\eqref{complex:intro:density-matrix} yields the joint density~\eqref{complex:R:jpdf} with weight functions
\begin{align}
w_n^{\beta=1}(x)&=\prod_{i=1}^n\int_{\R}d\lambda_i\,
\Big(\frac{\lambda_i^2}2\Big)^{\nu_i/2}e^{-\frac12\lambda_i^2}\delta(\lambda_n\cdots \lambda_1-x), 
\label{complex:R:weight1-proof}\\
W_n^{\beta=1}(Z)&=\prod_{i=1}^n\int_{\R^{2\times 2}}dZ_i\,\det \Big(\frac{Z_i^\transpose Z_i}2\Big)^{\nu_i/2}e^{-\frac12\tr Z_i^\transpose Z_i}\delta^{2\times 2}(Z_n\cdots Z_1-Z),
\end{align}
and normalisation constant
\begin{equation}
\frac1{\cZ_{K,L }^{\beta=1}}=\frac1{K!L !}\prod_{i=1}^n\frac{2^{\nu_i/2}}{Z_i^{\beta=1}}
\int d\mu(U_i)\int dT_i^{11}dT_i^{12}dT_i^{22}
e^{-\frac12\tr[ (T_i^{11})^\transpose T_i^{11}+(T_i^{12})^\transpose T_i^{12}+(T_i^{22})^\transpose T_i^{22}]},
\label{complex:R:normal-proof}
\end{equation}
where $Z_i^{\beta=1}$ is the normalisation constant for single induced real Ginibre matrix with charge $\nu_i$. 

It remains to perform the integrals in~\eqref{complex:R:weight1-proof} and in~\eqref{complex:R:normal-proof}. 

For $\nu_1=\cdots=\nu_n=0$, the weight function~\eqref{complex:R:weight1-proof} reduces to the (un-normalised) density for a product of $n$ independent Gaussian random variables, which was considered in section~\ref{sec:prologue:scalar}; the expression~\eqref{complex:R:weight1} is a straightforward generalisation.

In order to determine the normalisation constant, we first recall that~\cite{FBKSZ:2012}
\begin{equation}
Z_i^{\beta=1}=(2\pi)^{N^2/2}2^{N\nu_i/2}\prod_{j=1}^N\frac{\Gamma[(\nu_i+j)/2]}{\Gamma[j/2]}.
\end{equation}
The integral over the orthogonal group is well-known (see~\cite{Fischmann:2013} for a nice review)
\begin{equation}
\int d\mu(U_i)=\frac{\pi^{N(N+1)/4}}{(4\pi)^L }\prod_{j=1}^N\frac1{\Gamma[j/2]},
\end{equation}
while the integral over $T_i^{11}$, $T_i^{12}$, and $T_i^{22}$ is Gaussian, thus
\begin{equation}
\int dT_i^{11}dT_i^{12}dT_i^{22}
e^{-\frac12\tr[ (T_i^{11})^\transpose T_i^{11}+(T_i^{12})^\transpose T_i^{12}+(T_i^{22})^\transpose T_i^{22}]}
=(2\pi)^{N(N-1)/4-L /2}.
\end{equation}
Combining these formulae gives the normalisation~\eqref{complex:R:normal}.
\end{proof}

\begin{remark}\label{remark:complex:real}
The normalisation constants in proposition~\ref{thm:complex:jpdf-R} are such that
\begin{equation}
\P[\,\#\{\ev\in\R\}\geq K]= \prod_{i=1}^K\int_\R dx_i\prod_{j=1}^L  \int_{\R^{2\times 2}}dZ_j\,
\cP_\jpdf^{\beta=1}(x_1,\ldots,x_K,Z_1,\ldots,Z_L )
\end{equation}
gives the probability that there is at least $K$ real eigenvalues, while
\begin{equation}
\P[\,\#\{\ev\in\R\}= K]=  \prod_{i=1}^K\int_\R dx_i\prod_{j=1}^L  \int_{\det Z_j>(\tr Z_j/2)^2}dZ_j\,
\cP_\jpdf^{\beta=1}(x_1,\ldots,x_K,Z_1,\ldots,Z_L )
\end{equation}
gives the probability that there are exactly $K$ real eigenvalues. Normalisation on the full matrix space is ensured by the sum rule
\begin{equation}
\sum_{K=0}^N\P[\,\#\{\ev\in\R\}= K]=1.
\end{equation}
\end{remark}

\begin{remark}
To get the complete expression for the joint density of the eigenvalues, it remains to express the weight function~\eqref{complex:R:weight2} in terms of the eigenvalues only. One way to achieve this is to use the relation between the eigen- and singular values of a $2\times2$ matrix as was done in~\cite{IK:2014}. However, we will not repeat this derivation here, since it gives rise to a rather involved $n$-fold integral representation for the two-point weight function, which seems too complicated for further computations. It remains an open problem to obtain a more compact expression for the two-point weight function; for a product of two Ginibre matrices such an expression was obtained in~\cite{APS:2009,APS:2010}.  
\end{remark}

\subsection{Probability of a purely real spectrum}

In general, the spectrum of a real asymmetric random matrix will consist of both real and complex eigenvalues. Nonetheless, we will see in chapter~\ref{chap:lyapunov} that it will be fruitful to consider the special case where the spectrum is purely real. Thus, it is intriguing to ask: \emph{What is the probability that all eigenvalues are real?} For square matrices (i.e. $\nu_1=\cdots=\nu_n=0$) this question was answered by Forrester in~\cite{Forrester:2014a}. Here, we present the general result for rectangular matrices (for even matrix dimension this result was stated without proof in the review~\cite{AI:2015}).

\begin{proposition}\label{thm:complex:real}
Given a product matrix, $Y_n$, with parameters as in proposition~\ref{thm:complex:jpdf-R}, then the probability that all eigenvalues are real is
\begin{equation}
\P[\,\#\{\ev\in\R\}= N]= \det
\begin{bmatrix}
M_{1,1} & \ldots & M_{1,\frac N2} \\
\vdots & & \vdots \\
M_{\frac N2,1} & \ldots & M_{\frac N2,\frac N2} \\
\end{bmatrix}
\label{complex:R:real-even}
\end{equation}
for $N$ even and
\begin{equation}
\P[\,\#\{\ev\in\R\}= N]= \det
\begin{bmatrix}
M_{1,1} & \ldots & M_{1,\frac {N-1}2} & m_{1} \\
\vdots & & \vdots & \vdots \\
M_{\frac {N+1}2,1} & \ldots & M_{\frac {N+1}2,\frac {N-1}2} & m_{\frac{N+1}2} \\
\end{bmatrix}
\label{complex:R:real-odd}
\end{equation}
for $N$ odd, where
\begin{align}
M_{ij}:=
\frac{\MeijerG[\bigg]{n+1}{n}{n+1}{n+1}{\frac32-\frac{\nu_1}2-i,\ldots,\frac32-\frac{\nu_n} 2-i,1}{\frac{\nu_1}2+j,\ldots,\frac{\nu_n}2+j,0}{1}}
{\prod_{m=1}^n\Gamma[(\nu_m+2j-1)/2]\Gamma[(\nu_m+2j)/2]},\quad
m_{j}:=\prod_{i=1}^n\frac{\Gamma[(\nu_i+2j-1)/2]}{\Gamma[(\nu_i+N)/2]}.
\end{align}
(The Meijer $G$-function which appears in the definition of $M_{ij}$ is continuous at unity.)
\end{proposition}

\begin{proof}
We follow the same steps as in~\cite{Forrester:2014a}: From proposition~\ref{complex:R:jpdf}, we know that the probability that all eigenvalues are real is given by
\begin{equation}
 \P[\,\#\{\ev\in\R\}= N]=\frac1{\cZ_{N,0}^{\beta=1}}\prod_{a=1}^N\int_\R dx_a w_n^{\beta=1}(x_a)\prod_{1\leq i<j\leq N}\abs{x_j-x_i}.
\end{equation}
For $N$ even, it follows from the method of integration over alternate variables~\cite{Mehta:2004,Forrester:2010}
\begin{equation}
\P[\,\#\{\ev\in\R\}= N]=
\frac{2^{N/2}N!}{\cZ_{N,0}^{\beta=1}}\pf_{1\leq i,j\leq N}[\cI_{ij}]
\label{complex:R:real-proof-1}
\end{equation}
with
\begin{equation}
\cI_{ij}:=\int_{\R} dx\int_\R dy\, w_n^{\beta=1}(x)w_n^{\beta=1}(y)x^{i-1} y^{j-1}\sign(y-x).
\label{complex:R:real-cI}
\end{equation}
The weight function, $w_n^{\beta=1}(x)$ is an even function, thus it follows that $\cI_{2i,2j}=\cI_{2i-1,2j-1}=0$. This allows us to write the Pfaffian in~\eqref{complex:R:real-proof-1} as a determinant,
\begin{equation}
\P[\,\#\{\ev\in\R\}= N]=
\frac{2^{N/2}N!}{\cZ_{N,0}^{\beta=1}}\det_{1\leq i,j\leq N/2}[\cI_{2i-1,2j}].
\label{complex:R:real-even-pre}
\end{equation}
Following the usual modification for odd $N$ (see~\cite{Mehta:2004}) gives
\begin{equation}
\P[\,\#\{\ev\in\R\}= N]=\frac{2^{(N+1)/2}N!}{\cZ_{N,0}^{\beta=1}}
\det\bigg[\{\cI_{2i-1,2j}\}^{i=1,\ldots, \frac{N+1}2}_{j=1,\ldots, \frac{N-1}2}\ \bigg\vert\ \{I_{2i-1}\}_{i=1,\ldots,\frac{N+1}2}\bigg]
\label{complex:R:real-odd-pre}
\end{equation}
with
\begin{equation}
I_i:=\int_{\R} dx\, w_n^{\beta=1}(x) x^{i-1}.
\label{complex:R:real-I}
\end{equation}
The formulae~\eqref{complex:R:real-even-pre} and~\eqref{complex:R:real-odd-pre} have the same structure as~\eqref{complex:R:real-even} and~\eqref{complex:R:real-odd}, hence it only remains to compute the integrals~\eqref{complex:R:real-cI} and~\eqref{complex:R:real-I}.

After a simple change of variables, it follows from~\eqref{special:meijer:meijer-moment} that
\begin{equation}
I_{2j-1}=\frac{2^{n(2j-1)/2}}2\prod_{i=1}^n \Gamma[(\nu_i+2j-1)/2].
\label{complex:R:proof-int-1}
\end{equation}
To evaluate~\eqref{complex:R:real-cI}, we first note that $w_n^{\beta=1}(y)y^{2j-1}$ is an odd function, thus
\begin{equation}
\cI_{ij}=2\int_{-\infty}^\infty dx\int_x^\infty dy\, w_n^{\beta=1}(x)w_n^{\beta=1}(y)x^{i-1} y^{j-1}.
\end{equation}
The two remaining integrals can be performed using~\eqref{special:meijer:beta-conv} and~\eqref{special:meijer:int-meijer-meijer}, and it is seen that
\begin{equation}
\cI_{2i-1,2j}=\frac{2^{n(2i+2j-1)/2}}{2}
\MeijerG{n+1}{n}{n+1}{n+1}{\frac32-\frac{\nu_1}2-i,\ldots,\frac32-\frac{\nu_n}2-i,1}{\frac{\nu_1}2+j,\ldots,\frac{\nu_n}2+j,0}{1}.
\label{complex:R:proof-int-2}
\end{equation}
The proposition follows after inserting~\eqref{complex:R:proof-int-2} and~\eqref{complex:R:proof-int-1} into~\eqref{complex:R:real-even-pre} and~\eqref{complex:R:real-odd-pre}.
\end{proof}

\begin{remark}
Note that for $\nu_1=\cdots=\nu_n=0$, proposition~\ref{thm:complex:real} reduces to the result for square matrices presented in~\cite{Forrester:2014a}. Moreover, for $n=1$ and $\nu=0$ the Meijer $G$-functions are easily evaluated and it is seen that the probability that the spectrum is purely real is $2^{-N(N-1)/4}$ in agreement with~\cite{Edelman:1997}. It is slightly more challenging to evaluate the Meijer $G$-function for $n\geq2$ (see~\cite{Kumar:2015} for $n=2$). Table~\ref{table:complex:real} lists the numerical evaluation of~\eqref{complex:R:real-even} for various values of $N$ and $n$, which gives the probability of finding a purely real spectrum. The table suggests that the probabilities increase when $n$ increases, but decrease when $N$ increases. We will return to this question in chapter~\ref{chap:lyapunov} where we will look at the probability of finding a purely real spectrum for $n\to\infty$.   
\end{remark}

\begin{table}[tp]
\centering
\caption{Numerical evaluations of the probability that all eigenvalues are real~\eqref{complex:R:real-even} for various values of matrix dimension, $N$, and the number of factors, $n$, with charges $\nu_1=\cdots=\nu_n=0$ (top panel) and $\nu_1=\cdots=\nu_n=1$ (bottom panel).}
\label{table:complex:real}
\begin{tabular}{>{$}l<{$} | >{$}l<{$} >{$}l<{$} >{$}l<{$} >{$}l<{$} >{$}l<{$}}
\hline\hline
\nu_i=0 & n=2 & n=4 & n=6 & n=8 & n=10 \\
\hline
N=2 &  0.785 & 0.872 & 0.919 & 0.947 & 0.965 \\
N=4 &  0.242 & 0.418 & 0.540 & 0.629 & 0.696 \\
N=6 &  2.89\times 10^{-2} & 0.110 & 0.201 & 0.288 & 0.365 \\
N=8 &  1.32\times 10^{-3} & 1.57\times 10^{-2} & 4.79\times 10^{-2} & 9.18\times 10^{-2} & 0.141 \\
N=10 &  2.32\times 10^{-5} & 1.22\times 10^{-3} & 7.30\times 10^{-3} & 2.05\times 10^{-2} & 4.06\times 10^{-2} \\
\hline\hline
\nu_i=1 & n=2 & n=4 & n=6 & n=8 & n=10 \\
\hline
N=2 & 0.705 & 0.780 & 0.828 & 0.863 & 0.890 \\
N=4 & 0.152 & 0.285 & 0.391 & 0.476 & 0.546 \\
N=6 & 1.20\times 10^{-2} & 5.58\times 10^{-2} & 0.117 & 0.182 & 0.246 \\
N=8 & 3.56\times 10^{-4} & 5.93\times 10^{-3} & 2.22\times 10^{-2} & 4.86\times 10^{-2} & 8.19\times 10^{-2} \\
N=10 & 3.98\times 10^{-6} & 3.44\times 10^{-4} & 2.71\times 10^{-3} & 9.11\times 10^{-3} & 2.03\times 10^{-2} \\
\hline\hline
\end{tabular}
\end{table}

\section{Summary, discussion and open problems}
\label{sec:complex:discuss}

In this chapter, we showed that it is possible find explicit expressions for the joint probability density functions for eigenvalues of a product of an arbitrary number induced Ginibre matrices with arbitary matrix dimension (proposition~\ref{thm:complex:jpdf-C}, \ref{thm:complex:jpdf-H}, and~\ref{thm:complex:jpdf-R}). To this end, it should be mentioned that similar formulae have also been obtained for mixed products involving truncated unitary matrices, and inverse Ginibre matrices~\cite{ARRS:2013,IK:2014,ABKN:2014}, see~\cite{AI:2015} for a review.

In particular, we have verified that the exact formulae for product ensembles reproduce the (universal) macroscopic density previously obtained using other methods~\cite{BJW:2010,BJLNS:2010,GT:2010,OS:2011,GKT:2014}. In fact, the ``product law'' for the macroscopic density is known to be even more general. It was shown in~\cite{GNT:2014elliptic,ORSV:2014}, that the ``product law'' is reobtained when considering products of elliptic matrices as long as the product has at least two factors (see the aforementioned references for precise statements). It would be interesting to see whether it possible to obtain exact expressions for the joint density of the eigenvalues for a product of elliptic Gaussian matrices with arbitrary matrix dimension. We recall that an $N\times N$ random matrix is said to be an elliptic Gaussian matrix if it is distributed according to the density~\cite{SCSS:1988}
\begin{equation}
P(E)=\frac1{Z}\exp\Big[-\frac1{1-\tau^2}\tr\big(EE^\dagger-\frac\tau2(E^2+E^{\dagger\,2})\big)\Big],
\label{complex:discuss:elliptic}
\end{equation}
where $Z$ is a normalisation constant and $\tau\in[-1,1]$ is a parameter such that $E$ is a Ginibre matrix for $\tau=0$ and a Wigner-Dyson matrix (i.e. Hermitian) for $\tau=1$. Figure~\ref{fig:complex:scatter-elliptic} shows scatter plots for products of complex elliptic Gaussian matrices with $\tau=1$ (left column), $\tau=1/2$ (centre column), and $\tau=0$ (right column). The first row on the figure shows the eigenvalue distribution for a single matrix and we recognise the semi-circle law (the eigenvalues are real and in the interval $[-2,2]$), the elliptic law, and the circular law. At first glance the spectra seem rotational symmetric in the second row, which show products of two matrices (in agreement with~\cite{GNT:2014elliptic,ORSV:2014}), but a closer look reveals that the rotational invariance is not exactly true for finite matrix dimension. In particular, we note that the eigenvalues of a product of two Hermitian matrices are either real or come as a complex conjugate pair, which follows 
from the fact that given two $N\times N$ matrices, $E_1$ and $E_2$, then
\begin{equation}
\det(E_1E_2-\lambda)=\det(E_2E_1-\lambda).
\end{equation}
This symmetry between eigenvalues is similar to the symmetry between eigenvalues stemming from the real Ginibre ensemble. It is clear that some non-trivial scaling behaviour must occur if we were to study the microscopic correlations rather than the macroscopic densities. Perhaps, the easiest way to approach these microscopic scaling regimes is to start with elliptic Gaussian ensembles~\eqref{complex:discuss:elliptic}.

\begin{figure}[tp]
\centering

\includegraphics[width=.3\textwidth]{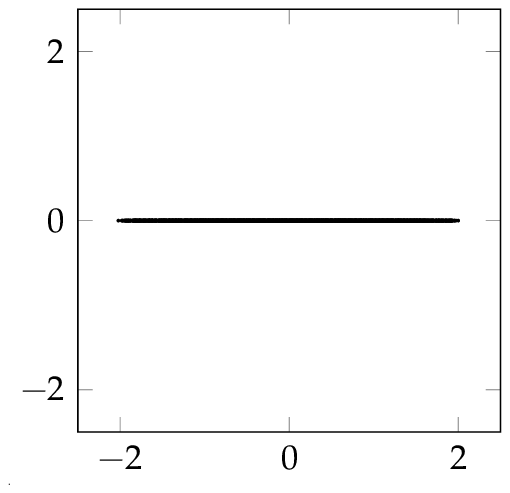}
\includegraphics[width=.3\textwidth]{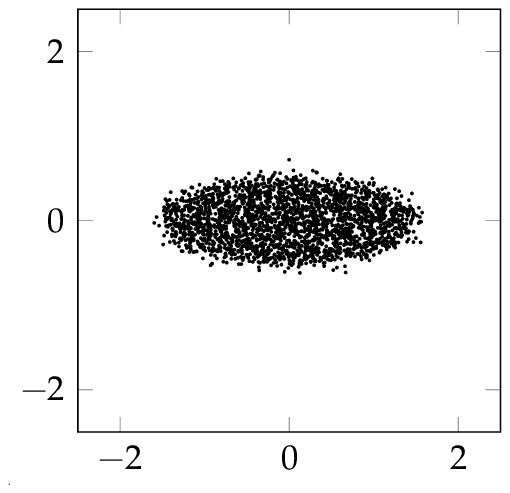}
\includegraphics[width=.3\textwidth]{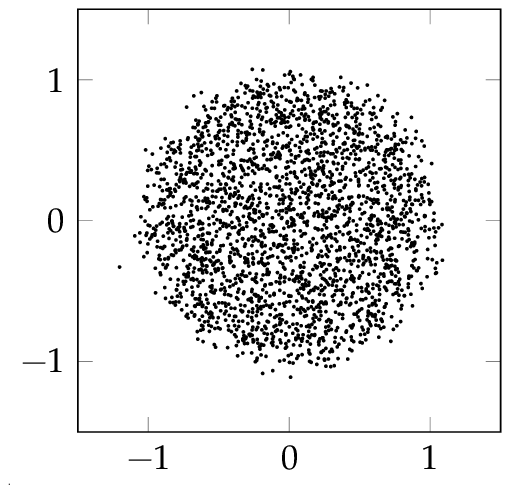}

\includegraphics[width=.3\textwidth]{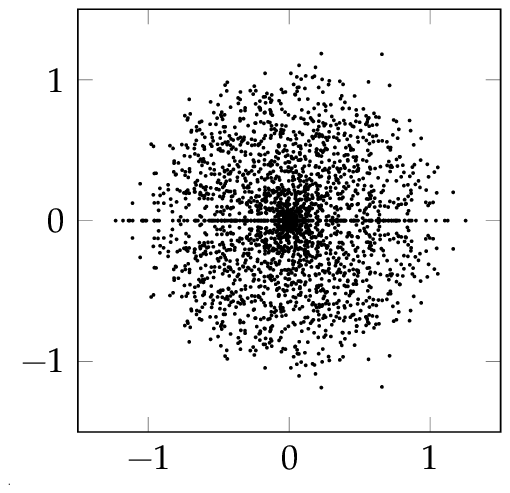}
\includegraphics[width=.3\textwidth]{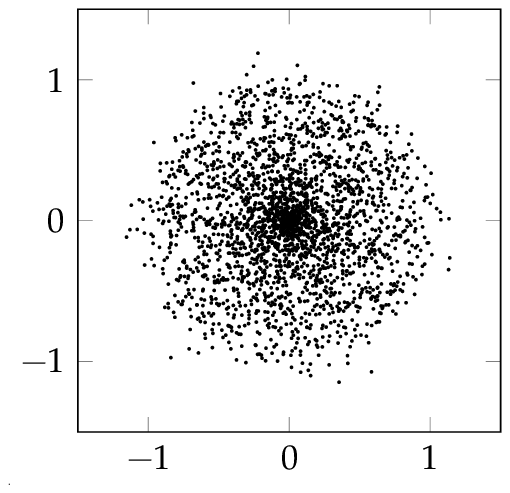}
\includegraphics[width=.3\textwidth]{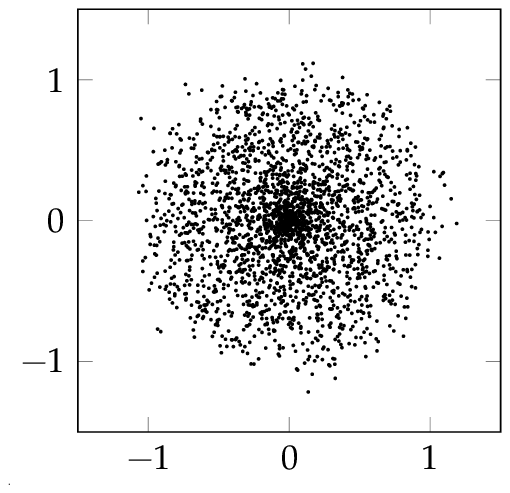}

\caption{The figure shows scatter plots for $50$ realisations for products of one (top row) and two (bottom row) $50\times 50$ complex elliptic Gaussian random matrices with $\tau=1$ i.e. Hermitian matrices (left column), $\tau=1/2$ (centre column), and $\tau=0$ i.e. Ginibre matrices (right column). Note that the top left and centre panels have different units than the rest.}
\label{fig:complex:scatter-elliptic}
\end{figure}
\begin{figure}[bp]
\centering
\includegraphics{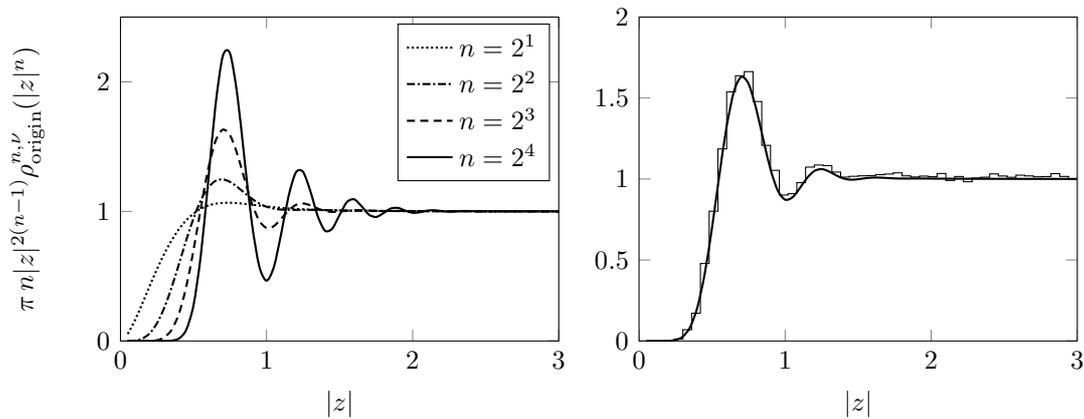}
\caption{The figure shows the microscopic density as described by the Meijer kernel, $\rho_\origin^{n,\nu}(z)=K_\origin^{n,\nu}(z,z)$. The left panel shows the density for $n=2,4,8,16$; for comparison between densities, we have unfolded the densities such that they all tend to unity as their argument tends to infinity. The right panel compares the $n=8$ case with numerical data from an ensemble of $20\,000$ realisations of a product of $8$ independent $100\times100$ complex Ginibre matrices.}
\label{fig:complex:origin}
\end{figure}

While formulae for the macroscopic densities of products of induced Ginibre matrices were obtained prior to the exactly solvable models presented in this chapter, the microscopic scaling regimes were complete unknown. In fact, so far the only results regarding microscopic correlations for the eigenvalues of product matrices have been obtained using the exactly solvable matrix models. 

In this chapter, we have seen that the microscopic correlations of the eigenvalues in the bulk and at the edge of spectrum for a product of (complex) induced Ginibre matrices are identical to those of the standard Ginibre ensemble (and they are, in this sense, universal); equivalent results hold for other product ensembles, see~\cite{ABKN:2014,LW:2014}. At the origin a new family of microscopic correlation kernels arose. These new kernels are most conviniently expressed in terms of Meijer $G$-functions (see proposition~\ref{thm:complex:jpdf-C}). This is very similar to the new family of kernels, which arose at the hard edge for the squared singular values (see proposition~\ref{thm:singular:meijer-kernel}). In fact, the two families of kernels have a striking resemblance, 
\begin{align}
K_\origin^{n,\nu}(x,y)&=\frac1\pi\MeijerG{n}{1}{1}{n+1}{0}{\nu_1,\ldots,\nu_n,0}{\abs{x}^2}\MeijerG{1}{1}{1}{n+1}{0}{0,-\nu_1,\ldots,-\nu_n}{-xy^*}, 
\label{complex:discuss:origin} \\
K_\meijer^{n,\nu}(x,y)&=\int_0^1du\,\MeijerG{n}{0}{0}{n+1}{-}{\nu_1,\ldots,\nu_n,0}{uy}\MeijerG{1}{0}{0}{n+1}{-}{0,-\nu_1,\ldots,-\nu_n}{ux}.
\end{align}
(In~\eqref{complex:discuss:origin} we have used~\eqref{special:meijer:reduce-1} to rewrite~\eqref{complex:C:kernel-origin}.) It is, of course, natural to expect that there is some type of relation between these two families of kernels, since they both consider microscopic spectral correlations near the origin of a product of induced Ginibre matrices (that is correlations for eigenvalues and squared singular values, respectively). However, it remains unclear whether some more precise statement can be made. The structure of the microscopic correlations near the origin is visualised on figure~\ref{fig:complex:origin}. We note that (similar to figure~\ref{fig:singular:meijer}) the amplitude of the oscillations in the microscopic density seems to increase as the number of factors increases. We investigate the large-$n$ limit in chapter~\ref{chap:lyapunov}.

%% file: lyapunov.tex
\chapter{Stability and Lyapunov exponents}
\label{chap:lyapunov}

Let us for a brief remark return to the dynamical system mentioned in section~\ref{sec:moti:stab}. We consider a system in the variables $u(n)=\{u_1(n),\ldots,u_N(n)\}$ which evolves according to $u_i(n+1)=f_i[u(n)]$, where each $f_i$ is some smooth function. Such systems arise in a vast variety of contexts in the physical, biological and social sciences. In order to study the separation between two nearby trajectories, we introduce the tangent vector $\delta u(t)$ which may be regarded as a small perturbation. This gives rise to the linearised problem
\begin{equation}
\delta u(n+1)=X_n\,\delta u(n) \qquad\text{with}\qquad (X_n)_{ij}:= \frac{\p f_i(u)}{\p u_j}\bigg\vert_{u(n)},
\label{lya:intro:dynamical}
\end{equation}
where $X_n$ is known as the stability matrix at ``time'' $n$. The evolution of the system is described by a product matrix,
\begin{equation}
Y_n:= X_nX_{n-1}\cdots X_1,\qquad n\in\N,
\label{lya:intro:product}
\end{equation}
which is obtained by evaluating the stability matrix along the trajectory. We will typically be interested in the large-$n$ limit.

The simplest examples occur when we consider the evolution in the vicinity of a fixed point. In these cases the stability matrices $X_i$ ($i=1,\ldots,n$) become delta-correlated and we have $Y_n=(X_1)^n$. For many large complex systems a reasonable approximation can be obtained by choosing $X_1$ as a random matrix with symmetries determined by the underlying theory~\cite{GA:1970}. An example is provided by the evolution of large ecosystems where minimal requirements demand that the stability matrix is an asymmetric real matrix~\cite{May:1972}.

More challenging problems occur when studying evolution away from the fixed points. A first non-trivial approximation is provided by choosing the stability matrices $X_i$ ($i=1,\ldots,n$) as independent random matrices~\cite{Benettin:1984,PV:1986}, see also~\cite{CPV:1993}. This approximation is reasonable in the limit of strong chaoticity where the stability matrices are only weakly correlated.

In this chapter we will consider one of the simplest toy models for such evolutions. The stability matrices $X_i$ ($i=1,\ldots,n$) are assumed to be independent $N\times N$ induced Ginibre matrices. The benefit of considering these models is, of course, that we (from previous chapters) have explicit formulae for joint densities valid for an arbitrary number of factors and arbitrary matrix dimension (proposition~\ref{thm:singular:jpdf-determinantal}, \ref{thm:complex:jpdf-C}, \ref{thm:complex:jpdf-H}, and~\ref{thm:complex:jpdf-R}). We will see that these exact formulae give new insights regarding such stability problems.

Contrary to the descriptions given in chapter~\ref{chap:singular} and~\ref{chap:complex}, our main focus in this chapter will be directed towards the limit $n\to\infty$ rather than $N\to\infty$. It is known from Oseledec's multiplicative ergodic theorem~\cite{Oseledec:1968,Raghunathan:1979} that if $Y_n:=X_n\cdots X_1$ where $X_i$ ($i=1,\ldots,n$) are independent $N\times N$ random matrices with measure $P(dX)$ such that $\lim_{n\to\infty}(\log\tr Y_nY_n^\dagger)/n<\infty$, then we have almost sure convergence
\begin{equation}
\lim_{n\to\infty} (Y_nY_n^\dagger)^{1/2n}=e^H,
\label{lya:intro:oseledec}
\end{equation}
where $H$ is a (non-random) Hermitian matrix with eigenvalues $\mu_1\leq\cdots\leq\mu_n$. Here, $\mu_i$ ($i=1,\ldots,n$) are known as the Lyapunov exponents. This can, of course, also be formulated as a statement about the singular values. If $x_1(n)\leq\cdots\leq x_N(n)$ denote the eigenvalues of the Wishart product matrix $Y_nY_n^\dagger$ (i.e. the squared singular values), then Oseledec's theorem tells that asymptotically we have
\begin{equation}
x_k(n)\sim e^{2n\,\mu_k},\qquad k=1,\ldots,N.
\label{lya:intro:singular-asymp}
\end{equation}
Due to this asymptotic relation, we introduce the finite-$n$ (or finite time) Lyapunov exponents defined by
\begin{equation}
\lambda_k(n):=\frac{\log x_k(n)}{2n},\qquad k=1,\ldots,N.
\label{lya:intro:lyapunov-def}
\end{equation}
The Lyapunov exponents are a natural measure for stability or instability of a dynamical system. Comparing with~\eqref{lya:intro:dynamical}, we see that if all Lyapunov exponents are less than zero then the system is stable in the sense that
\begin{equation}
\norm{\delta u(n)}\to0\qquad\text{for}\qquad n\to\infty,
\end{equation}
but if one or more of the Lyapunov exponents are positive then nearby trajectories diverge and the system is said to be unstable.

Although the Lyapunov exponents are the most common measure of stability, other measures exists (and are sometimes more appropriate). One choice is the so-called stability exponents~\cite{GSO:1987} related to the eigenvalues of the product matrix (instead of the singular values as for the Lyapunov exponents).  In most cases, it is expected that the absolute value of the eigenvalues will grow (decay) exponentially with the number of factors. Thus, if $z_1(n)\leq\cdots\leq z_N(n)$ denote the eigenvalues of the product matrix $Y_n$, then it is sensible to introduce the finite-$n$ (or finite time) stability exponents defined by
\begin{equation}
\zeta_k(n):=\frac{\log \abs{z_k(n)}}{n}, \qquad k=1,\ldots,N.
\label{lya:intro:stability-def}
\end{equation}
It has been conjectured in~\cite{GSO:1987} that if the spectrum is non-degenerate, then the Lyapunov and the stability exponents will agree in the large-$n$ limit, i.e. they are equivalent as measures of stability. 

The simplest possible example is, of course, the $N=1$ case, where we have a product of scalar-valued random variables, where the definitions of the Lyapunov exponents~\eqref{lya:intro:lyapunov-def} and stability exponents~\eqref{lya:intro:stability-def} are identical. In section~\ref{sec:prologue:scalar}, we recalled the well-known structure of a product of Gaussian distributed random \emph{scalars}. We saw that the product grew exponentially with a growth rate (i.e. Lyapunov/Stability exponent) subject to Gaussian fluctuations. The main purpose of this chapter is to extent this result to the case of Gaussian \emph{matrices}. 

The rest of the chapter is organised as follows: In section~\ref{sec:lya:finiteN} we consider stability as well as the Lyapunov exponents in large-$n$ limit when $N$ is kept fixed; we find their limiting values as well as their fluctuations. Section~\ref{sec:lya:macro} deals with macroscopic limit for the spectral density in the case where both $n$ and $N$ tend to infinity. A discussion of the results as well as open problems is given in section~\ref{sec:lya:discuss}.

\section{Stability and Lyapunov exponents at finite matrix dimension}
\label{sec:lya:finiteN}

\subsection{Stability exponents}
\label{sec:lya:stability}

In this section, we will consider a product of induced Ginibre matrices with our attention directed towards the statistical properties of the stability exponents~\eqref{lya:intro:stability-def}, while a similar description of the Lyapunov exponents are postponed to section~\ref{sec:lya:lyapunov}. For our purpose, the benefit of starting with the stability (rather than the Lyapunov exponents) is that the stability exponents are based on the eigenvalues of the product matrix~\eqref{lya:intro:stability-def} while the Lyapunov exponents are based on singular values~\eqref{lya:intro:lyapunov-def}. As we recall from chapter~\ref{chap:singular} and~\ref{chap:complex}, we have explicit formulae for the joint density of the \emph{eigenvalues} in all three Dyson classes ($\beta=1,2,4$), while an explicit expression for the joint density of the \emph{singular values} is only known in the $\beta=2$ case. For this reason, we can treat all Dyson classes equally when we consider the stability exponents, while we need to 
prescribe different treatments to the $\beta=2$ and the $\beta=1,4$ cases when we return to Lyapunov exponents in section~\ref{sec:lya:lyapunov}. 

We recall from propositions~\ref{thm:complex:jpdf-C}, \ref{thm:complex:jpdf-H}, and~\ref{thm:complex:jpdf-R}, that the joint densities for the eigenvalues on $\R$ ($\beta=1$), $\C$ ($\beta=2$), and $\C_+$ ($\beta=4$) are
\begin{subequations}
\label{lya:stab:jpdf-all}
\begin{align}
\cP_{\jpdf}^{\beta=1,n}(x_1,\ldots,x_N)&=\frac{1}{\cZ^{\beta=1}}\prod_{j=1}^N w_{n}^{\beta=1}(x_j)\prod_{1\leq k<\ell\leq N}\abs{x_k-x_\ell}, \label{lya:stab:jpdf-real} \\
\cP_{\jpdf}^{\beta=2,n}(z_1,\ldots,z_N)&=\frac{1}{\cZ^{\beta=2}}\prod_{j=1}^N w_{n}^{\beta=2}(z_j)\prod_{1\leq k<\ell\leq N}\abs{z_k-z_\ell}^2, \label{lya:stab:jpdf-complex} \\
\cP_{\jpdf}^{\beta=4,n}(z_1,\ldots,z_N)&=\frac{1}{\cZ^{\beta=4}}\prod_{j=1}^N w_{n}^{\beta=4}(z_j)\abs{z_j-z_j^*}^2\prod_{1\leq k<\ell\leq N}\abs{z_k-z_\ell}^2\abs{z_k-z_\ell^*}^2,
\label{lya:stab:jpdf-quaternion}
\end{align}
\end{subequations}
where $w_{n}^{\beta}(z)$ are positive weight functions,
\begin{equation}
w_{n}^{\beta}(z)=\MeijerG[\bigg]{n}{0}{0}{n}{-}{\frac{\beta\nu_1}{2},\ldots,\frac{\beta\nu_n}{2}}{\Big(\frac{\beta}{2}\Big)^n\,\abs z^2}, \label{lya:stab:weight} 
\end{equation}
and $\cZ^{\beta}$ are constants,
\begin{equation}
\cZ^{\beta}=\frac{N!\pi^{N(\beta-\gamma)/\gamma}}{2^{(2-\beta)\gamma nN(N+1)/4}}\prod_{k=1}^N\prod_{\ell=1}^n\Gamma\Big[\frac{\beta(\nu_\ell+k)}{2}\Big] \label{lya:stab:normalisation} 
\end{equation}
with $\gamma=1,1,2$ for $\beta=1,2,4$ and $\nu_i$ ($i=1,\ldots,n$) denoting non-negative constants.

Note that for $\beta=1$ we have taken the joint density which assumes that all eigenvalues are real and the normalisation constant, $\cZ^{\beta=1}$, is chosen such that integration over the eigenvalues yields the probability that all eigenvalues are real (cf. proposition~\ref{thm:complex:jpdf-R} and remark~\ref{remark:complex:real}). Although at first sight this seems like a strong restriction, it turns out to be no restriction at all, since (as long as the matrix dimension is kept fixed) all eigenvalues are real, almost surely, in the large-$n$ limit (as we will see below). This surprising property of real matrices was first observed numerically in~\cite{Lakshminarayan:2013}, while a proof for square matrices ($0=\nu_1=\nu_2=\cdots$) was provided in~\cite{Forrester:2014a}. Here, we will proof this property for the more general case, where $\{\nu_i\}_\N$ is a convergent non-negative sequence, i.e. $\nu_i\geq 0$ for all $i$ and $\nu_i\to\nu_\infty<\infty$ for $n\to\infty$. This result will be obtained using an idea presented in~\cite{Ipsen:2015}, which differs considerably from the approach used in~\cite{Forrester:2014a}.

It will be convenient to change to exponential variables when we consider the large-$n$ limit. Thus, we parametrise the eigenvalues as
\begin{equation}
z_k(n)=e^{n\zeta_k(n)+i\theta_k(n)}, \qquad k=1,\ldots,N,
\end{equation}
where $\zeta_k(n)$ ($k=1,\ldots,N$) are the finite-$n$ stability exponents~\eqref{lya:intro:stability-def}, while $\theta_k(n)$ ($k=1,\ldots,N$) are real phases which take the discrete values $0$ and $\pi$ for $\beta=1$, values in the interval $[0,2\pi)$ for $\beta=2$, and values in $[0,\pi]$ for $\beta=4$. The restriction of the phases appear since the eigenvalues belong to the real line, complex plane, and the complex upper half-plane for $\beta=1,2,4$, respectively.

We introduce the joint density of the finite-$n$ stability exponents and their phases defined by
\begin{equation}
p_\jpdf^{\beta,n}(\zeta_1,\theta_1,\ldots,\zeta_N,\theta_N):=\prod_{j=1}^Nne^{\beta n\zeta_j/\gamma}\chi^{\beta}(\theta_j)\,\cP_\jpdf^{\beta,n}(e^{n\zeta_1+i\theta_1},\ldots,e^{n\zeta_N+i\theta_N}),
\label{lya:stab:jpdf-stab}
\end{equation}
where $\cP_\jpdf^{\beta,n}(z_1,\ldots,z_N;t)$ is given by~\eqref{lya:stab:jpdf-all} and
\begin{equation}
\chi^{\beta}(\theta)=\delta_{\beta1}(\delta(\theta)+\delta(\pi-\theta))+\delta_{\beta2}+\delta_{\beta4}\one_{0\leq\theta\leq\pi}.
\label{lya:stab:indicator-func}
\end{equation}
is an indicator function which incorporates the fact that the eigenvalues belong to the line $\R$, the plane $\C$, and the half-plane $\C_+$ for real, complex, and quaternionic matrices, respectively. With this notation, we have $\zeta_i\in\R$ while $\theta_i\in[0,2\pi)$ ($i=1,\ldots,N$) for all three Dyson indices ($\beta=1,2,4$). Alternatively, we could have left out the indicator function~\eqref{lya:stab:indicator-func} and imposed the individual constraints  $\theta_i\in\{0,\pi\}$ for $\beta=1$, $\theta_i\in [0,2\pi)$ for $\beta=2$, and $\theta_i\in [0,\pi]$ for $\beta=4$. It is important to note that~\eqref{lya:stab:jpdf-stab} is nothing but a change of variables; the prefactor on the right hand side is simply the Jacobian arising from the change of variables.

It turns out that the stability exponents and their phases become independent in the large-$n$ limit if $N$ is kept fixed (as we will see). For this reason, it is often sufficient to consider the one-point correlation function (the spectral density), which is defined in the usual manner,
\begin{equation}
\rho^{\beta,n}(\zeta_1,\theta_1):=N\bigg[\prod_{i=2}^N\int_\R d\zeta_i\int_0^{2\pi}d\theta_i\bigg]p_\jpdf^{\beta,n}(\zeta_1,\theta_1,\ldots,\zeta_N,\theta_N).
\label{lya:stab:dens-stab}
\end{equation}

We are now ready to formulate the main statement of this section together with an important corollary. This statement was originally presented in~\cite{Ipsen:2015}, while the partial result for products of square complex matrices (i.e. $\beta=2$ and $0=\nu_1=\nu_2=\cdots$) was given previously by Akemann, Burda and Kieburg in~\cite{ABK:2014}. The striking statement of corollary~\ref{thm:lya:stab:-real-eigenvalues} was first shown for $N=2$ in~\cite{Lakshminarayan:2013}; based on numerical evidence it was conjectured to hold for $N\geq3$ as well. This was shown analytically for square matrices (i.e. $0=\nu_1=\nu_2=\cdots$) by Forrester in~\cite{Forrester:2014a}, while proposition~\ref{thm:lya:stab-exp} provides us with an alternative proof as observed in~\cite{Ipsen:2015}.

\begin{proposition}\label{thm:lya:stab-exp}
Let $\{\nu_i\}$ be a non-negative convergent sequence with the limit $\nu_i\to\nu_\infty<\infty$.
For $N$ fixed, the large-$n$ limit of the joint density for the stability exponents and their phases~\eqref{lya:stab:jpdf-stab} becomes
\begin{equation}
\lim_{n\to\infty}p_\jpdf^{\beta,n}(\zeta_1,\theta_1,\ldots,\zeta_N,\theta_N)
=\frac{1}{N!}\per_{1\leq k,\ell\leq N}\bigg[K^\beta(\theta_k)\delta(\zeta_\ell-\mu_k^{\beta}) \bigg],
\label{lya:stab:permanent}
\end{equation}
while the density~\eqref{lya:stab:dens-stab} allows the limit 
\begin{equation}
\lim_{n\to\infty}\frac{\sigma_k^\beta}{\sqrt n}\rho^{\beta,n}\Big(\mu_k^\beta+\zeta\frac{\sigma_k^\beta}{\sqrt{n}},\theta\Big)=K^\beta(\theta)\frac{e^{-\zeta^2/2}}{\sqrt{2\pi}}
,\qquad k=1,\ldots,N.
\label{lya:stab:dens-stab-limit}
\end{equation}
where
\begin{equation}
K^{\beta}(\theta):=\delta_{\beta1}\frac{\delta(\theta)+\delta(\pi-\theta)}{2}+\delta_{\beta2}\frac{1}{2\pi}+\delta_{\beta4}\frac{2\sin^2\theta}{\pi}\one_{0\leq\theta\leq\pi}.
\label{lya:stab:phase-func}
\end{equation}
and
\begin{equation}
\mu_k^\beta=\frac{1}{2}\log\frac{2}{\beta}+\frac{1}{2}\psi\Big[\frac{\beta(\nu_\infty+k)}{2}\Big],\qquad
(\sigma_k^\beta)^2=\frac{1}{4}\psi'\Big[\frac{\beta(\nu_\infty+k)}{2}\Big],
\label{lya:intro:mean+var}
\end{equation}
with $\psi(x)$ denoting the digamma function~\eqref{special:gamma:digamma-def}.
\end{proposition}

Before we provide a proof of proposition~\ref{thm:lya:stab-exp}, let us show a corollary which has fundamental importance for the interpretation of the proposition in the $\beta=1$ case.

\begin{corollary}\label{thm:lya:stab:-real-eigenvalues}
The eigenvalues of a product of independent real $N\times N$ induced Ginibre matrices (with charges $\nu_i$ as in proposition~\ref{thm:lya:stab-exp}) become real, almost surely, as the number of factors tends to infinity.
\end{corollary}

\begin{proof}
We recall that integration over the eigenvalues in~\eqref{lya:stab:jpdf-real}, and therefore also in the $\beta=1$ case of~\eqref{lya:stab:permanent}, gives the probability that all eigenvalues are real, but
\begin{equation}
\bigg[\prod_{i=1}^N\int_\R d\zeta_i\int_0^{2\pi}d\theta_i\bigg]\frac{1}{N!}\per_{1\leq k,\ell\leq N}\bigg[K^{\beta=1}(\theta)\delta(\zeta_\ell-\mu_k^{\beta=1})\bigg]=1.
\end{equation}
Thus the eigenvalues are real, almost surely.
\end{proof}

\begin{remark}
Corollary~\ref{thm:lya:stab:-real-eigenvalues} tells us that the imposed assumption that all eigenvalues are real is no real restriction in the large-$n$ limit, since this happens with probability one.
\end{remark}

\begin{proof}[Proof of proposition~\ref{thm:lya:stab-exp}]

The main idea of the proof is to show that the correlations between eigenvalues (stability exponents and their phases) decay exponentially for large $n$, while the distribution of individual stability exponents become Gaussians with variances of order $O(n^{-1/2})$ for large $n$. For accessibility, the proof is divided into three parts: \textit{(i) the complex case,} \textit{(ii) the quaternionic case,} and finally \textit{(iii) the real case.}

\paragraph{(i) Complex case. }
We write the joint density~\eqref{lya:stab:jpdf-stab} using~\eqref{lya:stab:jpdf-complex}
\begin{multline}
p_{\jpdf}^{\beta=2,n}(\zeta_1,\theta_1,\ldots,\zeta_N,\theta_N)=
\frac{1}{\cZ^{\beta=2}}\prod_{j=1}^N n w_{n}^{\beta=2}(e^{n\zeta_j}) \\
\times\det_{1\leq k,\ell\leq N}\Big[e^{k(n\zeta_\ell+i\theta_\ell)}\Big]
\det_{1\leq k,\ell\leq N}\Big[e^{k(n\zeta_\ell-i\theta_\ell)}\Big].
\label{lya:stab:jpdf-detdet}
\end{multline}
Here we have absorbed the prefactor from~\eqref{lya:stab:jpdf-stab} into the Vandermonde determinants; $w_n^{\beta=2}(z)$ and $\cZ^{\beta=2}$ denote the weight~\eqref{lya:stab:weight} and normalisation~\eqref{lya:stab:normalisation}, respectively. We want to rewrite the right hand side of~\eqref{lya:stab:jpdf-stab} as a sum involving a single determinant. To achieve this, we expand the first determinant as
\begin{equation}
\det_{1\leq k,\ell\leq N}[e^{k(n\zeta_\ell+i\theta_\ell)}]=\sum_{\sigma\in S_N} \sign\sigma \prod_{j=1}^N e^{j(n\zeta_{\sigma(j)}+i\theta_{\sigma(j)})},
\label{lya:stab:expansion}
\end{equation}
where the sum is over all permutations and `$\sign\sigma$' denotes the sign of the permutation. This expansion can be used to rewrite the joint density~\eqref{lya:stab:jpdf-detdet} as a sum involving a single determinant. After application of this expansion, we can pull the weight functions into the the second determinant and, by reordering the rows in the determinant, the `$\sign\sigma$' may be included in the determinant as well. Finally, we pull the product from the expansion~\eqref{lya:stab:expansion} into the determinant which yields
\begin{equation}
p_{\jpdf}^{\beta=2,n}(\zeta_1,\theta_1,\ldots,\zeta_N,\theta_N)=
\frac{1}{\cZ^{\beta=2}}\sum_{\sigma\in S_N} 
\det_{1\leq k,\ell\leq N}\Big[e^{n(k+\ell)\zeta_{\sigma(\ell)}}e^{i(k-\ell)\theta_{\sigma(\ell)}} n\, w_{n}^{\beta=2}(e^{n\zeta_{\sigma(\ell)}})\Big]
\label{lya:stab:jpdf-det}
\end{equation}
with weight and normalisation as above.

In order to simplify notation in the calculations below, we introduce the normalised non-negative function
\begin{equation}
f_{k\ell}^{\beta=2,n}(\zeta):= \frac{2n\,\MeijerG{n}{0}{0}{n}{-}{\nu_1+\frac{k+\ell}{2},\ldots,\nu_n+\frac{k+\ell}{2}}{e^{2n\zeta}}}{\prod_{i=1}^n\Gamma\left[\nu_i+(k+\ell)/2\right]}.
\label{lya:stab:f-complex}
\end{equation}
The normalisation may be checked using the integration formula~\eqref{special:meijer:meijer-moment}, while the non-negativity can be seen by writing the Meijer-$G$ function as an $n$-fold integral similar to the discussion of product of Gaussian random scalars in section~\ref{sec:prologue:scalar}.

With definition~\eqref{lya:stab:f-complex} in mind, we return to the joint density~\eqref{lya:stab:jpdf-det}. If we write the weight as a Meijer $G$-function~\eqref{lya:stab:weight} and use the identity~\eqref{special:meijer:meijer-shift}, then we see that the first exponential in~\eqref{lya:stab:jpdf-det} and the weight can be combined into a new Meijer $G$-function identical to the one appearing in~\eqref{lya:stab:f-complex}. We insert the explicit expression for the normalisation constant~\eqref{lya:stab:normalisation} and after some standard manipulations we find
\begin{equation}
p_{\jpdf}^{\beta=2,n}(\zeta_1,\theta_1,\ldots,\zeta_N,\theta_N)=
\frac{1}{N!}\sum_{\sigma\in S_N} 
\det_{1\leq k,\ell\leq N}\bigg[\frac{e^{i(k-\ell)\theta_{\sigma(\ell)}} D_{k\ell}^{\beta=2,n} f_{k\ell}^{\beta=2,n}(\zeta)}{2\pi}\bigg]
\label{lya:complex:jpdf-final-det}
\end{equation}
with the coefficient $D_{k\ell}^{\beta=2,n}$ defined by
\begin{equation}
D_{k\ell}^{\beta=2,n}:=\prod_{i=1}^n\frac{\Gamma\left[\nu_i+(k+\ell)/2\right]}{\Gamma[\nu_i+k]^{1/2}\Gamma[\nu_i+\ell]^{1/2}}.
\label{lya:complex:coefficient}
\end{equation}
So far no approximations have been used, hence the joint density~\eqref{lya:complex:jpdf-final-det} is valid at any $n$. In order to understand the large-$n$ asymptotics, we will look at $D_{k\ell}^{\beta=2,n}$ and $f_{k\ell}^{\beta=2,n}(\zeta)$ separately.

We start by evaluating the coefficient $D_{k\ell}^{\beta=2,n}$. We recall that $k$ and $\ell$ are positive integers and notice that
\begin{equation}
\Gamma[\nu_i+(k+\ell)/2]\leq \Gamma[\nu_i+k]^{1/2}\Gamma[\nu_i+\ell]^{1/2}
\label{lya:stab:gamma-inequality}
\end{equation}
with equality if and only if $k=\ell$. This can be seen by writing the gamma functions as their integral representations and interpreting the expression using the Cauchy--Schwarz inequality. Comparing this inequality with the expression for the coefficient~\eqref{lya:complex:coefficient}, we see that (recall that the sequence of charges is convergent)
\begin{equation}
D_{k\ell}^{\beta=2,n}\longrightarrow \delta_{k\ell}\qquad \text{for}\qquad n\to\infty,
\label{lya:complex:coefficient-limit}
\end{equation}
where $\delta_{k\ell}$ is a Kronecker delta. Note that $D_{kk}^{\beta=2,n}=1$ even for finite $n$, while the off-diagonal entries, $D_{k\ell}^{\beta=2,n}$ ($k\neq\ell$), decay exponentially with $n$.

Similar to the product of Gaussian random scalars discussed in section~\ref{sec:prologue:scalar}, we will use the cumulant expansion to argue that the function $f_{k\ell}^{\beta=2,n}(\zeta)$ reduces to a Gaussian in the large-$n$ limit. In the following, we will interpret $f_{k\ell}^{\beta=2,n}(\zeta)$ as the probability density for a random variable. The cumulant-generating function $g_{k\ell}^{\beta=2,n}(\xi)$ is given by
\begin{equation}
g_{k\ell}^{\beta=2,n}(\xi):= \log\Big( \int_{-\infty}^\infty d\zeta\, e^{\xi\zeta} f_{k\ell}^{\beta=2,n}(\zeta)\Big)
=\sum_{i=1}^n\log\frac{\Gamma[\nu_i+(k+\ell)/2+\xi/2n]}{\Gamma[\nu_i+(k+\ell)/2]}.
\end{equation}
Here we have used formula~\eqref{special:meijer:meijer-moment} to evaluate the integral. The cumulants are found by expanding the cumulant-generating function around $\xi=0$, hence the $j$-th cumulant is given by
\begin{equation}
\kappa_{k\ell;j}^{\beta=2,n}:= \frac{\p^jg_{k\ell}^{\beta=2,n}(\xi)}{\p\xi^j} \bigg\vert_{\xi=0}
=\sum_{i=1}^n\frac{1}{(2n)^j}\psi^{(j-1)}\Big(\nu_i+\frac{k+\ell}{2}\Big),
\end{equation}
where $\psi^{(j)}(x)$ is the $j$-th derivative of the digamma function (also known as the polygamma function). Due to convegence of the charges ($\nu_i\to\nu_\infty$) and continuity of the polygamma functions on the positive half-line, we have convergence of the arithmetic means. In particularly, we have
\begin{subequations}
\label{lya:complex:mean_kl}
\begin{align}
\lim_{n\to\infty}\kappa_{k\ell;1}^{\beta=2,n}&=\frac{1}{2}\psi\Big(\nu_\infty+\frac{k+\ell}{2}\Big)=:\mu_{k\ell}^{\beta=2}, \\
\lim_{n\to\infty}n\,\kappa_{k\ell;1}^{\beta=2,n}&=\frac{1}{4}\psi'\Big(\nu_\infty+\frac{k+\ell}{2}\Big)=:(\sigma_{k\ell}^{\beta=2})^2,
\end{align}
\end{subequations}
such that $\mu_{kk}^{\beta=2}$ and $(\sigma_{kk}^{\beta=2})^2$ are the mean and variance from~\eqref{lya:intro:mean+var}. To find the limiting distribution at large $n$, we switch from the variable $\zeta$ to the standardised variable $\widetilde \zeta:= \sqrt{n}(\zeta-\mu_{k\ell}^{\beta})/\sigma_{k\ell}^{\beta}$, which has mean zero and unit variance. The corresponding standardised cumulants are
\begin{equation}
\widetilde\kappa_{k\ell;1}^{\beta=2,n}=0,\qquad
\widetilde\kappa_{k\ell;2}^{\beta=2,n}=1 \qquad\text{and}\qquad
\widetilde\kappa_{k\ell;n}^{\beta=2,n}= O(n^{1-j/2}) \quad\text{for}\quad j\geq 3.
\end{equation}
We see that the higher order cumulants tend to zero for large $n$, and it follows by standard arguments that the limiting distribution is a Gaussian. In the original variables, we have
\begin{equation}
f_{k\ell}^{\beta=2,n}(\zeta)\sim\sqrt{\frac{n}{2\pi(\sigma_{k\ell}^{\beta=2})^2}} \exp\bigg[-n\frac{(\zeta-\mu_{k\ell}^{\beta=2})^2}{2(\sigma_{k\ell}^{\beta=2})^2}\bigg]
\label{lya:complex:f-limit}
\end{equation}
with mean and variance in terms by~\eqref{lya:complex:mean_kl}.

We can now return to the joint density. Without further rescaling, the function~\eqref{lya:complex:f-limit} becomes a delta peak located at $\mu_{k\ell}^{\beta=2}$ in the large-$n$ limit. Thus, using the asymptotic behaviour~\eqref{lya:complex:coefficient-limit} and~\eqref{lya:complex:f-limit} in the joint density~\eqref{lya:complex:jpdf-final-det} we obtain the $\beta=2$ case of~\eqref{lya:stab:permanent}; the dependence of the phases cancels out due to the Kronecker delta stemming from~\eqref{lya:complex:coefficient-limit}.

To see that~\eqref{lya:stab:dens-stab-limit} holds, we notice that the strength of correlations between exponents is determined by the coefficient $D_{k\ell}^{\beta=2,n}$ which is exponentially suppressed compared to the self-correlations. Thus, for large $n$, we have
\begin{multline}
\rho^{\beta,n}(\zeta_1,\theta_1)
\sim\frac{1}{(N-1)!}\bigg[\prod_{i=2}^N\int_\R d\zeta_i\int_0^{2\pi}d\theta_i\bigg]\sum_{\sigma\in S_N} \prod_{n=1}^N  \frac{1}{2\pi}f_{kk}^{\beta=2,n}(\zeta_{\sigma(n)})\\
=\sum_{k=1}^N  \frac{1}{2\pi}f_{kk}^{\beta=2,n}(\zeta_1).
\end{multline}
Now the limit density~\eqref{lya:stab:dens-stab-limit} follows since $\abs[\big]{\mu_k^{\beta=2}-\mu_\ell^{\beta=2}}$ is non-zero and independent of~$n$ for $k\neq\ell$.

\paragraph{(ii) Quaternionic case. }

The main idea for quaternionic matrices is the same as it was for complex matrices. Our starting point is the joint density~\eqref{lya:stab:jpdf-stab} expressed using~\eqref{lya:stab:jpdf-quaternion},
\begin{equation}
p_{\jpdf}^{\beta=4,n}(\zeta_1,\theta_1,\ldots,\zeta_N,\theta_N)
=\frac{1}{\cZ^{\beta=4}}
\det_{\substack{k=1,\ldots,2N \\ \ell=1,\ldots,N}}\begin{bmatrix} e^{k(n\zeta_\ell+i\theta_\ell)} \\ e^{k(n\zeta_\ell-i\theta_\ell)} \end{bmatrix}
\prod_{j=1}^N2ine^{n\zeta_j}\sin\theta_j w_n^{\beta=4}(e^{n\zeta_j}).
\end{equation}
Expanding the determinant yields
\begin{align}
p_{\jpdf}^{\beta=4,n}(\zeta_1,\theta_1,\ldots,\zeta_N,\theta_N)&= \label{lya:quaternion:jpdf-expand} \\
\frac{1}{\cZ^{\beta=4}}\sum_{\sigma\in S_{2N}} \sign\sigma
\prod_{j=1}^N&2in\sin\theta_j e^{i(\sigma(j)-\sigma(j+N))\theta_j}e^{n(\sigma(j)+\sigma(j+N)+1)\zeta_j}w_n^{\beta=4}(e^{n\zeta_j}). \nn
\end{align}
Analogously to the proof for complex matrices, we introduce a normalised non-negative function,
\begin{equation}
f_{k\ell}^{\beta=4,n}(\zeta):= \frac{2n\,\MeijerG[\bigg]{n}{0}{0}{n}{-}{2\nu_1+\frac{k+\ell+1}{2},\ldots,2\nu_n+\frac{k+\ell+1}{2}}{2^ne^{2n\zeta}}}{\prod_{i=1}^n\Gamma\left[2\nu_i+(k+\ell+1)/2\right]}.
\label{lya:quaternion:f_kl}
\end{equation}
Now, we can use the explicit form of the weight function~\eqref{lya:stab:weight} in the expression for the joint density~\eqref{lya:quaternion:jpdf-expand} and use~\eqref{special:meijer:meijer-shift} to absorb the prefactor involving Lyapunov exponents into the Meijer $G$-function. With the notation~\eqref{lya:quaternion:f_kl} and the explicit expression for the normalisation constant~\eqref{lya:stab:normalisation}, we write the joint density as  
\begin{multline}
p_{\jpdf}^{\beta=4,n}(\zeta_1,\theta_1,\ldots,\zeta_N,\theta_N)= \\
\frac{2^{N(N+1)n}}{\pi^NN!}\sum_{\sigma\in S_{2N}}\sign\sigma\, D_\sigma^{\beta=4,n}\prod_{j=1}^N \frac{i\sin\theta_j e^{i(\sigma(j)-\sigma(j+N))\theta_j}}{2^{(\sigma(j)+\sigma(j+N)+1)n/2}}f_{\sigma(j),\sigma(j+N)}^{\beta=4,n}(\zeta_n),
\label{lya:quaternion:jpdf-long}
\end{multline}
where the coefficient $D_\sigma^{\beta=4,n}$ depends on the permutation $\sigma$,
\begin{equation}
D_\sigma^{\beta=4,n}\equiv\prod_{i=1}^n\prod_{j=1}^N\frac{\Gamma[2\nu_i+(\sigma(j)+\sigma(j+N)+1)/2]}{\Gamma[2\nu_i+2j]}.
\label{lya:quaternion:coefficient}
\end{equation}
As in the proof for complex matrices, we will discuss the asymptotic limit of the coefficient~\eqref{lya:quaternion:coefficient} and the function~\eqref{lya:quaternion:f_kl} separately, and then combine the results.

In order to evaluate the coefficient~\eqref{lya:quaternion:coefficient} we will use the inequality~\eqref{lya:stab:gamma-inequality} once again to write
\begin{equation}
\Gamma[2\nu_i+(k+\ell+1)/2]\leq (2\nu_i+\ell)^{1/2}\Gamma[2\nu_i+k]^{1/2}\Gamma[2\nu_i+\ell]^{1/2},
\end{equation}
where we have also used the recursive property $\Gamma[x+1]=x\Gamma[x]$. Note that $k$ and $\ell$ are interchangeable. Given a permutation, $\sigma\in S_{2N}$, it follows that 
\begin{multline}
\prod_{j=1}^N\Gamma[2\nu_i+(\sigma(j)+\sigma(j+N)+1)/2]\leq \\
\prod_{n=1}^N(2\nu_i+\min\{\sigma(j),\sigma(j+N)\})^{1/2}\Gamma[2\nu_i+2j]^{1/2}\Gamma[2\nu_i+2j-1]^{1/2}.
\end{multline}
The maximal value of right hand side of this inequality only depends on the permutation through `$\min\{\sigma(j),\sigma(j+N)\}$'. It follows that the product takes its maximal value if and only if $\sigma(j)$ and $\sigma(j+N)$ comes in successive pairs, hence
\begin{equation}
\prod_{j=1}^N\Gamma[2\nu_i+(\sigma(j)+\sigma(j+N)+1)/2]\leq \prod_{j=1}^N\Gamma[2\nu_i+2j]
\end{equation}
with equality if and only if $\abs{\sigma(j)-\sigma(j+N)}=1$ for all $j$. Inserting this into the expression for the coefficient~\eqref{lya:quaternion:coefficient}, we see that (recall that the sequence of charges is convergent)
\begin{equation}
D_\sigma^{\beta=4,n}\longrightarrow \prod_{j=1}^N(\delta_{\sigma(j)+1,\sigma(j+N)}+\delta_{\sigma(j),\sigma(j+N)+1})
\qquad\text{for}\qquad
n\to\infty.
\label{lya:quaternion:coefficient-limit}
\end{equation}
This implies that in the large-$n$ limit the sum over all permutations, $S_{2N}$, reduces to a sum over permutations of pairs, $S_N$. The other permutations will be exponentially suppressed at large $n$.

The evaluation of the asymptotic behaviour of the function $f_{k\ell}^{\beta=4,n}(\zeta)$ proceeds exactly as for complex matrices and it follows that the limiting distribution is a Gaussian,
\begin{equation}
f_{k\ell}^{\beta=4,n}(\zeta)\sim\sqrt{\frac{n}{2\pi}} \frac1{\sigma_{k\ell}^{\beta=4}}\exp\bigg[-n\frac{(\zeta-\mu_{k\ell}^{\beta=4})^2}{2(\sigma_{k\ell}^{\beta=4})^2}\bigg]
\label{lya:quaternion:f-limit}
\end{equation}
with
\begin{equation}
\mu_{k\ell}^{\beta=4}:=\frac12\log\frac12+\frac{1}{2}\psi\Big(2\nu_\infty+\frac{k+\ell+1}{2}\Big)
\quad\text{and}\quad
(\sigma_{k\ell}^{\beta=4})^2:=\frac{1}{4}\psi'\Big(2\nu_\infty+\frac{k+\ell+1}{2}\Big),
\label{lya:quaternion:mean_kl}
\end{equation}
where the first term in the mean stems from the prefactor inside the argument of the Meijer $G$-function~\eqref{lya:quaternion:f_kl}. Note that with the above notation $\mu_{2k-1,2k}^{\beta=4}$ and $\sigma_{2k-1,2k}^{\beta=4}$ agree with~\eqref{lya:intro:mean+var}.

Combining the asymptotic behaviour for the coefficient~\eqref{lya:quaternion:coefficient-limit} and the function~\eqref{lya:quaternion:f-limit} in the expression for the joint density~\eqref{lya:quaternion:jpdf-long} we see that the large-$n$ limit yields the $\beta=4$ case of~\eqref{lya:stab:permanent}. Moreover, we notice that, similar to the complex case, the strength of the correlations between exponents is determined by the coefficient $D_{\sigma}^{\beta=4,n}$ which is exponentially suppressed compared to the self-correlations. Thus~\eqref{lya:stab:dens-stab-limit} follows.

\paragraph{(iii) Real case. }
Finally, let us look at a product of real matrices with eigenvalues distributed according to~\eqref{lya:stab:jpdf-real}. The joint density~\eqref{lya:stab:jpdf-stab} reads
\begin{equation}
p_{\jpdf}^{\beta=1,n}(\zeta_1,\theta_1,\ldots,\zeta_N,\theta_N)
=\frac{1}{\cZ^{\beta=1}}\prod_{1\leq k<\ell\leq N}\abs{e^{n\zeta_k+i\theta_k}-e^{n\zeta_\ell+i\theta_\ell}}\prod_{j=1}^N ne^{n\zeta_j}w_{n}^{\beta=1}(e^{n\zeta_j}).
\label{lya:real:jpdf-cov}
\end{equation}
Recall that the phases $\theta_j$ only take the discrete values $0$ and $\pi$ because we only consider the eigenvalues which are located on the real axis. 

Our first step is to rewrite the joint density~\eqref{lya:real:jpdf-cov} as a determinant,
\begin{equation}
p_{\jpdf}^{\beta=1,n}(\zeta_1,\theta_1,\ldots,\zeta_N,\theta_N)
=\frac{2^{-nN(N+1)/4}}{2^NN!}\abs[\bigg]{\det_{1\leq k,\ell\leq N}\bigg[\frac{e^{kn\zeta_\ell+i(k-1)\theta_\ell}2nw_{n}^{\beta=1}(e^{n\zeta_\ell})}{\prod_{i=1}^n\Gamma[(\nu_i+\ell)/2]}\bigg]},
\label{lya:real:jpdf-det}
\end{equation}
where we have used the explicit form of the normalisation constant~\eqref{lya:stab:normalisation}.
As for complex and quaternionic matrices, we introduce a normalised and non-negative function
\begin{equation}
f_{k}^{\beta=1,n}(\zeta):= \frac{2n\,\MeijerG{n}{0}{0}{n}{-}{(\nu_1+k)/2,\ldots,(\nu_n+k)/2}{2^{-n}e^{2n\zeta}}}{\prod_{i=1}^n\Gamma\left[(\nu_i+k)/2\right]}.
\label{lya:real:f_k}
\end{equation}
The joint density~\eqref{lya:real:jpdf-det} simplifies considerably if we use the definition~\eqref{lya:real:f_k} together with the explicit expression for the weight function~\eqref{lya:stab:weight} and the identity~\eqref{special:meijer:meijer-shift},
\begin{equation}
p_{\jpdf}^{\beta=1,n}(\zeta_1,\theta_1,\ldots,\zeta_N,\theta_N)
=\frac{1}{2^NN!}\abs[\Big]{\det_{1\leq k,\ell\leq N}\big[e^{i(k-1)\theta_\ell}f_{k}^{\beta=1,n}(\zeta_\ell)\big]}.
\label{lya:real:jpdf-simple}
\end{equation}
Note that the stability exponents in the joint density~\eqref{lya:real:jpdf-simple} are unordered. This can equivalently be expressed as a sum over all possible orderings, 
\begin{equation}
p_{\jpdf}^{\beta=1,n}(\zeta_1,\theta_1,\ldots,\zeta_N,\theta_N)=
\frac{1}{2^NN!}\sum_{\sigma\in S_N}\abs[\Big]{\det_{1\leq k,\ell\leq N}\big[e^{i(k-1)\theta_{\ell}}f_{k}^{\beta=1,n}(\zeta_{\ell})\big]}
\prod_{j=1}^{N-1}\one_{\zeta_{\sigma(j)}\leq\zeta_{\sigma(j+1)}},
\end{equation}
where $\one_{x\leq y}$ denotes the indicator function which is equal to unity if $x\leq y$ and zero otherwise.
Now, if we write the determinant as a sum over permutation, then the joint density becomes
\begin{multline}
p_{\jpdf}^{\beta=1,n}(\zeta_1,\theta_1,\ldots,\zeta_N,\theta_N)=\\
\frac{1}{2^NN!}\sum_{\sigma\in S_N}\abs[\Bigg]{\sum_{\omega\in S_N}\sign\omega\prod_{k=1}^N e^{i(k-1)\theta_{\omega(k)}}f_{k}^{\beta=1,n}(\zeta_{\omega(k)})
\prod_{j=1}^{N-1}\one_{\zeta_{\sigma(j)}\leq\zeta_{\sigma(j+1)}}}.
\label{lya:real:jpdf-long}
\end{multline}
So far no approximation have been used. 

At large $n$ the normalised function~\eqref{lya:real:f_k} can be evaluated using the same technique as for complex and quaternionic matrices; we find that
\begin{equation}
f_{k}^{\beta=1,n}(\zeta)\sim\sqrt{\frac{n}{2\pi(\sigma_{k}^{\beta=1})^2}} \exp\bigg[-n\frac{(\zeta-\mu_{k}^{\beta=1})^2}{2(\sigma_{k}^{\beta=1})^2}\bigg]
\label{lya:real:f-limit}
\end{equation}
with $\mu_{k}^{\beta=1}$ and $\sigma_{k}^{\beta=1}$ as in~\eqref{lya:intro:mean+var}.
In the strict $n\to\infty$ limit, the variances tends to zero and the Gaussians~\eqref{lya:real:f-limit} become delta-peaks. Inserting this into~\eqref{lya:real:jpdf-long}, we obtain
\begin{multline}
\lim_{n\to\infty}p_{\jpdf}^{\beta=1,n}(\zeta_1,\theta_1,\ldots,\zeta_N,\theta_N)=\\
\frac{1}{2^NN!}\sum_{\sigma\in S_N}\abs[\Bigg]{\sum_{\omega\in S_N}\sign\omega\prod_{k=1}^N e^{i(k-1)\theta_{\omega(k)}}\delta(\zeta_{\omega(k)}-\mu_k^{\beta=1})
\prod_{j=1}^{N-1}\one_{\zeta_{\sigma(j)}\leq\zeta_{\sigma(j+1)}}}.
\end{multline}
We know from~\eqref{lya:intro:mean+var} that $\mu_k^{\beta=1}<\mu_\ell^{\beta=1}$ for $k<\ell$, hence the only contribution to the sum within the absolute value comes from the term $\omega=\sigma$ (all other terms are zero due the indicator functions). The absolute value cancels any overall sign of the remaining term, such that
\begin{equation}
\lim_{n\to\infty}p_{\jpdf}^{\beta=1,n}(\zeta_1,\theta_1,\ldots,\zeta_N,\theta_N)
=\frac{1}{2^NN!}\sum_{\sigma\in S_N}\prod_{k=1}^N \delta(\zeta_{\sigma(k)}-\mu_k^{\beta=1})
\label{lya:real:lyapunov-delta}
\end{equation}
with the constraint that the phases only takes the distributed values $0$ and $\pi$, which is the $\beta=1$ case of~\eqref{lya:stab:permanent}.

At large but finite $n$, we have the Gaussian approximation~\eqref{lya:real:f-limit} rather than delta-peaks. Note that the means of the Gaussians are independent of $n$ while the variances decrease for increasing $n$, see~\eqref{lya:real:f-limit}. It follows that at sufficiently large $n$, $f_{k}^{\beta=1,n}(\zeta)$ ($k=1,\ldots,N$) are well-separated Gaussian-peaks. Inserting this approximation into~\eqref{lya:real:jpdf-long}, we see that at large $n$ the dominant term in the sum inside the absolute value is $\omega=\sigma$, while all other terms will be exponentially suppressed due the indicator functions. The simplest illustration of this is when there is only one indicator function, i.e. only one constraint. In this case the Gaussian approximation~\eqref{lya:real:f-limit} yields
\begin{equation}
\int_\R d\zeta_k \int_\R d\zeta_\ell\, f_{\ell}^{\beta=1,n}(\zeta_k) f_{k}^{\beta=1,n}(\zeta_\ell)\one_{\zeta_k\leq\zeta_\ell}
\sim\frac12\erfc\Bigg[\frac{\mu_k^{\beta=1}-\mu_\ell^{\beta=1}}{\big(2(\sigma_k^{\beta=1})^2+2(\sigma_\ell^{\beta=1})^2\big)^{1/2}}\Bigg].
\end{equation}
This integral gives the probability for a particular ordering of Lyapunov exponents and we see that it is exponentially close to unity for $k<\ell$ and exponentially suppressed for $k>\ell$. It is clear that the idea of this argument can be extended to the general case, and it follows that asymptotically we have
\begin{equation}
p_{\jpdf}^{\beta=1,n}(\zeta_1,\theta_1,\ldots,\zeta_N,\theta_N)
\sim\frac{1}{2^NN!}\sum_{\sigma\in S_N}\prod_{k=1}^Nf_{k}^{\beta=1,n}(\zeta_{\sigma(k)}).
\end{equation}
The density limit~\eqref{lya:stab:dens-stab-limit} follows similarly to the complex and quaternionic cases.
\end{proof}

\begin{remark}
Proposition~\ref{thm:lya:stab-exp} tells us that the eigenvalues of a product of independent induced Ginibre matrices become independent and exponentially separated when the number of factors tends to infinity. Furthermore, the absolute value of an individual eigenvalue becomes log-normal distributed. 
\end{remark}

\begin{remark}
If we are interested in the stability exponents only (and \emph{not} their phases), then the proofs for the complex and quaternionic case can be simplified using propositions~\ref{thm:complex:permanent-C} and~\ref{thm:complex:permanent-H}, respectively. After integration over the phases the joint densities for the stability exponents becomes independent random variables even at finite $n$. Their distribution is identical to that of a product of gamma-distributed random variables~\cite{AS:2013,AIS:2014}, thus the problem reduces to show that a product of gamma-distributed random variables becomes log-normal distributed in the large-$n$ limit, which is well-known. 
\end{remark}

\subsection{Lyapunov exponents}
\label{sec:lya:lyapunov}

In this section we will investigate the statistical properties of the Lyapunov exponents~\eqref{lya:intro:lyapunov-def} for a product of induced Ginibre matrices. We will focus on complex ($\beta=2$) case, where we have an exact formula for the joint density of the squared singular values. However, there exist results for all three Dyson classes ($\beta=1,2,4$), see~\cite{Newman:1986,Forrester:2013,Kargin:2014,Forrester:2015}. In fact, the analytic expression for the Lyapunov exponents (but not their fluctuations!) pre-dates the exact formula for the joint density of the singular values from chapter~\ref{chap:singular}. The approach used in~\cite{Newman:1986,Forrester:2013,Kargin:2014,Forrester:2015} is based on an idea  valid for isotropic matrices introduced by Cohen and Newman in~\cite{CN:1984}. We will return to this method in remark~\ref{remark:lya:cohen-newman}.

We recall from proposition~\ref{thm:singular:jpdf-determinantal} that if $x_1,\ldots,x_N$ denote squared singular values of the product matrix (the $n$-dependence is suppressed for notational simplicity), then the joint density is given by
\begin{equation}
\cP_\jpdf^{\beta=2}(x_1,\ldots,x_N)=\frac1{\cZ}
\prod_{1\leq i<j\leq N}(x_j-x_i)
\det_{1\leq i,j\leq N}\big[w_{j-1}(x_i)\big],
\label{lya:lya:jpdf-singular}
\end{equation}
where $\cZ$ is normalisation constant and $w_{j-1}(x)\ (j=1,\ldots,N)$ is a family of positive weight functions given by
\begin{equation}
\cZ=N!\prod_{k=0}^{N-1}k!\prod_{i=1}^n\Gamma[\nu_i+k+1]
\quad\text{and}\quad
w_{j}(x)=\MeijerG{n}{0}{0}{n}{-}{\nu_1+j,\nu_2,\ldots,\nu_n}{x},
\label{lya:lya:weight+normal}
\end{equation}
respectively. As above, $N$ denotes the matrix dimension and $n$ the number of factors.

Like for the stability exponents, it is convenient to change to exponential variables due to the expected asymptotic growth of the singular values~\eqref{lya:intro:singular-asymp}  when $n$ becomes large. For this reason, we introduce the joint density for the finite-$n$ Lyapunov exponents defined by
\begin{equation}
p_\jpdf^{n}(\lambda_1,\ldots,\lambda_N):=\prod_{j=1}^N2ne^{2n\lambda_j}\,\cP_\jpdf^{\beta=2}(e^{2n\lambda_1},\ldots,e^{2n\lambda_N}),
\label{lya:lya:jpdf-lya}
\end{equation}
where joint density on the right hand side is~\eqref{lya:lya:jpdf-singular}. Once again, it is important to note that~\eqref{lya:lya:jpdf-lya} is nothing but a change of variables; the prefactor on the right hand side is simply the Jacobian arising from the change of variables. The spectral density is (as always) given by
\begin{equation}
\rho^{n}(\lambda_1)=N\bigg[\prod_{i=2}^N\int_\R d\lambda_i\bigg]p_\jpdf^{n}(\lambda_1,\ldots,\lambda_N).
\label{lya:lya:dens-lya}
\end{equation}
We will not consider higher point correlations function, since it turns out that the Lyapunov exponents become independent in the large-$n$ if $N$ is kept fixed (as we will see).

Now, let us formulate the main statement of this section, which regards the Lyapunov exponents and their limits. This result was originally shown by Akemann, Burda and Kieburg in~\cite{ABK:2014} for square matrices (i.e. $0=\nu_1=\nu_2=\cdots$). The result given below, which is valid for a non-negative convergent sequence of charges, is a straightforward generalisation of their result, but it is included here for completeness and because it will be useful in the following discussion.

\begin{proposition}\label{thm:lya:lya}
Let $\{\nu_i\}$ be a non-negative convergent sequence with the limit $\nu_i\to\nu_\infty<\infty$.
For $N$ fixed, the large-$n$ limit of the joint density for the Lyapunov exponents~\eqref{lya:lya:jpdf-lya} becomes
\begin{equation}
\lim_{n\to\infty}p_\jpdf^{\beta=2,n}(\lambda_1,\ldots,\lambda_N)
=\frac{1}{N!}\per_{1\leq k,\ell\leq N}\big[\delta(\lambda_\ell-\mu_k^{\beta=2}) \big],
\label{lya:lya:permanent}
\end{equation}
while the density~\eqref{lya:lya:dens-lya} allows the limits
\begin{equation}
\lim_{n\to\infty}\frac{\sigma_k^{\beta=2}}{\sqrt{n}}\rho^{n}\Big(\mu_k^{\beta=2}+\lambda\frac{\sigma_k^{\beta=2}}{\sqrt{n}}\Big)=\frac{e^{-\lambda^2/2}}{\sqrt{2\pi}},\qquad k=1,\ldots,N.
\label{lya:lya:dens-limit}
\end{equation}
where $\mu_k^{\beta=2}$ and $\sigma_k^{\beta=2}$ are defined as in proposition~\ref{thm:lya:stab-exp}.
\end{proposition}

\begin{proof}
We follow the same steps as in~\cite{ABK:2014}: Our starting point is the joint density for the finite-$n$ Lyapunov exponents~\eqref{lya:lya:jpdf-lya}. After absorbing the prefactor in~\eqref{lya:lya:jpdf-lya} into the Vandermonde determinant stemming from~\eqref{lya:lya:jpdf-singular}, we write
\begin{equation}
p_\jpdf^{\beta=2,n}(\lambda_1,\ldots,\lambda_N)=\frac{(2n)^N}{\cZ}\det_{1\leq k,\ell\leq N}\big[e^{2n\,\ell\lambda_k}\big]\det_{1\leq k,\ell\leq N}\big[w_{\ell-1}(e^{2n\lambda_k})\big].
\end{equation}
Expanding the first determinant in this expression gives
\begin{equation}
p_\jpdf^{\beta=2,n}(\lambda_1,\ldots,\lambda_N)=\frac{(2n)^N}{\cZ}\sum_{\sigma\in S_N}\sign\sigma\, e^{2n\sigma(k)\lambda_k}\det_{1\leq k,\ell\leq N}\big[w_{\ell-1}(e^{2n\lambda_k})\big],
\end{equation}
where the sum is over all permutations and `$\sign\sigma$' denotes the sign of the permutation. 

For the next step we recall the explicit form of the normalisation constant~\eqref{lya:lya:weight+normal}. Absorbing factors into the remaining determinant and reordering columns within the determinant to cancel $\sign\sigma$, we see that
\begin{equation}
p_\jpdf^{\beta=2,n}(\lambda_1,\ldots,\lambda_N)=\frac1{N!}\sum_{\sigma\in S_N} \det_{1\leq k,\ell\leq N}
\bigg[\frac{2n\,e^{2nk\lambda_{\sigma(k)}}w_{\ell-1}(e^{2n\lambda_{\sigma(k)}})}{\Gamma[k]\prod_{i=1}^n\Gamma[\nu_i+k]}\bigg],
\label{lya:lya:jpdf-proof-1}
\end{equation}
with weight and normalisation as above.

In order to simplify notation, let us introduce the normalised, non-negative function
\begin{equation}
f_{k\ell}^{n}(\lambda):= \frac{2n\,\MeijerG{n}{0}{0}{n}{-}{\nu_1+k+\ell-1,\nu_2+k\ldots,\nu_n+k}{e^{2n\lambda}}}{\Gamma[\nu_1+k+\ell-1]\prod_{i=2}^n\Gamma[\nu_i+k]}.
\label{lya:lya:f-proof}
\end{equation}
Here, the normalisation may again be checked using the integration formula~\eqref{special:meijer:meijer-moment}, while the non-negativity can be seen by writing the Meijer-$G$ function as an $n$-fold integral similar to the discussion of product of Gaussian distributed random scalars in section~\ref{sec:prologue:scalar}. In fact, $f_{k\ell}^{n}(\lambda)$ has a natural interpretation as the density for a product of random variables, which explains its non-negativity.

Using the relation~\eqref{special:meijer:meijer-shift} together with the definition~\eqref{lya:lya:f-proof}, the joint density~\eqref{lya:lya:jpdf-proof-1} becomes
\begin{equation}
p_\jpdf^{\beta=2,n}(\lambda_1,\ldots,\lambda_N)=\frac1{N!}\sum_{\sigma\in S_N} \det_{1\leq k,\ell\leq N}
\bigg[\frac{\Gamma[\nu_1+k+\ell-1]}{\Gamma[k]\Gamma[\nu_1+k]}f_{k\ell}^n(\lambda_{\sigma(k)})\bigg].
\label{lya:lya:jpdf-proof-2}
\end{equation}
The prefactor within the determinant does not depend on $n$, thus we can restrict our attention to asymptotic behaviour of~\eqref{lya:lya:f-proof}. 

Completely analogously to the discussion of a product of Gaussian random scalars in section~\ref{sec:prologue:scalar} and the proof from the previous section, we can use the cumulant expansion to find the asymptotic behaviour of the function $f_{k\ell}^{n}(\lambda)$. In the following, we will interpret $f_{k\ell}^{n}(\lambda)$ as the probability density for a random variable. The corresponding cumulant-generating function $g_{k\ell}^{n}(\xi)$ is given by
\begin{equation}
g_{k\ell}^{n}(\xi):= \log\Big( \int_{-\infty}^\infty d\lambda\, e^{\xi\lambda} f_{k\ell}^{n}(\lambda)\Big)
\end{equation}
The integral within the logarithm may be performed using~\eqref{special:meijer:meijer-moment}, which yields
\begin{equation}
g_{k\ell}^{n}(\xi)=\log\frac{\Gamma[\nu_i+k+\ell-1+\xi/2n]}{\Gamma[\nu_i+k+\ell-1]}+\sum_{i=2}^n\log\frac{\Gamma[\nu_i+k+\xi/2n]}{\Gamma[\nu_i+k]}.
\end{equation}
The cumulants are found by expanding the cumulant-generating function around $\xi=0$, hence the $j$-th cumulant is given by
\begin{equation}
\kappa_{k\ell;j}^{n}:= \frac{\p^jg_{k\ell}^{n}(\xi)}{\p\xi^j} \bigg\vert_{\xi=0}
=\frac{1}{(2n)^j}\psi^{(j-1)}(\nu_1+k+\ell-1)+\sum_{i=2}^n\frac{1}{(2n)^j}\psi^{(j-1)}(\nu_i+k).
\end{equation}
In particular, we have (recall that $\nu_1$ is a finite $n$-independent constant)
\begin{subequations}
\begin{align}
\lim_{n\to\infty}\kappa_{k\ell;1}^{n}&=\frac{1}{2}\psi(\nu_\infty+k)=\mu_{k}^{\beta=2}, \\
\lim_{n\to\infty}n\,\kappa_{k\ell;2}^{n}&=\frac{1}{4}\psi'(\nu_\infty+k)=(\sigma_{k}^{\beta=2})^2.
\end{align}
\end{subequations}
To find the limiting distribution in the large-$n$ limit, we again switch from the variable $\lambda$ to the standardised variable $\widetilde \lambda:= \sqrt{n}(\lambda-\mu_{k\ell}^{\beta})/\sigma_{k\ell}^{\beta}$, which has mean zero and unit variance. The corresponding standardised cumulants are
\begin{equation}
\widetilde\kappa_{k\ell;1}^{n}=0,\qquad
\widetilde\kappa_{k\ell;2}^{n}=1 \qquad\text{and}\qquad
\widetilde\kappa_{k\ell;n}^{n}= O(n^{1-j/2}) \quad\text{for}\quad j\geq 3.
\end{equation}
We see that the higher order cumulants tend to zero for large $n$, and, as in the proof of proposition~\ref{thm:lya:stab-exp},
it follows that the limiting distribution is a Gaussian. In terms of the original variables, we have
\begin{equation}
f_{k\ell}^{n}(\lambda)\sim\sqrt{\frac{n}{2\pi(\sigma_k^{\beta=2})^2}}\exp\bigg[-n\frac{(\lambda-\mu_{k}^{\beta=2})^2}{2(\sigma_{k}^{\beta=2})^2}\bigg].
\end{equation}
Inserting this asymptotic form back into the joint density~\eqref{lya:lya:jpdf-proof-2} yields
\begin{multline}
p_\jpdf^{\beta=2,n}(\lambda_1,\ldots,\lambda_N)\sim\frac1{N!}\sum_{\sigma\in S_N} \det_{1\leq k,\ell\leq N}
\bigg[\frac{\Gamma[\nu_1+k+\ell-1]}{\Gamma[k]\Gamma[\nu_1+k]}\bigg]\\
\times\sqrt{\frac{n}{2\pi(\sigma_k^{\beta=2})^2}}\exp\bigg[-n\frac{(\lambda-\mu_{k}^{\beta=2})^2}{2(\sigma_{k}^{\beta=2})^2}\bigg].
\label{lya:lya:jpdf-proof-3}
\end{multline}
It is straightforward to show using an induction argument that the determinant in this expression is equal to unity independently of the value of $\nu_1$, thus the right hand side is recognised as a permanent and the limits~\eqref{lya:lya:permanent} and~\eqref{lya:lya:dens-limit} follows.
\end{proof}

\begin{remark}\label{remark:lya:cohen-newman}
As mentioned in the beginning of this section, the Lyapunov exponents for a product of $N\times N$ induced Ginibre matrices are actually known for all three Dyson classes ($\beta=1,2,4$). We have~\cite{Newman:1986,Forrester:2013,Kargin:2014}
\begin{equation}
\mu_k^\beta=\frac{1}{2}\log\frac{2}{\beta}+\frac{1}{2}\psi\Big(\frac{\beta(\nu_\infty+k)}{2}\Big),\qquad k=1,\ldots,N.
\label{lya:lya:lyapunov-general}
\end{equation}
In fact, this result pre-dates the proof for $\beta=2$ given above (at least for $0=\nu_1=\nu_2=\cdots$). The method used in~\cite{Newman:1986,Forrester:2013,Kargin:2014} is based on an approach introduced Cohen and Newman in~\cite{CN:1984} and is valid for isotropic matrices (we will refer to this approach as the Cohen--Newman method). 

Let us briefly recall the main idea of the Cohen--Newman method~\cite{CN:1984,Newman:1986} (here we follow a formulation of the Cohen--Newman method similar to that presented in~\cite{ABK:2014}). The starting point is to note that the sum of the $K$ largest Lyapunov exponents can be written as
\begin{equation}
\Lambda_K(n):=\lambda_N(n)+\cdots+\lambda_{N-K+1}(n)=\frac1{2n}\sup_{M\in\F_\beta^{N\times K}}\log\frac{\det M^\dagger Y_n^\dagger Y_n M}{\det M^\dagger M},
\end{equation}
where $Y_n=X_n\cdots X_1$ is the product matrix. We may write this as a telescopic sum
\begin{equation}
\Lambda_K(n)=\frac1{2n}\sup_{M\in\F_\beta^{N\times K}}\sum_{i=1}^n\log\frac{\det M_i^\dagger X_i^\dagger X_i M_i}{\det M_i^\dagger M_i}
\end{equation}
with $M_i:=Y_{i-1}M$ for $i\geq2$ and $M_1=M$. Each of these matrices may be decomposed as $M_i=U_iP\Sigma_iV_i$ with $U_i\in\gU(N,\F_\beta)$ and $V_i\in\gU(K,\F_\beta)$ while $\Sigma_i$ is a $K\times K$ positive semi-definite diagonal matrix and $P=[\delta_{ij}]$ is the $N\times K$ projection matrix, which projects the $N$ dimensional vector space down to a $K$ dimensional vector space. If the matrices $X_i$ ($i=1,\ldots,n$) are statistically isotropic then we have
\begin{equation}
\Lambda_K(n)
\stackrel{d}{=}\frac1{2n}\sum_{i=1}^n\log\det P^\dagger X_i^\dagger X_i P
\stackrel{d}{=}\frac1{2n}\sum_{i=1}^n\sup_{M\in\F_\beta^{N\times K}}\log\frac{\det M^\dagger X_i^\dagger X_i M}{\det M^\dagger M}.
\end{equation}
Here, we may apply Kolmogorov's law of large numbers to the right hand side, which in the case of induced Ginibre matrices gives~\eqref{lya:lya:lyapunov-general}.

Very recently, it was pointed out in~\cite{Forrester:2015} that the Cohen--Newman method also may be used to make predictions for the variances of the Lyapunov exponents. Here, we look at
\begin{equation}
\average{(\Lambda_K)^2}-\average{\Lambda_K}^2=\sum_{k=1}^K\var(\lambda_{N-k+1})+2\sum_{1\leq i<j\leq K}\cov(\lambda_{N-i+1},\lambda_{N-j+1}),
\label{lya:lya:var-CN}
\end{equation}
where expectation is with respect to the joint density~\eqref{lya:lya:jpdf-lya}. The left hand side may be evaluated using ideas similar to those discussed above. We note that if $K=1$ then the second sum on the right hand side disappears, which allow us to determine the variance of the largest Lyapunov exponent (this was done as early as~\cite{CN:1984}). It was argued in~\cite{Forrester:2015} that this approach may be extended to $K>1$ if the Lyapunov spectrum is discrete and non-degenerate, since the covariances are expected to decay exponentially fast in this case. If we believe this assumption, then the variances of the Lyapunov exponents are
 \begin{equation}
(\sigma_k^\beta)^2=\frac{1}{4}\psi'\Big(\frac{\beta(\nu_\infty+k)}{2}\Big),\qquad k=1,\ldots,N.
\label{lya:lya:variance-general}
\end{equation}
We note that this prediction is in agreement with the exact evaluation for $\beta=2$ given in proposition~\ref{thm:lya:lya} as well as the stability exponents from proposition~\ref{thm:lya:stab-exp}.
\end{remark}

\section{Macroscopic density for large matrix dimension}
\label{sec:lya:macro}

In this section, we will consider the macroscopic spectral density for a product of induced Ginibre matrices as both matrix dimension and the number of factors tends to infinity. We will see that the two limits commute and that the macroscopic spectral density converges in distribution to Newman's triangular law~\cite{Newman:1986,IN:1992}. We note that the triangular law also has been revisited in the context of free probability~\cite{Kargin:2008,Tucci:2010} and in the context of exactly solvable matrix models in~\cite{ABK:2014}. 

Here, we will present separate proofs for the convergence of the eigenvalue spectrum (stability exponents) and of the spectrum of the singular values (Lyapunov exponents). 

\begin{proposition}\label{thm:lya:triangular-stab}
Let $\{\nu_i\}$ be a non-negative convergent sequence with the limit $\nu_i\to\nu_\infty<\infty$. Consider the spectral density, $R_1^{\beta,n}(x)$, for the eigenvalues of a product of $n$ independent $N\times N$ real ($\beta=1$), complex ($\beta=1$), or quaternionic ($\beta=4$) induced Ginibre matrices with charges $\nu_1,\ldots,\nu_n$. We have
\begin{align}
\int_0^adr\,\rho_\vartriangle(r)
&=\lim_{N\to\infty}\lim_{n\to\infty}\int_0^adr\,2nN^{n-1}r^{2n-1}\int_{-\pi}^\pi  d\theta\,R_1^n(N^{n/2}r^{n}e^{i\theta})\label{lya:asymp:tri-limit-stab-Nn} \\
&=\lim_{n\to\infty}\lim_{N\to\infty}\int_0^adr\,2nN^{n-1}r^{2n-1}\int_{-\pi}^\pi  d\theta\,R_1^n(N^{n/2}r^{n}e^{i\theta})\label{lya:asymp:tri-limit-stab-nN}
\end{align}
pointwise for $a\geq0$, where
\begin{equation}
\rho_\vartriangle(x):=
\begin{cases}
2x & \textup{for }\ 0<x<1 \\
0 & \textup{otherwise}
\end{cases}
\label{lya:asymp:triangular}
\end{equation}
is the density of the triangular law.
\end{proposition}

\begin{proof}
Let us start with the iterated limit~\eqref{lya:asymp:tri-limit-stab-Nn}.
We know from proposition~\ref{thm:lya:stab-exp} that
\begin{equation}
\lim_{n\to\infty}2ne^{2n\zeta}\int_{-\pi}^\pi d\theta\, R_1^{\beta,n}(e^{n\zeta+i\theta})=\sum_{k=1}^N\delta(\zeta-\mu_k^\beta),
\label{lya:asymp:sum-stab-proof}
\end{equation}
where $\mu_k^\beta$ ($k=1,\ldots,N$) are given by~\eqref{lya:intro:mean+var}.
It follows by a change of variables that
\begin{equation}
\lim_{n\to\infty}2nN^{n-1}r^{2n-1}\int_{-\pi}^\pi d\theta\, R_1^{\beta,n}(N^{n/2}r^{n}e^{i\theta})=\frac{1}{N}\sum_{k=1}^N\delta(r-e^{\mu_k^\beta-\frac12\log N}).
\end{equation}
In order to evaluate the large-$N$ limit, we need to look at the asymptotic behavior of $\mu_k^\beta$ for large $k$. The asymptotic formula~\eqref{special:gamma:digamma-asymp} may be used to evaluate $\mu_k^\beta$ for large $k$, since $\nu_\infty$ is a finite constant. If $\alpha>0$ and $k=\alpha N$, then we have
\begin{equation}
e^{\mu_k^\beta-\frac12\log N}=\sqrt{\alpha}\,(1+O(N^{-1}))
\end{equation}
and it follows that
\begin{equation}
\lim_{N\to\infty}\lim_{n\to\infty}\int_0^adr\,2nN^{n-1}r^{2n-1}\int_{-\pi}^\pi d\theta\, R_1^{\beta,n}(N^{n/2}r^{n}e^{i\theta})
=
\begin{cases}
a^2 & \textup{for }\ 0<a<1 \\
1 & \textup{otherwise}
\end{cases},
\end{equation}
which confirms that~\eqref{lya:asymp:tri-limit-stab-Nn} converges as claimed.

We can now turn to the to the other iterated limit~\eqref{lya:asymp:tri-limit-stab-nN}.
We recall from proposition~\ref{thm:complex:macro-C} and~\ref{thm:complex:macro-H} that for $\beta=2,4$, we have
\begin{equation}
\lim_{N\to\infty}2nN^{n-1}r^{2n-1}\int_{-\pi}^\pi d\theta\, R_1^{\beta,n}(N^{n/2}r^{n}e^{i\theta})=4\pi nr^{2n-1}\rho_\macro^{n,\alpha=0}(r^{n}),
\end{equation}
where $\rho_\macro^{n,\alpha=0}(r^{n})$ is the macroscopic density~\eqref{complex:C:density-macro0}; the same is known to hold for $\beta=1$, see~\cite{GT:2010,OS:2011,GKT:2014}. In order to evaluate the large-$n$ limit, we first recall that the outer edge by definition is located at unity. Since we have integrated out the phase, the task of finding the limiting density is a Haussdorff moment problem if the moments converge. It follows from the explicit expression for the density~\eqref{complex:C:density-macro0} that
\begin{align}
\int_0^1 dr\,4\pi nr^{2n-1}\rho_\macro^{n,\alpha=0}(r^{n})r^s=\frac{2}{s+2},
\end{align}
thus the large-$n$ limit for the moments becomes trivial. The density~\eqref{lya:asymp:triangular} is obtained by an inverse Mellin transform or by a direct calculation of the moments of~\eqref{lya:asymp:triangular}.
\end{proof}

\begin{remark}
We emphasise that the phase-average included in~\eqref{lya:asymp:tri-limit-stab-Nn} and~\eqref{lya:asymp:tri-limit-stab-nN} is necessary for the proposition to hold, since the $N$- and $n$-limit do not commute without it. For instance, if we look at the spectrum for a product of real ($\beta=1$) matrices and we take $N\to\infty$ before we take $n\to\infty$ then we find a limiting density which is invariant under rotations in the complex plane; on the other hand, if we take $n\to\infty$ before we take $N\to\infty$ then we find a limiting density restricted to the real line. A similar discrepancy is present in the quaternionic case, while the spectrum is rotational symmetric in both cases for complex matrices.
\end{remark}

Let us formulate the equivalent statement for the singular values.

\begin{proposition}\label{thm:lya:triangular}
Let $\{\nu_i\}$ be a non-negative convergent sequence with the limit $\nu_i\to\nu_\infty<\infty$. The spectral density, $R_1^n(x)$, for the squared singular values for a product of $n$ independent $N\times N$ induced Ginibre matrices with charges $\nu_1,\ldots,\nu_n$ satisfies,
\begin{align}
\int_0^adx\,\rho_\vartriangle(x)
&=\lim_{N\to\infty}\lim_{n\to\infty}\int_0^adx\,2nN^{n-1}x^{2n-1}R_1^{\beta,n}(N^nx^{2n}) \label{lya:asymp:tri-limit-Nn}\\
&=\lim_{n\to\infty}\lim_{N\to\infty}\int_0^adx\,2nN^{n-1}x^{2n-1}R_1^{\beta,n}(N^nx^{2n}) \label{lya:asymp:tri-limit-nN}
\end{align}
pointwise for $a\geq0$, where $\rho_\vartriangle(x)$ is given by~\eqref{lya:asymp:triangular}
\end{proposition}

\begin{proof}
We will follow exactly the same idea as for the proof of~\ref{thm:lya:triangular-stab}. First, we look at the limit~\eqref{lya:asymp:tri-limit-Nn}.
We know from proposition~\ref{thm:lya:lya} and remark~\ref{remark:lya:cohen-newman} that
\begin{equation}
\lim_{n\to\infty}2ne^{2n\lambda}R_1^{\beta,n}(e^{2n\lambda})=\sum_{k=1}^N\delta(\lambda-\mu_k^\beta),
\end{equation}
where $\mu_k^\beta$ ($k=1,\ldots,N$) are as in~\eqref{lya:asymp:sum-stab-proof}. Thus, the large-$n$ limit can be evaluated in the same way as in the proof of proposition~\ref{thm:lya:stab-exp}.

It remains to investigate the limit~\eqref{lya:asymp:tri-limit-nN}.
We know that~\cite{Muller:2002,BJLNS:2010,BBTCC:2011,PZ:2011,GKT:2014}
\begin{equation}
\lim_{N\to\infty}2nN^{n-1}x^{2n-1}R_1^n(N^nx^{2n})=2nx^{2n-1}\rho_\fuss^n(x^{2n}),
\end{equation}
where $\rho_\fuss^n(x)$ is the Fuss--Catalan density~\eqref{singular:asymp:density-fc} (for $\beta=2$ this was stated as proposition~\ref{thm:singular:fuss-catalan}). In order to evaluate the large-$n$ limit, we first recall that the right edge for the Fuss--Catalan distribution is located at $(n+1)^{n+1}/n^n$. Thus, if the large-$n$ limit exists, then its edge must be located at unity, since
\begin{equation}
\lim_{n\to\infty}\bigg(\frac{(n+1)^{n+1}}{n^n}\bigg)^{1/2n}=1.
\end{equation}
This means that the task of finding the limiting density reduces to a Haussdorff moment problem if the moments converge. We have 
\begin{align}
\int_0^\infty dx\,2nx^{2n-1}\rho_\fuss^n(x^{2n})x^s&=\int_0^\infty dy\,\rho_\fuss^n(y)y^{s/2n}\nn\\
&=\frac{1}{n(s/2n)+1}\binom{(n+1)(s/2n)}{s/2n}
\xrightarrow{n\to\infty}\frac{2}{s+2},
\end{align}
where the second equality follows from~\eqref{singular:asymp:moments-fc}. As in the proof of proposition~\ref{thm:lya:lya}, the density~\eqref{lya:asymp:triangular} can be obtained by an inverse Mellin transform.
\end{proof}

\section{Summary, discussion and open problems}
\label{sec:lya:discuss}

In section~\ref{sec:lya:finiteN}, we have seen that the eigenvalues as well as the singular values for a product of induced Ginibre matrices diverge exponentially for a large number of factors. As a consequence, the stability and Lyapunov exponents become independent Gaussian random variables, asymptotically; a structure which conveniently can be written as a permanent.  Furthermore, it was shown that the stability and Lyapunov exponents converge to the same limit when the matrix dimension is kept fixed (in agreement with a conjecture from~\cite{GSO:1987}), i.e. the two approaches are (in this sense) equivalent. More precisely, (i) on the scale which includes all exponents, we have convergence to a non-random limit where the locations of the Lyapunov (stability) exponents may be expressed through a digamma function, while (ii) on the scale of the fluctuations of the individual exponents (which is suppressed compared to exponent interspacing for large $n$) we have Gaussian fluctuations with variances which 
can be expressed through a trigamma function. The Gaussian prediction for the smallest Lyapunov exponent is compared with numerical data for different values of $n$ on figure~\ref{fig:lya:gauss}.

\begin{figure}[tp]
\centering
\includegraphics{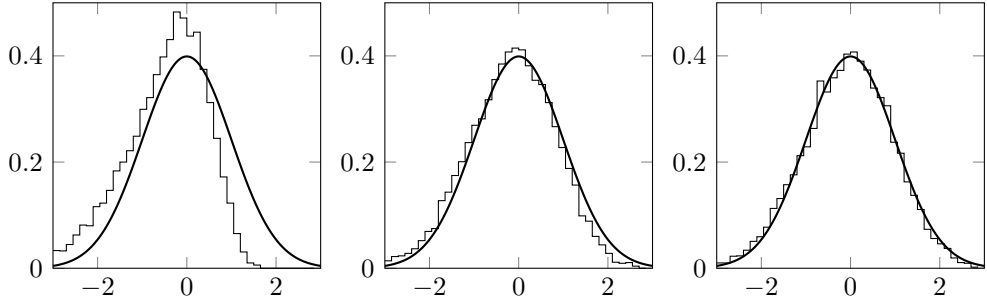}
\caption{Comparison of the analytic prediction for the fluctuations of the smallest Lyapunov exponent [i.e. a Gaussian with mean and variance given according to~\eqref{lya:intro:mean+var}] with numerical data stemming from $10\,000$ realisations of a product of $4$ (left panel), $40$ (centre panel), and $400$ (right panel) independent $4\times4$ real Ginibre matrices.}
\label{fig:lya:gauss}
\end{figure}

Furthermore, the complex phases of the eigenvalues have a striking dependence on the Dyson class in the large-$n$ limit: for $\beta=1$ the spectrum is purely real, for $\beta=2$ the spectrum is rotational symmetric, while for $\beta=4$ there is a sine square repulsion from the real axis (this corresponds to the Vandermonde repulsion for a single complex conjugate pair). The structure is visualised by a numerical simulation on figure~\ref{fig:lya:scatter}.

\begin{figure}[tp]
\centering
\includegraphics{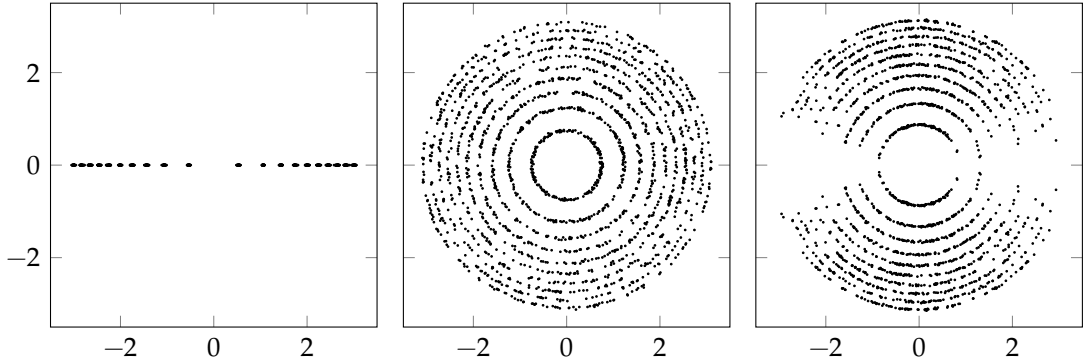}
\caption{Scatter plots of the eigenvalues of a product of $1\,000$ independent real (left panel), complex (centre panel), and quaternionic (right panel) $10\times 10$ Gaussian random matrices. The ensemble consists of $100$ matrices for $\beta=1,4$ and $200$ matrices for $\beta=2$. With the notation from section~\ref{sec:lya:stability}, the plotted values corresponds to $e^{\zeta_k+i\theta}$.}
\label{fig:lya:scatter}
\end{figure}

It might seem surprising that all eigenvalues of a product of real independent Gaussian matrices become real when taking the number of factors to infinity, assuming that the matrix dimension is kept fixed (as illustrated by the left panel on figure~\ref{fig:lya:scatter}). In fact, this phenomenon is far from fully explained. So far, only the Gaussian case has been proven to have a real spectrum in the limit, but the phenomenon is highly expected to hold for a more general class of matrices. This conjecture is supported by numerical results~\cite{HJL:2015}. A first analytical check for the generality of the statement could be to verify that the other known exactly solvable models (i.e. mixed products including truncated orthogonal and inverse Ginibre matrices~\cite{IK:2014,AI:2015}) also have purely real spectra in the large-$n$ limit (this can be done using either the method presented in this chapter or using an approach similar to~\cite{Forrester:2014a}). It is an interesting, but challenging, task to 
establish more general conditions implying that the spectrum becomes real. This would also be of physical interest since it would imply that the limiting eigenbasis (the so-called stability basis~\cite{GSO:1987}) is real.  It should be mentioned that in the opposite limit (i.e. large matrix dimension, fixed number of factors) the macroscopic eigenvalue spectrum becomes rotational invariant~\cite{GT:2010,OS:2011}.

In section~\ref{sec:lya:macro}, we saw that the macroscopic spectral density for both the Lyapunov exponents (singular values) and the stability exponents (phase-averaged eigenvalues) tends to the same $\beta$-independent distribution, known as the triangular law~\cite{Newman:1986,IN:1992}, in the limit where both matrix dimensions and the number of factors become large (see left panel on figure~\ref{fig:lya:triangular}). In particular, we showed that the two limits were interchangeable (a similar discussion was given for complex matrices in~\cite{ABK:2014}). This triangular law dates back to~\cite{Newman:1986}, where it was shown to hold for the singular values of a product of real Ginibre matrices under the assumption that the number of factors was taken to infinity before the matrix size.

\begin{figure}[tp]
\centering
\includegraphics{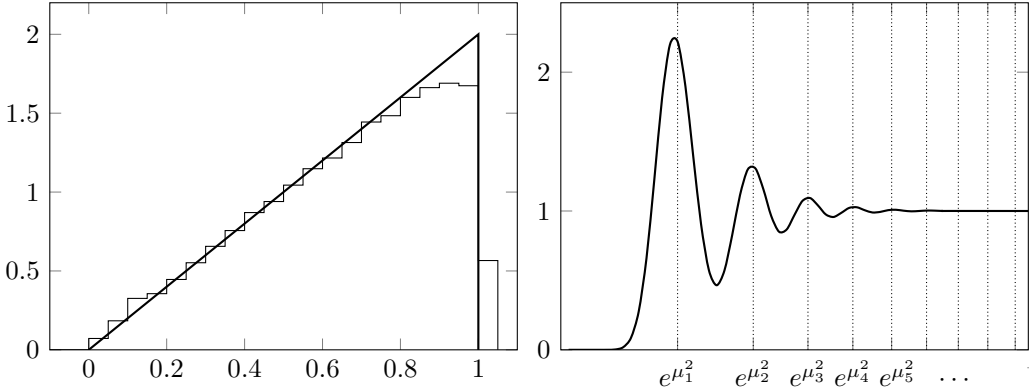}
\caption{The left panel compares the triangular law from proposition~\ref{thm:lya:triangular} with singular values numerically generated from $100$ realisations of a product of $100$ independent $100\times 100$ real Ginibre matrices. The right panel shows the radial microscopic density at the origin for $n=16$ (also seen on figure~\ref{fig:complex:origin}) together with asymptotic prediction of the ``crystallised eigenvalues for $n\to\infty$.}
\label{fig:lya:triangular}
\end{figure}

It is, of course, highly desirable to establish results for the microscopic regimes in the double scaling limit. It was argued in~\cite{ABK:2014}, that in general the large-$n$ and -$N$ limit will not commute in the microscopic scaling regime. This must be true, since the microscopic limits in the bulk and at the (soft) edge obtained for large $N$ (see proposition~\ref{thm:singular:sine}, \ref{thm:singular:airy}, and~\ref{thm:complex:bulk+edge}) are incompatible with limits obtained for large $n$ (see proposition~\ref{thm:lya:lya} and~\ref{thm:lya:stab-exp}). The distribution of the eigenvalue furthest from origin for a product of complex Ginibre matrices was considered in~\cite{JQ:2014}. Here, a crossover between the Gumbel distribution (Ginibre universality) and the log-normal distribution (stability exponent universality) was observed. 

Let us attempt to provide a heuristic understanding of the interplay between the two limits. In this chapter, we have seen that the Lyapunov (stability) exponents become independent Gaussians for large $n$ but fixed $N$. The intuitive interpretation is that the decorrelation occurs because the singular (eigen-) values separate exponentially fast with $n$ while the correlation length does not. Thus, it seems justified to believe that the Gaussian structure is a reasonable approximation as long as the separation between Lyapunov (stability) exponents is much larger than the average amplitude of their fluctuations. In other words, the approximation of independent Gaussians is expected to be reasonable if  
\begin{equation}
\frac{\mu_{k+1}^{\beta}-\mu_{k}^{\beta}}{\sigma_{k+1}^{\beta}+\sigma_{k}^\beta}\gg n^{-1/2}
\label{lya:discuss:CN-regime}
\end{equation}
with $\mu_k^\beta$ and $\sigma_k^\beta$ as in~\eqref{lya:intro:mean+var}.
Let us assume that we are away from the hard edge (origin) and set $k=\alpha N$ where $\alpha>0$ is a positive parameter. Using the asymptotic formula~\eqref{special:gamma:digamma-asymp}, we see that
\begin{equation}
\mu_{\alpha N+1}^{\beta}-\mu_{\alpha N}^{\beta}=O(N^{-1})
\qquad\text{and}\qquad
\sigma_{\alpha N+1}^{\beta}+\sigma_{\alpha N}^\beta=O(N^{-1/2})
\end{equation}
for large $N$. Comparing with~\eqref{lya:discuss:CN-regime}, we therefore expect that the Gaussian approximation is reasonable if $n\gg N$. On the other hand, at the hard edge (the origin) we have $k=O(1)$ independently of the value of $N$, hence the amplitude of the fluctuations of the Lyapunov (stability) exponents will always decay compared to their separation in this regime. This heuristic argument leads us to believe that the microscopic spectrum at the hard edge (the origin) will ``crystallise'' in a similar manner as in proposition~\ref{thm:lya:lya} (proposition~\ref{thm:lya:stab-exp}). This is consistent with our observation for small $n$ given on figure~\ref{fig:singular:meijer} (figure~\ref{fig:complex:origin}), where we saw that the amplitude of the oscillation in the spectrum grew with $n$. In the case of the complex spectrum, we can formalise this notion.

\begin{proposition}
Let $\{\nu_i\}$ be a non-negative convergent sequence with the limit $\nu_i\to\nu_\infty<\infty$. Introduce the microscopic eigenvalue density at the origin for a product of $n$ independent complex Ginibre matrices with charges $\nu_1,\ldots,\nu_n$ as given by proposition~\ref{thm:complex:origin}, $\rho_\origin^{n,\nu}(z):=K_\origin^{n,\nu}(z,z)$, then 
\begin{equation}
\lim_{n\to\infty}2ne^{2n\zeta}\rho_\origin^{n,\nu}(e^{n\zeta+i\theta})=\sum_{k=1}^\infty \frac1{2\pi}\delta(\zeta-\mu_k^{\beta=2})
\end{equation}
and
\begin{equation}
\lim_{n\to\infty}2\sqrt n\sigma^{\beta=2}_ke^{2n\zeta}\rho_\origin^{n,\nu}\Big(\exp\Big[n\Big(\mu_k^{\beta=2}+\zeta\frac{\sigma^{\beta=2}_k}{\sqrt n}\Big)+i\theta\Big]\Big)=\frac1{2\pi}\frac{e^{-\zeta^2/2}}{\sqrt{2\pi}},\qquad k=1,2,\ldots
\end{equation}
with $\mu_k^{\beta=2}$ and $\sigma_k^{\beta=2}$ defined as in~\eqref{lya:intro:mean+var}.
\end{proposition}

\begin{proof}
The main idea is the same as in the proofs of proposition~\ref{thm:lya:stab-exp} and~\ref{thm:lya:lya}. In this case, our starting point is the density
\begin{equation}
2ne^{2n\zeta}\rho_\origin^{n,\nu}(e^{n\zeta+i\theta})=
\frac{2ne^{2n\zeta}}{\pi}\MeijerG{n}{0}{0}{n}{-}{\nu_1,\ldots,\nu_n}{e^{2n\zeta}}\MeijerG{1}{1}{1}{n+1}{0}{0,-\nu_1,\ldots,\nu_n}{-e^{2n\zeta}},
\end{equation}
which we have from proposition~\ref{thm:complex:origin}. We can write the second Meijer $G$-function as a sum, using~\eqref{special:meijer:hyper-meijer} and~\eqref{special:meijer:hypergeometric}. This gives
\begin{equation}
2ne^{2n\zeta}\rho_\origin^{n,\nu}(e^{n\zeta+i\theta})=
\sum_{k=1}^\infty \frac1{2\pi}\frac{2n\MeijerG{n}{0}{0}{n}{-}{\nu_1,\ldots,\nu_n}{e^{2n\zeta}}}{\prod_{i=1}^n\Gamma[\nu_i+k]}e^{2nk\zeta}.
\end{equation}
Here, the last exponential can be absorbed into the Meijer $G$-function using~\eqref{special:meijer:meijer-shift}, which allow us to write the density as
\begin{equation}
2ne^{2n\zeta}\rho_\origin^{n,\nu}(e^{n\zeta+i\theta})=
\sum_{k=1}^\infty \frac1{2\pi}f^{\beta=2,n}_{kk}(\zeta)
\end{equation}
with $f^{\beta=2,n}_{kk}(\zeta)$ defined by
\begin{equation}
f^{\beta=2,n}_{kk}(\zeta):= \frac{2n\MeijerG{n}{0}{0}{n}{-}{\nu_1+k,\ldots,\nu_n+k}{e^{2n\zeta}}}{\prod_{i=1}^n\Gamma[\nu_i+k]}.
\end{equation}
This function is recognised as~\eqref{lya:stab:f-complex} from the proof of proposition~\ref{thm:lya:stab-exp} and the rest of this proof follows exactly as described there.
\end{proof}

It is remarkable that the same method used to study the large-$n$ limit for fixed $N$ can be used to study the microscopic behaviour of eigenvalues near the origin. On the right panel of figure~\ref{fig:lya:triangular}, we compare the microscopic density for $n=16$ (see also figure~\ref{fig:complex:origin}) with the asymptotic values for the eigenvalues. We see emerging peaks around the asymptotic values of the stability exponents even at relatively small $n$. It would be interesting to see whether a similar approach can be applied to the microscopic density for the singular values near the hard edge (given in terms of the Meijer $G$-kernel, see proposition~\ref{thm:singular:meijer-kernel}).

Finally, we emphasise that it is a non-trivial task to find an efficient algorithm for the numerical calculation of the Lyapunov exponents from a matrix product with large matrix dimensions and a large number of factors. A direct implementation of the mathematical definition of matrix multiplication gives an algorithm that takes time on the order $O(nN^3)$, where (as above) $n$ denotes the number of factors and $N$ denotes the matrix dimension. Moreover, high precision computations are often required for the smallest Lyapunov exponents (singular values) when using the direct implementation of the matrix product. In the regime where the Cohen--Newman method can be trusted ($n\gg N$) a more efficient (but less accurate) approach would be a Monte Carlo algorithm based on the Cohen--Newman method. We refer to~\cite{WT:2001,Bai:2007,Vanneste:2010,Pollicott:2010,Forrester:2013} for a discussion of numerical methods. Analytic results, such as those presented in this chapter provide a natural standard against which 
numerical methods can be compared. We stress that all numerical data presented here (and throughout the thesis) have been obtained using a direct implementation of the matrix product.

%% file: decomposition.tex
\chapter{Matrix decompositions and their Jacobians}
\label{app:decompositions}


In this appendix we introduce several matrix decompositions, which are used throughout the thesis.
The appendix is organised as follows: In section~\ref{sec:decomp:2by2} we briefly review some properties of $2\times 2$ matrices and their relation to the algebras of quaternions and split-quaternions, while section~\ref{sec:decomp:decomp} summarises some important matrix decompositions.

\section{Two-by-two matrices and their algebras}
\label{sec:decomp:2by2}

It is well-known from numerical linear algebra that $2\times 2$ matrices are at the foundation of many numerical computations. We believe that it is reasonable to say the same about analytical computations. For this reason, we devote this section to the description of two types of $2\times 2$ matrices with particular importance in random matrix theory; namely, matrices which are (ring-)isomorphic to either the quaternions or split-quaternions.

It is often useful to imagine the underlying algebraic structure as a consequence of imposing certain symmetry constraints on the system under consideration. From a physical point of view, we can imagine that we replace a physical operator by a random matrix restricted by some global symmetries determined by physical considerations. This approach is known to be very successful; sometimes it is even possible to establish an exact link between the random matrix model and some (universal) regime of the corresponding effective field theory, see e.g. the reviews~\cite{Beenakker:1997,VW:2000}. In fact, it often fruitful to classify random matrix models according to their symmetries~\cite{Dyson:1962,AZ:1997,BLC:2002,Magnea:2008}. Any symmetry classification, of course, depends on what type of symmetries we allow and on which types of matrices we consider to be equivalent. We will not give any details regarding such classification schemes here, but rather refer the reader to the aforementioned references.

Let us consider a $2N\times 2N$ complex matrix, $X$, and impose an anti-unitary symmetry given by
\begin{equation}
JX^*=XJ,
\label{decomp:2by2:antiU}
\end{equation}
where $J$ is either a symmetric ($\beta=1$) or an anti-symmetric ($\beta=4$) unitary matrix. We will use the canonical representations: an identity matrix for $\beta=1$ or an antisymmetric quasi-diagonal matrix
\begin{equation}
J=\diag(\big[\begin{smallmatrix} 0 & +1 \\ -1 & 0 \end{smallmatrix}\big],\ldots,\big[\begin{smallmatrix} 0 & +1 \\ -1 & 0 \end{smallmatrix}\big]\big)
\label{decomp:2by2:J}
\end{equation}
for $\beta=4$.
These are the so-called $K$-type symmetries in the classification from~\cite{BLC:2002,Magnea:2008}. If $J$ is the identity then the anti-unitary symmetry~\eqref{decomp:2by2:antiU} implies that $X$ is a $2N\times 2N$ real matrix, or equivalently an $N\times N$ split-quaternionic matrix as we will see in section~\ref{sec:decomp:real}. If $J$ is given by~\eqref{decomp:2by2:J} then symmetry implies that the $X$ consists $2\times 2$ blocks of the form
\begin{equation}
\begin{bmatrix} u & v \\ -v^* & u^*\end{bmatrix},\qquad u,v\in\C.
\label{decomp:2by2:H-structure}
\end{equation}
Such matrices will be referred to as quaternions for reasons which will be made clear in section~\ref{sec:decomp:H}.

\subsection{Symplectic symmetry and quaternions}
\label{sec:decomp:H}

Quaternions are described in several textbooks on random matrix theory, see e.g.~\cite{Mehta:2004,Forrester:2010,AGZ:2010}. Nonetheless, we briefly review some well-known properties of quaternions for easy reference and to avoid any confusion about notation. In addition to the aforementioned textbooks, we refer the reader to~\cite{Zhang:1997,GW:2009}.

\begin{definition}
The skew-field of quaternions, $\H$, is a four dimensional vector space $\R^4$ with basis $\{\one,\qi,\qj,\qk\}$ equipped with associative multiplication defined such that $\one$ is the identity and $\qi^2=\qj^2=\qk^2=\qi\qj\qk=-\one$.
\end{definition}

\begin{definition}
Given $q=a+b\qi+c\qj+d\qk\in\H$ ($a,b,c,d\in\R$), then we say that $q$ is a (real) quaternion and define
\begin{equation}
q^\dagger=a-b\qi-c\qj-d\qk,\qquad
q^*=a-b\qi+c\qj-d\qk \qquad\text{and}\qquad 
\abs{q}=\sqrt{q^\dagger q}
\end{equation}
called the quaternion conjugate, the partial conjugate, and the modulus, respectively.
\end{definition}

\begin{corollary}\label{thm:decomp:H-relations}
Let $q\in\H$, then 
\begin{inparaenum}[(i)]
\item the inverse of $q\neq0$ is $q^{-1}=q^\dagger/\abs{q}^2$, 
\item\label{thm:decomp:H-relations(ii)} the partial conjugate $q$ is $q^*=-\qj q\qj$, and 
\item the modulus, $\abs{\,\cdot\,}$, is a norm on $\H$.
\end{inparaenum}
\end{corollary}


\begin{definition}
Let $q_{ij}$ ($i=1,\ldots,N;j=1,\ldots,M$) be quaternions, then $X=[q_{ij}]\in\H^{N\times M}$ is an $N\times M$ quaternionic matrix, and $X^\dagger=[q_{ij}^*]^\transpose$ is said to be its quaternion dual. Moreover, if $N=M$ and either $X^\dagger=X$ or $X^\dagger=X^{-1}$, then we say that $X$ is either quaternion self-dual or quaternionic unitary, respectively.
\end{definition}

We can now return to the symmetries mentioned earlier. Let us consider matrices with the structure~\eqref{decomp:2by2:H-structure}, if we choose a basis
\begin{equation}
\one=\begin{bmatrix} +1 & 0 \\ 0 & +1 \end{bmatrix},\quad
\qi=\begin{bmatrix} +i & 0 \\ 0 & -i \end{bmatrix},\quad
\qj=\begin{bmatrix} 0 & +1 \\ -1 & 0 \end{bmatrix},\quad
\qk=\begin{bmatrix} 0 & +i \\ +i & 0 \end{bmatrix},
\end{equation}
then we see that the space of such matrices are (ring-)isomorphic to the skew-field of quaternions (note that $\qi^2=\qj^2=\qk^2=\qi\qj\qk=-\one$). The anti-unitary symmetry given by~\eqref{decomp:2by2:antiU} and~\eqref{decomp:2by2:J} becomes
\begin{equation}
JX^*=XJ\qquad\text{with}\qquad J=\diag(\qj,\ldots,\qj),\quad X\in\H^{N\times N}\subset\C^{2N\times2N}.
\end{equation}
Here, the equality may be thought of as a consequence of corollary~\ref{thm:decomp:H-relations} (\ref{thm:decomp:H-relations(ii)}). Also other concepts can be carried over to the matrix representation without too much effort, e.g. the groups of invertible and unitary quaternionic matrices become
\begin{equation}
\gGL(N,\H)\cong\gU^*(2N)\qquad \text{and}\qquad
\gU(N,\H)\cong\gUSp(2N),
\end{equation}
respectively. The former is the non-compact group related to the unitary group through Weyl's unitary trick, while the latter is the unitary symplectic group.

Henceforth, quaternion will always refer to the matrix representation. Thus, when we refer to the eigenvalues of an $N\times N$ quaternion matrix, then we mean the $2N$ eigenvalues (counted with multiplicity) of the corresponding $2N\times 2N$ complex matrix; and equivalently for singular values, determinants, traces, etc.
Note that we have chosen notation for the quaternions such that it matches our notation for matrices, i.e. the quaternion dual ``$\dagger$'' corresponds to the Hermitian conjugation, while ``$\ast$'' corresponds to complex conjugation. It follows that a quaternion self-dual matrix is Hermitian and a quaternionic unitary matrix is also unitary in the usual sense.

Keeping the notation from the $2\times 2$ matrix representation~\eqref{decomp:2by2:H-structure} it is trivially seen that the eigenvalues of a quaternion are
\begin{equation}
\lambda_\pm=\Re u\pm i\sqrt{(\Im u)^2+\abs{v}^2}.
\label{decomp:2by2:eigenvalue-H}
\end{equation}
Thus the eigenvalues of a quaternion are a complex conjugate pair with non-zero imaginary part unless it is proportional to the identity in which case there is one real double degenerate eigenvalue. We see that~\eqref{decomp:2by2:eigenvalue-H} is invariant under a transformation $v\mapsto ve^{i\theta}$, $\theta\in[0,2\pi)$, and it is easily verified that any quaternion~\eqref{decomp:2by2:H-structure} can be diagonalised by a unitary symplectic similarity transformation, $q\mapsto U^{-1}qU$, with $U\in\gUSp(2)/\gU(1)$.

The benefit of introducing the quaternion algebra is that many results which hold for vector spaces over the fields of real and complex numbers can be extended to include the skew-field of quaternions (essentially using the same proofs). Thus, these results hold for the matrices introduced above (at least up to $2\times2$ blocks). 

Finally, let us extend the concept of a standard Gaussian random variable to the quaternion case (we emphasise this, since not all authors use the same definition):

\begin{definition}
A standard Gaussian quaternionic random variable is defined by the probability density
\begin{equation}
p(q)=\Big(\frac{2}{\pi}\Big)^2e^{-\tr q^\dagger q},\qquad q=\begin{bmatrix}u&v\\ -v^*&u*\end{bmatrix}\in\H\subset\C^{2\times 2}
\label{decomp:2by2:gauss-H}
\end{equation}
with respect to the flat measure, $d^4q=d^2ud^2v$, on $\H$.
\end{definition}

\begin{remark}
As usual, \emph{standard} refers to the fact that the constants are chosen such that
\begin{equation}
\E[q]=\int_\H d^4q\, p(q)q=0 \qquad\text{and}\qquad
\E[q^\dagger q]=\int_\H d^4q\, p(q)q^\dagger q=\one.
\end{equation}
This seems like a natural extension from the scalars fields to the skew-field of quaternions and we will use this definition consistently throughout the thesis. However, it should be noted that many authors prefer to work with an exponential that differs from ours by a factor of two.
\end{remark}

\subsection{Real matrices and split-quaternions}
\label{sec:decomp:real}

The ring-isomorphism between the space of $2\times2$ real matrices and the algebra of split-quaternions is rarely mentioned in the random matrix literature. Nonetheless, it might be useful to keep this algebra in mind when considering real asymmetric matrices, where $2\times2$ blocks play an important r\^ole. Let us first recall the structure of the split-quaternion algebra.

\begin{definition}
The ring of split-quaternions, $\H_s$, is a four dimensional vector space $\R^4$ with basis $\{\one,\qi,\qj,\qk\}$ equipped with associative multiplication defined such that $\one$ is the identity and $-\qi^2=\qj^2=\qk^2=\qi\qj\qk=\one$.
\end{definition}

\begin{remark}
The split-quaternions constitute an associative algebra but (unlike the quaternions) not a division algebra.
\end{remark}

Let us return to real matrices. If our real matrix is even dimensional then we might divide it up into $2\times 2$ blocks and by choosing a basis (this is the same basis as in~\cite{Edelman:1997}),
\begin{equation}
\one=\begin{bmatrix} 1 & 0 \\ 0 & 1 \end{bmatrix},\quad
\qi=\begin{bmatrix} 0 & 1 \\ -1 & 0 \end{bmatrix},\quad
\qj=\begin{bmatrix} 1 & 0 \\ 0 & -1 \end{bmatrix},\quad
\qk=\begin{bmatrix} 0 & 1 \\ 1 & 0 \end{bmatrix},
\label{decomp:2by2:Hs-basis}
\end{equation}
we see that the $2\times 2$ blocks are ring-isomorphic to the split-quaternions\footnote{The split-quaternions have another commonly chosen representation as $2\times 2$ matrices,
\begin{equation*}
\one=\begin{bmatrix} 1 & 0 \\ 0 & 1 \end{bmatrix},\quad
\qi=\begin{bmatrix} i & 0 \\ 0 & -i \end{bmatrix},\quad
\qj=\begin{bmatrix} 0 & i \\ -i & 0 \end{bmatrix},\quad
\qk=\begin{bmatrix} 0 & 1 \\ 1 & 0 \end{bmatrix}.
\end{equation*}
This basis is natural if we consider matrices satisfying the anti-unitary symmetry
\begin{equation*}
JX^*=XJ\qquad\text{with}\qquad J=\diag(\qj,\ldots,\qj),\quad X\in\H_s^{N\times N}\subset\C^{2N\times2N},
\end{equation*}
which is equivalent to the case of real matrices within the classification scheme of~\cite{BLC:2002,Magnea:2008}.
}. If our real matrix has odd dimension then there will remain an unpaired row and column after dividing the matrix into $2\times 2$ blocks. However, this does not cause any major difficulties, since the block structure is most useful when considering complex eigenvalues and any odd dimensional real matrix has at least one real eigenvalue.

It is useful to study the eigenvalues of a $2\times 2$ block in greater detail. Due to the basis~\eqref{decomp:2by2:Hs-basis}, it is convenient to use a parametrisation,
\begin{equation}
\R^{2\times2}\ni X=\frac{a\one+b\qi+c\qj+d\qk}{\sqrt{2}}=
\frac1{\sqrt{2}}
\begin{bmatrix}
a+d & c+b \\ c-b & a-d
\end{bmatrix},
\qquad a,b,c,d\in\R.
\end{equation}
It follows immediately from the characteristic equation that the eigenvalues of $X$ are real if and only if $c^2+d^2\geq b^2$; otherwise the eigenvalues will be a complex conjugate pair. It is fruitful to consider the same problem from a different point of view. Let us look at a real orthogonal (unit determinant) similarity transformation,
\begin{equation}
X'=UXU^\transpose,\qquad
U=\begin{bmatrix}\cos\theta/2 & \sin\theta/2 \\ -\sin\theta/2 & \cos\theta/2\end{bmatrix}\in\gSO(2).
\label{decomp:2by2:ortho-R-matrix}
\end{equation}
Using our split-quaternionic basis, this may be written as
\begin{equation}
\begin{bmatrix} a' \\ b' \\ c' \\ d' \end{bmatrix}
=\begin{bmatrix}
 1 & 0 & 0 & 0 \\
 0 & 1 & 0 & 0 \\
 0 & 0 & \cos\theta & -\sin\theta \\
 0 & 0 & \sin\theta & \cos\theta
\end{bmatrix}
\begin{bmatrix} a \\ b \\ c \\ d \end{bmatrix}.
\label{decomp:2by2:ortho-R}
\end{equation}
If $X$ has real eigenvalues then $c^2+d^2\geq b^2$ and it follows from~\eqref{decomp:2by2:ortho-R} that we can choose the rotation $\theta$ so that $b'=c'$; in this case $X'$ becomes upper-triangular and the eigenvalues are simply given as the diagonal entries. If $X$ has non-real eigenvalues then $c^2+d^2<b^2$, hence there exists no rotation which brings $X$ to an upper-triangular form and therefore no obvious choice for fixing $\theta$. Nonetheless, two different choices appear frequently in the literature: $\theta$ is chosen such that either $c'=0$ or $d'=0$. 
However, neither of these choices, nor any other value of $\theta\in\R$, is such that the remaining algebraic structure is closed under multiplication, which suggests that new ideas are needed when considering \emph{products} of real random matrices with complex eigenvalues.

\section{Matrix decompositions}
\label{sec:decomp:decomp}

\subsection{Standard decompositions}
\label{sec:decomp:standard}

In this section, we recall the structure of some standard matrix decompositions and their corresponding Jacobians. Most results in this section are well-known and will therefore be presented without proof. The reader is referred to~\cite{GL:1996} and references therein for a comprehensive description of different matrix factorisations (and related algorithms). For derivation of the corresponding Jacobians we refer to~\cite{Mehta:2004,Forrester:2010,Muirhead:2009,Mathai:1997,Olkin:2002}.

In the following $\F_{\beta=1,2,4}=\R,\C,\H$ denotes the real numbers, the complex numbers, or the quaternions, while $\gU(N,\F_{\beta=1,2,4})=\gO(N),\gU(N),\gUSp(2N)$ denotes the unitary group over the corresponding (skew-)field.

\begin{proposition}[Spectral decomposition]\label{thm:decomp:spectral}
Let $H$ be a Hermitian (real symmetric, quaternionic self-dual) $N\times N$ matrix, then there exists a unitary (orthogonal, unitary symplectic) matrix $U$ and a real diagonal matrix $\Lambda$ such that
\begin{equation}
H=U\Lambda U^{-1}.
\end{equation}
The diagonal entries of $\Lambda$ are referred to as the eigenvalues and they are (up to an ordering) uniquely determined by $H$. Furthermore, the decomposition becomes unique if we fix the phases of the first column in $U$ and the ordering of the diagonal elements of $\Lambda$. The Jacobian of the decomposition reads
\begin{equation}
d^\beta H=\prod_{1\leq i<j\leq N}\abs{\lambda_j-\lambda_i}^\beta \prod_{k=1}^Nd\lambda_kd\mu(U),
\end{equation} 
where $\lambda_i$ denotes the $i$-th eigenvalue (double degenerate for $\beta=4$) and $d\mu(U)=[U^{-1}dU]$ is the Haar measure on $\gU(N,\F_\beta)/\gU(1,\F_\beta)^N$.
\end{proposition}

\begin{remark}
Before we go on, it is worthwhile to pause for a minute to see how the non-uniqueness of the spectral decomposition manifests itself in the unitary transformation and why this is cured by fixing phases in the unitary matrix and ordering the eigenvalues. Let $H$ be as in proposition~\ref{thm:decomp:spectral} and $U\in\gU(N,\F_\beta)$ be a unitary matrix which diagonalises $H$. If we introduce another unitary matrix, $V=\diag(v_1,\ldots,v_N)$ with $v_1,\ldots, v_N\in\gU(1,\F_\beta)$, then it is evident that
\begin{equation}
H=U\Lambda U^{-1}=UV\Lambda (UV)^{-1}.
\end{equation}
This invariance may be removed by inserting constraints in the first column of $U$. Next, consider the almost diagonal matrix
\begin{equation}
P_i=\diag(1,\ldots,1,\big[\begin{smallmatrix} 0 & 1 \\ 1 & 0 \end{smallmatrix}\big],1,\ldots,1)\in\gU(N,\F_\beta),
\end{equation}
which contains a $2\times 2$ block on the $i$-th position. We see that $\Lambda$ and $P_i\Lambda P_i^{-1}$ are identical except for an interchange of two eigenvalues. Moreover, $P_i$ ($i=1,\ldots,N-1$) are the generators of the permutation group of the eigenvalues. All such matrices are indistinguishable in the decomposition unless we introduce an ordering of the eigenvalues, e.g. $\lambda_1\leq \cdots\leq\lambda_N$ (the eigenvalues are, of course, double degenerate for $\beta=4$).
\end{remark}

\begin{proposition}[Singular value decomposition]\label{thm:decomp:SVD}
Given a complex (real, quaternionic) $(N+\nu)\times N$ matrix, $X$, then there exist two unitary (orthogonal, unitary symplectic) matrices, $U$ and $V$, and an $(N+\nu)\times N$ diagonal matrix $\Sigma$ with non-negative entries, such that
\begin{equation}
X=U\Sigma V^{-1}.
\end{equation}
We refer to the diagonal entries of $\Sigma$ as the singular values of $X$, and they are (up to an ordering) uniquely determined by $X$. If we fix the phases of the first $N$ entries in the first column of $U$ as well as a $\nu\times\nu$ matrix subgroup of $U$, then the decomposition is unique up to a ordering of the singular values. The Jacobian reads
\begin{equation}
d^\beta X=\prod_{1\leq i<j\leq N}\abs{\sigma_j^2-\sigma_i^2}^\beta \prod_{k=1}^N\sigma_k^{\beta(\nu+1)-1}d\sigma_kd\mu(U)d\mu(V),
\end{equation} 
with $\sigma_i$ denoting the $i$-th singular value (double degenerate for $\beta=4$) and $d\mu(U)=[U^{-1}dU]$ and $d\mu(V)=[V^{-1}dV]$ are the Haar measures on $\gU(N+\nu,\F_\beta)/\gU(1,\F_\beta)^N\times\gU(\nu,\F_\beta)$ and $\gU(N,\F_\beta)$, respectively.
\end{proposition}

\begin{proposition}[QR decomposition]\label{thm:decomp:QR}
Let $\widetilde X$ be a complex (real, quaternionic) $(N+\nu)\times N$ matrix, then there exists a decomposition,
\begin{equation}
\widetilde X=\widetilde U\begin{bmatrix} R+T \\ 0\end{bmatrix},
\end{equation}
where $\widetilde U$ is a unitary (orthogonal, unitary symplectic) matrix, $T$ is a complex (real, quaternionic) $N\times N$ strictly upper-triangular matrix, and $R$ is a positive semi-definite diagonal matrix. If $\widetilde X$ has full rank, then the decomposition can be made unique by fixing a $\nu\times\nu$ matrix subgroup of $\widetilde U$. The corresponding Jacobian is given by
\begin{equation}
d^\beta\widetilde X=\prod_{i=1}^N r_i^{\beta(N+\nu-i+1)-1}dr_id^\beta Td\mu(\widetilde U),
\end{equation}
where $r_i$ are the (positive) diagonal entries of $R$ (double degenerate for $\beta=4$) and $d\mu(\widetilde U)=[\widetilde U^{-1}d\widetilde U]$ is the Haar measure on $\gU(N+\nu,\F_\beta)/\gU(\nu,\F_\beta)$.
\end{proposition}

\begin{remark}
There is no permutation invariance of diagonal entries of $R$, such transformations require pivoting, see~\cite{GL:1996}.
\end{remark}

\begin{remark}
As we will see in section~\ref{sec:decomp:generalised}, the QR decomposition is essential for the study of products of random matrices. In addition to this, the QR decomposition also plays a prominent r\^ole in numerical linear algebra as a part of the QR algorithm which is used to obtain the Schur form (and therefore also the eigenvalues) of a matrix, see~\cite{GL:1996}. Here, input from random matrix theory is relevant as well; Edelman et al.~\cite{ESW:2014} comment: ``Notice that in earlier versions of \textsc{lapack} and \textsc{matlab} \verb![Q,R]=qr(randn(n))! did not always yield \verb!Q! with Haar measure.
Random matrix theory provided the impetus to fix this!'' 
%
\end{remark}

\begin{proposition}[Block-QR decomposition]\label{thm:decomp:QR-block}
Let $N$ be a positive integer and let $\nu,\mu$ be non-negative integers. Given an $(N+\nu)\times (N+\mu)$ complex (real, quaternionic) matrix, $\widetilde X$, then there exist complex (real, quaternionic) matrices $X$ and $T$ of size $N\times N$ and $(N+\nu)\times \mu$, respectively, as well as a unitary (orthogonal, unitary symplectic) matrix $U$, such that
\begin{equation}
\widetilde X=\widetilde U\bigg[\begin{matrix} X \\ 0 \end{matrix}\, \bigg\vert\, T\,\bigg]
\end{equation}
The decomposition becomes unique if the matrix $\widetilde X$ has full rank and we fix a diagonal sub-group $\gU(N,\F_\beta)\times\gU(\nu,\F_\beta)$. The Jacobian is
\begin{equation}
d^\beta\widetilde X=\det(X^\dagger X)^{\beta\nu/2\gamma}d^\beta Xd^\beta Td\mu(U),
\end{equation}
with $d\mu(\widetilde U)=[\widetilde U^{-1}d\widetilde U]$ denoting the Haar measure on $\gU(N+\nu,\F_\beta)/\gU(N,\F_\beta)\times\gU(\nu,\F_\beta)$.
\end{proposition}

Above we have considered only real-valued spectral properties. Now, let us turn to eigenvalues of non-Hermitian matrices which may be complex. One way to approach this is by means of an eigenvalue decomposition as in~\cite{Ginibre:1965}, see also~\cite{Mehta:2004}. However, we will use a more modern approach, namely the Schur decomposition (which is also valuable for numerical computation due to the QR algorithm). This decomposition is more involved than the decompositions given above and we will state the complex, quaternionic, and real case separately. 

\begin{proposition}[Complex Schur decomposition]\label{thm:decomp:schur-C}
Let $X$ be a complex $N\times N$ matrix; there exist a unitary matrix $U$, a diagonal matrix $Z$, and strictly upper-triangular complex matrix $T$, such that
\begin{equation}
X=U(Z+T)U^{-1}.
\end{equation}
The diagonal entries of $Z$ are the eigenvalues of $X$ and they are (up to an ordering) uniquely determined by $X$. If we fix the phases in the first column of $U$, then the decomposition becomes unique up to the order of the eigenvalues. The Jacobian reads
\begin{equation}
d^2 X=\prod_{1\leq i<j\leq N}\abs{z_j-z_i}^2 \prod_{k=1}^Nd^2z_kd^2Td\mu(U),
\end{equation} 
where $z_i$ denote the $i$-th eigenvalue and $d\mu(U)=[U^{-1}dU]$ is the Haar measure on $\gU(N)/\gU(1)^N$.
\end{proposition}

\begin{proposition}[Quaternionic Schur decomposition]\label{thm:decomp:schur-H}
Given a quaternionic $N\times N$ matrix, $X$, then there exist a unitary symplectic matrix $U$, a diagonal matrix $Z$ consisting of complex conjugate pairs, and strictly upper-triangular quaternionic matrix $T$, such that
\begin{equation}
X=U(Z+T)U^{-1}.
\end{equation}
Here, the diagonal entries of $Z$ are the eigenvalues of $X$ and they come in complex conjugate pairs. If the eigenvalues are non-degenerate, then the decomposition becomes unique by fixing their order as well as the phases in the first column of $U$. The corresponding Jacobian is
\begin{equation}
d^4 X=\prod_{1\leq i<j\leq N}\abs{z_j-z_i}^2 \abs{z_j-z_i^*}^2 \prod_{k=1}^N\abs{z_k-z_k^*}^2d^2z_kd^4Td\mu(U),
\end{equation} 
where $z_i$ denotes an eigenvalue in the upper-half plane stemming from the $i$-th eigenvalue pair and $d\mu(U)=[U^{-1}dU]$ is the Haar measure on $\gUSp(2N)/\gU(1)^N$.
\end{proposition}

\begin{proposition}[Real Schur decomposition]\label{thm:decomp:schur-R}
Let $K$ and $L$ be integers such that $K+2L=N$. Given $N\times N$ real asymmetric matrix, $X$, with $K$ real eigenvalues (and therefore $L$ complex conjugate pairs), then there exists a factorisation in terms of real matrices,
\begin{equation}
X=U\begin{bmatrix}\Lambda + T^{11} & T^{12} \\ 0 & Z + T^{22} \end{bmatrix}U^{-1}.
\end{equation}
Here $T^{11}$ is a $K\times K$ strictly upper-triangular matrix, $T^{12}$ is a $K\times 2L$ matrix, and $T^{12}$ is a $2L\times 2L$ matrix which is strictly upper-triangular in terms of $2\times 2$ blocks; $U$ is an orthogonal matrix, $\Lambda=\diag(\lambda_1,\ldots,\lambda_K)$ is a $K\times K$ diagonal matrix where the entries are (unordered) real eigenvalues, while $Z$ is a $2L\times 2L$ block diagonal matrix where the $j$-th entry given by
\begin{equation}
(Z)_{jj}=\begin{bmatrix} x_j & c_j+b_j \\ c_j-b_j & x_j \end{bmatrix},
\end{equation}
with $b_j^2>c_j^2$, such that $z_j=x_j+ i\sqrt{b_j^2-c_j^2}$ is a complex complex eigenvalue of $X$. The corresponding Jacobian is given by
\begin{multline}
dX=\prod_{1\leq i<j\leq K}\abs{\lambda_j-\lambda_i}
\prod_{1\leq i<j\leq L}\abs{z_j-z_i}^2\abs{z_j-z_i^*}^2
\prod_{i=1}^K\prod_{j=1}^L\abs{z_j-\lambda_i}^2\\
\times dT^{11}dT^{12}dT^{22}d\mu(U)\prod_{i=1}^Kd\lambda_i\prod_{j=1}^L dx_j\,db_j\,c_jdc_j,
\end{multline}
where $d\mu(U)$ is the Haar measure on $\gO(N)/\gO(1)^N$.
\end{proposition}

\begin{remark}
The Jacobian given in proposition~\ref{thm:decomp:schur-R} is closely related to the Jacobian presented in~\cite{Edelman:1997}, but other parametrisations for the $2\times 2$ block matrices are possible, see e.g.~\cite{LS:1991}.
\end{remark}

\subsection{Generalised decompositions}
\label{sec:decomp:generalised}

In this section, we present some generalisations of the decompositions from the previous section. These decompositions are closely related to products of random matrices and they have only recently appeared in the literature.

\begin{proposition}[Generalised block-QR decomposition]\label{thm:decomp:gen-QR-block}
Let each $\widetilde X_i$ ($i=1,\ldots,n$) be an $(N+\nu_i)\times (N+\nu_{i-1})$ complex (real, quaternionic) matrix, where $N$ is a positive integer and $\nu_0,\nu_1,\ldots,\nu_n$ are non-negative integers. Then we have the decomposition,
\begin{equation}
\widetilde X_i=\widetilde U_i\bigg[\begin{matrix} X_i \\ 0 \end{matrix}\, \bigg\vert\, T_i\,\bigg]\widetilde U_{i-1}^{-1},\qquad
i=1,\ldots,n,
\label{decomp:decomp:gen-QR-block}
\end{equation}
where $\widetilde U_i$ are unitary (orthogonal, unitary symplectic) matrices, $X_i$ are $N\times N$ complex (real, quaternionic) matrices, and $T_i$ are $(N+\nu_i)\times \nu_{i-1}$ complex (real, quaternionic) matrices for $i=1,\ldots,n$, while $\widetilde U_0$ is the $(N+\nu_0)\times(N+\nu_0)$ identity matrix. Furthermore, if each individual matrix has full rank and we fix a diagonal sub-group $\gU(N,\F_\beta)\times\gU(\nu_i,\F_\beta)$ for each $\widetilde U_i$, then the decomposition becomes unique and the corresponding Jacobian is
\begin{equation}
\prod_{i=1}^nd^\beta \widetilde X_i=\prod_{i=1}^n\det(X_i^\dagger X_i)^{2\beta\nu_i/2\gamma}d^\beta  X_id^\beta T_id\mu(\widetilde U_i),
\label{decomp:decomp:gen-QR-block-jacob}
\end{equation}
where $d\mu(\widetilde U_i)=[\widetilde U_i^{-1}d\widetilde U_i]$ is the Haar measure on $\gU(N+\nu_i,\F_\beta)/\gU(N,\F_\beta)\times\gU(\nu_i,\F_\beta)$ and $\gamma=1,1,2$ for $\beta=1,2,4$.
\end{proposition}

\begin{proof}
We first notice that the $n=1$ case is nothing but the `ordinary' block-QR decomposition (proposition~\ref{thm:decomp:QR-block}). Next, let us assume that~\eqref{decomp:decomp:gen-QR-block} holds for $n-1$. Using this we may write
\begin{equation}
\widetilde X_n\cdots \widetilde X_1=\widetilde X_n\widetilde U_{n-1}
\bigg[\begin{matrix} X_{n-1} \\ 0 \end{matrix}\, \bigg\vert\, T_{n-1}\,\bigg]\cdots\bigg[\begin{matrix} X_1 \\ 0 \end{matrix}\, \bigg\vert\, T_1\,\bigg]
\end{equation}
with matrices defined as in the proposition. Here, we may use the block-QR decomposition of the matrix $\widetilde X_{n}\widetilde U_{n-1}$ which yields
\begin{equation}
\widetilde X_{n}\widetilde U_{n-1}=\widetilde U_{n}\bigg[\begin{matrix} X_n \\ 0 \end{matrix}\, \bigg\vert\, T_n\,\bigg].
\end{equation}
Thus, the complete decomposition follows by induction. Clearly, the generalised block-QR decomposition becomes unique, if make each of $n$ individual block-QR decompositions unique.

Now, let us turn to the Jacobian~\eqref{decomp:decomp:gen-QR-block-jacob}. We assume that~\eqref{decomp:decomp:gen-QR-block-jacob} holds for $n=k-1$ and let $[d\widetilde X_k]$ denote the matrix of one-forms. We have
\begin{equation}
\widetilde U_k^{-1}[d\widetilde X_k]\widetilde U_{k-1}
=[\widetilde U_k^{-1}d\widetilde U_k]\bigg[\begin{matrix} X_k \\ 0 \end{matrix}\, \bigg\vert\, T_k\,\bigg]-\bigg[\begin{matrix} X_k \\ 0 \end{matrix}\, \bigg\vert\, T_k\,\bigg][\widetilde U_{k-1}^{-1}d\widetilde U_{k-1}]
+\bigg[\begin{matrix} [dX_k] \\ 0 \end{matrix}\, \bigg\vert\, [dT_k]\,\bigg].
\label{decomp:decomp:gen-QR-block-proof-1}
\end{equation}
First, we notice that we may ignore the second term since all independent one-forms from $[\widetilde U_{k-1}^{-1}d\widetilde U_{k-1}]$ (by assumption) already have appeared in the wedge product
\begin{equation}
\bigwedge_{1\leq i < k} \widetilde U_i^{-1}d\widetilde X_i\widetilde U_{i-1}.
\end{equation}
Second, since $\widetilde U_{k}\in\gU(N+\nu_k,\F_\beta)/\gU(N,\F_\beta)\times\gU(\nu_k,\F_\beta)$ it follows that
\begin{equation}
[\widetilde U_k^{-1}d\widetilde U_k]=\begin{bmatrix} 0 & -[dA]^\dagger \\ [dA] & 0 \end{bmatrix},
\end{equation}
where $[dA]$ is $\nu_k\times N$ matrix of independent one-forms. Using these two observations, we write~\eqref{decomp:decomp:gen-QR-block-proof-1} as
\begin{equation}
\widetilde U_k^{-1}[d\widetilde X_k]\widetilde U_{k-1}\simeq
\bigg[\begin{matrix} [dX_k] \\ [dA]X_k \end{matrix}\, \bigg\vert\, [dT_k]+[\widetilde U_k^{-1}d\widetilde U_k]T_k\,\bigg].
\end{equation}
Here, `$\simeq$' means equal up to terms depending on $[\widetilde U_{k-1}^{-1}d\widetilde U_{k-1}]$ which do not contribute to the Jacobian. Writing out the wedge product, we have
\begin{equation}
\bigwedge\widetilde U_k^{-1}[d\widetilde X_k]\widetilde U_{k-1}=
\bigwedge[dX_k]\bigwedge[dA]X_k\bigwedge[dT_k]
\label{decomp:decomp:gen-QR-block-proof-2}
\end{equation}
and the well-known identity 
\begin{equation}
\bigwedge_{\substack{a=1,\ldots,\nu_n\\ b=1,\ldots,N}}([dA]X_k)_{ab}=\det(X_k^\dagger X_k)^{2\beta\nu_k/2\gamma}\bigwedge_{\substack{a=1,\ldots,\nu_n\\ b=1,\ldots,N}}[dA]_{ab}.
\label{decomp:decomp:gen-QR-block-proof-3}
\end{equation}
Combining~\eqref{decomp:decomp:gen-QR-block-proof-3} with~\eqref{decomp:decomp:gen-QR-block-proof-2}, we see that if the Jacobian~\eqref{decomp:decomp:gen-QR-block-jacob} is valid for $n=k-1$ then it holds for $n=k$ as well. Since the $n=1$ case is the familiar block-QR decomposition, the Jacobian follows by induction.
\end{proof}

\begin{remark}
Note that if we have $\nu_0=0$ in proposition~\ref{thm:decomp:gen-QR-block}, then it follows that
\begin{equation}
\widetilde X_n\cdots\widetilde X_1=
\widetilde U_n\bigg[\begin{matrix} X_n \\ 0 \end{matrix}\, \bigg\vert\, T_n\,\bigg]\cdots\bigg[\begin{matrix} X_2 \\ 0 \end{matrix}\, \bigg\vert\, T_2\,\bigg]\begin{bmatrix} X_1 \\ 0\end{bmatrix}
=\widetilde U_n\begin{bmatrix} X_n\cdots X_1 \\ 0\end{bmatrix}.
\end{equation}
This is an important property in section~\ref{sec:prologue:rectangular}.
\end{remark}

\begin{corollary}\label{thm:decomp:gen-QR-block-2}
Let $\widetilde X_i$ ($i=1,\ldots,n$) be as in proposition~\ref{thm:decomp:gen-QR-block}. Then we have the decomposition,
\begin{equation}
\widetilde X_i=\widetilde U_i\bigg[\begin{array}{@{}c @{}} \begin{matrix} X_{i} & 0 \end{matrix} \\ \hline T_i \end{array}\bigg]\widetilde U_{i-1}^{-1},\qquad
i=1,\ldots,n,
\end{equation}
where $U_{i-1}$ are unitary (orthogonal, unitary symplectic) matrices, $X_i$ are $N\times N$ complex (real, quaternionic) matrices, and $T_i$ are $\nu_i\times (N+\nu_{i-1})$ complex (real, quaternionic) matrices for $i=1,\ldots,n$, while $\widetilde U_n$ is the $(N+\nu_n)\times(N+\nu_n)$ identity matrix. Fixing a diagonal sub-group $\gU(N,\F_\beta)\times\gU(\nu_{i-1},\F_\beta)$ for each $\widetilde U_i$ makes the decomposition unique assuming that the individual matrices has full rank. The corresponding Jacobian reads
\begin{equation}
\prod_{i=1}^nd^\beta \widetilde X_i=\prod_{i=1}^n\det(X_i^\dagger X_i)^{2\beta\nu_{i-1}/2\gamma}d^\beta  X_id^\beta T_id\mu(\widetilde U_{i-1}),
\end{equation}
where $d\mu(\widetilde U_i)=[\widetilde U_i^{-1}d\widetilde U_i]$ is the Haar measure on $\gU(N+\nu_i,\F_\beta)/\gU(N,\F_\beta)\times\gU(N,\nu_i)$ and $\gamma=1,1,2$ for $\beta=1,2,4$.
\end{corollary}

\begin{proof}
Follows from proposition~\ref{thm:decomp:gen-QR-block} by taking the Hermitian conjugate of~\eqref{decomp:decomp:gen-QR-block}.
\end{proof}

\begin{proposition}[Generalised complex Schur decomposition]\label{thm:decomp:gen-schur-C}
Let $X_i$ ($i=1,\ldots, n$) be a family of non-singular $N\times N$ complex matrices, then there exists a family of unitary, diagonal, and strictly upper-triangular matrices, $U_i$, $Z_i$ and $T_i$ ($i=1,\ldots, n$), such that
\begin{equation}
X_i=U_i(Z_i+T_i)U_{i-1}^{-1},\qquad i=1,\ldots,n\qquad U_0=U_n.
\label{decomp:decomp:gen-schur-C}
\end{equation}
The decomposition becomes unique (up to an ordering) if we fix the phases of the first column in each $U_i$ ($i=1,\ldots, n$). The corresponding Jacobian reads
\begin{equation}
\prod_{k=1}^nd^2X_k=\prod_{1\leq i<j\leq N}\abs{z_j-z_i}^2 \prod_{k=1}^nd^2T_kd\mu(U_k)\prod_{\ell=1}^Nd^2x_{k,\ell},
\end{equation} 
where $x_{k,\ell}\in\C$ is the $\ell$-th diagonal entry in $Z_k$ and $z_\ell:=x_{1,\ell}\cdots x_{n,\ell}$ is the $\ell$-th diagonal entry in $Z:=Z_n\cdots Z_1$ and, thus, is an eigenvalue of the product matrix $X_n\cdots X_1$. Finally, $d\mu(U_k)=[U_k^{-1}dU_k]$ denotes the Haar measure on $\gU(N)/\gU(1)^N$.
\end{proposition}

\begin{proof}
It follows from the complex Schur decomposition (proposition~\ref{thm:decomp:schur-C}), that there exists a unitary matrix $U_n$, such that
\begin{equation}
U_0^{-1}X_n\cdots X_1U_0=Z+T,
\label{decomp:decomp:proof-schur-C-product}
\end{equation}
where $Z$ is a diagonal matrix and $T$ is a strictly upper-triangular matrix. Since $X_1U_0$ is a complex matrix, we know from the QR decomposition (proposition~\ref{thm:decomp:QR}), that there exists a unitary matrix $U_1$ such that
\begin{equation}
X_1U_0=U_1(Z_1+T_1)
\end{equation}
with $Z_1$ and $T_1$ denoting a diagonal and a strictly upper-triangular matrix, respectively. Inserting this decomposition back into~\eqref{decomp:decomp:proof-schur-C-product} and repeating the same idea for $X_2,\ldots,X_{n-1}$ yields
\begin{gather}
X_i=U_i(Z_i+T_i)U_{i-1}^{-1},\qquad i=1,\ldots,n-1,\\
U_0^{-1}X_nU_{n-1}(Z_{n-1}+T_{n-1})\cdots(Z_1+T_1)=Z+T
\end{gather}
with $U_i$, $Z_i$ and $T_i$ denoting unitary, diagonal, and strictly upper-triangular matrices as before.
We recall that each $X_i$ is invertible and therefore so is $Z_i+T_i$, thus
\begin{equation}
U_0^{-1}X_nU_{n-1}=(Z+T)(Z_1+T_1)^{-1}\cdots(Z_{n-1}+T_{n-1})^{-1}.
\end{equation}
Here, the right hand side is an upper-triangular matrix and the decomposition~\eqref{decomp:decomp:proof-schur-C-product} follows with
\begin{equation}
Z_n:=ZZ_1^{-1}\cdots Z_{n-1}^{-1},\qquad
T_n:=(Z+T)(Z_1+T_1)^{-1}\cdots(Z_{n-1}+T_{n-1})^{-1}-Z_n.
\label{decomp:decomp:proof-schur-final}
\end{equation}
As seen, the generalised Schur decomposition is a combination of an `ordinary' Schur decomposition and $n-1$ QR decompositions, thus in order to make the generalised decomposition unique we need to make each of the individual decompositions unique. Up to an ordering of the eigenvalues of the product matrix, this can be done by fixing the phases in the first column in $U_i$ for every $i$.

The derivation of the Jacobian is given in~\cite{AB:2012}.
\end{proof}

\begin{remark}
We note that the generalised Schur decomposition for $n=2$ was previously known in the literature (sometimes called the QZ decomposition), see e.g.~\cite{GL:1996}. As we have seen, the extension to arbitrary $n$ can be done without too much effort; this was originally presented in~\cite{ARRS:2013,Strahov:2013} with inspiration from~\cite{AB:2012}. More surprising is the fact that the Jacobian for the generalised Schur decomposition has the same form as the `ordinary' Schur decomposition, see proposition~\ref{thm:decomp:schur-C}.  
\end{remark}

\begin{proposition}[Generalised quaternionic Schur decomposition]\label{thm:decomp:gen-schur-H}
Given a family of non-singular $N\times N$ quaternionic matrices, $X_i$ ($i=1,\ldots, n$), then there exists a family of unitary symplectic, diagonal, and strictly upper-triangular quaternionic matrices, $U_i$, $Z_i$ and $T_i$ ($i=1,\ldots, n$), such that
\begin{equation}
X_i=U_i(Z_i+T_i)U_{i-1}^{-1},\qquad i=1,\ldots,n\qquad U_0=U_n.
\label{decomp:decomp:gen-schur-H}
\end{equation}
Furthermore, if we order the eigenvalues of the product matrix $X_n\cdots X_1$ and fix the phases of the first column in each $U_i$ ($i=1,\ldots, n$) then the decomposition becomes unique and the Jacobian reads
\begin{equation}
\prod_{k=1}^nd^4X_k=\prod_{1\leq i<j\leq N}\abs{z_j-z_i}^2 \abs{z_j-z_i^*}^2 \prod_{k=1}^nd^4T_kd\mu(U_k)\prod_{\ell=1}^N\abs{z_\ell-z_\ell^*}^2d^2z_{k,\ell}.
\label{decomp:decomp:gen-schur-H-jacob}
\end{equation} 
Here $d\mu(U)=[U^{-1}dU]$ denotes the Haar measure on $\gUSp(2N)/\gU(1)^N$, while $z_{k,\ell}$ and $z_{k,\ell}^*$ are the $(2\ell-1)$-th and the $2\ell$-th diagonal entry in $Z_k$, respectively. Thus $z_\ell:=z_{1,\ell}\cdots z_{n,\ell}$ is an eigenvalue of the product matrix $X_n\cdots X_1$ which by fixing of the phases of the first column in $U_n$ is chosen to belong the upper-half plane.
\end{proposition}

\begin{proof}
The derivation~\eqref{decomp:decomp:proof-schur-C-product} to~\eqref{decomp:decomp:proof-schur-final} carries over from the complex to the quaternionic case by simple replacement of the field $\C$ with the skew-field $\H$; finally~\eqref{decomp:decomp:gen-schur-H} follows due to the isomorphism between the $2\times 2$ block structure and the quaternions, see~\cite{AIS:2014} for details.

In order to find the Jacobian we follow~\cite{Ipsen:2013}. Denoting the matrix of one-forms by $[dX_i]$, we have (recall that $U_0=U_n$)
\begin{equation}
U_i^{-1}[dX_i]U_{i-1}=[U_i^{-1}dU_i](Z_i+T_i)-(Z_i+T_i)[U_{i-1}^{-1}dU_{i-1}]+[dZ_i]+[dT_i],
\label{decomp:decomp:gen-schur-H-proof-1}
\end{equation}
for $i=1,\ldots,n$. Here, we have used the anti-Hermiticity, $[U_{i-1}dU_{i-1}^{-1}]=-[U_{i-1}^{-1}dU_{i-1}]$, which follows from the constraint $U_{i-1}^\dagger U_{i-1}=1$. We will consider the strictly lower-triangular, the diagonal, and the strictly upper-triangular part of~\eqref{decomp:decomp:gen-schur-H-proof-1} separately (interpreted  in terms of quaternions or equivalently $2\times 2$ block matrices). 

We start with the upper-triangular part, which is the simplest case. We have
\begin{equation}
\bigwedge_{k>\ell} (U_i^{-1}[dX_i]U_{i-1})_{k\ell}=\bigwedge_{k>\ell}[dT_i]_{k\ell}+\cdots
\end{equation}
but we notice that the one-forms $[dT_i]_{k\ell}$ (the `${k\ell}$' subscript refer to the $(k,\ell)$-th entry of the matrix) only appear above the diagonal in~\eqref{decomp:decomp:gen-schur-H-proof-1} and since they all must appear in the final exterior product it immediately follows that the contribution to the Jacobian is trivial,
\begin{equation}
\prod_{k=1}^nd^4T_k.
\label{decomp:decomp:gen-schur-H-proof-jacob-above}
\end{equation}
Henceforth, we will use `$\simeq$' to mean equal up to terms that can not contribute to the Jacobian.


Now, we can turn the lower-triangular part, which is significantly more difficult. We have
\begin{equation}
(U_i^{-1}[dX_i]U_{i-1})_{k\ell}=\sum_{m\geq \ell}\Big([U_i^{-1}dU_i]_{km}(Z_i+T_i)_{m\ell}-(Z_i+T_i)_{km}[U_{i-1}^{-1}dU_{i-1}]_{m\ell}\Big)
\end{equation}
with $k<\ell$. Here the sum starts at $m=\ell$ since $Z_i+T_i$ consists of only zeros below the diagonal. We would like to take the wedge product with respect to indices $k<\ell$. Since each $[U_i^{-1}dU_i]_{k\ell}$ can appear at most once in the final wedge product (and similarly for $[U_{i-1}^{-1}dU_{i-1}]_{k\ell}$), we have
\begin{equation}
\bigwedge_{k<\ell}(U_i^{-1}[dX_i]U_{i-1})_{k\ell}=\bigwedge_{k<\ell}\Big([U_i^{-1}dU_i]_{k\ell}(Z_i)_{\ell\ell}-(Z_i)_{kk}[U_{i-1}^{-1}dU_{i-1}]_{k\ell}\Big).
\end{equation}
The next step is easier if we introduce some new notation. Let `$\vec$' denote the vectorisation operator which transform a matrix into a vector by stacking its columns; and let $\otimes$ denote the Kronecker product (tensor product). Using the identity $\vec(ABC)=(C^\transpose\otimes A)\vec B$ (see e.g.~\cite{HJ:2012,vanLoan:2000,Edelman:2005}), we see that
\begin{multline}
\bigwedge_{k<\ell}\vec(U_i^{-1}[dX_i]U_{i-1})_{k\ell}=\\
\bigwedge_{k<\ell}\Big(((Z_i)_{\ell\ell}^\transpose\otimes\one)\vec[U_i^{-1}dU_i]_{k\ell}-(\one\otimes(Z_i)_{kk})\vec[U_{i-1}^{-1}dU_{i-1}]_{k\ell}\Big)
\end{multline}
with $\one$ denoting the $2\times 2$ identity matrix. 
We can now take the exterior product with respect to the index $i$, which yields
\begin{multline}
\bigwedge_{i=1}^n\bigwedge_{k<\ell}\vec(U_i^{-1}[dX_i]U_{i-1})_{k\ell}=\\
\prod_{k<\ell}\det\bigg[\prod_{i=1}^n((Z_i)_{\ell\ell}^\transpose\otimes\one)-\prod_{i=1}^n(\one\otimes(Z_i)_{kk})\bigg]\bigwedge_{i=1}^n\bigwedge_{k<\ell}\vec[U_{i}^{-1}dU_{i}]_{k\ell}.
\end{multline}
Let $Z=Z_n\cdots Z_1$, i.e. a (block) diagonal matrix with entries $(Z)_{\ell\ell}=(Z_n)_{\ell\ell}\cdots (Z_1)_{\ell\ell}$, then using identity $(A\otimes B)(C\otimes D)=(AC)\otimes(BD)$ we find
\begin{equation}
\bigwedge_{i=1}^n\bigwedge_{k<\ell}\vec(U_i^{-1}[dX_i]U_{i-1})_{k\ell}=
\prod_{k<\ell}\det\big[(Z)_{\ell\ell}^\transpose\otimes\one-\one\otimes(Z)_{kk}\big]\bigwedge_{i=1}^n\bigwedge_{k<\ell}\vec[U_{i}^{-1}dU_{i}]_{k\ell}.
\end{equation}
Thus, the contribution to the Jacobian~\eqref{decomp:decomp:gen-schur-H-jacob} from the terms below the diagonal in~\eqref{decomp:decomp:gen-schur-H-proof-1} is given by
\begin{equation}
\prod_{1\leq i<j\leq N}\abs*{\det\big[(Z)_{jj}^\transpose\otimes\one-\one\otimes(Z)_{ii}\big] }\prod_{k=1}^nd\mu(U_k),
\label{decomp:decomp:gen-schur-H-proof-jacob-below}
\end{equation} 
where $d\mu(U_k)$ is the Haar measure on $\gUSp(2N)/\gUSp(2)$. Note that the unitary integral differs from~\eqref{decomp:decomp:gen-schur-H-jacob}, since the contribution from $[U_i^{-1}dU_i]_{kk}$ has not been included yet. The determinantal prefactor is the $2\times 2$ block analogue of the Vandermonde determinant and due to the quaternionic structure, $(Z)_{ii}=\diag(z_i,z_i^*)$, we have
\begin{equation}
\prod_{1\leq i<j\leq N}\abs*{\det\big[(Z)_{jj}^\transpose\otimes\one-\one\otimes(Z)_{ii}\big] }
=\prod_{1\leq i<j\leq N} \abs{z_j-z_i}^2 \abs{z_j-z_i^*}^2.
\end{equation}
by evaluation of the determinant.

It remains only to investigate the diagonal terms in~\eqref{decomp:decomp:gen-schur-H-proof-1}, i.e. terms of the form
\begin{equation}
(U_i^{-1}[dX_i]U_{i-1})_{kk}=[U_i^{-1}dU_i]_{kk}(Z_i)_{kk}-(Z_i)_{kk}[U_{i-1}^{-1}dU_{i-1}]_{kk}+[dZ_i]_{kk}.
\label{decomp:decomp:gen-schur-H-proof-diagonal}
\end{equation}
Here, we will exploit the internal structure of the $2\times 2$ block matrices. To do so, we first recall that $(U_i)_{kk}\in\gUSp(2)/U(1)$ and $(Z_i)_{kk}=\diag(z_i,z_i^*)$, hence
\begin{equation}
[U_j^{-1}dU_j]_{kk}=\begin{bmatrix} 0 & -d\alpha_{j,k}^* \\ d\alpha_{j,k} & 0 \end{bmatrix}
\qquad\text{and}\qquad
[dZ_i]_{kk}=\begin{bmatrix} dz_{j,k} & 0 \\ 0 & dz_{j,k}^*  \end{bmatrix},
\end{equation}
where $d\alpha_{j,k}$ and $dz_{j,k}$ are complex one-forms (different for all $j=1,\ldots,n$ and $k=1,\ldots,N$). Writing~\eqref{decomp:decomp:gen-schur-H-proof-diagonal} in matrix form, we find
\begin{equation}
(U_i^{-1}[dX_i]U_{i-1})_{kk}=
\begin{bmatrix} dz_{j,k} & -(d\alpha_{j,k}^*z_{j,k}^*-z_{j,k}d\alpha_{j-1,k}^*) \\ (d\alpha_{j,k}z_{j,k}-z_{j,k}^*d\alpha_{j-1,k}) & dz_{j,k}^*  \end{bmatrix}.
\end{equation}
Now, we can to take the exterior product with respect to the indices $i$ and $k$, and we read off the contribution from the diagonal part of~\eqref{decomp:decomp:gen-schur-H-proof-1} to the Jacobian,
\begin{equation}
\prod_{k=1}^N \abs{z_k-z_k^*}^2d\mu(U_k)\prod_{j=1}^nd^2z_{j,k}
\label{decomp:decomp:gen-schur-H-proof-jacob-diag}
\end{equation}
with $z_k=z_{n,k}\cdots z_{1,k}$ as above and $d\mu(U_k)$ denoting the Haar measure on $\gUSp(2)/U(1)$.

Finally, we can combine~\eqref{decomp:decomp:gen-schur-H-proof-jacob-above}, \eqref{decomp:decomp:gen-schur-H-proof-jacob-below}, and~\eqref{decomp:decomp:gen-schur-H-proof-jacob-diag} which gives~\eqref{decomp:decomp:gen-schur-H-jacob} and completes the derivation of the Jacobian.
\end{proof}

\begin{proposition}[Generalised real Schur decomposition]\label{thm:decomp:gen-schur-R} 
Let $K$ and $L$ be integers such that $K+2L=N$. Consider a family $X_i$ ($i=1,\ldots, n$) of non-singular $N\times N$ real matrices, such that the product matrix $X_n\cdots X_1$ has at least $K$ real eigenvalues. We can decompose each matrix $X_i$ as
\begin{equation}
X_i=U_i\begin{bmatrix}\Lambda_i + T^{11}_i & T^{12}_i \\ 0 & Z_i + T^{22}_i \end{bmatrix}U_{i-1}^{-1},
\label{decomp:decomp:gen-schur-R}
\end{equation}
where $\Lambda_i=\diag(\lambda_{i,1},\ldots,\lambda_{i,K})$ is a $K\times K$ real diagonal matrix, $Z_i=\diag(Z_{i,1},\ldots,Z_{i,L})$ is a block diagonal matrix constructed from $L$ real $2\times2$ block matrices, $T^{11}_i$ is a $K\times K$ real upper-triangular matrix, $T^{12}_i$ is a $K\times 2L$ real matrix, $T^{22}_i$ is a $2L\times 2L$ real $2\times 2$ block upper-triangular matrix, while $U_i$ and $U_{i-1}$ are orthogonal matrices with the constraint $U_0=U_n$. The decomposition becomes unique up to reordering of the (block) diagonal elements of the product matrix if we fix the sign in the first $K$ column as well as a diagonal subgroup, $\gO(2)^L$, of each $U_i$.

Furthermore, $\lambda_j:=\lambda_{1,j}\cdots\lambda_{n,j}$ $(j=1,\ldots,K)$ are real eigenvalues of $X_n\cdots X_1$, while the remaining $2L$ eigenvalues may be either real or complex conjugate pairs; these are given as the eigenvalues of the $2\times2$ matrices defined by $(Z)_{jj}=Z_{1,j}\cdots Z_{n,j}$ ($j=1,\ldots,L$). The corresponding Jacobian reads 
\begin{multline}
\prod_{i=1}^ndX_i=\prod_{1\leq i<j\leq K}\abs{\lambda_j-\lambda_i}
\prod_{1\leq i<j\leq L}\abs{\det[(Z)_{jj}^\transpose\otimes\one-\one\otimes (Z)_{ii}]}
\prod_{i=1}^K\prod_{j=1}^L\abs{\det[(Z)_{jj}-\one\otimes \lambda_i]} \\
 \times \prod_{i=1}^ndT^{11}_idT^{12}_idT^{22}_id\mu(U_i)\prod_{j=1}^Kd\lambda_{i,j}\prod_{h=1}^L dZ_{i,h},
\label{decomp:decomp:gen-schur-R-jacob}
\end{multline} 
where $d\mu(U_i)$ denotes the Haar measure on $\gO(N)/\gO(1)^K\times\gO(2)^L$ and $dZ_{i,h}$ denotes the flat measure on $\R^{2\times 2}$. 
\end{proposition}

\begin{proof}
If $K=N$, then the derivation~\eqref{decomp:decomp:proof-schur-C-product} to~\eqref{decomp:decomp:proof-schur-final} carries over from the complex to the real case by simple replacement of the field $\C$ with the field $\R$. If the product matrix has less than $N$ real eigenvalues and, thus, $K<N$, then there is no orthogonal similarity transformation which brings the product matrix to upper-triangular form, cf. section~\ref{sec:decomp:real} and proposition~\ref{thm:decomp:schur-R}. However, if we replace a sufficient number of diagonal elements with $2\times 2$ blocks, then we can follow exactly the same step as in the derivation~\eqref{decomp:decomp:proof-schur-C-product} to~\eqref{decomp:decomp:proof-schur-final}, which gives us~\eqref{decomp:decomp:gen-schur-R}.

In order to find the Jacobian, we essentially follow the same steps as in the proof of proposition~\ref{thm:decomp:gen-schur-H}. Denoting the matrix of one-forms by $[dX_i]$, we have (recall that $U_0=U_n$)
\begin{multline}
U_i^{-1}[dX_i]U_{i-1}=[U_i^{-1}dU_i]\begin{bmatrix}\Lambda_i + T^{11}_i & T^{12}_i \\ 0 & Z_i + T^{22}_i \end{bmatrix}
-\begin{bmatrix}\Lambda_i + T^{11}_i & T^{12}_i \\ 0 & Z_i + T^{22}_i \end{bmatrix}[U_{i-1}^{-1}dU_{i-1}]\\
+\begin{bmatrix}[d\Lambda_i] + [dT^{11}_i] & [dT^{12}_i] \\ 0 & [dZ_i] + [dT^{22}_i] \end{bmatrix}.
\label{decomp:decomp:gen-schur-R-proof-1}
\end{multline}
We need to look at the exterior product $\bigwedge_{a,b}(U_i^{-1}[dX_i]U_{i-1})_{ab}$. For notational simplicity, it is convenient to let the $ab$-index denote $1\times 2$ block matrices for $a\leq K < b$,  $2\times 1$ block matrices for $b\leq K < a$, and $2\times 2$ block matrices for $K < a,b$, i.e. we use the following block structure
\begin{equation}
\left[
\begin{array}{c|c}
1\times1 & 1\times2 \\ \hline
2\times1 & 2\times2
\end{array}
\right].
\end{equation}
Since each of the one-forms must appear in the final exterior product, we immediately see that the contribution to the Jacobian from the (block) upper-triangular part is
\begin{align}
\bigwedge_{a\leq b\leq K} (U_i^{-1}[dX_i]U_{i-1})_{ab}&\simeq \bigwedge_{a< b\leq K}[dT^{11}_i]_{ab}\bigwedge_{c=1}^Kd\lambda_{i,c}, \\
\bigwedge_{K< a\leq b} (U_i^{-1}[dX_i]U_{i-1})_{ab}&\simeq \bigwedge_{K< a\leq b} [dT^{22}_i]_{ab}\bigwedge_{h=1}^L dZ_{i,h}, \\[.5em]
\bigwedge_{a\leq K< b} (U_i^{-1}[dX_i]U_{i-1})_{ab}&\simeq \bigwedge_{a\leq K<b}[dT^{12}_i]_{ab}.
\end{align}
Here `$\simeq$' means equal up to terms that do not contribute to the Jacobian.

The contribution to the Jacobian from the (block) lower-triangular part can be obtained following essentially the same steps as in the proof of proposition~\ref{thm:decomp:gen-schur-H}. This gives
\begin{align}
\bigwedge_{i=1}^n\bigwedge_{b<a\leq K} (U_i^{-1}[dX_i]U_{i-1})_{ab}
&= \prod_{1\leq b<a\leq K}(\lambda_a-\lambda_b)\bigwedge_{i=1}^n\bigwedge_{b< a\leq K}[U_i^{-1}dU_i]_{ab}, \\
\bigwedge_{i=1}^n\bigwedge_{K<b<a} (U_i^{-1}[dX_i]U_{i-1})_{ab}
&= \prod_{1\leq b<a\leq L}\det[(Z)_{aa}^\transpose\otimes\one-\one\otimes (Z)_{bb}]\bigwedge_{i=1}^n\bigwedge_{K<b<a}[U_i^{-1}dU_i]_{ab}, \\
\bigwedge_{i=1}^n\bigwedge_{b\leq K< a} (U_i^{-1}[dX_i]U_{i-1})_{ab}
&= \prod_{a=1}^L\prod_{b=1}^K\det[(Z)_{aa}-\one\otimes\lambda_a]\bigwedge_{i=1}^n\bigwedge_{b\leq K< a}[U_i^{-1}dU_i]_{ab}, 
\end{align}
where we have used the notation $\lambda_a=\lambda_{1,a}\cdots\lambda_{n,a}$ and $(Z)_{aa}=Z_{1,a}\cdots Z_{n,a}$ introduced earlier.
Combining the contributions from the lower- and upper-triangular part yields the Jacobian~\eqref{decomp:decomp:gen-schur-R-jacob} and completes the proof. 
\end{proof}

\begin{remark}\label{remark:decomp:vandermonde-real}
We emphasise that even though the Jacobian~\eqref{decomp:decomp:gen-schur-R-jacob} might look unfamiliar, it is nothing but a Vandermonde determinant. Let us assume that the product matrix has exactly $K$ real eigenvalues such that the eigenvalues of each $(Z)_{ii}$ is a complex conjugate pair denoted by $z_i$ ($z_i^*$), then we have
\begin{align}
\prod_{1\leq i<j\leq K}&\abs{\lambda_j-\lambda_i}
\prod_{1\leq i<j\leq L}\abs{\det[(Z)_{jj}^\transpose\otimes\one-\one\otimes (Z)_{ii}]}
\prod_{\substack{i=1,\ldots,K\\j=1,\ldots,L}}\abs{\det[(Z)_{jj}-\one\otimes \lambda_i]} 
\nn \\
&=
\prod_{1\leq i<j\leq K}\abs{\lambda_j-\lambda_i}
\prod_{1\leq i<j\leq L}\abs{z_j-z_i}^2\abs{z_j-z_i^*}^2
\prod_{\substack{i=1,\ldots,K\\j=1,\ldots,L}}\abs{z_j-\lambda_i}^2.
\label{decomp:decomp:vandermonde-real-2} 
\end{align}
The equality 
is easily seen to hold, if we use the identities
\begin{align}
(B\otimes A)(X\otimes\one-\one\otimes Y)(B^{-1}\otimes A^{-1})&=(BXB^{-1}\otimes\one-\one\otimes AY A^{-1}), \\
(B\otimes A)^{-1}&=B^{-1}\otimes A^{-1},
\end{align}
where $A$ and $B$ are non-singular matrices. Due to these identities, the two matrices $(Z)_{ii}$ and $(Z)_{jj}$ in the second determinant on the first line in~\eqref{decomp:decomp:vandermonde-real-2} may be diagonalised independently and the evaluation of the determinant becomes a trivial matter.
\end{remark}

%% file: special.tex
\chapter{Higher transcendental functions}
\label{app:special}

For easy reference, we will in this appendix collect some definitions and identities for the higher transcendental functions which are used frequently throughout this thesis; namely the gamma, the digamma, the hypergeometric and the Meijer $G$-function. The Meijer $G$-function is probably the less well-known of these functions but it is nonetheless of the utmost importance for the topic of this thesis. For general reference the reader is referred to~\cite{NIST:2010,Bateman:1955function,GR:1987,Luke:1975,WW:1996}.

\section{Gamma and digamma functions}

This section summarises some formulae for the gamma and digamma function which are used repeatedly in this thesis. All the formulae stated in this section can be found in~\cite{NIST:2010}.

The well-known gamma function is a meromorphic function with no zeros and with simple poles at the non-positive integers. It is (typically) defined as
\begin{equation}
\Gamma[z]=\int_0^\infty dt\, e^{-t}t^{z-1}
\label{special:gamma:gamma-def}
\end{equation}
for $\Re z>0$ and by analytic continuation for $\Re z\leq0$. Similarly, the digamma function (also known as the Euler psi function)
\begin{equation}
\psi(z):=\frac{\Gamma'[z]}{\Gamma[z]}.
\label{special:gamma:digamma-def}
\end{equation}
is a meromorphic function with no zeros and with simple poles of residue $-1$ at the non-positive integers. We need the following simple relations,
\begin{align}
\Gamma[z+1]&=z\Gamma[z], &
\psi(z+1)&=\psi(z)+\frac1z, \label{special:gamma:recurrence}\\
\Gamma[z]\Gamma[1-z]&=\frac\pi{\sin\pi z}, &
\psi[z]\psi[1-z]&=-\frac{\pi}{\tan\pi z}, \label{spacial:gamma:reflection}
\end{align}
which are known as the recurrence and reflection formulae, respectively. Additionally, we will use the Gauss multiplication formula
\begin{equation}
\Gamma[nz]=(2\pi)^{(1-n)/2}n^{nz-1/2}\prod_{k=1}^n\Gamma\big[z+\tfrac{k-1}{n}\big],
\label{special:gamma:multiplication}
\end{equation}
where $n$ is a positive integer.

We also need some asymptotic formulae for the gamma function. Stirling's formula states that
\begin{equation}
\Gamma[z]=(2\pi)^{1/2}z^{z-1/2}e^{-z}\Big(1+\frac1{12z}+O(z^{-2})\Big),
\label{special:gamma:stirling}
\end{equation}
for $z\to\infty$ in the sector $\abs{\arg z}<\pi$. We have an equivalent formula for ratios of gamma functions,
\begin{equation}
\frac{\Gamma[z+a]}{\Gamma[z+b]}=z^{a-b}(1+O(z^{-1}))
\label{special:gamma:ratio-asymp}
\end{equation}
for $z\to\infty$ in the sector $\abs{\arg z}<\pi$ with $a$ and $b$ denoting complex constants.

The analogue of Stirling's formula for the digamma function is
\begin{equation}
\psi(z)=\log z-\frac1{2z}+O(z^{-2})
\label{special:gamma:digamma-asymp}
\end{equation}
for $z\to\infty$ in the sector $\abs{\arg z}<\pi$.

\section{Hypergeometric and Meijer \textit{G}-Functions}
\label{sec:special:meijer}

Let us first consider the (generalised) hypergeometric function, which
is defined as a power series
\begin{equation}
\hypergeometric{p}{q}{a_1,\ldots,a_p}{b_1,\ldots,b_q}{z}:=\sum_{k=0}^\infty \frac{\prod_{i=1}^p(a_i)_k}{\prod_{j=1}^q(b_j)_k}\frac{z^k}{k!}
\label{special:meijer:hypergeometric}
\end{equation}
in the region of convergence. Here $p$ and $q$ are non-negative integers, $a_i$ ($i=1,\ldots,p$) and $b_j$ ($j=1,\ldots,q$) are real or complex parameters, and the Pochhammer symbol is defined as $(a)_0=1$ and $(a)_k=a(a+1)\cdots(a+k-1)$ for $k\geq1$ (the $b_j$'s cannot be non-positive integers unless the corresponding singularities cancel due to zeros in the numerator). In the special case when one of the $a_i$'s is a non-positive integer, the sum terminates and the hypergeometric function becomes a polynomial,
\begin{equation}
\hypergeometric{p+1}{q}{-n,a_1,\ldots,a_p}{b_1,\ldots,b_q}{z}:=\sum_{k=0}^n \frac{(-1)^kn!}{(n-k)!}\frac{\prod_{i=1}^p(a_i)_k}{\prod_{j=1}^q(b_j)_k}\frac{z^k}{k!},\quad
n=0,1,2,\ldots
\end{equation}
For our purpose, a particularly important example of a hypergeometric polynomial is the (associated) Laguerre polynomial which may be written as
\begin{equation}
L^\nu_n(x)=\frac{(\nu+1)_n}{n!}\hypergeometric{1}{1}{-n}{\nu+1}{x}.
\end{equation}
The Laguerre polynomials appear naturally in the study of the Wishart ensemble (see e.g.~\cite{Forrester:2010}). In section~\ref{sec:singular:correlations}, where we consider the Wishart product ensemble, we will face a generalised version of these polynomials given as a more complicated hypergeometric function or, equivalently, a Meijer $G$-function.

\begin{definition}\label{def:special:meijer}
The Meijer $G$-function is defined through a contour integral in the complex plane~\cite{Bateman:1955function},
\begin{equation}
\MeijerG{m}{n}{p}{q}{a_1,\ldots,a_p}{b_1,\ldots,b_q}{z}:=
\frac{1}{2\pi i}\int_L du\,z^u\frac{\prod_{i=1}^m\Gamma[b_i-u]\prod_{i=1}^n\Gamma[1-a_i+u]}{\prod_{i=n+1}^p\Gamma[a_i-u]\prod_{i=m+1}^q\Gamma[1-b_i+u]}.
\label{special:meijer:meijer-def}
\end{equation}
Here an empty product is interpreted as unity, $0\leq m\leq q$ and $0\leq n\leq p$ are integers, while $a_i$ ($i=1,\ldots,p$) and $b_j$ ($j=1,\ldots,q$) are real or complex parameters with the constraint that no pole of $\prod_{i=1}^m\Gamma[b_i-u]$ coincides with any pole of $\prod_{i=1}^n\Gamma[1-a_i+u]$. There are three possible choices for the contour $L$:
\begin{enumerate}[(i)]
 \item The contour runs from $-i\infty$ to $+i\infty$ and the path is chosen such that all poles $\prod_{i=1}^m\Gamma[b_i-u]$ are to the right of $L$, while all poles of $\prod_{i=1}^n\Gamma[1-a_i+u]$ are to the left of $L$. The integral converges if $m+n>\frac12(p+q)$ and $\abs{\arg z}<\pi(m+n-\frac12(p+q))$.
 \item The contour forms a loop starting and ending at $+\infty$, such that it encircles all poles $\prod_{i=1}^m\Gamma[b_i-u]$ once in the negative direction, but none of the poles stemming from $\prod_{i=1}^n\Gamma[1-a_i+u]$. In this case, the integral converges for all $z(\neq0)$ if $q>p$  and for $0<\abs z<1$ if $q=p\geq1$. 
 \item The contour forms a loop starting and ending at $-\infty$, such that it encircles all poles $\prod_{i=1}^n\Gamma[1-a_i+u]$ once in the positive direction, but none of the poles stemming from $\prod_{i=1}^m\Gamma[b_i-u]$. In this case, the integral converges for all $z(\neq0)$ if $q<p$  and for $1<\abs z$ if $q=p\geq1$. 
\end{enumerate}
The Meijer $G$-function is defined when the integral converges.
\end{definition}

\begin{remark}
The criteria of convergence stated in the definition may be established using Stirling's approximation for the gamma functions in the integrand. Note that if more than one of the contours (i-iii) result in a convergent integral then they will also yield the same result. 
\end{remark}

\begin{remark}
In order to avoid too much repetition, we will assume that the above stated conditions of convergence are fulfilled for all formulae introduced below. We additionally note that, as a consequence of its definition, the Meijer $G$-function is an analytic function for all $z$ (where it is defined) with the possible exceptions of the origin, $z=0$, the unit circle, $\abs z=1$, and infinity. 
\end{remark}

The Meijer $G$-function generalises the hypergeometric function in the sense that any hypergeometric function can be written as Meijer $G$-function. Due to~\eqref{special:meijer:meijer-def}, this also gives an explicit contour integral representation of the hypergeometric functions, which will be useful for analytical computations. In the non-polynomial case, we have~\cite{Bateman:1955function}
\begin{equation}
\hypergeometric{p}{q}{a_1,\ldots,a_p}{b_1,\ldots,b_q}{z}=\frac{\prod_{j=1}^q\Gamma[b_j]}{\prod_{i=1}^p\Gamma[a_i]}\MeijerG{1}{p}{p}{q+1}{1-a_1,\ldots,1-a_p}{0,1-b_1,\ldots,1-b_q}{-z}
\label{special:meijer:hyper-meijer}
\end{equation}
with parameters as above. In the polynomial case, we will need the formula
\begin{equation}
\hypergeometric{p+1}{q}{-n,a_1,\ldots,a_p}{b_1,\ldots,b_q}{z}=\frac{\prod_{j=1}^q\Gamma[b_j]}{\prod_{i=1}^p\Gamma[a_i]}\MeijerG{1}{p}{p+1}{q+1}{1-a_1,\ldots,1-a_p,n+1}{0,1-b_1,\ldots,1-b_q}{z},
\label{special:meijer:poly-meijer}
\end{equation}
which is valid for $p<q$. Both~\eqref{special:meijer:hyper-meijer} and~\eqref{special:meijer:poly-meijer} are obtained by writing the Meijer $G$-function as a sum over the residues.

In this thesis, we will several times be faced with Meijer $G$-functions with $m=p=q$ and $n=0$. In this case the contour changes discontinuously when $z$ crosses the unit circle, which typically gives rise to non-analytic behaviour. Let us mention two simple but important examples on the positive half-line:
\begin{equation}
\MeijerG{1}{0}{1}{1}{\,1}{\,0}{x}=\begin{cases} 1 & \text{for } 0\leq x\leq 1 \\ 0 & \text{for } 1< x \end{cases},\quad
\MeijerG{1}{0}{1}{1}{\,1}{\,-\frac12}{x}=\begin{cases} \sqrt{\frac{4(1-x)}{\pi x}} & \text{for } 0\leq x\leq 1 \\ 0 & \text{for } 1<x \end{cases}.
\end{equation}
The former is the Heaviside step function which is discontinuous at unity, while the latter is the Mar\v cenko--Pastur density (on the unit interval) which is continuous but not differentiable at unity. Both cases are immediate consequences of the residue theorem.

In addition to the cases mentioned above, the following relations between the Meijer $G$-function and other special functions will be used in this thesis~\cite{Bateman:1955function},
\begin{equation}
\begin{aligned}
\MeijerG{1}{0}{0}{1}{-}{\nu}{x}&=x^\nu e^{-x}, &
\MeijerG{1}{0}{0}{2}{-}{\nu,0}{x}&=x^{\nu/2}J_\nu(2\sqrt x), \\
\MeijerG{1}{1}{1}{2}{-n-\nu}{0,-\nu}{x}&=n!e^{-x}L_n^\nu(x), &
\MeijerG{2}{0}{0}{2}{-}{\nu,0}{x}&=2x^{\nu/2}K_\nu(2\sqrt x).
\end{aligned}
\label{special:meijer:meijer-other}
\end{equation}
Here $L_n^\nu(x)$ denote the Laguerre polynomial, while $J_\nu(x)$ and $K_\nu(x)$ are Bessel functions.

Furthermore, we need the following simple identities for the Meijer $G$-function:
by cancellations of gamma functions in the integrand in the definition~\eqref{special:meijer:meijer-def}, we have
\begin{align}
\MeijerG{m}{n}{p}{q}{a_1,\ldots,a_p}{b_1,\ldots,b_{q-1},a_1}{z}&=\MeijerG{m}{n-1}{p-1}{q-1}{a_2,\ldots,a_p}{b_1,\ldots,b_{q-1}}{z},\quad n,p,q\geq1, \label{special:meijer:reduce-1} \\
\MeijerG{m}{n}{p}{q}{a_1,\ldots,a_{p-1},b_1}{b_1,\ldots,b_q}{z}&=\MeijerG{m-1}{n}{p-1}{q-1}{a_1,\ldots,a_{p-1}}{b_2,\ldots,b_{q}}{z},\quad m,p,q\geq1; \label{special:meijer:reduce-2} 
\end{align}
by shifting the contour, we have
\begin{equation}
x^c\MeijerG{m}{n}{p}{q}{a_1,\ldots,a_p}{b_1,\ldots,b_q}{x}=\MeijerG{m}{n}{p}{q}{a_1+c,\ldots,a_p+c}{b_1+c,\ldots,b_q+c}{x}
\label{special:meijer:meijer-shift}
\end{equation}
where $c$ is a constant; and finally by changing the direction of the contour
\begin{equation}
\MeijerG{m}{n}{p}{q}{a_1,\ldots,a_p}{b_1,\ldots,b_q}{\frac1z}=\MeijerG{n}{m}{q}{p}{1-a_1,\ldots,1-a_p}{1-b_1,\ldots,1-b_q}{z}
\label{special:meijer:meijer-inverse}
\end{equation}
which maps zero to infinity and vice versa.

As discussed in section~\ref{sec:prologue:scalar} the Meijer $G$-function plays an important r\^ole when considering products of independent random variables due to its close relation with the Mellin transform. Let us consider the Mellin transform of a Meijer $G$-function~\cite{Bateman:1955function},
\begin{equation}
\int_0^\infty dx\,x^{s-1}\,\MeijerG{m}{n}{p}{q}{a_1,\ldots,a_p}{b_1,\ldots,b_q}{xy}=
\frac{1}{y^s}\,\frac{\prod_{i=1}^m\Gamma(b_i+s)\prod_{j=1}^n\Gamma(1-a_j-s)}{\prod_{k=m+1}^q\Gamma(1-b_k-s)\prod_{\ell=n+1}^p\Gamma(a_\ell+s)}.
\label{special:meijer:meijer-moment}
\end{equation}
This means that we may think of the Meijer $G$-function as the inverse Mellin transform with respect to $s$ of the right hand side in~\eqref{special:meijer:meijer-moment} whenever this transform exists and is unique. 

The relation between the Meijer $G$-function and products of independent random variables may be used to establish the following $n$-fold integral representations~\cite{Springer:1979},
\begin{align}
\MeijerG{n}{0}{0}{n}{-}{\nu_1,\ldots,\nu_n}{x}&=\bigg[\prod_{i=1}^n\int_0^\infty dx_i\, x_i^{\nu_i}e^{-x_i}\bigg]\delta(x_n\cdots x_1-x)  
\label{special:meijer:gamma} \\
\MeijerG{n}{0}{n}{n}{\mu_1,\ldots,\mu_n}{\nu_1,\ldots,\nu_n}{x}&=\bigg[\prod_{i=1}^n\int_0^1 dx_i\, \frac{x_i^{\nu_i}(1-x_i)^{\mu_i-\nu_i-1}}{\Gamma[\mu_i-\nu_i]}\bigg]\delta(x_n\cdots x_1-x), \label{special:meijer:beta}
\end{align}
where $x\in(0,\infty)$ and $-\mu_i>\nu_i>-1$ ($i=1,\ldots,n$). After multiplication with a factor $\prod_{i=1}^n\Gamma[\nu_i+1]^{-1}$, the Meijer $G$-function~\eqref{special:meijer:gamma} becomes the density of a product of $n$ independent gamma-distributed random scalars, while~\eqref{special:meijer:beta} becomes the density of a product of $n$ independent beta-distributed random scalars after multiplication by $\prod_{i=1}^n\Gamma[\nu_i+\mu_i+2]/\Gamma[\nu_i+1]$. 

With a change of variables like in~\eqref{prologue:scalar:y-product}, the representations~\eqref{special:meijer:gamma} and~\eqref{special:meijer:beta} may be written as ($n-1$)-fold Mellin convolutions, where each convolution transforms a lower order Meijer $G$-function into a higher order Meijer $G$-function. These results are part of a more general integration formula which says the Mellin convolution of two Meijer $G$-functions gives another Meijer $G$-function~\cite{PBM:1998},
\begin{multline}
\MeijerG{m+\mu}{n+\nu}{p+\sigma}{q+\tau}{a_1,\ldots,a_n,c_1,\ldots,c_\sigma,\,a_{n+1},\ldots,a_p}{b_1,\ldots,b_m,d_1,\ldots,d_\tau,b_{m+1},\ldots,b_q}{\eta\omega} \\
\begin{aligned}
&=\int_0^\infty \frac{dx}x \MeijerG{m}{n}{p}{q}{a_1,\ldots,a_p}{b_1,\ldots,b_q}{\eta x}\MeijerG{\mu}{\nu}{\sigma}{\tau}{c_1,\ldots,c_\sigma}{d_1,\ldots,d_\tau}{\frac\omega x} \\
&=\int_0^\infty \frac{dx}\omega \MeijerG{m}{n}{p}{q}{a_1,\ldots,a_p}{b_1,\ldots,b_q}{\eta x}\MeijerG{\mu}{\nu}{\sigma}{\tau}{-c_1,\ldots,-c_\sigma}{-d_1,\ldots,-d_\tau}{\frac x\omega};
\end{aligned}
\label{special:meijer:int-meijer-meijer}
\end{multline}
this integration formula is valid whenever each of the individual Meijer $G$-functions is well-defined. For the purpose of this thesis two cases are of particular importance:
\begin{align}
\MeijerG{m+1}{n}{p}{q+1}{a_1,\ldots,a_p}{\nu,b_1,\ldots,b_q}{y}&=\int_0^\infty \frac{dx}x\, e^{-x}x^{\nu}\MeijerG{m}{n}{p}{q}{a_1,\ldots,a_p}{b_1,\ldots,b_q}{\frac yx} \nn \\
&=\int_0^\infty \frac{dx}x\, e^{-1/x}x^{-\nu}\MeijerG{m}{n}{p}{q}{a_1,\ldots,a_p}{b_1,\ldots,b_q}{xy},
\end{align}
which is related to products of gamma distributed random variables~\eqref{special:meijer:gamma} and
\begin{align}
\MeijerG{m+1}{n}{p+1}{q+1}{a_1,\ldots,a_p,\mu}{\nu,b_1,\ldots,b_q}{y}&=\int_0^1 \frac{dx}x\, \frac{x^{\nu}(1-x)^{\mu-\nu-1}}{\Gamma[\mu-\nu]} \MeijerG{m}{n}{p}{q}{a_1,\ldots,a_p}{b_1,\ldots,b_q}{\frac yx} \nn \\
&=\int_1^\infty \frac{dx}x\, \frac{x^{1-\mu}(x-1)^{\mu-\nu-1}}{\Gamma[\mu-\nu]} \MeijerG{m}{n}{p}{q}{a_1,\ldots,a_p}{b_1,\ldots,b_q}{xy}
\label{special:meijer:beta-conv}
\end{align}
which is related to products of beta distributed random variables~\eqref{special:meijer:beta}.

Many integrals found in this thesis either are performed using the above given integration formula for Meijer $G$-functions. One of the benefits of using~\eqref{special:meijer:int-meijer-meijer} to perform (even quite simple) integrals is that the method is very systematic. This is also the reason why similar methods are often favoured in computer implementations; Wolfram's \textsc{Mathematica} writes in \textit{Notes on Internal Implementation}:
``Many other definite integrals are done using Marichev--Adamchik Mellin transform methods. The results are often initially expressed in terms of Meijer $G$-functions, which are converted into hypergeometric functions using Slater's Theorem and then simplified.''

The last property of the Meijer $G$-function that we will need is the asymptotic behaviour for large argument. We have~\cite{Fields:1972,Luke:1975}
\begin{equation}
\MeijerG{m}{0}{0}{m}{-}{b_1,\ldots,b_m}{z}\sim \frac1{\sqrt m}\Big(\frac{2\pi}z\Big)^{\frac{m-1}2}e^{-mz^{1/z}}z^{(b_1+\cdots+b_m)/m}(1+O(z^{-1/m}))
\label{special:meijer:meijer-asymp}
\end{equation}
as $z\to\infty$.

We end this appendix by mentioning the Fox $H$-function which generalises the Meijer $G$-function similarly to the way the Fox--Wright function generalises the hypergeometric function. Analogously to the Meijer $G$-function, the Fox $H$-function is defined as a contour integral in the complex plane,
\begin{equation}
\FoxH[\bigg]{m}{n}{p}{q}{(a_1,A_1),\ldots,(a_p,A_p)}{(b_1,B_1),\ldots,(b_q,B_q)}{z}=\!
\int_L \frac{du\,z^u}{2\pi i}\frac{\prod_{i=1}^m\Gamma[b_i-B_iu]\prod_{i=1}^n\Gamma[1-a_i+A_iu]}{\prod_{i=n+1}^p\Gamma[a_i-A_iu]\prod_{i=m+1}^q\Gamma[1-b_i+B_iu]},
\end{equation}
where $0\leq m\leq q$ and $0\leq n\leq p$ are integers, $a_i,A_i$ ($i=1,\ldots,p$) and $b_j,B_j$ ($j=1,\ldots,q$) are constants and the contour, $L$, separates the poles of the two products in the numerator; see~\cite{MSH:2009} for precise statements of constraints and criteria for convergence. In many aspects the Fox $H$-function is no harder to work with than the Meijer $G$-function, e.g. all the identities for the Meijer $G$-functions mentioned above have direct analogues in terms of the Fox $H$-function (see~\cite{MSH:2009}). A comparison between~\eqref{singular:asymp:density-fc} and \eqref{singular:asymp:fox} shows that the use of the Fox $H$-function sometimes simplifies formulae considerably. Moreover, the Fox $H$-function was given a prominent r\^ole in study of the algebra of random variables~\cite{Springer:1979}. The main drawback of the function is that it is not included in standard tables of special functions, such as~\cite{NIST:2010} and~\cite{GR:1987}, neither is it predefined in 
mathematical software such as \textsc{Mathematica} and \textsc{Maple}.

%% file: thesis.bbl
\providecommand{\bysame}{\leavevmode\hbox to3em{\hrulefill}\thinspace}
\providecommand{\MR}{\relax\ifhmode\unskip\space\fi MR }
\providecommand{\MRhref}[2]{%
  \href{http://www.ams.org/mathscinet-getitem?mr=#1}{#2}
}
\providecommand{\href}[2]{#2}
\begin{thebibliography}{100}

\bibitem{ARRS:2013}
K.~Adhikari, N.~K. Reddy, T.~R. Reddy, and K.~Saha, \emph{Determinantal point
  processes in the plane from products of random matrices}, arXiv:1308.6817 [to
  appear in Ann. Inst. Henri Poincar{\'e}] (2013).

\bibitem{Akemann:2001}
G.~Akemann, \emph{Microscopic correlations for non-{H}ermitian {D}irac
  operators in three-dimensional {QCD}}, Phys. Rev. D \textbf{64} (2001),
  114021.

\bibitem{Akemann:2002}
\bysame, \emph{Microscopic correlation functions for the {QCD} {D}irac operator
  with chemical potential}, Phys. Rev. Lett. \textbf{89} (2002), 072002.

\bibitem{Akemann:2005}
\bysame, \emph{The complex {L}aguerre symplectic ensemble of non-{H}ermitian
  matrices}, Nucl. Phys. B \textbf{730} (2005), 253.

\bibitem{Akemann:2011}
\bysame, \emph{Non-{H}ermitian extensions of {W}ishart random matrix
  ensembles}, Acta Phys. Pol. B \textbf{42} (2011).

\bibitem{ABF:2011}
G.~Akemann, J.~Baik, and Ph. Di~Francesco (eds.), \emph{The {O}xford handbook
  of random matrix theory}, Oxford University Press, 2011.

\bibitem{AB:2012}
G.~Akemann and Z.~Burda, \emph{Universal microscopic correlation functions for
  products of independent {G}inibre matrices}, J. Phys. A \textbf{45} (2012),
  465201.

\bibitem{ABK:2014}
G.~Akemann, Z.~Burda, and M.~Kieburg, \emph{Universal distribution of
  {L}yapunov exponents for products of {G}inibre matrices}, J. Phys. A
  \textbf{47} (2014), 395202.

\bibitem{ABKN:2014}
G.~Akemann, Z.~Burda, M.~Kieburg, and T.~Nagao, \emph{Universal microscopic
  correlation functions for products of truncated unitary matrices}, J. Phys. A
  \textbf{47} (2014), 255202.

\bibitem{ADMN:1997}
G.~Akemann, P.~H. Damgaard, U.~Magnea, and S.~Nishigaki, \emph{Universality of
  random matrices in the microscopic limit and the {D}irac operator spectrum},
  Nucl. Phys. B \textbf{487} (1997), 721.

\bibitem{AI:2015}
G.~Akemann and J.~R. Ipsen, \emph{Recent exact and asymptotic results for
  products of independent random matrices}, Acta Phys. Pol. B \textbf{46}
  (2015), 1747.

\bibitem{AIK:2013}
G.~Akemann, J.~R. Ipsen, and M.~Kieburg, \emph{Products of rectangular random
  matrices: singular values and progressive scattering}, Phys. Rev. E
  \textbf{88} (2013), 052118.

\bibitem{AIS:2014}
G.~Akemann, J.~R. Ipsen, and E.~Strahov, \emph{Permanental processes from
  products of complex and quaternionic induced {G}inibre ensembles}, Random
  Matrices: Theory Appl. \textbf{03} (2014), 1450014.

\bibitem{AKW:2013}
G.~Akemann, M.~Kieburg, and L.~Wei, \emph{Singular value correlation functions
  for products of {W}ishart random matrices}, J. Phys. A \textbf{46} (2013),
  275205.

\bibitem{AOSV:2005}
G.~Akemann, J.~C. Osborn, K.~Splittorff, and J.~J.~M. Verbaarschot,
  \emph{Unquenched {QCD} {D}irac operator spectra at nonzero baryon chemical
  potential}, Nucl. Phys. B \textbf{712} (2005), 287.

\bibitem{APS:2009gap}
G.~Akemann, M.~J. Phillips, and L.~Shifrin, \emph{Gap probabilities in
  non-{H}ermitian random matrix theory}, J. Math. Phys. \textbf{50} (2009),
  no.~6, 063504.

\bibitem{APS:2009}
G.~Akemann, M.~J. Phillips, and H.-J. Sommers, \emph{Characteristic polynomials
  in real {G}inibre ensembles}, J. Phys. A \textbf{42} (2009), 012001.

\bibitem{APS:2010}
\bysame, \emph{The chiral {G}aussian two-matrix ensemble of real asymmetric
  matrices}, J. Phys. A \textbf{43} (2010), 085211.

\bibitem{AP:2004}
G.~Akemann and A.~Pottier, \emph{Ratios of characteristic polynomials in
  complex matrix models}, J. Phys. A \textbf{37} (2004).

\bibitem{AS:2013}
G.~Akemann and E.~Strahov, \emph{Hole probabilities and overcrowding estimates
  for products of complex {G}aussian matrices}, J. Stat. Phys. \textbf{151}
  (2013).

\bibitem{AV:2003}
G.~Akemann and G.~Vernizzi, \emph{Characteristic polynomials of complex random
  matrix models}, Nucl. Phys. B \textbf{660} (2003), 532.

\bibitem{AZ:1997}
A.~Altland and M.~R. Zirnbauer, \emph{Nonstandard symmetry classes in
  mesoscopic normal-superconducting hybrid structures}, Phys. Rev. B
  \textbf{55} (1997), 1142.

\bibitem{AJM:1990}
J.~Ambj{\o}rn, J.~Jurkiewicz, and Yu.~M. Makeenko, \emph{Multiloop correlators
  for two-dimensional quantum gravity}, Phys. Lett. B \textbf{251} (1990), 517.

\bibitem{AHM:2011}
Y.~Ameur, H.~Hedenmalm, and N.~Makarov, \emph{Fluctuations of eigenvalues of
  random normal matrices}, Duke Math. J \textbf{159} (2011), 31.

\bibitem{AGZ:2010}
G.~W. Anderson, A.~Guionnet, and O.~Zeitouni, \emph{An introduction to random
  matrices}, Cambridge University Press, 2010.

\bibitem{Andreief:1883}
C.~Andr{\'e}ief, \emph{Note sur une relation les int{\'e}grales d{\'e}finies
  des produits des fonctions}, M{\'e}m. de la Soc. Sci. Bordeaux \textbf{2}
  (1883), 1.

\bibitem{Bai:2007}
Z.-Q. Bai, \emph{On the cycle expansion for the {L}yapunov exponent of a
  product of random matrices}, J. Phys. A \textbf{40} (2007), 8315.

\bibitem{BDS:2003}
J.~Baik, P.~Deift, and E.~Strahov, \emph{Products and ratios of characteristic
  polynomials of random {H}ermitian matrices}, J. Math. Phys. \textbf{44}
  (2003), 3657.

\bibitem{BBTCC:2011}
T.~Banica, Serban~T. Belinschi, M.~Capitaine, and B.~Collins, \emph{Free
  {B}essel laws}, Canad. J. Math \textbf{63} (2011), 3.

\bibitem{BT:1993}
E.~L. Basor and C.~A. Tracy, \emph{Variance calculations and the {B}essel
  kernel}, J. Stat. Phys. \textbf{73} (1993), 415.

\bibitem{Beenakker:1993}
C.~W.~J. Beenakker, \emph{Universality in the random-matrix theory of quantum
  transport}, Phys. Rev. Lett. \textbf{70} (1993), 1155.

\bibitem{Beenakker:1994}
\bysame, \emph{Universality of {B}r{\'e}zin and {Z}ee's spectral correlator},
  Nucl. Phys. B \textbf{422} (1994), no.~3, 515--520.

\bibitem{Beenakker:1997}
\bysame, \emph{Random-matrix theory of quantum transport}, Rev. Mod. Phys.
  \textbf{69} (1997), 731.

\bibitem{Bellman:1954}
R.~Bellman, \emph{Limit theorems for non-commutative operations {I}.}, Duke
  Math. J. \textbf{21} (1954), 491.

\bibitem{BP:2005}
G.~Ben~Arous and S.~P{\'e}ch{\'e}, \emph{Universality of local eigenvalue
  statistics for some sample covariance matrices}, Comm. Pure Appl. Math.
  \textbf{58} (2005), no.~10, 1316.

\bibitem{BG:2010}
F.~Benaych-Georges, \emph{On a surprising relation between the
  {M}archenko--{P}astur law, rectangular and square free convolutions}, Ann.
  Inst. Henri Poincar{\'e} Probab. Stat \textbf{46} (2010), 644.

\bibitem{Benettin:1984}
G.~Benettin, \emph{Power-law behavior of {L}yapunov exponents in some
  conservative dynamical systems}, Physica D \textbf{13} (1984), 211.

\bibitem{Berman:2008}
R.~J. Berman, \emph{Determinantal point processes and fermions on complex
  manifolds: bulk universality}, arXiv:0811.3341 (2008).

\bibitem{BLC:2002}
D.~Bernard and A.~LeClair, \emph{A classification of non-{H}ermitian random
  matrices}, Statistical Field Theories (A.~Cappelli and G.~Mussardo, eds.),
  NATO Science Series, vol.~73, Springer Netherlands, 2002, p.~207.

\bibitem{BB:2015}
M.~Bertola and T.~Bothner, \emph{Universality conjecture and results for a
  model of several coupled positive-definite matrices}, Comm. Math. Phys.
  \textbf{337} (2015), 1077.

\bibitem{BGS:2009}
M.~Bertola, M.~Gekhtman, and J.~Szmigielski, \emph{The {C}auchy two-matrix
  model}, Comm. Math. Phys. \textbf{287} (2009), 983.

\bibitem{BGS:2014}
\bysame, \emph{Cauchy--{L}aguerre two-matrix model and the {M}eijer-{G} random
  point field}, Comm. Math. Phys. \textbf{326} (2014), 111.

\bibitem{BN:2008}
J.-P. Blaizot and M.~A. Nowak, \emph{Large-{$N_c$} confinement and turbulence},
  Phys. Rev. Lett. \textbf{101} (2008), 102001.

\bibitem{BC:2012}
C.~Bordenave and D.~Chafa{\"\i}, \emph{Around the circular law}, Probab. Surv.
  \textbf{9} (2012).

\bibitem{Borodin:1998}
A.~Borodin, \emph{Biorthogonal ensembles}, Nucl. Phys. B \textbf{536} (1998),
  704.

\bibitem{BS:2009}
A.~Borodin and C.~D. Sinclair, \emph{The {G}inibre ensemble of real random
  matrices and its scaling limits}, Comm. Math. Phys. \textbf{291} (2009), 177.

\bibitem{BH:2000}
E.~Br{\'e}zin and S.~Hikami, \emph{Characteristic polynomials of random
  matrices}, Comm. Math. Phys. \textbf{214} (2000), 111.

\bibitem{BIPZ:1978}
E.~Br{\'e}zin, C.~Itzykson, G.~Parisi, and J.-B. Zuber, \emph{Planar diagrams},
  Comm. Math. Phys. \textbf{59} (1978), 35.

\bibitem{BZ:1993}
E.~Br{\'e}zin and A.~Zee, \emph{Universality of the correlations between
  eigenvalues of large random matrices}, Nucl. Phys. B \textbf{402} (1993),
  no.~3, 613.

\bibitem{Burda:2013}
Z.~Burda, \emph{Free products of large random matrices--a short review of
  recent developments}, J. Phys: Conference Series, vol. 473, IOP Publishing,
  2013, p.~012002.

\bibitem{BJW:2010}
Z.~Burda, R.~A. Janik, and B.~Waclaw, \emph{Spectrum of the product of
  independent random {G}aussian matrices}, Phys. Rev. E \textbf{81} (2010),
  041132.

\bibitem{BJLNS:2010}
Z.~Burda, A.~Jarosz, G.~Livan, M.~A. Nowak, and A.~Swiech, \emph{Eigenvalues
  and singular values of products of rectangular {G}aussian random matrices},
  Phys. Rev. E \textbf{82} (2010), 061114.

\bibitem{CKN:1986}
J.~E. Cohen, H.~Kesten, and C.~M. Newman (eds.), \emph{Random matrices and
  their applications}, Amer. Math. Soc., 1986.

\bibitem{CN:1984}
J.~E. Cohen and C.~M. Newman, \emph{The stability of large random matrices and
  their products}, Ann. Probab. (1984), 283.

\bibitem{CPV:1993}
A.~Crisanti, G.~Paladin, and A.~Vulpiani, \emph{Products of random matrices},
  Springer, 1993.

\bibitem{deBruijn:1955}
N.~G. de~Bruijn, \emph{On some multiple integrals involving determinants}, J.
  Indian Math. Soc \textbf{19} (1955), 133.

\bibitem{DKMVZ:1999}
P.~Deift, T.~Kriecherbauer, K.~T.-R. McLaughlin, S.~Venakides, and X.~Zhou,
  \emph{Uniform asymptotics for polynomials orthogonal with respect to varying
  exponential weights and applications to universality questions in random
  matrix theory}, Comm. Pure Appl. Math. \textbf{52} (1999), 1335.

\bibitem{DF:2006}
P.~Desrosiers and P.~J. Forrester, \emph{Hermite and {L}aguerre
  $\beta$-ensembles: Asymptotic corrections to the eigenvalue density}, Nucl.
  Phys. B \textbf{743} (2006), no.~3, 307.

\bibitem{Dyson:1962}
F.~J. Dyson, \emph{The threefold way. algebraic structure of symmetry groups
  and ensembles in quantum mechanics}, J. Math. Phys. \textbf{3} (1962), 1199.

\bibitem{DZ:1996}
J.~D’Anna and A.~Zee, \emph{Correlations between eigenvalues of large random
  matrices with independent entries}, Phys Rev E \textbf{53} (1996), 1399.

\bibitem{Edelman:1997}
A.~Edelman, \emph{The probability that a random real {G}aussian matrix has k
  real eigenvalues, related distributions, and the circular law}, J.
  Multivariate Anal. \textbf{60} (1997), 203.

\bibitem{Edelman:2005}
\bysame, \emph{Finite random matrix theory}, MIT handouts (2005).

\bibitem{EKS:1994}
A.~Edelman, E.~Kostlan, and M.~Shub, \emph{How many eigenvalues of a random
  matrix are real?}, J. Amer. Math. Soc. \textbf{7} (1994), 247.

\bibitem{ER:2005}
A.~Edelman and N.~R. Rao, \emph{Random matrix theory}, Acta Numer. \textbf{14}
  (2005), 233.

\bibitem{ESW:2014}
A.~Edelman, B.~D. Sutton, and Y.~Wang, \emph{Random matrix theory, numerical
  computation and applications}, Proc. Sympos. Appl. Math. (2014), 53.

\bibitem{Epstein:1948}
B.~Epstein, \emph{Some applications of the {M}ellin transform in statistics},
  Ann. Math. Stat. \textbf{19} (1948), 370.

\bibitem{Bateman:1955function}
A.~Erd{\'e}lyi (ed.), \emph{Higher transcendental functions [volumes {I}, {II}
  \& {III}]}, McGraw-Hill, 1953-1955.

\bibitem{Erdos:2011}
L.~Erd{\"o}s, \emph{Universality of {W}igner random matrices: a survey of
  recent results}, Russian Math. Surveys \textbf{66} (2011), 507.

\bibitem{EK:1995}
B.~Eynard and C.~Kristjansen, \emph{Exact solution of the {O}(n) model on a
  random lattice}, Nucl. Phys. B \textbf{455} (1995), 577.

\bibitem{EM:1998}
B.~Eynard and M.~L. Mehta, \emph{Matrices coupled in a chain {I}. {E}igenvalue
  correlations}, J. Phys. A \textbf{31} (1998), 4449.

\bibitem{EZ:1992}
B.~Eynard and J.~Zinn-Justin, \emph{The {O}(n) model on a random surface:
  critical points and large-order behaviour}, Nucl. Phys. B \textbf{386}
  (1992), no.~3, 558.

\bibitem{FZ:1997}
J.~Feinberg and A.~Zee, \emph{Non-{G}aussian non-{H}ermitian random matrix
  theory: phase transition and addition formalism}, Nucl. Phys. B \textbf{501}
  (1997), 643.

\bibitem{Feller:1968}
W.~Feller, \emph{An introduction to probability theory and its applications}, 3
  ed., John Wiley and Sons, 1968.

\bibitem{Fields:1972}
J.~L. Fields, \emph{The asymptotic expansion of the {M}eijer {G}-function},
  Math. Comp. (1972), 757.

\bibitem{Fischmann:2013}
J.~Fischmann, \emph{Eigenvalue distributions on a single ring}, Ph.D. thesis,
  Queen Mary University of London, 2013.

\bibitem{FBKSZ:2012}
J.~Fischmann, W.~Bruzda, B.~A. Khoruzhenko, H.-J. Sommers, and
  K.~{\.Z}yczkowski, \emph{Induced {G}inibre ensemble of random matrices and
  quantum operations}, J. Phys. A \textbf{45} (2012), 075203.

\bibitem{Forrester:1993}
P.~J. Forrester, \emph{The spectrum edge of random matrix ensembles}, Nucl.
  Phys. B \textbf{402} (1993), 709.

\bibitem{Forrester:2010}
\bysame, \emph{Log-gases and random matrices}, Princeton University Press,
  2010.

\bibitem{Forrester:2013}
\bysame, \emph{Lyapunov exponents for products of complex {G}aussian random
  matrices}, J. Stat. Phys. \textbf{151} (2013), 796.

\bibitem{Forrester:2014b}
\bysame, \emph{Eigenvalue statistics for product complex {W}ishart matrices},
  J. Phys. A \textbf{47} (2014), 345202.

\bibitem{Forrester:2014a}
\bysame, \emph{Probability of all eigenvalues real for products of standard
  {G}aussian matrices}, J. Phys. A \textbf{47} (2014), 065202.

\bibitem{Forrester:2015}
\bysame, \emph{Asymptotics of finite system {L}yapunov exponents for some
  random matrix ensembles}, arXiv:1501.05702 (2015).

\bibitem{FK:2014}
P.~J. Forrester and M.~Kieburg, \emph{Relating the {B}ures measure to the
  {C}auchy two-matrix model}, arXiv:1410.6883 (2014).

\bibitem{FL:2014}
P.~J. Forrester and D.-Z. Liu, \emph{Raney distributions and random matrix
  theory}, J. Stat. Phys. (2014), 1.

\bibitem{FM:2012}
P.~J. Forrester and A.~Mays, \emph{Pfaffian point process for the {G}aussian
  real generalised eigenvalue problem}, Probab. Theory Related Fields
  \textbf{154} (2012), 1.

\bibitem{FN:2007}
P.~J. Forrester and T.~Nagao, \emph{Eigenvalue statistics of the real {G}inibre
  ensemble}, Phys. Rev. Lett. \textbf{99} (2007), 050603.

\bibitem{FW:2015}
P.~J. Forrester and D.~Wang, \emph{Muttalib--{B}orodin ensembles in random
  matrix theory---realisations and correlation functions}, arXiv:1502.07147
  (2015).

\bibitem{FK:1960}
H.~Furstenberg and H.~Kesten, \emph{Products of random matrices}, Ann. Math.
  Stat. \textbf{31} (1960), 457.

\bibitem{FS:2003}
Y.~V. Fyodorov and E.~Strahov, \emph{An exact formula for general spectral
  correlation function of random {H}ermitian matrices}, J. Phys. A \textbf{36}
  (2003), 3203.

\bibitem{GA:1970}
M.~R. Gardner and W.~R. Ashby, \emph{Connectance of large dynamic (cybernetic)
  systems: critical values for stability}, Nature \textbf{228} (1970), 784.

\bibitem{Ginibre:1965}
J.~Ginibre, \emph{Statistical ensembles of complex, quaternion, and real
  matrices}, J. Math. Phys. \textbf{6} (1965), 440.

\bibitem{Glasser:1994}
M.~L. Glasser, \emph{The quadratic formula made hard: A less radical approach
  to solving equations}, arXiv:math/9411224 (1994).

\bibitem{GSO:1987}
I.~Goldhirsch, P.-I. Sulem, and S.~A. Orszag, \emph{Stability and {L}yapunov
  stability of dynamical systems: {A} differential approach and a numerical
  method}, Physica D \textbf{27} (1987), 311.

\bibitem{GL:1996}
G.~H. Golub and Ch.~F. van Loan, \emph{Matrix computations}, 3 ed., JHU Press,
  2012.

\bibitem{GW:2009}
R.~Goodman and N.~R. Wallach, \emph{Symmetry, representations, and invariants},
  vol.~66, Springer, 2009.

\bibitem{GKT:2014}
F.~G{\"o}tze, H.~K{\"o}sters, and A.~Tikhomirov, \emph{Asymptotic spectra of
  matrix-valued functions of independent random matrices and free probability},
  Random Matrices: Theory Appl. \textbf{4} (2015), 1550005.

\bibitem{GNT:2014}
F.~G{\"o}tze, A.~Naumov, and A.~Tikhomirov, \emph{Distribution of linear
  statistics of singular values of the product of random matrices},
  arXiv:1412.3314 (2014).

\bibitem{GNT:2014elliptic}
\bysame, \emph{On one generalization of the elliptic law for random matrices},
  arXiv:1404.7013 (2014).

\bibitem{GT:2010}
F.~G{\"o}tze and A.~Tikhomirov, \emph{On the asymptotic spectrum of products of
  independent random matrices}, arXiv:1012.2710 (2010).

\bibitem{GR:1987}
K.~I. Gross and D.~St.~P. Richards, \emph{Special functions of matrix argument
  {I}. {A}lgebraic induction, zonal polynomials, and hypergeometric functions},
  Trans. Amer. Math. Soc. \textbf{301} (1987), 781.

\bibitem{HM:2013}
U.~Haagerup and S.~M{\"o}ller, \emph{The law of large numbers for the free
  multiplicative convolution}, Operator algebra and dynamics, Springer, 2013,
  p.~157.

\bibitem{HJL:2015}
S.~Hameed, K.~Jain, and A.~Lakshminarayan, \emph{Real eigenvalues of
  non-{G}aussian random matrices and their products}, J. Phys. A \textbf{48}
  (2015), no.~38, 385204.

\bibitem{HC:1957}
Harish-Chandra, \emph{Differential operators on a semisimple {L}ie algebra},
  Am. J. Math. (1957), 87.

\bibitem{HJ:2012}
R.~A. Horn and Ch.~R. Johnson, \emph{Matrix analysis}, 2nd ed., Cambridge
  University Press, 2012.

\bibitem{HKPV:2009}
J.~B. Hough, M.~Krishnapur, Y.~Peres, and B.~Vir{\'a}g, \emph{Zeros of
  {G}aussian analytic functions and determinantal point processes}, vol.~51,
  Amer. Math. Soc., 2009.

\bibitem{IWZ:1990}
S.~Iida, H.~A. Weidenm{\"u}ller, and J.~A. Zuk, \emph{Statistical scattering
  theory, the supersymmetry method and universal conductance fluctuations},
  Ann. Phys. \textbf{200} (1990), 219.

\bibitem{Ipsen:2013}
J.~R. Ipsen, \emph{Products of independent quaternion {G}inibre matrices and
  their correlation functions}, J. Phys. A \textbf{46} (2013), 265201.

\bibitem{Ipsen:2015}
\bysame, \emph{Lyapunov exponents for products of rectangular real, complex and
  quaternionic {G}inibre matrices}, J. Phys. A \textbf{48} (2015), 155204.

\bibitem{IK:2014}
J.~R. Ipsen and M.~Kieburg, \emph{Weak commutation relations and eigenvalue
  statistics for products of rectangular random matrices}, Phys. Rev. E
  \textbf{89} (2014), 032106.

\bibitem{IS:2012}
J.~R. Ipsen and K.~Splittorff, \emph{Baryon number {D}irac spectrum in {QCD}},
  Phys. Rev. D \textbf{86} (2012), 014508.

\bibitem{Ishitani:1977}
H.~Ishitani, \emph{A central limit theorem for the subadditive process and its
  application to products of random matrices}, Publ. Res. Inst. Math. Sci.
  \textbf{12} (1977), 565.

\bibitem{Ismail:2005}
M.~Ismail, \emph{Classical and quantum orthogonal polynomials in one variable},
  Cambridge University Press, 2005.

\bibitem{IN:1992}
M.~Isopi and C.~M. Newman, \emph{The triangle law for {L}yapunov exponents of
  large random matrices}, Comm. Math. Phys. \textbf{143} (1992), 591--598
  (English).

\bibitem{Itoi:1997}
C.~Itoi, \emph{Universal wide correlators in non-{G}aussian orthogonal, unitary
  and symplectic random matrix ensembles}, Nucl. Phys. B \textbf{493} (1997),
  651.

\bibitem{IZ:1980}
C.~Itzykson and J.-B. Zuber, \emph{The planar approximation. {II}}, J. Math.
  Phys. \textbf{21} (1980), 411.

\bibitem{JW:2004}
R.~A. Janik and W.~Wieczorek, \emph{Multiplying unitary random
  matrices---universality and spectral properties}, J. Phys. A \textbf{37}
  (2004), 6521.

\bibitem{JQ:2014}
T.~Jiang and Y.~Qi, \emph{Spectral radii of large non-{H}ermitian random
  matrices}, arXiv:1411.1833 (2014).

\bibitem{Kallenberg:1997}
O.~Kallenberg, \emph{Foundations of modern probability}, Springer Science {\&}
  Business Media, 1997.

\bibitem{Kanzieper:2002}
E.~Kanzieper, \emph{Eigenvalue correlations in non-{H}ermitean symplectic
  random matrices}, J. Phys. A \textbf{35} (2002), 6631.

\bibitem{AK:2005}
E.~Kanzieper and G.~Akemann, \emph{Statistics of real eigenvalues in
  {G}inibre’s ensemble of random real matrices}, Phys. Rev. Lett. \textbf{95}
  (2005), 230201.

\bibitem{KF:1998}
E.~Kanzieper and V.~Freilikher, \emph{Random-matrix models with the
  logarithmic-singular level confinement: method of fictitious fermions},
  Philos. Mag. B \textbf{77} (1998), 1161.

\bibitem{KS:2010}
E.~Kanzieper and N.~Singh, \emph{Non-{H}ermitean {W}ishart random matrices
  ({I})}, J. Math. Phys. \textbf{51} (2010), 103510.

\bibitem{Kargin:2008}
V.~Kargin, \emph{Lyapunov exponents of free operators}, J. Funct. Anal.
  \textbf{255} (2008), 1874.

\bibitem{Kargin:2014}
\bysame, \emph{On the largest {L}yapunov exponent for products of {G}aussian
  matrices}, J. Stat. Phys. \textbf{157} (2014), 70.

\bibitem{KKP:1995}
A.~Khorunzhy, B.~Khoruzhenko, and L.~Pastur, \emph{On the 1/{N} corrections to
  the {G}reen functions of random matrices with independent entries}, J. Phys.
  A \textbf{28} (1995), L31.

\bibitem{Kieburg:2015}
M.~Kieburg, \emph{Supersymmetry for products of random matrices}, Acta Pys.
  Pol. B \textbf{46} (2015), 1709.

\bibitem{KKS:2015}
M.~Kieburg, A.~B.~J. Kuijlaars, and D.~Stivigny, \emph{Singular value
  statistics of matrix products with truncated unitary matrices},
  arXiv:1501.03910 (2015).

\bibitem{Kostlan:1992}
E.~Kostlan, \emph{On the spectra of {G}aussian matrices}, Lin. Alg. Appl.
  \textbf{162} (1992), 385.

\bibitem{Krishnapur:2009}
M.~Krishnapur, \emph{From random matrices to random analytic functions}, Ann.
  Probab. (2009), 314.

\bibitem{Kuijlaars:2010}
A.~B.~J. Kuijlaars, \emph{Multiple orthogonal polynomial ensembles}, Contemp.
  Math. \textbf{507} (2010), 155.

\bibitem{Kuijlaars:2011}
\bysame, \emph{Universality, {C}hapter~6 in "{T}he {O}xford handbook of random
  matrix theory"}, Oxford University Press, 2011.

\bibitem{Kuijlaars:2015}
\bysame, \emph{Transformations of polynomial ensembles}, arXiv:1501.05506
  (2015).

\bibitem{KS:2014}
A.~B.~J. Kuijlaars and D.~Stivigny, \emph{Singular values of products of random
  matrices and polynomial ensembles}, Random Matrices \textbf{3} (2014).

\bibitem{KV:2002}
A.~B.~J. Kuijlaars and M.~Vanlessen, \emph{Universality for eigenvalue
  correlations from the modified {J}acobi unitary ensemble}, Int. Math. Res.
  Notices \textbf{2002} (2002), 1575.

\bibitem{KV:2003}
\bysame, \emph{Universality for eigenvalue correlations at the origin of the
  spectrum}, Comm. Math. Phys. \textbf{243} (2003), 163.

\bibitem{KZ:2014}
A.~B.~J. Kuijlaars and L.~Zhang, \emph{Singular values of products of {G}inibre
  random matrices, multiple orthogonal polynomials and hard edge scaling
  limits}, Comm. Math. Phys. \textbf{332} (2014), 759.

\bibitem{Kumar:2015}
S.~Kumar, \emph{Exact evaluations of some {M}eijer {$G$}-functions and
  probability of all eigenvalues real for product of two {G}aussian matrices},
  arXiv:1507.05571 (2015).

\bibitem{Lakshminarayan:2013}
A.~Lakshminarayan, \emph{On the number of real eigenvalues of products of
  random matrices and an application to quantum entanglement}, J. Phys. A
  \textbf{46} (2013), 152003.

\bibitem{lePage:1982}
{\'E}.~Le~Page, \emph{Th{\'e}oremes limites pour les produits de matrices
  al{\'e}atoires}, Probability measures on groups, Springer, 1982, p.~258.

\bibitem{LS:1991}
N.~Lehmann and H.-J. Sommers, \emph{Eigenvalue statistics of random real
  matrices}, Phys. Rev. Lett. \textbf{67} (1991), 941.

\bibitem{LWZ:2014}
D.-Z. {Liu}, D.~{Wang}, and L.~{Zhang}, \emph{{Bulk and soft-edge universality
  for singular values of products of {G}inibre random matrices}},
  arXiv:1412.6777 (2014).

\bibitem{LW:2014}
D.-Z. Liu and Y.~Wang, \emph{Universality for products of random matrices {I}:
  {G}inibre and truncated unitary cases}, arXiv:1411.2787 (2014).

\bibitem{Lubinsky:2008}
D.~S. Lubinsky, \emph{Universality limits at the hard edge of the spectrum for
  measures with compact support}, Int. Math. Res. Notices \textbf{2008} (2008),
  rnn099.

\bibitem{LSZ:2006}
T.~L{\"u}ck, H.-J. Sommers, and M.~R. Zirnbauer, \emph{Energy correlations for
  a random matrix model of disordered bosons}, J. Math. Phys. \textbf{47}
  (2006), 103304.

\bibitem{Luke:1975}
Y.~L. Luke, \emph{Mathematical functions and their approximations}, Academic
  Press, 2014.

\bibitem{Magnea:2008}
U.~Magnea, \emph{Random matrices beyond the {C}artan classification}, J. Phys.
  A \textbf{41} (2008), 045203.

\bibitem{MS:2014}
S.~N. Majumdar and G.~Schehr, \emph{Top eigenvalue of a random matrix: large
  deviations and third order phase transition}, J. Stat. Mech. \textbf{2014}
  (2014), P01012.

\bibitem{MP:1967}
V.~A. Mar{\v{c}}enko and L.~A. Pastur, \emph{Distribution of eigenvalues for
  some sets of random matrices}, Sbornik: Mathematics \textbf{1} (1967), 457.

\bibitem{Mathai:1997}
A.~M. Mathai, \emph{Jacobians of matrix transformations and functions of matrix
  argument}, World Scientific, 1997.

\bibitem{MSH:2009}
A.~M. Mathai, R.~K. Saxena, and H.~J. Haubold, \emph{The {H}-function: theory
  and applications}, Springer Science \& Business Media, 2009.

\bibitem{May:1972}
R.~M. May, \emph{Will a large complex system be stable?}, Nature \textbf{238}
  (1972), 413.

\bibitem{Mays:2013}
A.~Mays, \emph{A real quaternion spherical ensemble of random matrices}, J.
  Stat. Phys. \textbf{153} (2013), 48.

\bibitem{Mehta:2004}
M.~L. Mehta, \emph{Random matrices}, Elsevier, 2004.

\bibitem{MS:1966}
M.~L. Mehta and P.~K. Srivastava, \emph{Correlation functions for eigenvalues
  of real quaternian matrices}, J. Math. Phys. \textbf{7} (1966), 341.

\bibitem{MPK:1988}
P.~A. Mello, P.~Pereyra, and N.~Kumar, \emph{Macroscopic approach to
  multichannel disordered conductors}, Ann. Phys. \textbf{181} (1988), 290.

\bibitem{Muirhead:2009}
R.~J. Muirhead, \emph{Aspects of multivariate statistical theory}, John Wiley
  \& Sons, 2009.

\bibitem{Muller:2002}
R.~R. M{\"u}ller, \emph{On the asymptotic eigenvalue distribution of
  concatenated vector-valued fading channels}, IEEE T. Inform. Theory
  \textbf{48} (2002), 2086.

\bibitem{Muttalib:1995}
K.~A. Muttalib, \emph{Random matrix models with additional interactions}, J.
  Phys. A \textbf{28} (1995), L159.

\bibitem{NR:2007}
R.~Narayanan and H.~Neuberger, \emph{Universality of large {N} phase
  transitions in {W}ilson loop operators in two and three dimensions}, JHEP
  \textbf{2007} (2007), 066.

\bibitem{Neuschel:2014}
T.~Neuschel, \emph{Plancherel--{R}otach formulae for average characteristic
  polynomials of products of {G}inibre random matrices and the
  {F}uss--{C}atalan distribution}, Random Matrices: Theor. Appl. \textbf{3}
  (2014), 1450003.

\bibitem{Newman:1986}
C.~M. Newman, \emph{The distribution of {L}yapunov exponents: {E}xact results
  for random matrices}, Comm. Math. Phys. \textbf{103} (1986), 121.

\bibitem{NS:2006}
A.~Nica and R.~Speicher, \emph{Lectures on the combinatorics of free
  probability}, Cambridge University Press, 2006.

\bibitem{Olkin:2002}
Ingram Olkin, \emph{The 70th anniversary of the distribution of random
  matrices: a survey}, Linear Algebra Appl. \textbf{354} (2002), 231.

\bibitem{NIST:2010}
F.~W.~J. Olver, D.~W. Lozier, R.~F. Boisvert, and C.~W. Clark (eds.),
  \emph{{{NIST} Handbook of Mathematical Functions}}, Cambridge University
  Press, 2010.

\bibitem{Osborn:2004}
J.~C. Osborn, \emph{{Universal results from an alternate random matrix model
  for {QCD} with a baryon chemical potential}}, Phys. Rev. Lett. \textbf{93}
  (2004), 222001.

\bibitem{OSV:2005}
J.~C. Osborn, K.~Splittorff, and J.~J.~M. Verbaarschot, \emph{Chiral symmetry
  breaking and the {D}irac spectrum at nonzero chemical potential}, Phys. Rev.
  Lett. \textbf{94} (2005), 202001.

\bibitem{Oseledec:1968}
V.~I. Oseledec, \emph{A multiplicative ergodic theorem. {L}yapunov
  characteristic numbers for dynamical systems}, Trans. Moscow Math. Soc.
  \textbf{19} (1968), 197--231.

\bibitem{ORSV:2014}
S.~O’Rourke, D.~Renfrew, A.~Soshnikov, and V.~Vu, \emph{Products of
  independent elliptic random matrices}, J. Stat. Phys. (2014), 1.

\bibitem{OS:2011}
S.~O’Rourke and A.~Soshnikov, \emph{Products of independent non-{H}ermitian
  random matrices}, Electron. J. Probab \textbf{16} (2011), 2219.

\bibitem{PV:1986}
G.~Paladin and A.~Vulpiani, \emph{Scaling law and asymptotic distribution of
  {L}yapunov exponents in conservative dynamical systems with many degrees of
  freedom}, J. Phys. A \textbf{19} (1986), 1881.

\bibitem{PZ:2011}
K.~A. Penson and K.~{\.Z}yczkowski, \emph{Product of {G}inibre matrices:
  {F}uss--{C}atalan and {R}aney distributions}, Phys. Rev. E \textbf{83}
  (2011), 061118.

\bibitem{Pollicott:2010}
M.~Pollicott, \emph{Maximal {L}yapunov exponents for random matrix products},
  Invent. Math. \textbf{181} (2010), 209.

\bibitem{PBM:1998}
A.~P. Prudnikov, I.~U.~A. Brychkov, and O.~I. Marichev, \emph{Integrals and
  series: special functions}, vol. 1-5, CRC Press, 1998.

\bibitem{Raghunathan:1979}
M.~S. Raghunathan, \emph{A proof of {O}seledec's multiplicative ergodic
  theorem}, Israel J. Math. \textbf{32} (1979), 356.

\bibitem{Rider:2004}
B.~Rider, \emph{Deviations from the circular law}, Probab. Theory Rel.
  \textbf{130} (2004).

\bibitem{Schmidt:2006}
Christian Schmidt, \emph{{Lattice {QCD} at finite density}}, PoS
  \textbf{LAT2006} (2006), 021.

\bibitem{Simon:1998}
B.~Simon, \emph{The classical moment problem as a self-adjoint finite
  difference operator}, Adv. Math. \textbf{137} (1998), 82.

\bibitem{Sommers:2007}
H.-J. Sommers, \emph{Symplectic structure of the real ginibre ensemble},
  arXiv:0706.1671 (2007).

\bibitem{SCSS:1988}
H.-J. Sommers, A.~Crisanti, H.~Sompolinsky, and Y.~Stein, \emph{Spectrum of
  large random asymmetric matrices}, Phys. Rev. Lett. \textbf{60} (1988), 1895.

\bibitem{SW:2008}
H.-J. Sommers and W.~Wieczorek, \emph{General eigenvalue correlations for the
  real ginibre ensemble}, J. Phys. A \textbf{41} (2008), 405003.

\bibitem{Splittorff:2006}
K.~Splittorff, \emph{{The Sign problem in the {$\epsilon$}-regime of {QCD}}},
  PoS \textbf{LAT2006} (2006), 023.

\bibitem{Springer:1979}
M.~D. Springer, \emph{The algebra of random variables}, Wiley, 1979.

\bibitem{Strahov:2013}
E.~Strahov, Unpublished notes (2013).

\bibitem{Szego:1939}
G.~Szeg{\"o}, \emph{Orthogonal polynomials}, Amer. Math. Soc., 1939.

\bibitem{Tao:2012}
T.~Tao, \emph{Topics in random matrix theory}, Amer. Math. Soc., 2012.

\bibitem{TV:2010}
T.~Tao and V.~Vu, \emph{Random matrices: {T}he distribution of the smallest
  singular values}, Geom. Func. Anal. \textbf{20} (2010), 260.

\bibitem{TV:2012}
\bysame, \emph{Random matrices: {U}niversality of local spectral statistics of
  non-{H}ermitian matrices}, Ann. Probab. \textbf{43} (2015), 782.

\bibitem{Telatar:1999}
I.~E. Telatar, \emph{Capacity of multi-antenna {G}aussian channels}, Euro.
  Trans. Tel. \textbf{10} (1999), 585.

\bibitem{TW:1994}
C.~A. Tracy and H.~Widom, \emph{Level-spacing distributions and the {A}iry
  kernel}, Comm. Math. Phys. \textbf{159} (1994), 151.

\bibitem{Tucci:2010}
G.~H. Tucci, \emph{Limits laws for geometric means of free random variables},
  Indiana Univ. Math. J. \textbf{59} (2010), 1.

\bibitem{TV:2004}
A.~M. Tulino and S.~Verd{\'u}, \emph{Random matrix theory and wireless
  communications}, vol.~1, Now Publishers Inc, 2004.

\bibitem{vanLoan:2000}
Ch.~F. van Loan, \emph{The ubiquitous {K}ronecker product}, J. Comput. Appl.
  Math. \textbf{123} (2000), 85.

\bibitem{Vanneste:2010}
J.~Vanneste, \emph{Estimating generalized {L}yapunov exponents for products of
  random matrices}, Phys. Rev. E \textbf{81} (2010), 036701.

\bibitem{Verbaarschot:1994}
J.~J.~M. Verbaarschot, \emph{Spectrum of the {QCD} {D}irac operator and chiral
  random matrix theory}, Phys. Rev. Lett. \textbf{72} (1994), 2531.

\bibitem{VW:2000}
J.~J.~M. Verbaarschot and T.~Wettig, \emph{{Random matrix theory and chiral
  symmetry in {QCD}}}, Ann. Rev. Nucl. Part. Sci. \textbf{50} (2000), 343.

\bibitem{WW:1996}
E.~T. Whittaker and G.~N. Watson, \emph{A course of modern analysis}, Cambridge
  university press, 1996.

\bibitem{Wishart:1928}
J.~Wishart, \emph{The generalised product moment distribution in samples from a
  normal multivariate population}, Biometrika \textbf{20A} (1928), 32.

\bibitem{WT:2001}
T.~G. Wright and L.~N. Trefethen, \emph{Computing {L}yapunov constants for
  random recurrences with smooth coefficients}, J. Comp. Appl. Math.
  \textbf{132} (2001), 331.

\bibitem{Yeomans:1992}
J.~M. Yeomans, \emph{Statistical mechanics of phase transitions}, Oxford
  University Press, 1992.

\bibitem{Zabrodin:2006}
A.~Zabrodin, \emph{Matrix models and growth processes: from viscous flows to
  the quantum {H}all effect}, Applications of random matrices in physics,
  Springer, 2006, p.~261.

\bibitem{Zhang:1997}
F.~Zhang, \emph{Quaternions and matrices of quaternions}, Linear Algebra Appl.
  \textbf{251} (1997), 21.

\bibitem{Zhang:2015}
L.~Zhang, \emph{Local universality in biorthogonal {L}aguerre ensembles},
  arXiv:1502.03160 (2015).

\end{thebibliography}
